\documentclass[11pt,vi]{mitthesis}
\usepackage[hidelinks, bookmarks=true]{hyperref}
\usepackage{pdfpages}
\usepackage{lgrind}
\usepackage[english]{babel}
\usepackage[fleqn]{amsmath}
\usepackage{array}
\usepackage{graphicx}
\usepackage{bbm}
\usepackage[numbered,autolinebreaks,useliterate]{mcode}
\usepackage{bookmark}
\usepackage{color}
\usepackage{wrapfig}
\usepackage{caption}
\usepackage{subcaption}
\usepackage{amssymb}
\usepackage{xcolor,colortbl}
\usepackage{amsthm}
\usepackage{tabularx}
\usepackage{booktabs}
\usepackage{multirow}
\usepackage{afterpage}
\usepackage{tocloft}
\usepackage{mdframed}
\usepackage{hhline}
\usepackage{upgreek}
\usepackage{pifont}% http://ctan.org/pkg/pifont
\usepackage{enumitem}
\usepackage{mathtools}
\usepackage{setspace}
\newcommand{\finishchapquote}[3]{\begin{quotation} \textit{#1} \end{quotation} \begin{flushright} - #2. \textit{#3}\end{flushright} }

\newtheorem{lemma}{Lemma}
\newtheorem{theorem}{Theorem}
\newtheorem{corollary}{Corollary}

\newtheorem{definition}{Definition}

\newcommand{\Rho}{\mathrm{P}}
\newcommand{\Beta}{\mathrm{B}}
\newcommand{\Alpha}{\mathrm{A}}
\newcommand{\Nu}{\mathrm{V}}
\newcommand{\W}{\mathrm{W}}

\newcommand {\param} {\mathrm}

\newcommand {\Rmax} {\param{R}_{\param{max}}}
\renewcommand {\r} {\param{r}}
\newcommand {\n}{\param{n}}
\newcommand {\N}{\param{N}}
\renewcommand {\W}{\param{W}}
\newcommand {\R}{\param{R}}
\renewcommand {\a}{\param{a}}
\newcommand {\s}{\param{s}}
\renewcommand {\b}{\param{b}}
\newcommand {\m}{\param{m}}
\newcommand {\xbar}{\bar{\param{x}}}
\newcommand {\ybar}{\bar{\param{y}}}
\newcommand {\zbar}{\bar{\param{z}}}
\newcommand {\ubar}{\bar{\param{u}}}
\newcommand {\wbar}{\bar{\param{w}}}
\newcommand {\w}{\param{\upomega}}
\newcommand {\A}{\param{A}}
\newcommand {\B}{\param{B}}
\newcommand {\D}{\param{D}}
\renewcommand {\d}{\param{d}}

\renewcommand {\P}{\param{P}}
\newcommand {\Q}{\param{Q}}
\newcommand {\C}{\param{C}}
\newcommand {\zo}{\zbar}
\newcommand {\uo}{\ubar}
\newcommand {\wo}{\wbar}
\newcommand {\p}{\param{p}}
\newcommand {\del}{\updelta}

\newcommand{\RN}[1]{%
  \textup{\uppercase\expandafter{\romannumeral#1}}%
}

\def\meas{{\rm meas}\,}

\makeatletter
\def\@fnsymbol#1{\ensuremath{\ifcase#1\or \,a\or \,b\or \,c\or \,d\or \,e\else\@ctrerr\fi}}
\makeatother

\DeclareMathOperator*{\argmax}{argmax}

\newlist{papers}{enumerate}{1}
\setlist[papers, 1]{label = {\bf Paper \arabic*}, align=left}

\begin{document}
\setcounter{page}{1}
\pagenumbering{roman}

%\makeatletter
%\@addtoreset{section}{part}
%\makeatother
%\newlength\mylen
%\renewcommand\thepart{\arabic{part}}
\renewcommand\cftpartpresnum{Part~}
%\settowidth\mylen{\bfseries\cftpartpresnum\cftpartaftersnum}
%\addtolength\cftpartnumwidth{\mylen}

\setcounter{tocdepth}{2}
\setcounter{secnumdepth}{2}

\newcommand {\titlet} {Towards a Cybernetic Foundation for Natural Resource Governance}
\newcommand {\authort} {Talha Manzoor}
\newcommand {\schoolt} {Syed Babar Ali School of Science and Engineering}
\newcommand {\departmentt}{Department of Electrical Engineering}
\newcommand {\degreet}{Doctor of Philosophy }
\newcommand {\degreefieldt}{Electrical Engineering}
\newcommand{\supervisornamet}{Abubakr Muhammad }
\newcommand {\supervisoraffiliationt}{LUMS}
\newcommand{\cosupervisornamet}{Elena Rovenskaya }
\newcommand {\cosupervisoraffiliationt}{IIASA}
\newcommand {\degreemontht}{December}
\newcommand {\degreeyeart}{2017}

\cleardoublepage
% Uncomment the next line if you do NOT want a page number on your
% abstract and acknowledgments pages.
 \pagestyle{empty}
%\setcounter{savepage}{\thepage}
%\makeatletter
\begin{titlepage}
\Large
{\def\baselinestretch{1.2}\Huge\bf \choosecase{\titlet} \par}
\vspace{80pt}
\Large{A Thesis}\\
\choosecase{Presented to}\\
\choosecase{the Academic Faculty}\\
\vspace{10pt}
\choosecase{ by}\\
\vspace{30pt}
{\huge \choosecase{\bf \authort}}
\par
\vspace{30pt}
\choosecase{In Partial Fullfilment}\\
\choosecase{of the Requirements for the Degree of }\\
\choosecase{\degreet in \\ \degreefieldt}\\
\vspace{20pt}
\choosecase{Supervisor: \supervisornamet (\supervisoraffiliationt)}\\
\choosecase{Co-supervisor: \cosupervisornamet (\cosupervisoraffiliationt)}\\
%\choosecase{Supervisor: Abubakr Muhammad}\\
%\choosecase{Co-supervisor: Elena Rovenskaya}\\
\vfill
\includegraphics[width=0.25\linewidth]{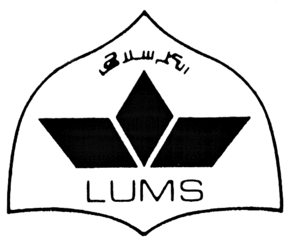}\par
\schoolt\\
\lums\\
\degreemontht\ \degreeyeart
\par
%\@copyrightnotice
\copyright\ \uppercase\expandafter{\degreeyeart} by \authort
\end{titlepage}
%\makeatother

\cleardoublepage

\includepdf[link,pages=-]{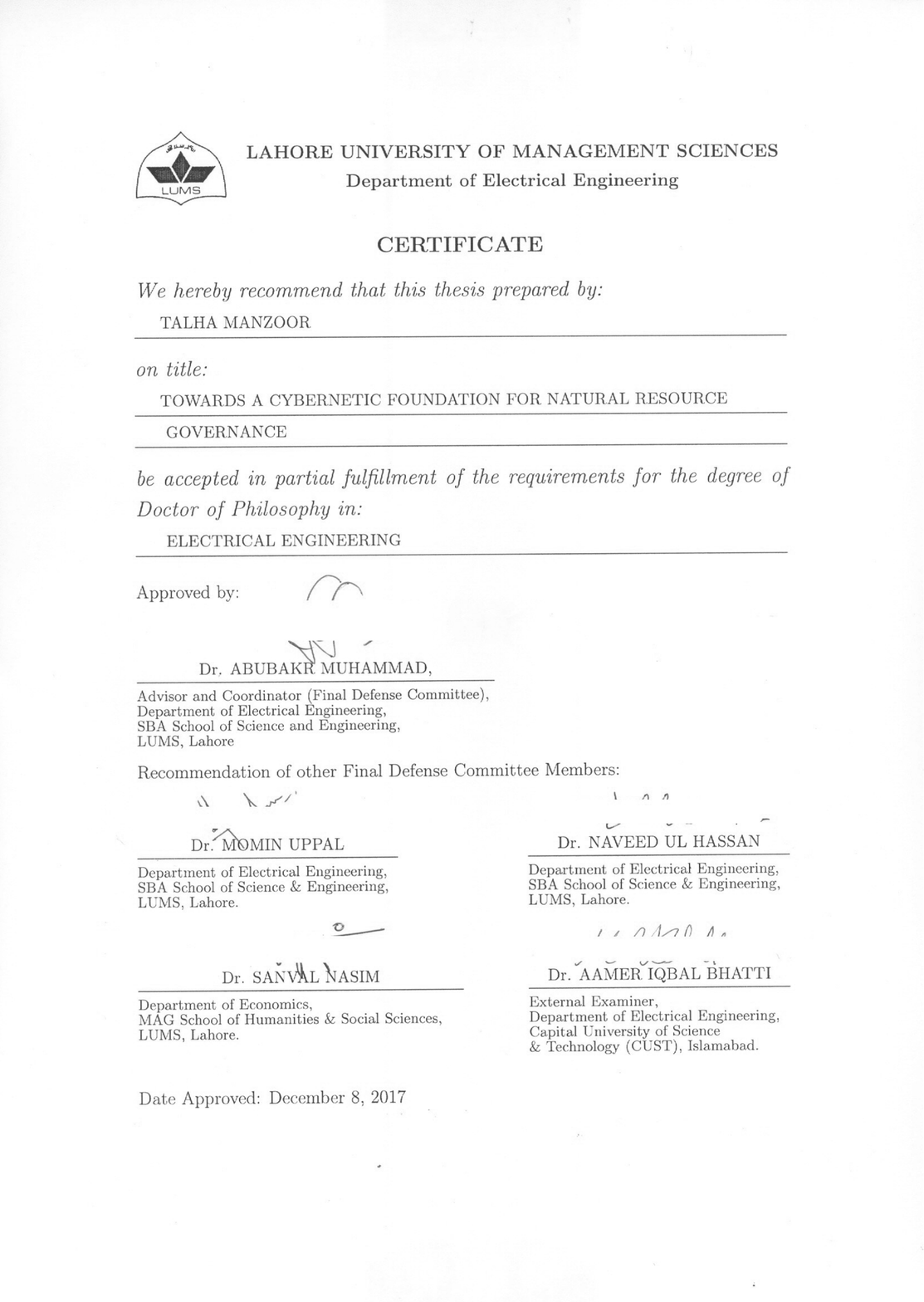}
%\addcontentsline{toc}{chapter}{Signature Page}

\cleardoublepage
\par

\clearpage
\begin{center}
    \thispagestyle{empty}
    \vspace*{.35 \textheight}
    	\emph{\Large To Marwa and her never-ending quest for adventure.}
    \vspace*{\fill}
\end{center}
\clearpage

%%%%%%%%%%%%%%%% ABSTRACT %%%%%%%%%%%%%%%%%%%%%%

\begin{abstract}
	% $Log: abstract.tex,v $
% Revision 1.1  93/05/14  14:56:25  starflt
% Initial revision
% 
% Revision 1.1  90/05/04  10:41:01  lwvanels
% Initial revision
% 
%
%% The text of your abstract and nothing else (other than comments) goes here.
%% It will be single-spaced and the rest of the text that is supposed to go on
%% the abstract page will be generated by the abstractpage environment.  This
%% file should be \input (not \include 'd) from cover.tex.

\addcontentsline{toc}{chapter}{Abstract}

This study explores the potential of the cybernetic method of inquiry for the problem of natural resource governance. The systems way of thinking has already enabled scientists to gain considerable headway in framing global environmental challenges. On the other hand, technical solutions to environmental problems have begun to show significant promise, driven by the advent of technology and its increased proliferation in coupled human and natural systems. Such settings lie on the interface of engineering, social and environmental sciences, and as such, require a common language in order for natural resources to be studied, managed and ultimately sustained. In this dissertation, we argue that the systems theoretic tradition of cybernetics may provide the necessary common ground for examining such systems.

After discussing the relevance of the cybernetic approach to natural resource governance, we present a mathematical model of resource consumption, grounded in social psychological research on consumer behavior. We also provide interpretations of the model at various levels of abstraction in the social network of the consuming population. We demonstrate the potential of the model by examining it in various theoretic frameworks which include dynamical systems, optimal control theory, game theory and the theory of learning in games. Each framework yields different policy guidelines to avoid Tragedy of the Commons like scenarios in the natural resource system. 

Mainly, we find that a high importance attached to social information (rather than ecological information) on part of the consumers helps overcome free-riding behavior and achieves affluence in both the resource stock and consumption levels. Moreover, we observe that discounting future utility beyond a specific threshold results in unsustainable consumption patterns according to a pre-defined notion of sustainability. We study the optimal control law for both sustainable and unsustainable cases, and give a rigorous criterion for sustainable growth. Later we examine the long-term effects of rational behavior on the part of the consumers and compare it with the collectively optimal outcome via a non-cooperative game and see how different societal attributes effect the ``tragicness" of the game. We show that the adoption of a basic fictitious play learning scheme by the consumers results in the equilibrium solution in a way that avoids free-riding behavior. All aspects of the analysis are conducted with one single question in mind -- what are the favorable conditions for sustainability?

\end{abstract}
\newpage

% Additional copy: start a new page, and reset the page number.  This way,
% the second copy of the abstract is not counted as separate pages.
% Uncomment the next 6 lines if you need two copies of the abstract
% page.
% \setcounter{page}{\thesavepage}
% \begin{abstractpage}
% \input{abstract}
% \end{abstractpage}

\cleardoublepage
%%%%%%%%%%%%%%%%%%%%%%%%%%%%%%%%%%%%%%%%%%%%%%%%%%%%%%%%%%%%%%%%%%%%%%
% -*-latex-*-

\pagestyle{plain}

%% This is an example first chapter.  You should put chapter/appendix that you
%% write into a separate file, and add a line \include{yourfilename} to
%% main.tex, where `yourfilename.tex' is the name of the chapter/appendix file.
%% You can process specific files by typing their names in at the 
%% \files=
%% prompt when you run the file main.tex through LaTeX.
\chapter*{Preface}
\addcontentsline{toc}{chapter}{Preface}

My first exposure to environmental problems began in the summer of 2013 when I participated in the Young Scientists Summer Program at the International Institute for Applied Systems Analysis (IIASA) in Laxenburg, Austria. I had just been admitted in the PhD program of Electrical Engineering at the Lahore University of Management Sciences (LUMS), which was set to begin in the fall of 2013. As a typical graduate student just beginning a PhD, I was still going back and forth over ideas for my doctorate and the work I started at IIASA four years ago resulted in this dissertation. 

The systems way of thinking made me realize how the proper level of abstraction can reveal similarities between problems originating from completely different disciplines. How scientific methods and tools developed by scientists of a particular field could be used by those working in completely different areas. It is exactly the reason how a mechatronics engineer with a strict background in robotics ended up writing a thesis in natural resource governance, something that many of my colleagues have found to be peculiar to say the least. To them, I say that when viewed through the lens of systems theory, robotics (and all engineering disciplines for that matter) and resource governance do not turn out to be as disparate as they may appear at first.

I must admit here that my initial comprehension of systems analysis was little more than just a buzzword that is always fascinating and politically correct. Writing this dissertation has given me an opportunity to understand better what systems thinking really is, and made me appreciate the role of a unified articulation for systems independent of their discipline of origin. Such a thing is necessary for the engineers as they have begun to encounter more and more human-centric applications of their technologies. The perspective with which they are accustomed to study their machines has deep philosophical roots which cover not only man-made but also naturally occurring systems and even living organisms. Such is the language of cybernetics, a distinct tradition of systems theory that has sadly disappeared from the forefront of the engineering community, hidden behind its descendants of control and communication. Even while we are introduced to the basic concepts of feedback and control, our vision is so narrow that we do not realize that the underlying philosophy of the field is so rich. For me this was mind boggling. In contrast however, to quote my advisor Abubakr Muhammad -- \emph{It doesn't surprise me one bit that a mechatronics engineer finally discovers his machinistic origins of thinking in a different topic}. 

In this dissertation I argue that the cybernetic way of thinking may provide the common language that the engineering community requires to study the intricate couplings between society, technology and the environment. Indeed, as I explain, small factions within the engineering community are already doing this without explicitly referring to the term cybernetics in its original form. It is my hope that this dissertation may serve as a spring board to further propagate this view and to guide future work in this area.

\section*{List of Contributions}

Below is a list of the salient contributions of this thesis. Where applicable, the contributions are linked in parenthesis to the publications list of the following section.

\begin{itemize}

\item A unique discourse on the relevance of cybernetics to natural resources.

\item A psychologically justified mathematical model of human behavior in resource consumption {\bf(Paper 2)}.

\item An aggregation mechanism to lump communities in the social network to ease tractability of the model {\bf (Paper 4)}.

\item A mathematical criterion of sustainability {\bf (Papers 3 \& 6)}.

\item The application of Pontryagin's Maximum Principle to an optimal control problem with unbounded control and non-concave Hamiltonian {\bf (Papers 3, 5 \& 6)}.

\item A notion of ``tragedy" in a non-cooperative game describing the Commons dilemma in our resource consumption model {\bf (Paper 2)}.

\item Some results illustrating the dependence of ``tragicness" on various societal parameters and guidelines on how to reduce it {\bf (Paper 2)}.

\item A demonstration of the relevance of the rational outcome for the commons game via application of the theory of learning in games {\bf (Paper 1)}. 

\item Various results describing the emergence of free-riding behavior and the role of different societal attributes in discouraging such behavior {\bf (Papers 1 \& 2)}.

\item A Lyapunov function based proof for the asymptotic global stability of the model, which may also help in determining stability for a class of related systems {\bf (Paper 1)}.

\end{itemize}

\section*{Publications}

Below is a list of publications that have resulted directly from the research described in this thesis. The chapters from the dissertation that include the material from each publication have also been given in the parenthesis.

\subsection*{Journal articles}

\begin{papers}
   \item (Chapters \ref{chap:open} \& \ref{chap:learning})\\
{\bf Talha Manzoor}, Elena Rovenskaya, Alexey Davydov, Abubakr Muhammad, ``Learning through Fictitious Play in a Game-theoretic Model of Natural Resource Consumption", {\it IEEE Control Systems Letters}, vol. 2, pp. 163-168, Jan 2018.
\end{papers}

\begin{papers}[resume]
   \item (Chapters \ref{chap:model}, \ref{chap:open} \& \ref{chap:game})\\
{\bf Talha Manzoor}, Elena Rovenskaya, Abubakr Muhammad, ``Game-theoretic insights into the role of environmentalism and social-ecological relevance: A cognitive model of resource consumption", {\it Ecological Modelling}, vol. 340, pp. 74-85, Nov 2016.
\end{papers}

\subsection*{Book chapters}

\begin{papers}[resume]
     \item (Chapters \ref{chap:optimal} \& \ref{chap:learning})\\ 
Sergey Aseev, {\bf Talha Manzoor}, ``Optimal Exploitation of Renewable Resources: Lessons in Sustainability from an Optimal Growth Model of Natural Resource Consumption", {\it Lecture Notes in Economics and Mathematical Systems}, Springer (in press). 
\end{papers}

\subsection*{Conference proceedings}

\begin{papers}[resume]
    \item (Chapter \ref{chap:lump})\\
{\bf Talha Manzoor}, Elena Rovenskaya, Abubakr Muhammad, ``Structural Effects and Aggregation in a Social-Network Model of Natural Resource Consumption", {\it 20th World Congress of the International Federation of Automatic Control (IFAC)}, Toulouse, France, Jul 2017.
    \item (Chapter \ref{chap:optimal})\\
{\bf Talha Manzoor}, Sergey Aseev, Elena Rovenskaya, Abubakr Muhammad,  ``Optimal Control for Sustainable Consumption of Natural Resources", {\it 19th World Congress of the International Federation of Automatic Control (IFAC)}, Cape Town, South Africa, Aug 2014. 
 \end{papers}

\subsection*{Working papers}

\begin{papers}[resume]
    \item (Chapters \ref{chap:optimal} \& \ref{chap:learning})\\
Sergey Aseev, {\bf Talha Manzoor}, ``Optimal Growth, Renewable Resources and Sustainability", {\it IIASA Working Paper} WP-16-017, IIASA, Laxenburg, Austria, Dec 2016.
\end{papers}

\chapter*{Acknowledgments}
\addcontentsline{toc}{chapter}{Acknowledgments}

During the course of my academic career, I have been blessed with many great teachers. The work presented herein, is a result of their guidance, perseverance and patience towards me, and my many flaws as a student. Foremost on this long list is my advisor Abubakr Muhammad who took me under his wing when I was a young graduate and provided the environment that was necessary for me to conduct my research. I am thankful to him for allowing me flexibility in my work, and insulating me from many, many hindrances (both financial and logistical) that could otherwise have proven detrimental to my focus as a graduate student. I view this dissertation purely as the outcome of his own intellectual capability and his capacity as an outstanding advisor.

Next I wish to thank my co-advisor Elena Rovenskaya whose supervision has been instrumental especially during the foundational stage of my doctoral research. I stumbled into the 2013 YSSP program under her supervision at IIASA with only a minute idea of what I wanted to do in my PhD. She contributed in establishing a direction for the research and also helped me in navigating across a multitude of academic disciplines that were well beyond my comfort zone.

I am also grateful to Sergey Aseev who led the research on optimal growth that eventually formed a vital part of this dissertation. This would not have been possible without his extensive technical expertise and also his patience with my lack of experience in the area. I also want to thank Alexey Davydov for providing the proof on stability of our model that we had struggled with unsuccessfully for months on our own. I am thankful to Momin Uppal and Naveed Ul Hassan for their guidance as part of my PhD committee and informal advice at numerous occasions. Also to Tariq Samad, Magnus Egerstedt and Andries Richter for sparing time from their busy schedules to review an initial draft of the dissertation and also for providing their constructive feedback and criticism.

The Electrical Engineering Department at LUMS provided the infrastructure and support that was necessary to facilitate me in my work. I am thankful to Ijaz Naqvi, the Chair of the Graduate Program Committee who supported me in all the logistics and paperwork especially during my final phase as a graduate student; to Tariq Jadoon who provided all necessary support (both within and beyond his job description) as Chair of the Department; to Soban Hameed and his tedious hours of administrative assistance. I also want to give a special mention to the LUMS library and its staff that has helped me greatly during my studies and still continues to surprise me with the diversity and range of its excellent collection.

During my stay at LUMS, I have enjoyed the company of many great colleagues and strong friends. I wish to thank all members of the CYPHYNETS lab for making this time memorable: Syed Muhammad Abbas, Mudassir Khan, Zahoor Ahmad, Atif Adnan, Hasan Arshad, Bilal Talat, Zeeshan Shareef, Faiz Alam, Adnan Munawar, Zubair Ahmed, Sajjad Haider, Ansir Ilyas, Waqas Riaz, Bilal Haider, Hamza Jan, Ali Ahmed, Saad Hassan, Muhammad Abdullah and Allah Bakhsh (for those whose name I have failed to mention, please accept my humble apology).

As my advisor says, a PhD is a joint effort on part of the student and his family. In my case, I would say that the role of my family far exceeds that of my own. I first want to reach out to my father, Manzoor Hussain, who not only allowed me to pursue a career of my own choosing, but also supported me through all hardships and difficulties that I faced. I am rarely able to express how thankful I am to him for providing me a life full of comfort and void of any liabilities. His support has enabled me to maintain focus on my work and to enjoy my personal life to the fullest, a luxury seldom available to a doctoral candidate supporting a wife and child along with his research.

I wish to mention my late grandfather, Chaudhry Rehmat Ali who I am sure would have been extremely proud at this moment. Also my grandmother, Ramzan Bibi whose constant prayers and relentless concern have motivated me greatly throughout my studies. 

I want to thank my sister, Bazla Manzoor, for supplying a lifetime of joy and memories as a wonderful sibling. Even as I am here in Lahore, she takes care of our mother to the best of her ability as a married woman with a promising professional career. This has been a source of great content and consolation for me during the recent years of my studies.

For my mother, Naheed Manzoor, no measure of text would be sufficient to contain the ocean of praise that she deserves from me. Even so, I am at a total loss of words to express my gratitude towards her, for nurturing a young soul into the grown man typing in these words right now. To her I can only say

\begin{center}
\emph{Thank you mother, for everything}.
\end{center}

In the end I wish to acknowledge my wife, Mehwish Pervaiz for providing me the comfort of a clean house, ironed clothes, warm meals, and most importantly, a loving family. Our daughter was born while my PhD was undergoing an extremely critical phase. I could not be more thankful to Mehwish for taking care of her all through the tedious household responsibilities and for her moral and emotional support during the difficult times of my studies. To happily endure the long hours I had to work in the lab, weekends I could not spend at home and the uncertainty of my finances. I am truly grateful to her, and hope I can make it worth her while, as we move ahead to the next chapter of our lives together.\\

\vspace{40pt}
\hfill\begin{minipage}[r]{0.45\linewidth}
\flushright Talha Manzoor,\\
LUMS, 2017.
\end{minipage}

\pdfbookmark{\contentsname}{Contents}
  % -*- Mode:TeX -*-
%% This file simply contains the commands that actually generate the table of
%% contents and lists of figures and tables.  You can omit any or all of
%% these files by simply taking out the appropriate command.  For more
%% information on these files, see appendix C.3.3 of the LaTeX manual. 
\tableofcontents
\newpage
\listoffigures
\newpage
\listoftables
\newpage
\chapter*{Glossary}

\begin{tabular}{r p{10cm}}
{\bf LTG} & Limits To Growth: An influential text authored by Donella Meadows et. al. in 1972. LTG models investigate the effects of exponential growth on the finite resources of the Earth\\
{\bf TOC} & Tragedy of the Commons\\
{\bf CPR} & Common Pool Resource\\
{\bf NRM} & Natural Resource Management\\
{\bf SA} & Systems Analysis\\
{\bf CPS} & Cyber Physical Systems\\
{\bf CPSS} & Cyber Physical Social Systems\\
{\bf SES} & Social Ecological System\\
{\bf HITL} & Human In The Loop\\
{\bf IoT} & Internet of Things\\
{\bf IoE} & Internet of Everything\\
{\bf NSF} & {National Science Foundation: A United States Government Agency that supports Research and Education in the Non-Medical Sciences and Engineering}\\
{\bf CHANS} & Coupled Human And Natural System\\
{\bf SER} & Social-Ecological Relevance\\
{\bf OCP} & Optimal Control Problem\\
{\bf BR} & Best Response\\
{\bf NE} & Nash Equilibrium\\
\end{tabular}
\afterpage{\null\newpage}

\setcounter{page}{1}
\pagenumbering{arabic}

\part{Foundations and Background}
%% This is an example first chapter.  You should put chapter/appendix that you
%% write into a separate file, and add a line \include{yourfilename} to
%% main.tex, where `yourfilename.tex' is the name of the chapter/appendix file.
%% You can process specific files by typing their names in at the 
%% \files=
%% prompt when you run the file main.tex through LaTeX.

\chapter{Lessons in Natural Resource Governance}

\label{chap:res_gov}

This chapter marks the beginning of the first part of this dissertation. The three chapters included in this part are qualitative in nature and are intended to provide context to the technical work presented afterwards. The first chapter, Chapter \ref{chap:res_gov}, highlights the resource governance problem in general. The second chapter, Chapter \ref{chap:sa}, discusses the difficulties associated with the analysis and control of resource based systems and showcases the relevant theoretical tools developed, and research conducted by scientists working in the area. The third chapter, Chapter \ref{chap:cyber} focuses on the role of technology in resource based systems and how engineers may employ cybernetic methods to analyze these systems in an effective and integrative manner.

\section{Respecting the Limits of our Habitat}

In his thought-provoking book ``A Green History of the World" \cite{ponting1991green}, Clive Ponting presents a narrative on how the uncontrolled exploitation and eventual depletion of natural resources has played a pivotal role in the collapse of many leading civilizations of the world. It is not so difficult to comprehend the immense importance of preserving the habitat on which the human population is increasingly reliant. The finiteness of the planet's resources dictates that in order for the human race to satisfy the demands of its increasing population, it must devise a way of life that does not eradicate the resource base which is so essential for its survival. As Ponting most aptly recounts, there exist numerous lessons for us throughout our history in the form of societies that collapsed as a result of their direct or indirect abuse of the environment. Perhaps the most appropriate example to discuss here is that of the Polynesian society of Easter Island. 

Easter Island is located in one of the most remote regions on Earth. It is situated in the south Pacific Ocean over 3,500 km off the west coast of Chile in South America. Populated by a few dozen Polynesians in the 5th century A.D., the society gradually flourished untill the population peaked to a few thousand in the 16th century. Their diet consisted mainly of chicken and sweet potatoes, ensuring the provision of which, did not consume much of their time. This left a majority of the islanders' time open for social customs and rituals, avidly directed by the individual chiefs of the numerous clans that had taken shape over time. The center of all rituals and worship were the \emph{ahu}, stone platforms on which large statues (the \emph{Moai}) were placed as a symbol of authority and power. These Moai were extremely large in size and required an enormous amount of labour and also substance in the form of trees in order to transport them from the quarry to the ahu. As the island population grew in magnitude, the clans started competing for the felling of trees to transport the ever-growing Moais which led to the eventual depletion of the forests and degradation of their ecosystem. By the time that the European explorers first visited the island in the 18th century, the Easter Islanders had been reduced to a meager population engaged in perpetual warfare over the dwindling resources that remained, reduced to a mere shadow of the culturally complex and prosperous society they once used to be. Ponting compares the case of the Polynesians on Eater Island to that of all mankind on planet Earth and points out striking similarities albeit on different scales. Like Easter Island, the Earth is also isolated with nowhere the humans can escape to in the event that the environment is degraded to the point of inhabitability. The question quite rightly put is: have humans ``\emph{been any more successful than the islanders in finding a way of life that does not fatally deplete the resources that are available to them and irreversibly damage their life support system}?" 

As human beings, it is easy for us to realize the absolute importance of maintaining a balanced relationship with the surrounding environment in small isolated territories. For instance, Easter Island comprises a small area (roughly 160 sq. km.) which can be traversed in a single day. Moreover, due to the depleting timber and the state of their technology, the inhabitants were effectively isolated from other lands and so had nowhere to go in the event of a crisis. The causality between the deteriorating resource base and the devastating events that transpired must have been painstakingly apparent to the existing population. However, the same cannot be said when we think of our planet as a whole. In his famous essay \cite{boulding1996}, Kenneth Boulding discusses the need for a change in how we perceive our biosphere. The primitive society viewed itself as living in a limitless world i.e. a world with unbounded resources spread out over an infinite plane. Whenever primordial humans were faced with deterioration of the natural habitat or their social structure, they always had the luxury of migrating to areas with better prospects. Thus, there always existed a ``frontier", which separated the known limits of human domain from the promising territory which was yet to be explored. However as mankind would eventually discover, the Earth is not a limitless plane, but in fact a closed sphere with limited resources whose boundaries we have now begun to increasingly stress. Hence there is a need to abandon the ``\emph{cowboy}" view of the economy which symbolizes the limitless plains and the sense of adventure and exploration that always comes with the existence of a frontier. The inability of mankind to incorporate the finiteness of nature's resources in its economic activities, and the inexorable requirement to satisfy the needs of a growing world population  has led to an extraordinary strain being put on the Earth's ecosystem.

Boulding notes that abandoning the image of a limitless Earth in our economic principles is not an easy task as this image has been inscribed in our minds for centuries. Nevertheless we must now adopt the idea of a ``spaceship" economy, whereby realizing that the earth is like a spaceship, limited in its resources and isolated from its surroundings. There is no other planet we can escape to \footnote{One solution that immediately comes to mind is undertaking space and sea explorations to expand the frontiers of the available resources. But even if we extend to the immediate solar system or find new resources under the sea, it pushes these frontiers only for a finite period of time. It is important to realize that our unbounded desire to growth may eventually overtake even apparently unbounded resources. So the basic dilemma is physical and always holds.} and this is becoming increasingly clear as we are finally beginning to bear the brunt of our inconsiderate economic activities in the form of phenomena such as environmental pollution and global warming. As opposed to the cowboy economy, the magnitude of consumption and production can no longer be held as a metric for success, but in fact, it is ``\emph{the nature, extent, quality, and complexity of the total capital stock}" that really matters. 

In light of the debate on the potentially lethal relationship between mankind's economic activities and the Earth's environment, one wonders what would happen if human beings continue the pursuit of economic growth in the current fashion. Jay Forrester was perhaps the first to address this issue in a mathematically rigorous framework. His book ``World Dynamics" \cite{forrester1971world} presents a dynamical model of certain global variables that represent the state of the world economy. These include world population, pollution, quantity of non-renewable resources, world capital investment and fraction of investment in the agriculture sector. His model known as World2 marks the beginning of the field of global modeling. Perhaps the most influential work that resulted as a direct consequence of his efforts is the book ``Limits To Growth" (LTG) \cite{meadows1972limits} with Donella Meadows as the lead author. 

Using Forrester's system dynamics technique and building on his work, the authors of LTG construct a detailed computer model of the world economy called World3. The purpose of the model is to demonstrate the effect of limitless exponential growth in a world whose resources are finite. The relationship between the variables are assumed to be highly non-linear with parameters regressed from statistical yearbooks. The model also includes various feedback effects, for example, population growth is positively related to food production. However population growth is positively related to industrial output, which is positively related to pollution, which in turn, is negatively related to food production and hence population growth. 

%The authors run various scenarios and observe that the world population may exhibit any of the four different modes listed as follows
%\begin{enumerate}
%	\item The population smoothly approaches an equilibrium by means of a gradual decrease in growth till the growth is effectively zero.
%	\item The population overshoots the world's carrying capacity and then settles down in a smooth manner.
%	\item The population overshoots the carrying capacity and settles down in an oscillatory manner.
%	\item The population overshoots in a manner that consumes necessary non-renewable resources, ultimately decreasing the carrying capacity of the world.
%\end{enumerate}

Meadows et. al. ran the model under 12 different scenarios and found that at the time (in 1972) the world economy was operating inside the planetary boundaries still with room to grow and an opportunity for mankind to contemplate long-term corrective measures. However the study predicted a sudden collapse of the world population, a massive die-out somewhere in the 21st century, if the world economy continued to grow at a positive rate. Thus the author's proposed an operating point of zero growth as the only way to avoid the massive die-out. The original book received a few updates with the latest one \cite{meadows2004synopsis} arguing that humanity has already passed the planetary boundaries bringing it in a state of serious overshoot and much must be changed in order to minimize the impact of the overshoot. Various studies also observe that the historical data of events after 1972 support the key features of the predictions made by LTG \cite{turner2008comparison, turner2012cusp}. Although there has been much criticism on the methods of the LTG models and interpretation of the results (for instance, see \cite[Chapter 7]{bardi2011limits}), one thing is painstakingly clear: there is a dire need for change if we hope to live in harmony with the Earth's natural system which is at the foundation of the economical system and already showing signs of fragility. 

\section{Inability to Act: The Tragedy of the Commons}

After going through the history of Easter Island, one very intriguing thought strikes the mind. As mentioned earlier, effects of the rapid deforestation and the impending doom it would bring, would have been clearly visible to the people of the island at that time. One might ask what stopped the individual clans from abstaining from deforestation and altering their lifestyle so as to maintain the intricate balance between mother nature and human society. After all, if the environmental deterioration stopped it would provide a chance for the Islanders to rebuild and prosper once more, a situation that would unquestionably be beneficial for all clans. Yet as we know this did not happen. The history of human civilization is filled with accounts of societies that collapsed due to direct or indirect environmental abuse. For instance, the collapse of the Sumerians in Mesopotamia, the fall of the Mohenjo Daro and Harappa civilizations of the Indus, Baghdad and the Mediterranean, deforestation in China and medieval Ethiopia, and the decline of the Mayans in Mesoamerica were all linked to uncontrolled exploitation of the environment in one form or the other.  It is not unreasonable to assume that the effects of human activity on its surroundings and the eventual backlash would have been apparent to the individuals of each society in question at the time that the activities were taking place. Yet somehow, the detrimental activities continued till the ultimate disintegration of the societies. What prevented those communities from acting collectively in their own mutual interest? The answer lies in a concept that has gained the attention of social scientists for the past many decades. 

Garrett Hardin introduced the phrase ``Tragedy of the Commons" for the first time in his famous article \cite{hardin1968tragedy} in 1968. The concept is presented by envisioning a pasture or grazing commons, filled with grass. The local herdsman bring their cattle to the commons to graze upon each day. The commons is open to everyone and so their is no restriction as to how many cattle one can bring to graze upon there. The addition of a single cow on the commons will present a benefit to the particular herdsman that cow belongs to. However this addition also holds a cost in the form of additional overgrazing as a result of the increase in animals. While the benefit of an additional cow is for a single herdsman to enjoy all by himself, the cost is shared by all herdsmen which means that the cost of an additional animal (for the individual herdsman) is only a fraction of the benefit. Thus the only rational action for the herdsman is to expand his herd. This course of action is followed by all herdsmen and eventually the number of cattle exceeds the capacity of the commons thus ruining it for everyone. Rational action on behalf of the herdsmen results in collective destruction and herein lies the tragedy. What aggravates the tragedy even more is the fact that the herdsmen continue to act this way even with full knowledge of the consequence of their actions (hence a tragedy and not a shock or disaster). The commons dilemma is a specific instance of a social dilemma where individual rationality is in contrast to the collective benefit of all.

While the concept of the Tragedy of the Commons (TOC) has been defined using the story of the commons, it has many diverse manifestations in real life and is certainly not restricted to the grazing pastures only. One example that Hardin discusses in his article as well is that of overpopulation. An individual couple may embrace parenthood due to the psychological satisfaction and other benefits that come with offspring, however each additional human being on the planet increases the strain being put on its resources. If the increase in population is not stopped soon enough, it may grow beyond what the Earth is able to support. Here again, the benefits of offspring are enjoyed by the parents only, while the cost (however small) is incurred by the remaining world population. Another example is traffic congestion. Here the common resource is public roads. An individual may quite understandably prefer to commute via public roads in order to save time. However, if everyone follows this reasoning it would result in congestion of traffic on the public roads, eventually causing everyone to be late. Another illustration of the Tragedy occurs in developing communities with shortage of electric supply. In such settings, households often install an Uninterruptible Power Supply (UPS) system which stores backup electricity for use during the power outages. While this may ensure that the household gets electricity even during the outage, it puts additional strain on the supply grid due to its inefficiency (the power returned by a UPS during an outage is less than the power consumed in order to charge it). This means that if every household installs a UPS, it would result in even more outages than before, thus the tragedy. Other real life examples of the commons dilemma include, but are not limited to, pollution, overfishing (the bluefin tuna is perhaps the most eyeopening manifestation of the TOC), deforestation, depleting groundwater aquifers (a prevalent problem in Pakistan and the South-Asian region \cite{shah2003sustaining}), animal poaching and sending spam emails. 

The commons serves as a symbol for all resources that are openly accessible to the general consumer (more on this in Section \ref{subsec:res_typ}) and Hardin argues that in the absence of external regulation, rational consumers will eventually and most certainly lead to destruction of the resource. The Tragedy of the commons is one of three metaphorical models which have until recently been most influential for policy prescription in resource governance problems. We discuss the others in the following section.

\subsection{Other Models of Inevitable Doom}

The tragedy of the commons has been popularly formalized as a Prisoner's dilemma game \cite{dawes1974formal}. The motivation behind this game is the following story. Consider two robbers who are known to have robbed a bank together. However the police does not have enough evidence to convict any of the robbers without a confession from at least one of them. Upon capturing the robbers the police offers each of them the following two choices (in isolation), either they confess to the crime (the defect strategy) or not (the cooperate strategy). If that prisoner confesses to the crime and the other does not, then he is freed while the other receives 10 years in jail. If both prisoners confess then both of them receive 5 years in jail. If neither of them confesses then each will receive a single year in jail on minor charges. The game is presented in matrix form in Figure \ref{fig:pris_dil} with the Nash equilibrium shaded in gray (see Appendix \ref{app:games} for a basic introduction to games).
\begin{figure}[!htb]
\begin{center}
	\centering
	\begin{tabular}[t]{ >{\centering}m{10pt} | >{\centering}m{30pt} | >{\centering}m{30pt} | m{0pt}}
		\multicolumn{1}{c}{ }& \multicolumn{1}{c}{ C} & \multicolumn{1}{c}{D}\\
		\cline{2-3}
		C & -1,-1 & -10,0 &\\[30pt]
		\cline{2-3}
		D & 0,-10 & \cellcolor{gray}-5,-5 &\\[30pt]
		\cline{2-3}
	\end{tabular}\\[10pt]
\end{center}
\caption{The prisoner's dilemma game}
\label{fig:pris_dil}
\end{figure}

From the structure of the game, we see that each player has a dominating strategy of defecting. However if each player plays this strategy then the resulting outcome will be pareto sub-optimal i.e. there exists some other strategy (in this case: both cooperating) in which at least one player is better-off without the other being worse-off. The definition of the more general N-person dilemma which exhibits the same structure can be found in \cite{dawes1980social}. It is easy to see how this game also serves as an illustration of the Tragedy of the Commons. Assume that the players are now two herdsmen deciding on how many cattle to keep in their respective herds. If the critical level of total cattle (after which the commons is destroyed) is $L$ then let the cooperate strategy for each herdsman be to restrict his herd to under $L/2$, and the defect strategy be to add as many cattle to his herd as he pleases. The resulting game also has the same structure as shown in Figure \ref{fig:pris_dil}. 

The prisoner's dilemma is fascinating because it contradicts the inherent belief of humans that actions that are individually rational also lead to collectively rational outcomes. In his popular book ``The Logic of Collective Action"  \cite{olson2009logic},  Mancur Olson challenges the (up till then widely accepted)  basic premise of group theory, that individuals who aim to fulfill their personal benefits will rationally strive to fulfill the collective benefit of the group as well. Olson asserts that without any external coercion, large groups of individuals will not voluntarily act in the common interest of the group. In order for them to collectively achieve the group's most beneficial outcome, they would have to exhibit at least some altruistic tendency which is more of an exception than the norm, especially when there are economic principles involved. Groups and organizations tend to exist solely to serve the common goals or benefits of the group itself. If they existed to serve the interests of the individuals, it would be pointless, since independent, unorganized action on behalf of the individuals themselves would be most effective in achieving that interest. For instance, the Labor Union exists to the serve the common interest of higher wages for the workers, but each worker also has the individual interest in her own personal income which not only depends on the wages but also on the number of hours she works everyday. 

So what circumstances lead to a group of rational individuals directing their actions towards achieving their collective interest? Olson argues that economic incentives alone are not enough, there must be some social incentive as well. For instance, individuals might serve the group in order to win ``prestige" within it. Such social dynamics are more significant in groups of smaller sizes. This implies that one of the most important factors in determining whether or not individual action will result in the fulfillment of the group's mutual benefit is the size of the group. Thus smaller groups are more efficient and viable than larger ones, which is one of the major conclusions of Olson's work.

\subsection{Open-access: The Nature of ``Tragic" Resources}

\label{subsec:res_typ}

Here we take some time to discuss the nature of the resources prone to the Commons Dilemma. As we shall soon see, natural resources that are vulnerable to tragedy-like situations have particular characteristics that are not exhibited by every resource. For example, their is no TOC in the direct consumption of solar energy which will be available in the same quantity tomorrow, irrespective of how much it is consumed today. We elaborate more formally as follows.

Natural resources can be classified as either renewable or non-renewable/exhaustive resources. While renewable resources have the capacity to regenerate over time, non-renewable resources either do not possess a regenerative capacity or if so, the timescale of the regeneration is of geological magnitude (e.g. fossil fuels). While non-renewable resources are bound to be exhausted at some future point in time (unless of course consumption is stopped altogether), renewable resources can be consumed indefinitely while maintaining a positive stock quantity, provided the consumption is below a certain critical level. This possibility is what makes the study of renewable resources so much more interesting especially in the context of sustainability.

Natural resources are also fundamentally classified as either public or private. This classification \cite{perman2003natural} is based on the categorization of goods as rivalrous/non-rivalrous and excludable/non-excludable.  An excludable good is one in which individuals who do not pay for the good can be prevented from consuming it. All private goods such as cellphones, personal vehicles and clothing, etc, are excludable. Other resources such as public parks, cinemas and zoos are also excludable.  Goods such as fish in the open sea, solar energy, national television and radio are goods that are non-excludable. It is not possible to prevent someone from consuming a non-excludable resource, it is open for consumption to anyone who can do so. A rivalrous good is one in which consumption of the good by one individual affects the amount available to other individuals. The earlier example of fish in the open sea is a rivalrous resource. Solar energy is not. The described classification is tabularized in Table \ref{tab:goods}. The category of major interest here is \textbf{open-access} or \textbf{common-pool resources (CPR)}. Examples of such resources include fisheries, forests, bio-fuels, geo-thermal energy, minerals, oil pools and many more. Common pool resources may also be either renewable or non-renewable (see Figure \ref{fig:res_typology}). It is the open-access, rivalrous nature of these resources that make them amenable to phenomena such as the Tragedy of the Commons.

Let us revisit the example of the grazing commons. As the commons is not owned privately, any new herdsman may approach the commons so that his cattle may graze upon it. There is also no possibility of removing any existing cattle already on the commons. These are all characteristics of a non-excludable resource. Furthermore, despite the grass being a renewable entity, the total amount at any given instant  is finite and so, the grass consumed by some cattle is not available for consumption to the others, i.e., the resource is rivalrous. Thus according to Table \ref{tab:goods}, the grazing commons is an open-access or a common pool resource. It is this attribute of the commons that makes the tragedy possible.

\begin{table}[t!]
\centering
	\begin{tabular}[t]{ >{\flushleft}m{80pt} | >{\centering}m{150pt} | >{\centering}m{150pt} | m{0pt}}
	\multicolumn{1}{c}{ }& \multicolumn{1}{c}{Excludable} & \multicolumn{1}{c}{Non-excludable}\\[10pt]
	\cline{2-3}
	Rivalrous & \begin{minipage}{\linewidth}\centering \textbf{Purely private good} \\ cellphones\end{minipage} & \begin{minipage}{\linewidth}\centering \textbf{Common pool resource (CPR)} \\ Fisheries\end{minipage} &\\ [25pt]
	\cline{2-3}
	Non-rivalrous & \begin{minipage}{\linewidth}\centering \textbf{Congestible resource} \\ public parks\end{minipage} & \begin{minipage}{\linewidth}\centering \textbf{Purely public good} \\ national television\end{minipage} &\\ [25pt]
	\cline{2-3}
	\end{tabular}\\[10pt]
\caption{public and private goods}
\label{tab:goods}
\end{table}
\begin{figure}[b!]
	\captionsetup{font=normal,width=0.8\textwidth}
	\begin{center}
		\includegraphics[width=0.8\linewidth]{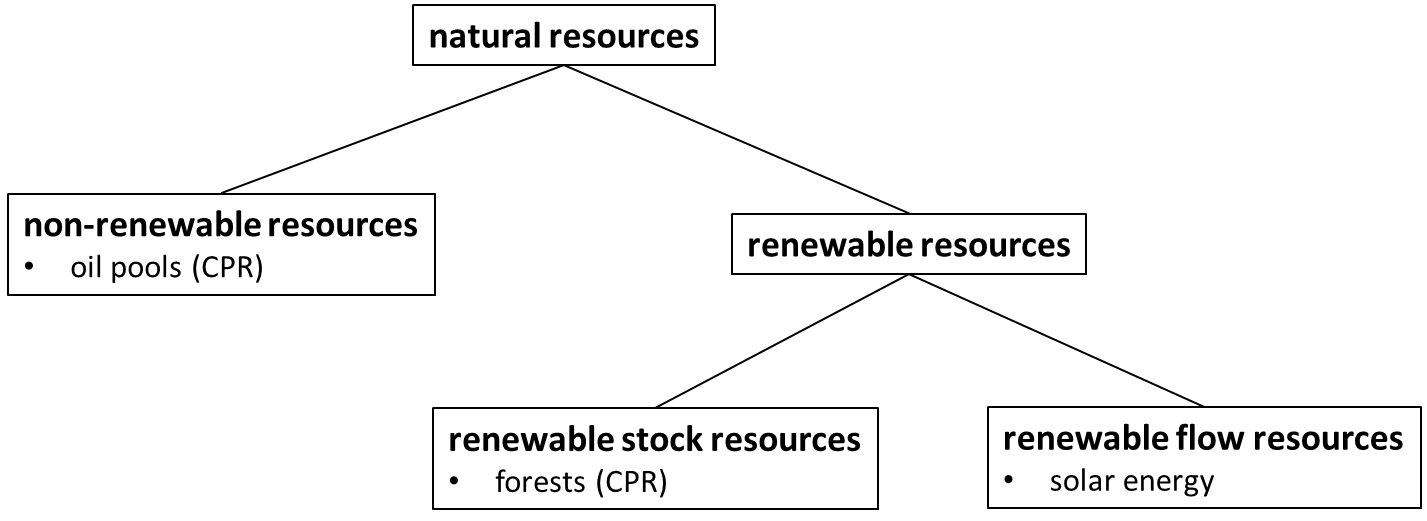}
	\end{center}
	\caption{A typology of natural resources.} 
	\label{fig:res_typology} 
\end{figure}

Another classification of interest is shown in Figure \ref{fig:res_typology}, where renewable resources are classified as either renewable flow resources or renewable stock resources. Renewable stock resources are resources for which the amount consumed in the present, affects the amount available for consumption in the future. Examples include fisheries, water resources, forests, arable and grazing land, etc. Renewable flow resources on the other hand are ones for which the amount available for consumption in the future is not affected by the amount consumed in the present. Examples include solar, wind and geo-thermal energy (although the sun itself, and the Earth that harbors the winds and geo-thermal reservoirs are physically depleting albeit over extremely large stellar and geological timescales).

It should be noted that for all practical purposes flow resources (in contrast to stock resources) are non-depletable. Thus the issues that arise in governing stock resources are not encountered in flow resources. The concept of a common-pool resource covers only stock resources (here a non-renewable resource is also considered a stock resource) and thus renewable flow resources are not generally categorized as CPRs.

\subsection{Privatization and State Regulation: The Only Solutions?}

The implication of Hardin's Tragedy of the Commons and the other metaphorical models of doom (the prisoner's dilemma and Olson's logic of collective action) imply that in the absence of external coercion by some central authority, consumers of a CPR are trapped in a downward spiral of resource depletion and ultimate ruin. In his original article, Hardin posited that the commons dilemma belonged to a class of ``no technical solution problems", where a technical solution is understood to mean ``\emph{one that requires a change only in the techniques of the natural sciences, demanding little or nothing in the way of change in human values or ideas of morality}". Up until recently, it was commonly conceived that the only way to ensure sustainable consumption of a commons was either privatization or state regulation, thus fundamentally altering the very open-access nature of the resource. Indeed these are the very solutions that Hardin advocated for the TOC. No community-based solution was thought to be feasible. However, as we shall soon see, Elinor Ostrom was yet to compile several real-world examples of communities, which had collectively devised their own institutions to use the commons in a renewable and sustainable manner. Not only this, but she would also enumerate several common characteristics of the institutions as guidelines for policy prescription, a work which would eventually win her the Nobel prize in economics.

\section{Cooperative Solutions to the Commons Dilemma}

In her book ``Governing the Commons" \cite{ostrom1990governing}, Elinor Ostrom presents her work on community-based governance of natural resources in open-access settings. She begins by accepting that the three models namely, The Tragedy of the Commons, The Prisoner's Dilemma and The Logic of Collective Action all serve very nicely as illustrations of situations in which societies are not able to fulfill their mutual interests. The underlying aspect of such situations that these models highlight is the so called tendency of \emph{free-riding}. Such a phenomenon occurs when individuals of a group cannot be excluded from enjoying the group's mutual benefit, regardless of how much effort he or she exerted towards fulfilling that benefit. Thus fishers are tempted to individually harvest large stocks of fish while the others abstain to maintain the stock at a renewable level and factories spill their waste in untreated form leaving it to others to worry about spending money on waste disposal systems. And so these models are able to explain (at least up to some convincing extent) why individually rational actions may lead to collectively irrational outcomes.

\begin{table}[b!]
\centering
\begin{tabular}{| >{\small}m{0.47\linewidth} | >{\small}m{0.47\linewidth} | m{0pt}}
\cline{1-2}
\centering {\cellcolor[HTML]{9B9B9B}{\color[HTML]{FFFFFF} \bf \large Design Principles}} & \centering {\cellcolor[HTML]{9B9B9B}{\color[HTML]{FFFFFF} \bf \large Examples}} & \\[0.5 cm]  \cline{1-2}
\flushleft \begin{minipage}{\linewidth} \centering \textbf{Clearly Defined Boundaries} \\ The individuals who have the right to harvest the CPR and the boundaries of the CPR itself are well defined \vspace{10 pt}\end{minipage}   & Fishing licenses, Demarcation of grazing pastures &\\  \cline{1-2}
\flushleft \begin{minipage}{\linewidth} \centering \textbf{Proportional Equivalence Between Benefits and Costs} \\ Individuals who benefit more from the CPR must incur more of the operational cost \vspace{10 pt}\end{minipage}   & Large benefactors providing more labor for communal activities, Imposition of taxes proportional to consumption &\\  \cline{1-2}
\flushleft \begin{minipage}{\linewidth} \centering \textbf{Collective-choice arrangements} \\ Individuals effected by operational rules also have influence in changing those rules \vspace{10 pt}\end{minipage}   & Local elections, Referendums, Rotational rosters for irrigation&\\  \cline{1-2}
\flushleft \begin{minipage}{\linewidth} \centering \textbf{Monitoring} \\ Monitors/auditors are accountable to the users \vspace{10 pt}\end{minipage}   & Auditors elected from within the community itself, Technological solutions such as smart meters &\\  \cline{1-2}
\flushleft \begin{minipage}{\linewidth} \centering \textbf{Graduated sanctions} \\ Violation of rules receive penalties proportional to the degree of violation \vspace{10 pt}\end{minipage}   & Extra labor in the case of violating a maintenance
contribution, Predefined fines or seasonal bans on resource use
stealing &\\  \cline{1-2}
\flushleft \begin{minipage}{\linewidth} \centering \textbf{Conflict-resolution mechanisms} \\ Users have access to low cost local arenas to resolve conflicts \vspace{10 pt}\end{minipage}   & Communal judicial systems such as Community Courts in the Western countries and the Panchayati Raj in South Asia &\\  \cline{1-2}
\flushleft \begin{minipage}{\linewidth} \centering \textbf{Minimal recognition of rights to organize} \\ Institutions devised by locals to regulate the CPR are recognized by external government authorities \vspace{10 pt}\end{minipage}   & Authoritative powers entrusted to the municipalities, Promotion of autonomous tribal area systems for resource governance &\\  \cline{1-2}
\flushleft \begin{minipage}{\linewidth} \centering \textbf{Nested enterprises} \\ Multiple tiers of nested organization \vspace{10 pt}\end{minipage}   &  Defining layers of organization, e.g., in irrigation systems, for users in the same canal branch, same headworks, same river basin, and so on    &\\  \cline{1-2}
\end{tabular}
\caption{Design principles illustrated by long-enduring CPR institutions (as identified by Ostrom \cite{ostrom1990governing})}
\label{tab:ostrom_principles}
\end{table}
Ostrom continues by asserting that while the aforementioned models of inevitable doom are extremely interesting and powerful they are also extremely dangerous. What makes them dangerous is the manner in which they are used in policy prescription. Instead of being used as a guiding tool for the policy-makers the conclusions of the models are taken as a starting point, a dogma to say the least, and it is accepted without question that the community itself is unable to sustain its resources without interference of the governing body. It is not taken into account that the hard constraints that are in place in the assumed setting of the models are not in fact hard constraints in empirical settings. For instance, the prisoners in the prisoners dilemma might be able to negotiate beforehand, reaching a binding agreement to always play the cooperate strategy.

The metaphorical use of these models in Ostrom's point of view eventually leads to prescribing solutions that are not appropriate in practice. Consider for instance the commonly prescribed solutions to the CPR problem: privatization and state regulation. While hectares of land are possible to privatize, how would one effectively do this in the case of fisheries? Ostrom posits that for effective governance, the resources must belong to institutions that are collectively owned by the consumers themselves. Such institutions are more likely to succeed in governance as compared to institutions that are owned by the state or private enterprises. This is because the users (as stakeholders) themselves have more incentives to make the arrangement successful. Furthermore they are less likely to make informational errors e.g. under or over-estimating the carrying capacity of the resource and the wrong identification of cooperators and defaulters, which may cause the institution to fail in its objective. An example of such policy making is the nationalization of forests in Third World countries, which has been advocated on the grounds that the local villagers are unable to manage their forests effectively. In order to regulate the forests, the governments usually deploy officials with salaries that are low enough to propel them to accept bribes (this is of course not intentional but an implication of the state of the national economies). This has in many cases resulted in the forests being effectively transformed to an open-access resource which had previously been limited-access.

The highlight of Ostrom's work is the enumeration of several communities from diverse geographical and demographical regions of the world that have devised solutions to sustain their commons which neither fall in to the category of state-ownership nor privatization. In all the studied cases, the local communities were able to avoid tragedies without any external coercion by coming up with regulating institutions of their own. She then presents the common characteristics of these institutions as guidelines for governance of CPRs. These characteristics have been reproduced in Table \ref{tab:ostrom_principles}. An adaptation of the design principles to irrigation systems is given by Ostrom in \cite{ostrom1993design} and have been upheld by successful irrigation systems in multiple parts of the world \cite{lam2010analyzing, sarker2001design}. The principles have also been studied in the context of other resources such as forests \cite{gibson2000people}, fisheries \cite{schlager1999property} and rangelands \cite{mwangi2009top} among various others.

Ostrom's work contains an important lesson: successful management of open-access resources is possible through local institutions without the coercion of a central authority. Knowing that feasible solutions to a problem exist serves as a stimulus for the quest of those solutions. The case studies of real-world communities living in harmony with their respective resources provide exactly the inspiration needed as we contemplate a scientific enquiry of the NRM problem in the subsequent chapters of this dissertation.

%% This is an example first chapter.  You should put chapter/appendix that you
%% write into a separate file, and add a line \include{yourfilename} to
%% main.tex, where `yourfilename.tex' is the name of the chapter/appendix file.
%% You can process specific files by typing their names in at the 
%% \files=
%% prompt when you run the file main.tex through LaTeX.

\chapter{Eyeing CPRs through the Systems Lens}

\label{chap:sa}

%\chapter{Evolution of Scientific Inquiry: Eyeing CPRs through the Systems Lens}

%\chapter{Systems Analysis: Evolution of the Scientific Method to Confront Complex Problems}

%\chapter{Systems Analysis: An Integrative Approach to Scientific Inquiry and its Application to Complex Environmental Issues}

Common pool resources, like all natural resources, are vital components of the complex ecosystem that sustains the development of human society. Resources are used in the provision of food \& drink, housing \& infrastructure, mobility and so many other necessities for our survival. As we saw in Chapter \ref{chap:res_gov} they can be a cause to initiate war and even spur the collapse of civilizations. Resources effect the way we live, move and even think. On the other hand, humans effect the resource through consumption and conservation acts. Hence common pool resources, like all natural resources, are intricately tied to the human society. Neither can be studied in isolation without considering the effects of the other. CPRs and society thus constitute what is known in scientific terms as a \emph{system}. In this chapter, after defining the construct of a system, we study the method of inquiry commonly employed to examine complex systems. This method of inquiry is known as \emph{Systems Analysis} (or SA for short).

\section{The Concept of a System}

With the expansion of science and human knowledge, we are being confronted with systems that are increasingly complex to understand and eventually manage and control. Methods for scientific inquiry have evolved over the years with the goal of understanding such systems. Since we will henceforth be employing the concept of a system extensively, it is beneficial to reflect on what we actually mean by the term.
%The scientific method can be categorized broadly into two techniques: analysis and synthesis. A complete understanding of both modes of inquiry is essential for a complete investigation of any system. As we shall see shortly, while both methods may appear to be contrary to each other, they are in fact complementary and one cannot be without the other. Each method opens up a different view of the system under investigation, both of which are necessary to build a complete picture for the system's working.

A system can be thought of as a whole that consists of more than one component (these components may also be systems in their own right). According to Ackoff \cite{johnson1997mechanistic}, each component of a system must have the following characteristics
\begin{enumerate}
	\item Each component is capable of influencing the behavior of the whole system.
	\item The way each component influences the overall system is dependent on the state of the other components. In other words, the components are all interconnected.
	\item Any set of  components also has the characteristics (as mentioned in the above two points) of the individual components.
\end{enumerate}
The definition above serves to underpin a very critical property of a system. Namely, the system holds properties that none of the individual components have. In other words, a system is more than just a sum of its individual parts.

We now move on to analysis and synthesis, the two modes of inquiry for understanding a system, both of which form the foundation for Systems Analysis.

\section{Analysis and Synthesis: The Fundamental Ingredients of Scientific Inquiry}

Analysis stems from the Greek \emph{analusis} which means `` to loosen up". As a method of inquiry it dictates that in order to understand any system, we must decompose it into its parts and gain an understanding of the individual parts in isolation. For instance, in order to comprehend the working of a computer we decompose it into parts until we get a CPU, hard drives, CD ROM, external ports, data and memory buses etc. In order to understand the computer an understanding of all these parts in isolation is necessary. Thus in order to understand life, we must understand a cell, the fundamental unit of all living organisms; and to understand the working of the universe, we must first understand an atom, the fundamental unit of all matter. 

Analysis is hence a process by which a system is taken apart, and the understanding of the individual parts is then aggregated to gain an understanding of the whole. Consider the example of the computer again. We already discussed that in order to understand the computer, we must gain an understanding of the CPU, hard drives, CD ROM, external ports etc. However, if we take all these things individually and simply put them together in one place in front of us, we will surely not get a working computer. The ultimate purpose of the computer is to perform calculations (admittedly a very crude delineation of the computer's purpose, but nonetheless sufficient for the point of this discussion). But if we take the memory unit and pull it out from the computer, it will cease to serve its purpose.  Analysis prompts us to take the system apart in order to understand it. But as we have just seen, taking apart the system (in this case a computer) takes away its essential properties. For  it is not only the individual components that make the system, but it is also the interconnections between them that are equally (if not more) important in its formation. Synthesis enables us to understand the system with respect to these interconnections. 

Synthesis comes from the Greek \emph{sunthesis} which means ``to place together" and as the literal meaning indicates, it is the opposite of analysis. Let us look at synthesis more closely by comparing it with analysis. In analysis, we first disintegrate the system into its components, understand each component in isolation and then combine the isolated understanding to form an understanding of the whole. In synthesis, we first view the system as a part of the larger system which contains it, understand that system and then disintegrate the understanding of the larger system into that of the system under investigation by identifying its role in the larger system.

Let us return to the example of the computer. We already discussed that in order to understand the computer via analysis, we must gain an understanding of its various parts which include the data, memory and control buses. If we look at each bus in isolation, we would find them to be identical. However if we understand the larger system i.e. the system bus, of which all these are a part of, and see how it connects with the CPU, memory unit and input/output we will be able to identify and understand the individual buses through these interconnections. Here it is the role each individual bus plays as part of the larger system  (the system bus) that truly defines its purpose. 

It is important to note that synthesis is not ``better" than analysis for any common notion of the word. In fact analysis and synthesis are not even comparable, they are complementary to each other. Let us consider the example of a research institute. If we insist on employing analysis to understand the working of the institute we disintegrate it into its constituent programs. But in order to understand the programs we must use analysis again and so we break them into their constituent clusters. Each cluster then consists of individual staff members and to understand each staff member we would again have to disintegrate perhaps along the lines of that member's professional profile. Each profile would indicate association with a particular field of study and so on and so forth in a never-ending process of continuous disintegration. On the other hand, if we try to understand the institute on the basis of synthesis, we look at the larger system that institute belongs to. Perhaps it is part of a university. But the university is again a part of a larger system which is the educational system, which is in turn part of the governing body and we continue in this manner ad infinitum. Analysis or synthesis alone cannot furnish a complete understanding of a system.

We shall now return our focus to natural resource systems. In this dissertation, a resource system is understood to be composed of two components: human society (or the fraction of society consuming the resource) and the resource itself. Thus if we intend to dissect CPR systems for the purpose of analysis, we would obtain at the first step the resource and the society. The resource itself is a complex system. For example, forests consist of trees, wildlife and other vegetation. Trees in turn are composed of leaves, branches, roots, etc. We could delve even deeper till we reach the biological cell and even then, the process does not end here. The same holds for the consuming society, which is divided into small communities. The communities consist of households which in turn consist of individual human beings. Now the human being itself is such a complex entity that one cannot hope to even begin to understand it fully until the disintegrating process is stopped somewhere and the context of the study is brought into play to simplify matters. On the other side, to understand CPR systems by synthesis we search for the containing system. Conceivably, the CPR system is part of a larger economy belonging to a particular geographical location. That economy is a part of the global economy which in turn is a part of our planet's ecosystem, which is again part of a larger system, and so on. CPRs as systems comprise of human society and a natural resource which are both extremely complex systems in their own right. This is what makes the understanding of CPRs so challenging. 

So how does one then actually understand a system? The answer, as we shall see next, is a blend of analysis and synthesis.

\section{Systems Analysis: Tackling Complex Systems with an Integrative Perspective}

When we apply analysis to any system, what we obtain is a set of the individual components. As Ackoff \cite{johnson1997mechanistic} would say, it yields an understanding of the system's structure.  When we apply synthesis to a system we understand the role it plays in the larger system of which it is a part. It thus yields an understanding of the system's behavior. Ritchey \cite{ritchey1991analysis} distinguishes between two levels of abstraction of a system. One level views the system as a primary functioning unit, a ``black box" and what this presents is information on the system's behavior (thus the level synthesis is concerned with). Another level views the system as a set of components working together to achieve some common purpose, yielding the overall construction or structure of the system (thus the level analysis is concerned with). As we saw in the case of the research institute, it is only possible to understand the system at one of these levels. So then the question really is: what type of understanding do we want to gain about the system?

Perhaps one of the first applications of this intricate relationship between analysis and synthesis is included in Riemann's study ``The Mechanism of the Ear" \cite{riemannmechanism}. The study is concerned with how the sound sensation of the human ear works and is one of the assignments Riemann was engaged in just prior to his death. The departing point for the study is the late 19th century work of Hermann Helmholtz \cite{von2005sensations} who used the analytical method of scientific inquiry to develop a theory of the human ear, and subsequently applied it to the theory of musical composition. While Helmholtz's work contains many significant contributions and has gained the status of a classic in human physiology, it also contained many flaws. Going into specific details is beyond the scope of this dissertation however it is sufficient to mention that one of the biggest flaws concerned the role of the middle ear. Riemann begins his study from a synthetic point of view i.e. the purpose of the ear. He does not conduct his study with perspective of the anatomy of the human ear but with the perspective of what the ear actually accomplishes. Indeed this is the same Georg Friedrich Bernhard Riemann who is primarily known for his significant contributions in mathematical analysis, differential geometry and number theory. Despite the incomplete form of Riemann's manuscript, it includes important conclusions for the principles of human anatomy. While Riemann did not actually use the term ``system" he had unknowingly employed what would later come to be known as systems analysis.

Without the knowledge of the ear's anatomy that Helmholtz's work contained, Riemann would not have been able to complete his study. For in his words, ``every synthesis rests upon the results of a preceding analysis, and every analysis requires a subsequent synthesis". It is exactly this thought that is at the very core of the Systems approach to scientific inquiry.

Barton and Haslett \cite{barton2007analysis} describe the scientific method as a dialect between analysis and synthesis. The purpose of Systems Theory is to frame this dialect. Thus, as Ackoff posits \cite{johnson1997mechanistic}, Systems Theory is a  fusion of both analysis and synthesis. How much reduction to apply to the system via analysis before performing the final synthesis is a question of what we wish to gain from our understanding of the system and must be inferred from context of the inquiry. The early proponents of the systems approach asserted that every system has certain characteristics that are independent of the discipline that particular system is associated with and sought to develop a theory of those characteristics. Ludwig Von Bertalanffy in his famous paper \cite{von1950outline} laid the foundations for General Systems Theory which later on came to be known as Systems Science. Even with the diversity of systems that scientists encounter, ranging from living organisms to social systems to man made machines, there exist several universal properties that all systems possess. Via his General Systems Theory, Bertalanffy gave mathematical descriptions of some of these properties such as wholeness, growth, mechanization, centralization, sum and finality to name just a few. 

The term Systems Theory has grown to encompass several related terms such as Systems Thinking, Systems Science, Systems Analysis, Systems Dynamics, Systems Engineering and so on (Skyttner \cite{skyttner2005general} gives an excellent review on the subject). We next give an overview of how the Systems perspective has been utilized by those engaged in natural resource systems. 

\section{Systems Analysis and Coupled Human Natural Systems}

\label{sec:chans_sa}

The immense significance of TOCs in the future of human development has been the primary focus of Chapter \ref{chap:res_gov}. At this point, we hope the reader has begun to appreciate that the nature of TOC problems are approachable only through an integrative, multidisciplinary approach (which is exactly what Systems Analysis has to offer). The problem of the Commons is one of the major global challenges that scientists working in the environmental and related sciences have confronted through systems analysis. These include issues such as climate change, population growth, environmental pollution and resource management in general \cite{marsh1965man, grant1997ecology, ison1997systems, watt2013systems}. Simon Levin \cite{Levin2015current} considers TOCs as one the three grand challenges of systems theory. The complete list is enumerated as follows: \emph{1)}: robustness and resilience to critical transitions, \emph{2)}: scaling from the microscopic to the macroscopic, and \emph{3)}: dealing with problems of the Commons. We now give a brief overview of the application of Systems Analysis to TOC problems as part of the more general problem of natural resource management.

Natural resource systems, as coupled systems between resources and society, represent a larger class of systems, namely, Coupled Human Natural Systems (or CHANS). CHANS are characterized by settings in which people interact with natural components. These components may be natural resources or other elements of the environment which include wildlife, vegetation, the atmosphere, landscapes, climate, etc. CHANS are characterized by complex feedbacks between social and ecological components, high non-linearities and significant time-lags. As such, understanding the system is of critical importance before any kind of intervention is undertaken. The law of unintended consequences often comes into play in such systems. For instance, in a bid to increase crop productivity, the ``Four-Pests" campaign of China \cite{shapiro2001mao} sought to eradicate the Eurasian Tree Sparrow (along with mosquitoes, flies and rats) which ate grain, seed and fruit. It was later realized that besides grain, the sparrows also ate many insects harmful to rice, and as a result, the rice yield actually decreased. Another example is the policy to eradicate wildfires in American forests \cite{calkin2015negative}. The policy did lead to fewer fires, but, also led to growth conditions such that when fires did occur, they were much larger and more damaging. As a result the total acreage of forests affected by fires increased. 

The lesson of the Four-Pest campaign and the American wildfire policy dictates a full understanding of the system before any intervention is done. For such systems, traditional ecological research often excludes human impacts while traditional sociological research often neglects ecological effects. However, as Liu et. al. \cite{liu2007complexity} assert, any effective study on CHANS must incorporate not only variables related to the social and ecological components, but also those related to their connections (such as resource consumption and provision of ecosystem services). The systems perspective is thus a natural way to approach CHANS and an increasing number of studies are now integrating the ecological and social sciences in order to study these systems. 

Due to the applicability and effectiveness of systems thinking for analyzing CHANS problems, numerous interdisciplinary frameworks have emerged to study such systems with an integrative perspective. Binder et. al. \cite{binder2013comparison} provide a comprehensive comparison of the different frameworks used to analyze coupled human-natural systems \footnote{In the cited publication, the authors refer to CHANS as social-ecological systems. Such interchangeable terminology is common in the literature.}. The authors classify ten different frameworks according to three criteria \emph{1)}: whether the framework conceptualizes the interactions between the social and ecological components as uni- or bi-directional, \emph{2)}: whether the ecological system is defined according to its utility to humans (anthropocentric) or its internal functioning, and \emph{3)}: whether the framework is intended to provide a common language for formulating and approaching research problems (analysis-oriented), or with the goal of acting upon or intervening in the system (action-oriented).

\begin{table}
\centering
\begin{tabularx}{\linewidth}{| >{\centering}m{0.3\linewidth} | >{\small \centering}m{0.21\linewidth} | >{\small \centering}m{0.20\linewidth} | >{\small \centering}m{0.21\linewidth} | m{0pt}}
\hhline{|*{4}{-}}
\centering {\cellcolor[HTML]{9B9B9B} { \color[rgb]{1 1 1} \normalsize \bf Framework}} & \centering{\cellcolor[HTML]{9B9B9B} \color[rgb]{1 1 1} \normalsize \bf Interactions} &  \hspace{0pt} \centering { \cellcolor[HTML]{9B9B9B} \color[rgb]{1 1 1} \normalsize \bf Conceptualization} &  \centering {\cellcolor[HTML]{9B9B9B}} \color[rgb]{1 1 1} \normalsize \bf Orientation & \\[1cm] \hhline{|*{4}{-}} 
\flushleft \begin{minipage}{\linewidth} \textbf{DPSIR} \\ Driver, Pressure, State, Impact, Resource \cite{carr2007applying} \vspace{10 pt}\end{minipage}   & Social $\rightarrow$ Ecological &Anthropocentric&Action Oriented&\\  \hhline{|*{4}{-}}
\flushleft \begin{minipage}{\linewidth} \textbf{ESA} \\ Earth Systems Analysis \cite{schellnhuber1999earth} \vspace{10 pt}\end{minipage}   &Social$\rightarrow$Ecological&Ecocentric&Analysis Oriented&\\  \hhline{|*{4}{-}}
\flushleft \begin{minipage}{\linewidth} \textbf{ES} \\ Ecosystems Services \cite{de2002typology}  \vspace{10 pt}\end{minipage}   &Social$\rightarrow$Ecological&Ecocentric&Analysis Oriented&\\  \hhline{|*{4}{-}}
\flushleft \begin{minipage}{\linewidth} \textbf{HES} \\ Human Environment Systems Framework \cite{scholz2004principles}  \vspace{10 pt}\end{minipage} &Social$\rightleftarrows$Ecological&Anthropocentric&Analysis Oriented&\\  \hhline{|*{4}{-}}
\flushleft \begin{minipage}{\linewidth} \textbf{MEFA} \\  Material and Energy Flow Analysis \cite{haberl2004progress} \\  \vspace{10 pt}\end{minipage}   &Social$\rightarrow$Ecological&Ecocentric&Analysis Oriented&\\  \hhline{|*{4}{-}}
\flushleft \begin{minipage}{\linewidth} \textbf{MTF} \\ Management and Transition Framework \cite{pahl2009conceptual}  \vspace{10 pt}\end{minipage}   &Social$\rightleftarrows$Ecological&Anthropocentric&Analysis Oriented&\\  \hhline{|*{4}{-}}
\flushleft \begin{minipage}{\linewidth} \textbf{SESF} \\ Social-Ecological Systems Framework \cite{ostrom2009general}  \vspace{10 pt}\end{minipage}   &Social$\rightleftarrows$Ecological&Anthropocentric&Analysis Oriented&\\  \hhline{|*{4}{-}}
\flushleft \begin{minipage}{\linewidth} \textbf{SLA} \\ Sustainable Livelihood Approach \cite{ashley1999sustainable} \vspace{10 pt}\end{minipage}   &Social$\leftarrow$Ecological&Anthropocentric&Action Oriented&\\  \hhline{|*{4}{-}}
\flushleft \begin{minipage}{\linewidth} \textbf{TNS} \\ The Natural Step \cite{burns1999natural} \vspace{10 pt}\end{minipage}   &Social$\rightarrow$Ecological&Ecocentric&Action Oriented&\\  \hhline{|*{4}{-}}
\flushleft \begin{minipage}{\linewidth} \textbf{TVUL} \\ Turner's Vulnerability Framework \cite{turner2003framework} \vspace{10 pt}\end{minipage}   &Social$\leftarrow$Ecological&Anthropocentric&Action Oriented&\\  \hhline{|*{4}{-}}
\end{tabularx}
\caption{Integrative frameworks for analyzing coupled human natural systems \cite{binder2013comparison}}
\label{tab:chans_fw}
\end{table}

Table \ref{tab:chans_fw} enumerates the list of frameworks and their categorization as given by Binder et. al. \cite{binder2013comparison}. Three of these frameworks conceptualize the interactions between the ecological and social sub-systems in both directions: the Human Environment Systems (HES) Framework \cite{scholz2004principles}, the Management and Transition Framework (MTF) \cite{pahl2009conceptual} and the Social-Ecological Systems Framework (SESF) \cite{ostrom2009general}. The three are similar as far as the conceptualization of the ecological subsystem and the orientation of the frameworks' goal are concerned. All three are analysis-oriented with the ecological subsystem defined according to its utility to humans (anthropocentric). However among the frameworks, there exist subtle differences in the disciplines they evolved from, their theoretical origins, fundamental research questions and the settings to which they have been applied to.

The Management and Transition Framework (MTF) was developed for the application to water governance problems with a focus on studying the transitional behavior under regime change. Binder et. al. state that MTF has no specific disciplinary origin but it may be attributed to complex systems science. The theoretical foundations of MTF are related to CPR theory, institutional analysis \& development and social psychology. As can be seen in Table \ref{tab:chans_fw}, the MTF conceptualizes the resource in an anthropocentric manner which means that only those aspects are considered which impact humans significantly. The MTF goes one step further than the other anthropocentric frameworks and explicitly considers environmental hazards which pose a significant threat to human well-being. While MTF was originally intended for water governance systems, it has also been applied to risk governance and integrated landscape management. In the authors' opinion, the framework may be tailored for application to other domains as well.

The Human Environment Systems (HES) framework, as its description in Table \ref{tab:chans_fw} suggests, was developed as a tool to structure CHANS research. The framework is heuristic in nature and provides operative concepts to study human-centric environmental problems (which may include resource governance). The HES framework considers the details of the ecological system only with reference to the social system. The standard procedure for analysis is to select a state-of-the-art model for the ecological system and then set its scale according to the corresponding social system. The framework has emerged from the environmental decision making and psychology community and is theoretically based on systems science, decision theory, game theory and sustainability science. It is applicable to any system where the interactions between humans and the environment play a significant role e.g., energy, water, waste, etc.

The Social Ecological Systems Framework {SESF} specifies a set of multi-hierarchical variables relevant to CHANS. The framework has been proposed by Ostrom \cite{ostrom2009general}, who selects the variables on the basis that they proved to be pivotal in her studies of successful CPR governance schemes. She places a significant emphasis on hierarchies (note principle no. 8: nested enterprises, in Table \ref{tab:ostrom_principles} of features of successful CPR governance schemes), with multiple tiers of the variables corresponding to multiple levels of organization. This aspect of the SESF enables it to support analysis on multiple scales. The SESF also stands out from the frameworks of Table \ref{tab:chans_fw} with respect to the degree to which the social and ecological systems are teated in equal depth. As Binder et.al. \cite{binder2013comparison} note, while the ecocentric frameworks incorporate the ecological system in more detail and the anthropocentric frameworks incorporate the society in more detail, SESF incorporates both systems in equal detail. The framework has emerged from political sciences and at its foundation, includes theories such as collective choice, CPR and natural resource management. The framework has been developed for resources that are common-pool in nature and has been applied mainly to forests, pastures, fisheries and water.

The three frameworks HES, MTF and SESF described above have either been developed specifically for natural resource governance, or can potentially be applied to resource governance problems. The rest of the frameworks in Table \ref{tab:chans_fw} also address environmental problems (in which humans play a significant role) of immense importance. The interested reader is referred to the cited references of each framework for specific details.  Each framework includes explicit consideration of the non-linear nature of resource and social dynamics, uncertainties for key variables, imperfect knowledge of human actors, effects of social linkages, and so on. As applications of the Systems approach to critical environmental problems, they touch on many disparate disciplines each of which contribute different aspects to the study of the coupled human natural system. 

Perhaps one of the most fundamental proposition of the systems approach is an abandonment of cause and effect as the only possible relationship between two entities, whether viewed independently or as components of a larger system. This of course is a reference to the feedback loop, a concept at the foundation of what is known as cybernetics. Cybernetics is a manifestation of systems theory that was conceived almost parallel to Systems Science. While the concept has somewhat faded away from the forefront of the scientific community, the engineering community, motivated by recent advances in technology, is now bringing it back into focus. The development of cybernetics and its potential application to NRM is what shall capture our attention for the next chapter.

%% This is an example first chapter.  You should put chapter/appendix that you
%% write into a separate file, and add a line \include{yourfilename} to
%% main.tex, where `yourfilename.tex' is the name of the chapter/appendix file.
%% You can process specific files by typing their names in at the 
%% \files=
%% prompt when you run the file main.tex through LaTeX.

\chapter{From Control of Machines to Cybernetics of Resources}

\label{chap:cyber}

The engineering community has achieved great technological advancement by application of the concepts of control and communication to machines and technical systems. The philosophy of control and communication is however quite general and is manifested not only in technical systems but can also be observed in living organisms, society and the object of interest for this dissertation: natural resource systems. Natural resource systems consist of an intricate coupling between natural and social systems and as such, have immense potential to benefit from different lines of research in the ecological, social and engineering communities. The conception of smart cities and the Internet of Things (IOT) have greatly stimulated the interest of engineers towards coupled systems of humans and technology. The growing pervasiveness of such systems has resulted in the emergence of what has come to be known as Human In The Loop (HITL) control. 

In this chapter, after discussing the convoluted relationship between technology and the environment, we introduce the fundamental concepts of cybernetics that have spawned the disciplines of communication and control. We then give an account of how recent progress in technology has led to more and more human-centric realizations of control, including resource related applications. We then consider a cybernetic perspective for NRM, thus framing the research presented in the remaining part of this dissertation. 

\section{Technical Advancement and the Environment}

The Industrial Revolution of the 19th century saw great technological advancement and led to an (at that time) unfathomable improvement in the quality of the common individuals' life. However the subsequent gains in economic prosperity have come at a heavy price which the natural environment has had to pay in the form of increased global emissions, melting glaciers, depleting ozone, deforestation and so on. The complete list of these effects is quite long, well-known and depressing. The 20th century thus saw a growing awareness of the negative impact of technology on the environment and ultimately our collective living condition. Publications such as Forrester's ``World Dynamics" \cite{forrester1971world} and Meadows and colleagues' ``Limits to Growth" \cite{meadows1972limits} also enhanced the view of technology (and the economic growth it enabled) as the factor responsible for our deteriorating habitat. Thus the contemporary environmental movement saw technology as a major source of environmental degradation and a continuous cause for unintended and unwanted side effects.

Interestingly, one of the major objections to the LTG world models is the conceptualization of technology. As Foray and Gr{\"u}bler \cite{foray1996technology} assert, more technology in the future does not necessarily imply more of today's technology or simply more technological ``hardware". Moreover, the temporal scale over which we consume our resources is not determined by their geological abundance, but is determined through technological, social and economical factors which are all subject to a process of continuous change. Technology enables us to positively influence the environment by \emph{1)}: increasing efficiency in production, \emph{2)}: augmenting our natural resources by unearthing new reservoirs, and \emph{3)}: enabling new trends in material and resource use. For example, through technology we have been able to reduce vehicle emissions by manufacturing more efficient automobiles. Greenhouse gas emissions have been reduced by adopting new technologies such as fuel cells \& single-gas turbines and by shifting from coal and petroleum to renewable energy sources.  Geographical Information Systems (GIS) technology has enabled the remote detection of sources for environmental deterioration. Genetically engineered micro-organisms are now used to treat toxic chemicals in waste material. Precision agriculture and the revolution in information technology has greatly enhanced the farmer's capability to effectively monitor, treat and make better decisions about crop production. It has eventually been realized now that feasible solutions to our environmental problems call for more proficiency in technology and not less \cite{ausubel1989technology}. Ironically, technology has demonstrated a capacity to both inflict and remedy damage to the environment \cite{gray1989paradox} and as such can be perceived as a ``double edged sword".

The potential of technology in redeeming the environment is not limited to the three roles mentioned above. An increased emphasis is now being placed on the fact that systemic changes in the economical system are required. For example, the ``circular economy" \cite{schlesinger2014technology} promotes a regenerative economy in which consumer goods are designed to be of such nature that  after the end of their useful life cycle, certain components could be restored and made to re-enter the biosphere instead of being wasted. This is in contrast to the current make-dispose-waste linear model of the economy. 
Such a thing is now possible due to the availability of new materials and sophisticated manufacturing techniques.

Although the effectiveness of technological solutions to environmental problems may seem promising at first, the practical world is full of examples of unintended consequences and, as such, caution must be practiced at all times. As Moriarty and Honnery assert \cite{moriarty2015reliance}, technical advances are sometimes responsible in reinforcing the same effects that they were brought about to mitigate. For example, increasing forested area in an attempt to reduce global warming, risks reducing the Earth's capacity to reflect part of the incident solar energy back into space \cite{keller2014potential}. Thus the Earth would absorb more solar energy, partly offsetting the intended reduction in global warming. An apparent technical solution to one problem may also produce unwanted effects in another domain. For instance, the increased production of biofuels is forecasted to compete with existing agricultural goods which can result in a rise in food prices. Indeed the use of corn to produce ethanol has occasionally caused the global price of corn to spike in the past \cite{moriarty2015reliance}.

The unintended consequences of forestation and bioenergy are similar to those of the Four-Pest campaign and the American wildfire policy we saw in Section \ref{sec:chans_sa}. This reinforces the notion that any intervention, including those technological in nature, in coupled human and natural systems must not be made without fully understanding the system in context of the issue at hand. Such an understanding requires a systems perspective. Due to the potential impact of technology in environmental problems, the engineering community is now showing an increased interest in CHANS. It is thus natural to extend their traditional way of approaching systems (which, up until recently, have been mostly technical in nature) to such settings. Here we are referring to Cybernetics, an old paradigm of Systems Theory that has shaped the foundations of communication \& control.

\section{Cybernetics and Concepts Defining Feedback Control}

Rosenblueth, Wiener and Bigelow \cite{rosenblueth1943behavior} present a classification of system behavior in an essay published in 1943. Behavior is defined as ``any change of an entity with respect to its surroundings". Interestingly, the authors use the term \emph{behavioristic study} to describe the method of study which is akin to the synthetic method. The behavioristic approach is contrasted with \emph{functional analysis}, which as the term suggests is akin to what we have dubbed simply analysis. Behavior can be classified as active or passive. Active behavior is one where the system is the source of output energy generated as a consequence of its operation. In passive behavior either no output energy is produced to begin with, or energy flows from the system input(s) directly to its output(s). Active behavior is further divided into purposeful and non-purposeful behavior, where purpose denotes the existence of a specific objective or goal for the system. Purposeful behavior is categorized as either feedback (teleological) or non-feedback (non-teleological) with feedback including both negative and positive feedback. Rosenblueth et. al. state that all purposeful behavior can be considered as requiring negative feedback. Feedback may be extrapolative (or predictive) so that when a falcon dives towards a rabbit, its aim is not towards where the rabbit is at the time the dive is executed but where the rabbit is conjectured to be when the falcon will reach it. Feedback may also be non-extrapolative (or non-predictive) just as the amoeba is able to move towards light but is not able to track a moving source. Thus purposeful feedback behavior can be further divided into predictive and non-predictive behavior. The order of predictive behavior denotes the complexity of prediction, for instance, the falcon needed only to predict the path of the rabbit (hence first order prediction), whereas an archer hunting the same rabbit would have to predict both the paths of the rabbit and the arrow (thus second order prediction). The classification of behavior just described above is depicted in Figure \ref{fig:behavior}.

\begin{figure}[h!]
	\captionsetup{font=normal,width=0.8\textwidth}
	\begin{center}
		\includegraphics[width=0.8\linewidth]{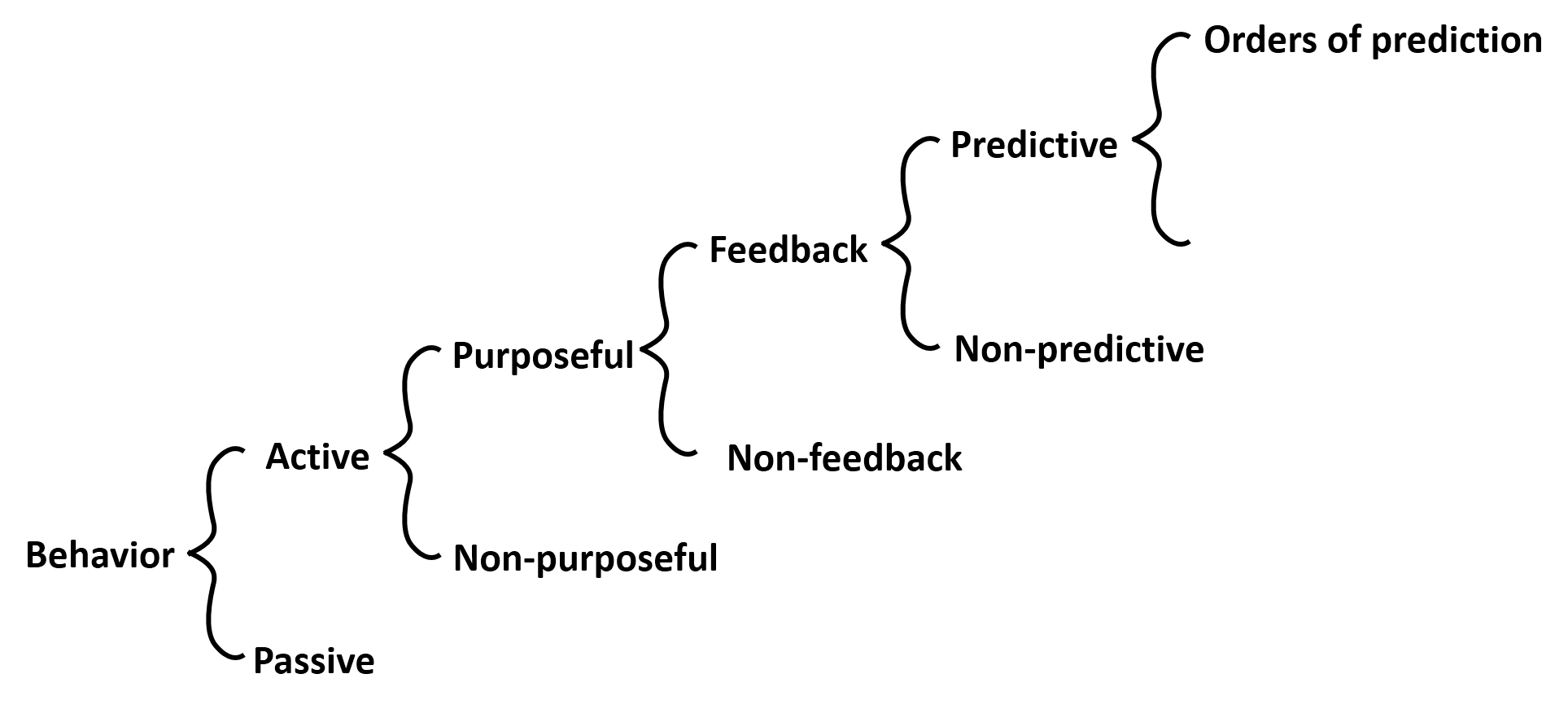}
	\end{center}
	\caption{A classification of behavior (reproduced from \cite{rosenblueth1943behavior}).} 
	\label{fig:behavior} 
\end{figure}

The postulate that all purposeful behavior requires feedback lays the foundation of Norbert Wiener's \emph{Cybernetics}. Weiner introduced the term in his acclaimed book \cite{wiener1961cybernetics} in 1948, a couple of years before Bertalanffy published his paper on General Systems Theory \cite{von1950outline}. The cybernetics movement has roots that are quite distinct from those of systems theory. While Bertalanffy's ideas stemmed from biology, cybernetics took inspiration from technology, using the feedback circuit as its basic model. However there exists a significant overlap in the interests of both movements. Both challenge the purely analytical methodology of the pre-systems era and stress the importance of viewing systems from a holistic angle. Although both were developed in parallel, cybernetics became a part of systems science shortly after its inception. On a related note, systems thinking has given rise to a number of traditions each of which are now established theories in their own right. Early examples include Operations Research, Systems Dynamics, Organizational Learning and Cybernetics (refer to \cite{umpleby1999origins} for further details). More recent developments include (but are not limited to) Control and Communication Theory, Signals and Systems, Multi-agent systems and Network Science. All the above can be thought of as different traditions of systems theory, each having its own (often overlapping) language, methodology and outlook on the systems perspective. 

Cybernetics stems from the Greek \emph{kubeman} which means ``to steer". Indeed one of the earliest examples of control technology is the steering mechanism of a ship's rudder. In his published works \cite{wiener1961cybernetics,wiener1954human} Wiener formalized the concept of cybernetics by abstracting the principles of feedback, control and communication from machines to systems in general. These systems include living organisms, society, organizations and of course machines. Through his discourse, Weiner was able to show that the concepts of feedback control and communication were applicable to all the above mentioned systems. Thus cybernetics was able to unify seemingly disparate disciplines such as psychology, anthropology, sociology and engineering, all of which were (and continue to be) subject to an increasing amount of specialization.

Cybernetics is concerned with the study of all systems which exhibit purposeful (or teleological) behavior. The purpose of the system is specified through some objective or goal that must be reached. This objective is usually translated to a certain state of the system output(s) with the desired state being called the reference. The output of the system is then compared with the reference in order to apply a corrective action depending on how far the system is in reaching its goal. Thus the basic ingredients of a purposeful system are

\begin{figure}[t!]
	\captionsetup{font=normal,width=0.8\textwidth}
	\begin{center}
		\includegraphics[width=0.8\linewidth]{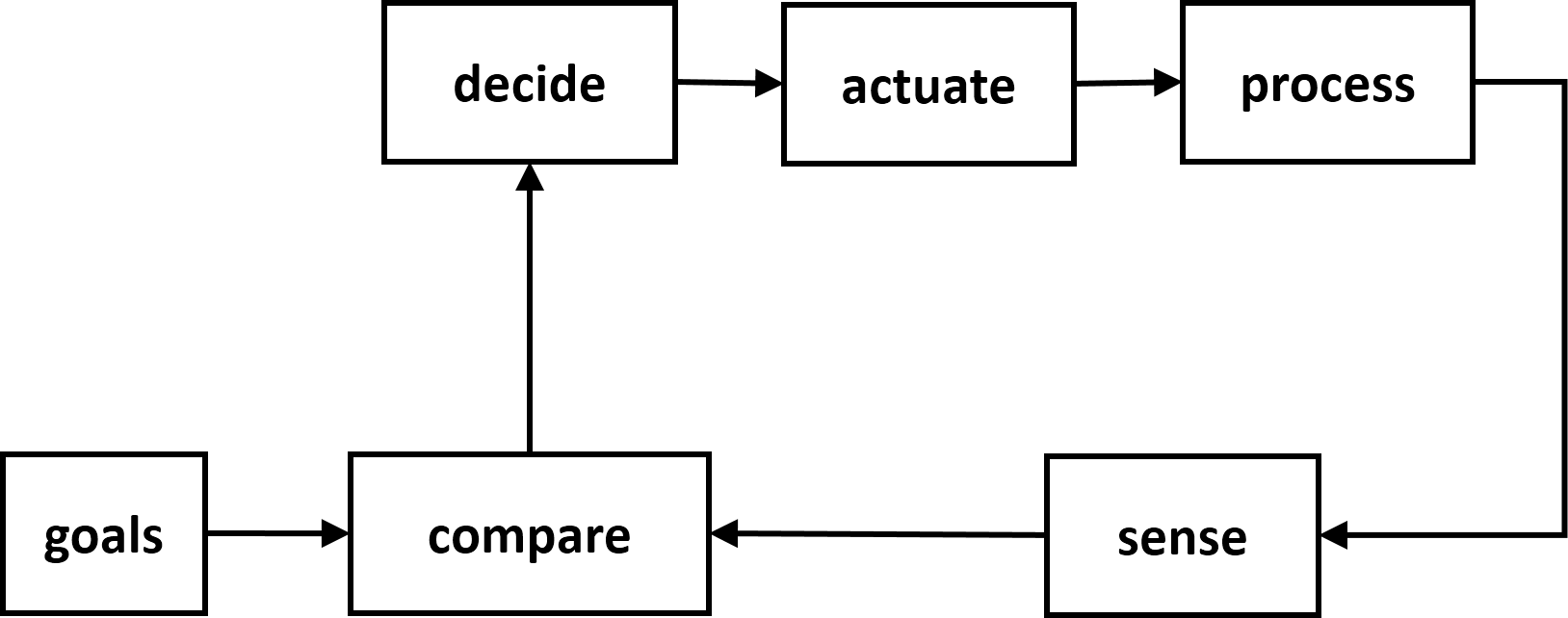}
	\end{center}
	\caption{The feedback control loop.} 
	\label{fig:fb_loop} 
\end{figure}

\begin{enumerate}
	\item The {\bf goal-setter} which determines the system objective and translates it into a desired state for the system output. This corresponds to the ``goals" block of the feedback loop in Figure \ref{fig:fb_loop}. 
	\item The {\bf receptor} which receives information on the output of the system and forwards it for comparison with the desired output. This corresponds to the ``sense" block in Figure \ref{fig:fb_loop}.
	\item The {\bf comparator} which compares the desired output specified by the goal-setter with the output as received from the receptor and formulates a difference or error. This corresponds to the ``compare" block in Figure \ref{fig:fb_loop}.
	\item The {\bf control unit} which decides a corrective action to be taken in order to reduce the error. This corresponds to ``decide" block in Figure \ref{fig:fb_loop}.
	\item The {\bf effector} which implements the action determined by the control unit. This corresponds to the ``actuate" block of Figure \ref{fig:fb_loop}.
\end{enumerate}
In graphical representations of the feedback loop, the ``compare" and ``goals" blocks are sometimes omitted (like Figure \ref{fig:cybernet_pic}). However their presence is always understood.

This type of corrective behavior can be seen, as Wiener posits, in every purposeful system. A person driving a car will continuously read the deviation of the car from the center of the road and apply the necessary corrective action through the steering wheel in order to keep the car in the center. Two people engaged in conversation also exhibit the same behavior. The conversation has a purpose and if one of the individuals deviates from that purpose, someone applies a corrective action by steering the conversation in the required direction. Living organisms exhibit this behavior in the process of homeostasis, which maintains the parameters of the body within a nominal range of values. Needless to say, the principles of cybernetics are at play all around us.

The concepts of cybernetics, as a paradigm of Systems Analysis, has had a profound influence in the development of both control and communication \cite{gao2014engineering}. Over the passage of time, although the term cybernetics has gone out of fashion, it continues to live on in the sense that the discourse in many diverse fields are still being broken down into systems, information, feedback and networks. As Asaro \cite{asaro2010ever} puts it, ``... cybernetics is alive and well. Of course, nobody calls it cybernetics anymore, they simply call it control theory or information theory." However, as Asaro himself also notes, the original ideology of cybernetics is much broad and is limited not just to machines. As technology has advanced, the proliferation of humans in control applications continues to grow. In what follows we look at how control theorists have incorporated humans, society and the environment in their cybernetic formalism of systems. The umbrella term for such settings is called Human In The Loop (HITL) control. The discourse is presented to build towards how CPR systems may be perceived from the cybernetic approach.

\section{Human-in-the-loop Control}

The involvement of human beings in control systems can be broadly categorized into the following levels \cite{wilfrid2017social}

\begin{enumerate}
	\item {\bf Human machine symbiosis:} Symbiosis is a state of two different organisms living together in a close interactive manner. Examples of such systems include prosthetics, exoskeletons, neurostimulators, etc.
	\item {\bf Human as operators or decision makers:} These systems include those in which regular intervention of a human operator may be expected e.g., self-driving vehicles, remote control applications, auto-pilot systems. They also include systems in which humans are involved in the high level decision making or goal setting such as town municipalities, supply chains and resource management.
	\item{\bf Humans as agents in multi-agent systems:} Such systems include automated transportation and pedestrian management, semi autonomous manufacturing, human vs machine games and so on.
	\item{\bf Humans as the objects of control:} These systems include home comfort systems, networks of opinion formation and social media to name just a few.
\end{enumerate}

Resource management systems are special in the sense that they involve human beings at two levels of involvement. Firstly, the humans are the objects of control as it is their consumption which must be regulated. Secondly, humans are involved in the decision making process of any control that may be subsequently applied. This may include setting discount rates, devising effective means to implement the control, engaging in philosophical exercises of formulating social welfare and so on.

While economists, ecologists and sociologists have all been working on coupled human natural systems for a significant time now, the engineering community (in particular the control faction) have recently gained interest in these systems through the increasing integration of technology in society and the opportunities that it offers. A small outline of the journey is given below.  

\subsection{Cyber-Physical Systems and the Internet of Things}
 
Cyber Physical Systems (or CPSs for short) pertain to the application of computing systems to the control of physical processes. The overall behavior of the system is then determined by both the computing element and the physical element. Early applications of engineered systems for the control of intricate physical processes may be traced back to the steam engine governor and the industrial revolution initiated by it. In recent decades the advancement in computing and communication technologies have given rise to the Information Technology (IT) revolution. This has resulted in a proliferation of computing elements and the world wide web at the core of control design for complex physical processes. The term Cyber-Physical Systems, coined by Helen Gill \cite{baheti2011cyber} in 2006 at the National Science Foundation, thus lies at the interface between control, communication and computing technology. Kim and Kumar \cite{kim2012cyber} identify multiple pathways leading to this area of interest via multiple scientific communities which has also led to the birth of closely related terms such as Ubiquitous Computing, Wireless Sensor Networks and the Internet of Things. The root of the term CPS is, widely and befittingly, considered to be Wiener's Cybernetics \cite{lee2015past}. 

The Internet of Things (or IoT for short) relates to the extension of the internet and its integration in the physical realm \cite{miorandi2012internet}. The notion of IoT was conceived by the computer science community and envisions a world of ``smarts" i.e., smart transport, smart phones, smart infrastructure\footnote{see Elsevier's special issue on smart infrastructures \cite{elsevier2017smart}} and smart cities \cite{salim2015urban}. The IoT concept shares many similarities with CPS and both terms are often used as synonyms \cite{bojanova2014imagineering}. Both envision a strong coupling of computation, intelligence and the environment in the form of a large scale distributed computing system of systems. However both have evolved from different communities, and as such, there exist subtle differences between the philosophy of both. The IoT focuses more on virtual problems such as architecture, security, protocols, openness and privacy \cite{stankovic2014research}. CPS however places more focus on control, automation and the underlying physical process, often formulated as a closed loop control system. While IoT consists of the internet at its core technology, CPS consists of intelligent embedded systems in general. Thus IoT may be viewed as an enabling technology for CPS. The IoT has also be portrayed as a CPS connected to the internet \cite{jazdi2014cyber}. 

The application of CPS based technologies are numerous, ranging from robotics, factory automation and smart grids, to medical systems, citizen actuation and environmental control. Such diverse and convoluted applications pose many challenges for research in theories for analysis and design, integration of cross-disciplinary works, developing technologies for control actuation and sensing, meeting real-time computing requirements and abstracting the complexities of physical systems to name just a few \cite{baheti2011cyber, kim2012cyber, lee2015past, rajkumar2010cyber}.

\subsection{Cyber Physical Social Systems: Addressing the social component in CPSs}

Cyber-Physical Social Systems (or CPSS) brings together the cyber, physical and social components of a controlled process. The social dimension of CPSSs is rather broad and encompasses not only human societies per se, but also mental capabilities and human knowledge. Thus one realizes that all CPSs include a social dimension in one form or another. Robots do not operate in their own world, but in the physical world where interaction with humans is no longer a possibility but a necessity. Auto pilot systems and self-driving vehicles rely on high level inputs from a human operator. Even highly automated control systems which may seem to operate without any human intervention are frequently effected by human behavior, either in the form of disturbances, or high level decision making and goal setting. The goal of CPSSs is to effectively add and address this human component in CPSs \cite{wang2010emergence}. The fundamental difference then between a CPS and a CPSS is that in the latter humans are modeled as fundamental units of the system and not objects that exist outside the system boundaries. While CPS and IoT make connections across the physical space  and the cyber space, CPSSs connect the physical and cyber spaces with the social space. This space not only includes humans in their physical form, but also includes their cognitive behavior, social interactions and collective institutions \footnote{A parallel concept is the extension of IoT to the Internet of Everything (IoE) which goes beyond just ``things" to incorporate humans and society in what is envisioned to be a global sensing and actuating utility connected to the internet \cite{yang2017internet}.}. It has been argued that since acknowledging human behavior is a key component of CPS design, a CPS should infact be more accurately referred to as a CPSS \cite{cassandras2016smart}.

\subsection{The Imperative to Understand Human Behavior}

\label{sec:cps_challenges}

As social applications of control and automation increase with the advancement of technology, we are faced with many new challenges and research opportunities. Stankovic et. al. \cite{stankovic2014research, munir2013cyber} identify the following CPS challenges for human in the loop control of CPSs

\begin{enumerate}
	\item {\bf Challenge 1: } Gaining an understanding of the complete spectrum of human-in-the-loop control systems.
	\item {\bf Challenge 2: } Extending system identification techniques to adequately capture human behavior.
	\item {\bf Challenge 3: } Incorporating human behavior models into the formal methodology of feedback control.
\end{enumerate}

The common aspect in all three challenges listed above is a focus on understanding human behavior in the formalism of control. Once it is understood that a system of concern is HITL, the first challenge is to model the involvement of human beings. As stated previously, it is absolutely essential that the social elements are modeled not as external disturbances or noise but as primary components of the system. Using the formalisms that exist separately in the social, technical and computational domains is not sufficient to deal with these systems as each model highlights only certain aspects relevant to that domain and ignores others. Such an approach to CPS design may prove to be fatal for verifying the overall correctness and safety of design at the system level \cite{baheti2011cyber}. Thus an integrative approach is required which combines the research from social sciences to engineering in a formalism that is coherent with the language of feedback control.

\section{The Cybernetics of Natural Resource Systems}

In Chapter \ref{chap:sa}, we discussed the development of systems analysis and its application to the field of natural resource management. Despite the aforementioned advancements however, NRM has yet to benefit from the full potential of the cybernetic perspective. Accordingly there have been attempts to draw attention to the application of the cybernetics philosophy to environmental problems \cite{hoffman2010cybernetic, ben2012cybernetics}, but they have largely gone unnoticed. Cybernetics focuses on the key question of how systems regulate themselves, how they evolve and adapt to disturbances and the underlying structure and mechanisms that govern their dynamics. Moreover, the framework is mathematically intensive and, as such, permits rigorous deductions and confirmation (or refusal) of theory. The focus of the cybernetic approach on regulating system outputs and states to desired values uncovers profound insights on the required system structure and behavior in order to yield the desired results. Needless to say, NRM can benefit greatly from the cybernetics perspective especially due to a growing interest in the field shown by the control and automation community.

The second aspect of the relevance of cybernetics to natural resource management lies with the controls community itself. The involvement of human beings in control and automation systems is growing rapidly with new applications being realized everyday. Such applications have been made possible through the advancement of technology. New and improved sensors, actuators and the evolution of design techniques has transformed the concept of smart cities and IoE from mere aspirations to a reality within reach. Many researchers have pointed out that the missing link in the way forward is the seamless integration of human behavior with the technical aspects of system design \cite{lamnabhi2017systems}. One of the most important factors responsible for the existent gap between the social sciences and control and automation is the dearth of mathematical models describing social dynamics \cite{proskurnikov2017tutorial}. While significant progress is being made in control of opinion formation and related problems \cite{friedkin2015problem}, much ground has yet to be covered in control and automation for natural resource management. Caragliu et. al. \cite{caragliu2011smart} define a smart city as one where ``investments in human and social capital and traditional (transport) and modern (ICT) communication infrastructure fuel sustainable economic growth and a high quality of life, with a wise management of natural resources, through participatory governance". Thus a smart city cannot be realized without effective management of its natural resource base. Natural resources are coupled human-natural systems and in order to integrate them effectively in the formalism of feedback control, we must first translate their dynamics in a form that is amenable to  to the cybernetic way of thinking.

\subsection{Imagineering the Cybernetic Picture}

Cybernetics concerns itself with systems as regulators that actively compensating for any disturbances that might occur. As such, while applying cybernetics to any system, we must address the following questions: What is the goal or objective of the system in context of the analysis? What factors and decisions influence these goals? What information is required by the analyst (or observer) and decision makers? What technologies are needed to implement the system? In the following, we address the above and similar questions in the context of NRM. A synopsis is given by Figure \ref{fig:cybernet_pic}.

\subsubsection{Defining System Goals}

The research enclosed in this dissertation strives towards uncovering favorable conditions for sustainable development. This relates to the setting of goals for the system and although the ``goals" block is not shown explicitly in Figure \ref{fig:cybernet_pic}, it is a crucial part in the overall control implementation. Here we adopt the following definition of sustainability, contained in the well known report of the Brundtland Commission titled ``Our Common Future" \cite{brundtland1987report}

\begin{quotation}
``Sustainable development is development that meets the needs of the present without compromising the ability of future generations to meet their own needs"
\end{quotation}

The above definition requires further specification of two elements. First we must specify what is meant by needs. In the context of NRM we interpret needs as consumption from the natural resource stock. Second we interpret the ability to meet ones needs as the availability of the resource stock for consumption. In the mathematical language which we employ to frame the system model, this implies a steady-state with finite resource stock and consumption level (see \cite{valente2005sustainable, beltratti1993sustainable} for similar examples).

\subsubsection{Factors and Decisions Influencing Achievement of the Goals}

The sustainability of natural resources is threatened by the overexploitation of its stocks through phenomena similar to the TOC. Since sustainability of the stock implies an exploitation that is at most equal to the resource regeneration, it is effected both by resource and consumption dynamics. Resource dynamics for renewable resources are specified by environmental parameters such as carrying capacity and growth rate \cite{perman2003natural}. Dynamics of human consumption are specified by the cognitive decision-making process of resource use \cite{mosler2003integrating}. It is important to realize that there exist a plethora of variables and parameters that effect consumption patterns in the real world. However it is not possible to represent all influencing factors and keep the analysis tractable at the same time. Therefore one must select only those factors which are most relevant to the underlying study. This pertains to understanding of the physical process in context of the problem at hand (and therefore the social and environmental factors are depicted alongside the social-ecological system in Figure \ref{fig:cybernet_pic}). Moreover, as discussed previously in Section \ref{sec:chans_sa} for CHANS, it is important to not only consider the ecological and social components separately, but also to consider the various interactions between them. 
 
\begin{figure}[t!]
	\captionsetup{font=normal,width=0.9\textwidth}
	\begin{center}
		\includegraphics[width=\linewidth]{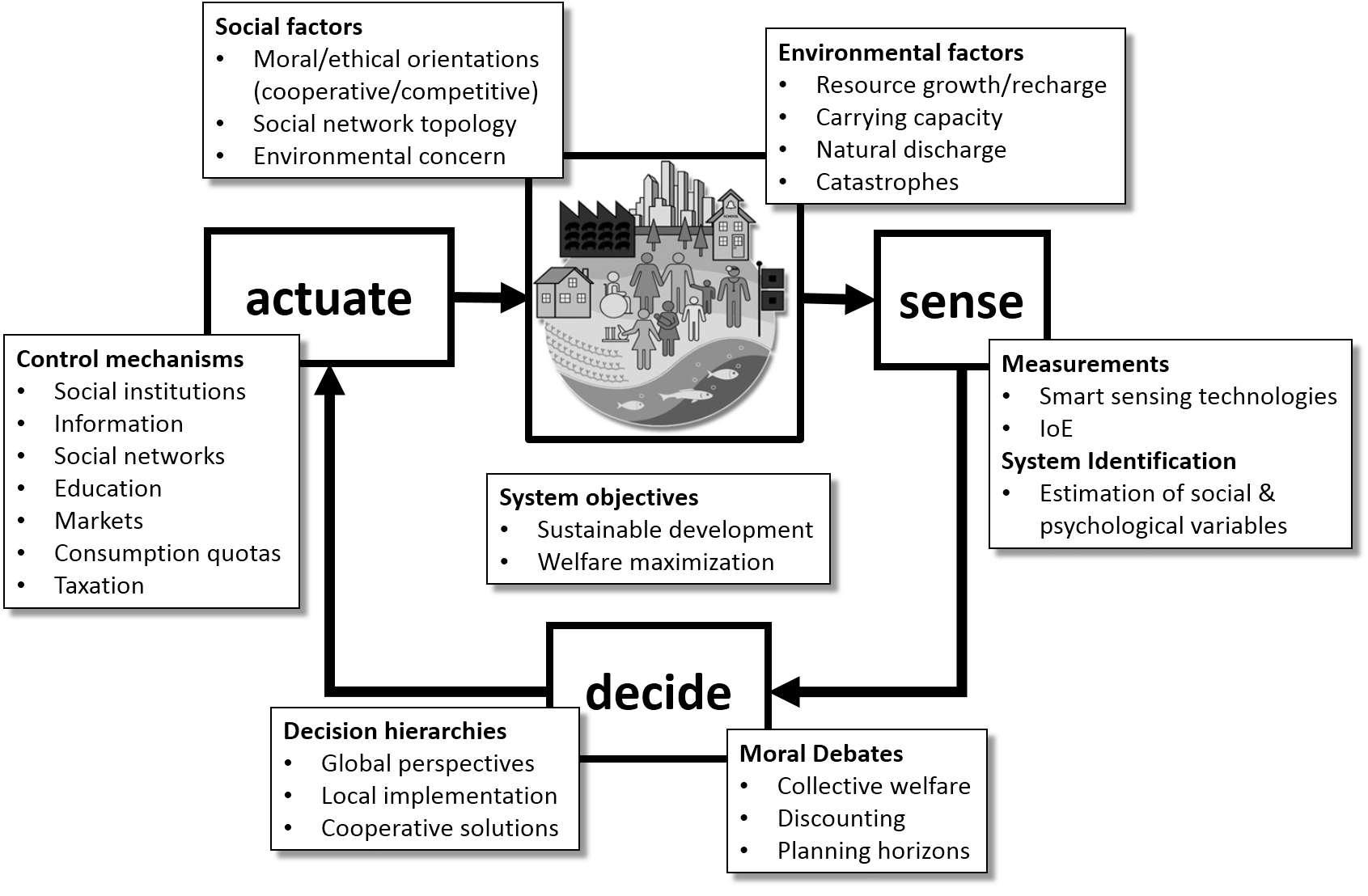}
	\end{center}
	\caption{The cybernetic perspective of NRM. The boldly outlined blocks represent the basic components of the feedback control loop. The other blocks list the relevant factors and possible realizations for the different components as discussed in the text.} 
	\label{fig:cybernet_pic} 
\end{figure}

\subsubsection{Mathematical Models of Human Behavior}

As stated previously perhaps the factor most responsible for the gap between research in social sciences and engineering is the lack of mathematical models of human behavior. Depending on the type of analysis, human behavior may be incorporated in objective functions in optimal control problems \cite{acemoglu2008introduction} or in the system dynamics themselves \cite{anderies2000modeling}. Depending on the scale of the problem, the model may be a block model (representing tightly knit communities as a single entity) \cite{friedkin2006structural} or an agent-based model where the dynamics of each individual consumer are represented explicitly \cite{roopnarine2013ecology}. Note that obtaining mathematical models of human behaviors is not easy, in fact, this is exactly one of the challenges for HITL control of CPSs we saw in Section \ref{sec:cps_challenges}. Nonetheless, a mathematical model (whether conceptual or data driven) is essential for the rigorous framework of control.

It is important to note here, a common misconception about modeling i.e., the goal of modeling is always prediction. This is a tendency that is commonly found in both technical and social disciplines; if a model has been made, then it is eventually going to be used to make projections, or at least be ``calibrated" or ``validated" using some data. Epstein  treats these enduring misconceptions in his acclaimed lecture \cite{epstein2008model} quite judiciously. Aside from prediction, he lists 16 different uses to build a model. Any model can in principle be used for prediction, but it is important to recognize the limitations of the model and the resulting influence it has on the validity of the predictions. We do not  advocate the use of our model (see Chapter \ref{chap:model}) to predict resource consumption patterns in actual scenarios, as this is not the intent with which the model is created. Indeed there exist many models that have been created for this exact purpose (for instance, see the discussion in \cite{bots2008participatory, brugnach2008broadened}). We hope to demonstrate in the material that follows, the use of our model as a potential tool to illuminate core dynamics of the system, highlight trade-offs in decision making, and most importantly, take a step towards incorporating CHANS dynamics effectively in control processes. The findings of the model might be also be used to guide the collection of data and even discipline policy dialog for natural resource management.

\subsubsection{Identifying the Mechanisms of Control}

There exist many forms of control for NRM. While some are relatively straight-forward to identify, others may not be so due to the complex nature of social dynamics and its intricate coupling with the environment. Examples include devising social institutions of collective behavior \cite{ostrom1990governing}, manipulating dissemination of information to alter consumption patterns \cite{stephan2002environmental}, government interventions in local markets \cite{fligstein2002architecture}, and influencing the social network so that consumers are connected in some desirable way \cite{bodin2011social}. There also exists long-term measures such as educating the masses to nurture a certain type of mindset \cite{arbuthnot1977roles} and more coercive methods such as enforcing quotas, steering prices or imposing taxes \cite{hardin1968tragedy}. The main challenge of course is to incorporate these controls in models of human behavior. Control mechanisms fall in the ``actuate" block of Figure \ref{fig:cybernet_pic} which also lists the aforementioned mechanisms of control.

\subsubsection{Measurements and System Identification}

In order to evaluate performance, we must have some representation of the current state of the system. This is achieved by the observation of relevant outputs of the process being controlled. Thus the ``decide" block (Figure \ref{fig:cybernet_pic}) is preceded by a ``sense" block in the overall feedback loop of the control system. Measurement and sensing technologies for physical quantities such as resource stocks and individual consumption have matured well beyond their infancy. So much so that we are now living in a world of ``smarts" \cite{bohn2004living}. Portable hand-held devices such as smart phones contain dozens of sensors which include the capability to log and keep track of behavioral patterns of each individual. Smart meters record and send data of resource stock and flows from extremely remote locations with time resolutions that far exceed the time scales of environmental processes \cite{soloman2009sensors}. The bottleneck in sensing and system identification for NRM then lies in the parameters concerning human psychology such as environmental awareness and social orientations \cite{munir2013cyber}. While sociologists are able to estimate such psychological data in controlled experiments \cite{rutte1987scarcity}, we are still much behind in devising system identification techniques for such parameters in real world CPSSs.

\subsubsection{Evaluating System Performance}

Evaluating system performance is critical for regulation and control. While obtaining benchmarks for machines is relatively straight forward, it is not the case for NRM. Evaluating how well a society is doing requires intense philosophical and moral dialogue, for example, in establishing notions of social welfare and selecting discount rates for future generations \cite{sen1997choice}. Such issues are heavily debated in the economics literature \cite{baujard2015economic, perman2003natural} and must be addressed before any type of social control is envisioned. In the overall control system this step may be seen as part of the decide process as shown in Figure \ref{fig:cybernet_pic}, which implicitly includes the ``compare" block from Figure \ref{fig:fb_loop}.

\section{Cybernetics and Society: A Word of Caution}

In his original text \cite{wiener1961cybernetics}, Norbert Weiner advises caution in the use of the cybernetics philosophy for social systems. His major concern is the non-isolation of the observer from the system. While in machines, the observer is external and is able to investigate the system without disturbing it, it is not so in the case of social systems where the observer is part of the system she analyzes (for instance, any attempt to investigate the stock exchange has the potential to upset the stock exchange itself). Thus the principles governing the control and communication of machines (subsequently called first-order cybernetics) are different from those of society (subsequently called second-order cybernetics). Weiner further goes on to state that this approach of investigations in the social sciences can never furnish us with a quantity of verifiable, significant information which can even be compared with what we have learned to expect in the natural sciences. 

Although philosophical advancements in second-order cybernetics have gained significant ground \cite{von2003cybernetics}, the cybernetics of social systems (also called socio-cybernetics\cite{mancilla2013introduction}) is far from yielding the rigorous theory of control and communication that has spawned from its first order antecedent. Heylighen and Joslyn  \cite{heylighen2001cybernetics} go as far to state that ``the emphasis on the irreducible complexity of the various system-observer interactions and on the subjectivity of modeling has led many to abandon formal approaches and mathematical modeling altogether, limiting themselves to philosophical or literary discourses".

In the material of the following chapters, a model for natural resource consumption is derived from psychological roots and formulated as a dynamical system. Keeping in mind Wiener's cautionary remark, the formulation of the model and the analysis that follows has been done by placing a focus on qualitative results. Accordingly, an emphasis has been placed on the relationship between variables and uncovering sustainable trends in the system. The association of numerics with social and psychological parameters has been kept to a minimum, and where possible, avoided altogether. We hope that the research enclosed in this dissertation will serve as a small step and an inspiration towards a fully mathematical framework for NRM in the cybernetic context.

\part{Modeling the Dynamics of Resource Based Systems}
\clearpage
%% This is an example first chapter.  You should put chapter/appendix that you
%% write into a separate file, and add a line \include{yourfilename} to
%% main.tex, where `yourfilename.tex' is the name of the chapter/appendix file.
%% You can process specific files by typing their names in at the 
%% \files=
%% prompt when you run the file main.tex through LaTeX.

\chapter{A Psychologically Relevant Mathematical Model of Consumer Behavior}

\label{chap:model}

In Chapter \ref{chap:cyber}, we discussed the relevance of cybernetics to natural resource systems and how it may provide a common ground for engineers to analyze and study the various societal and environmental applications of their technologies. Conant and Ashby \cite{conant1970every} argue that a mathematical model of any system is essential in order to achieve any purposeful behavior on its part. Whether the model is conceptual or data-driven (see \cite{schmidhuber2015deep} for an advanced application of empirical models), it is always understood to exist at the core of any regulating mechanism. In this chapter, we introduce a mathematical model of resource consumption, derived conceptually from social psychological research on consumer behavior. After discussing the first principles used to derive the model, we present the model in mathematical form, followed by a reduction in dimensionality for the sake of ease in analysis and interpretations carried out in the subsequent chapters of this dissertation.

\section{Modeling Paradigms in Socio-ecological Systems} 

Understanding and taking into account consumer behavior in socio-ecological systems is of prime importance for designing an effective solution to phenomena, which are usually referred to as “The Tragedy of the Commons” \cite{hardin1968tragedy}. Social psychologists have carried out a wealth of research to identify factors driving the decision processes of individuals. For instance, Festinger \cite{festinger1954theory} hypothesizes, that human beings are continuously driven to evaluate their decisions, and, in the absence of any objective non-social means, they do so by comparison with decisions of other individuals. Thus, if a consumer is uncertain about the state of a resource, she may make her decision based on the information gathered from other consumers. Simon postulates \cite{simon1959administrative}, that individuals seldom optimize their outcomes over different available alternatives, mainly due to their limited cognitive resources, and instead act according to automated habitual behavior. This implies, that although a consumer might be aware of the consequences of over-utilization of a natural resource, she may still do so according to habit, as long as the outcome is satisfying to her for the time-being (assuming that there is enough of the resource to permit her over-utilization). 

These theoretical propositions have been supplemented by various laboratory and field experiments of common-pool resource settings. Samuelson et al. \cite{samuelson1984individual} conduct an experiment to observe how consumers respond to information on the overall consumption of the resource. They find that individual consumptions tend to increase over time, however there is little or no increase in consumption when the resource is being overused. Rutte et al. \cite{rutte1987scarcity} demonstrate, through an experiment, that consumer behavior is determined by whether the society or the environment is held responsible for the scarcity or abundance of resource. They conclude that when the environment is held responsible, consumers give ecological information more preference than social information, for evaluation of their consumption levels. Conversely, when the society is held responsible, social information is given more preference than ecological information. Brucks and Mosler \cite{brucks2011information} conduct an experiment to find what type of information is important to consumers while making decisions on consumption. They observe, among other findings, that when the resource availability is low, the importance of the information regarding the situation of the resource increases, while the importance of the information regarding consumption of other individuals decreases.

Building on the aforementioned psychological research, scientists have proposed various computational models that simulate consumer behavior during resource crises. Much of these models have been developed under the umbrella of “multi-agent simulations”, or “agent-based modeling” (these terms are often used interchangeably), in which several heterogeneous individuals (agents) are programmed according to some psychological rules. Such a framework can be used to identify key behavioral elements that shape large-scale societal transitions. Deadman \cite{deadman1999modelling} presents one of the first agent-based models of the tragedy of the commons. Since then, various computational agent-based models of resource dilemmas have been developed – see, for instance, \cite{jager2000behaviour, feuillette2003sinuse, jager2007simulating}. These models differ in the principles used to simulate the behavior of individual agents, as each group of authors, depending upon the posed research question, attempts to incorporate different psychological findings into the behavior rules of agents. Bosquet and Le Page \cite{bousquet2004multi} give an extensive review of the development of this field and early applications to ecosystem management.

Despite the power of computational agent-based models to reveal collective phenomena from individual interactions, they lack rigor, generality and elegance that is provided by mathematical models. Notable achievements on the latter front include, for example, Anderson \cite{anderson1974model} who describes the coupled dynamics of a regenerating commons and the physical capital, as a system of a few ordinary differential equations, and uses calculus to show that the model exhibits Hardin’s conclusions (Hardin 1968) explicitly. Anderies \cite{anderies2000modeling} applies bifurcation analysis to two separate models of resource exploitation: (1) the slash-and-burn agricultural system of the Tsembaga tribesmen of New Guinea, and (2) the exploitation of palm forests on Easter Island by the Polynesians. The analysis is used to study the long-term behavior of both systems under different management/ behavioral regimes. The study suggests that successful institutional designs are highly site-specific and that a careful understanding of the “geometry” of the system is necessary for successful resource governance. Roopnarine \cite{roopnarine2013ecology} suggests simple differential equation models which separately illuminate three different aspects of the tragedy of the commons:  (1) the compulsion of individual users to act based on the action of other users, (2) the opportunities and limitations imposed by the networked nature of consumers, and (3) mutualisms as solutions to tragedies. Thus, these mathematical models are able to demonstrate and explain a certain phenomenon in a clear and tractable fashion, where a major focus is usually placed on the dynamics of consumption. Although no objective function or decision variables are explicitly included in their formulation, they typically assume some notion of rationality prevalent in the consumers. However, the theoretical findings provided by social psychologists are seldom incorporated in these models. 

The commons dilemma has also been studied extensively in the framework of game theory. In her seminal work, Ostrom states that all institutional arrangements that can be used to avert the commons dilemma are expressible as games in extensive form, and presents several such examples \cite{ostrom1990governing, ostrom1994rules}. While she expresses institutionalization for resource governance as a non-cooperative game, a large number of analysts have also used cooperative game theory to study environmental resource issues (see, for example, \cite{parrachino2006cooperative}). Ostrom’s research has stimulated a vast amount of work, concerning use of the game theoretic framework for modelling strategic interactions between consumers of a natural resource (see for instance \cite{madani2010game, cole2010institutions, diekert2012tragedy}). Similar to system dynamic models, game-theoretic models also assume a notion of rationality, which is made explicit by including objective functions and decision variables in the problem formulation. However, most game-theoretic models, either seek an optimal strategy of resource extraction, or seek to design models that yield strategies similar to those observed in real-world scenarios, rather than focusing on strategies that accurately depict the cognitive process of the consumers’ decision making. We argue, that in order to increase the relevance of such models to real-world scenarios, it is necessary to incorporate the cognitive principles that govern consumer behaviour, as revealed by the psychological research.

In this chapter, we present a dynamic model of natural resource consumption, taking into account the psychology of consumer behavior in an open-access setting. We depart from the computational model by Mosler and Brucks \cite{mosler2003integrating} and put forward a stylized version that is based on the same psychological principles employed in the original model. Essential psychological variables have been maintained in our model, which include the scarcity thresholds as perceived by consumers, the extent to which the consumers hold nature (or society) responsible for the state of the resource, and the social value of the consumers. While the original model assumes unlimited growth of the resource, we assume standard logistic growth, which is more realistic. The logistic growth model and its variants are commonly used to model many renewable resources that saturate at a certain carrying capacity. Thus the resources which lie in the scope of our study include (but are not limited to) fisheries \cite{gordon1954economic}, forests \cite{fekedulegn1999parameter}, vegetation \cite{birch1999new}, foliage \cite{werker1997modelling}, saffron \cite{torabi2014evaluating}, and so on. We do not, in this study, incorporate the effects of uncertainty in our model. Thus we assume, that perfect information is available to each consumer, regarding the resource quantity and the consumption of other individuals present in the system. While this is a strong and restrictive assumption, incorporating uncertainty is beyond the scope of this paper, as it would considerably increase the complexity of the model, and must be dealt with separately. 

%Thus, we present the reformulation, which enables formal tractable mathematical treatment of the model. We carry out the steady-state analysis, consider our system in equilibrium and employ the game-theoretic framework to study what conditions lead to the commons problem. Namely, we introduce a non-cooperative continuous-kernel game to analyze the rational decisions of consumers for different combinations of key model parameters describing the resource dynamics and the society. Furthermore, we define a notion of “tragedy” in the commons game, based on the distance between the Nash equilibrium and the Pareto optimum. We then use exhaustive numeric simulations to reveal such trends in the system parameters, which are helpful for decreasing “tragicness” and are also beneficial to the resource stock.

\section{The Mathematical Model}

Our starting point is the computational consumer behavior model of Mosler and Brucks \cite{mosler2003integrating}, which we stylize by formalizing it into a mathematical form. The resource is supposed to regenerate according to the classical logistic growth model \cite{perman2003natural}. To begin with, we consider a society of $\n$ individual consumers, each of whom have open-access to the resource which they consume by exerting some effort. The dynamics of the consumers' effort are modeled based on our understanding of the cognitive process depending on certain psychological characteristics each consumer. These characteristics depict: (1) the environmentalism of the consumer (the perceived benchmark quantity that the consumer uses to evaluate whether the resource is in abundance or in scarcity), (2) the social values of the consumer (how cooperative is the consumer), and (3) causal attributions of the consumer (to what extent does the consumer attribute the current condition of the resource to nature-induced reasons, as compared to society-induced reasons). The model is specified through the dynamics of the resource (the ecological sub-model) and those of the individual consumption efforts for each consumer (the social sub-model). We describe each sub-model in turn as follows.

%To begin with, we consider a society consisting of two consumer groups characterized by significantly different psychological characteristics. Within the groups these characteristics are similar, thus the groups are assumed to be internally homogeneous. This phenomenon of social polarization has been well studied in the past and also observed in real-world settings \cite{chakravarty2009inequality, zwiers2015divided}. Under this assumption it is possible to describe their behaviour through two aggregated equations. 

\subsection{The Ecological Sub-model}

Consider a society with $\n$ consumers, each having open-access to a single renewable resource. We assume that the resource quantity $R(\tau)$, available for consumption by the society at time $\tau$, has an associated growth function, which is logistic in nature. We further assume, that in the absence of consumption, $R(\tau)$ increases over time at an intrinsic (positive, constant) growth rate $\r$ and saturates at the given carrying capacity $\Rmax$ . The individual consumers, each identified by their respective index $i \in \{1,\dots, \n \}$, exert effort $e_i (\tau)$ \footnote{see \cite{perman2003natural} for interpretations of the word ``effort" in this context} to consume the resource. The resource growth is, therefore, given by
\begin{align}
\label{eq:ec_mod} %ecological sub model
	\frac{d R(\tau)}{d \tau} = \r\, R(\tau)\left( 1 - \frac{R(\tau)}{\Rmax} \right) - \sum_{i=1}^{\n} e_i(\tau)\,R(\tau),
\end{align}
where each individual's harvest of the resource equals $e_i (\tau)R(\tau)$. Note that Equation \eqref{eq:ec_mod} is equivalent to the standard Gordon-Schaefer model \cite{gordon1954economic}, with the catch coefficient set to unity.

Following \cite{perman2003natural}, we use a general notion of ``efforts": ``all the different dimensions of harvesting activity can be aggregated into one magnitude called effort". For example, in fisheries, factors which define efforts may include ``the number of boats deployed and their efficiency, the number of days when fishing is undertaken and so on". While positive efforts refer to the extraction of a resource, negative efforts can be, for instance, growing trees or breeding fish, and also beyond these – namely, any effort exerted for the sustenance of the resource. In fisheries, the latter can be related to the restriction of fishing gear, closed area management, awareness programs and so on \cite{worm2009rebuilding}; in forest management it can be related to the restoration of soil fertility, preservation of remnant vegetation, and promoting community forest enterprises \cite{chazdon2003tropical, bray2003mexico}. One can imagine similar realizations for other renewable resources.

An important implication of the aggregate consumption effort being negative, is that the resource may grow over the natural carrying capacity, i.e., $R(\tau)$ may take on values greater than $\Rmax$, which in the context of the model means that the resource grows beyond the carrying capacity. However even then there is still a limit to how much beyond the natural carrying capacity the resource quantity can be stretched. In the case of forests, for example, using fertilizers and planting trees can act as negative extraction efforts leading to the forest biomass growing over the natural carrying capacity. However, the limited amount of light coming onto the Earth, competition for light between trees and simply the space constraint will not allow the biomass to grow unboundedly and thus there will be some limit to it; similar examples can be envisioned for other resources as well. 

\subsection{The Social Sub-model}

Here we define the dynamics of the effort $e_i (\tau)$ over time. We assume that the change in consumption of each individual, is based on their weighing of two different factors, one pertaining to information about the resource quantity (the ecological factor), and second pertaining to information about the use of others (the social factor) \cite{brucks2011information}. Which factor gains precedence over the other in the cognitive process of harvesting decisions depends on the individual characteristics of each consumer, depicted here through multiple psychological variables. The final change in consumption is the sum of the ecological and social impacts, where we define the impact as the product of the corresponding factor and its weight. Thus the change in effort is given by
\begin{align*}
	\frac{d e_i(\tau)}{d \tau}=\text{ecological weight}_i \times \text{ecological factor}_i + \text{social weight}_i \times \text{social factor}_i. 
\end{align*}
In what follows, we specify the ecological \& social factors, and the psychological variables that determine the relative weighing of these factors to determine the change in effort $e_i (\tau)$. The final equation is given by \eqref{eq:cons_mod}.

Rutte et al. conclude from their study \cite{rutte1987scarcity} that consumers harvest more from the resource in abundance than in scarcity. They further observe that harvests for each individual tend to increase over time, except when the resource is scarce, in which case it decreases. We implement this effect in our model by comparing the current resource stock with a perceived-by-the-consumer ``scarcity threshold", denoted as $\R_i \in \mathbb{R}$. Thus we define the ecological factor for $i$ as the difference $R(\tau) - \R_i$. A positive ecological factor represents an abundant state of the resource, whereas a negative factor represents a scarce state. A scarce resource in the perception of individual $i$ results in a decrease in consumption, whereas an abundant resource results in an increase. Note that that abundance or scarcity represent only the beliefs of the consumers and may or may not depict the objective state of the resource, making it possible for $\R_i$ to lie outside the interval $[0, \Rmax]$. A negative value of $\R_i$ simply implies that $i$ considers the resource to always be in abundance no matter how low the actual stock is. Increasingly negative values amplify this effect via the factor $R(\tau)-\R_i$ (a similar argument holds for $\R_i > \Rmax$). While Mosler and Brucks \cite{mosler2003integrating} assume the scarcity threshold to be constant for the entire population, we generalize by allowing it to be different for each consumer (thereby ascribing more influence of the heterogeneity of consumers). Indeed, consumers with different characteristics (age, sex, social class, political orientation, etc.) also tend to have different levels of environmental concern \cite{liere1980social}. 

In the same study, Rutte et al. \cite{rutte1987scarcity} show that consumer behavior differs, depending on the extent, to which they attribute scarcity of the resource to nature-induced reasons relative to society-induced reasons. For example in the case of a forest, nature-induced reasons may include less rainfall, volcanic eruption, and so on. Society-induced reasons may include extensive overuse by the consumers, pollution, and other practices that are harmful to the forest. We represent this attribution of $i$ as $\a_i \geq 0$; where $\a_i=0$ represents a consumer, who associates the scarcity of the resource entirely with society, with increasing values of $\a_i$ representing increasing association of the resource scarcity to nature. As suggested by the aforementioned study, the consumers that attribute the condition of the resource more to nature, tend to give ecological information more importance. Thus we define the ecological weight to be equal to the attribution $\a_i$. The ecological impact is therefore given by $\a_i (R(\tau)-\R_i )$, which means that a consumer, who attributes the condition of the resource more to nature (thus having a high value of $\a_i$) will manifest a higher ecological impact, than a consumer, who attributes the condition of the resource to society (having a low value of $\a_i$). 

In his theory of social comparison processes, Festinger \cite{festinger1954theory} postulates that people are less attracted to situations, where others are very divergent from them, than to situations, where others are close to them regarding both abilities and opinions. Therefore, we define the social factor by the term $\sum_{j=1}^\n \w_{ij} ( e_j - e_i)$ which represents the difference in $i$'s consumption w.r.t all the other connected agents\footnote{Henceforth, the terms individuals, consumers and agents are used interchangeably.}. Here $\w_{ij} > 0$ is the directed tie-strength between $i$ and $j$; it represents the influence that $j$ has on $i$'s consumption. Furthermore $\omega_{ii} = 0$ and $\sum_{j} \omega_{ij} = 1 \,\, \forall \,\, i$.. Thus the social factor is a measure of equality, with a low magnitude representing a society with equal consumption (dictating a lesser change in consumption) and a high magnitude representing unequal consumption (dictating a greater change in consumption). Furthermore, a negative social factor represents a higher consumption of the target individual relative to the other consumers (dictating a decrease in the target individual's consumption) and a positive social factor represents a relatively lower consumption of the target individual (dictating an increase in consumption).

The social value of consumer $i$ is represented by $\s_i \geq 0$. Here $\s_i = 0$ represents an extremely non-cooperative individual, with increasing values of $\s_i$ representing an increasingly cooperative individual. A consumer with higher social value will weigh equality more heavily than a consumer with lower social value. The influence of the social value of consumer $i$, represented by $\s_i \geq 0$, is based on the assumption that cooperative individuals are more concerned with maximizing equality and respond with anger to violations in equality regardless of the effect on their own outcomes \cite{van2013psychology}. Thus the social impact is given by the product $\s_i \left(\sum_{j=1}^\n \w_{ij} ( e_j - e_i)\right)$. 

Thus in our assumed society of $\n$ consumers, the following equation gives the change in effort of $i$ as a function of her individual characteristics, the current stock level and the effort of the other consumers.
\begin{align}
\label{eq:cons_mod} %original consumer model
	\hspace{0pt}\frac{d e_i(\tau)}{d \tau} = \a_i \left(R(\tau) - \R_i\right) + \s_i  \sum_{j=1}^\n  \w_{ij}  \left(e_j(\tau) -  e_i(\tau)\right),
\end{align} 
where $i \in \{1,\dots, \n \}$. The first term on the right hand side of \eqref{eq:cons_mod} represents the weighted ecological factor, where $\left(R(\tau) - \R_i\right)$  is the difference between the perceived optimum and actual level of the stock, which is weighed by the attribution $\a_i$. The second term on the right hand side of \eqref{eq:cons_mod} represents the weighted social factor, where $\sum_{j=1}^\n \w_{ij} ( e_j - e_i)$  is the difference in effort between the consumers, which is weighed by the social value $\s_i$.

\subsection{The Coupled Socio-ecological System}
\label{sec:ses}

Together \eqref{eq:ec_mod} and \eqref{eq:cons_mod} describe the overall dynamics of the considered socio-ecological system of the consuming population harvesting an open-access resource.  In what follows, we undertake some additional transformations, after which the model is able to capture the first principles more accurately, and the new variables have clearer interpretations in terms of the underlying theory. The transformations also reduce the overall dimensionality of the parameter space.

Let $x(\tau)$ be the quantity of the resource relative to the carrying capacity $\Rmax$  of the environment and $y_i (\tau)$ be $i$'s effort relative to the intrinsic growth rate $\r$: $x(\tau) = \displaystyle \frac{R(\tau)}{\Rmax}$ and $y_i (\tau) = \displaystyle \frac{e_i (\tau)}{\r}$. Define $\uprho_i=\displaystyle \frac{\R_i}{\Rmax}$ as $i$'s scarcity threshold level relative to $\Rmax$. Next define $t= \r \tau$ as the new, non-dimensional time scale. We can now express the system as follows
\begin{align}
\hspace{0pt}
\label{eq:ses}%socioecological system
\begin{split}
	&\dot{x} =  (1-x)x - x\sum_{i=1}^{\n}  y_i,\\
	&\dot{y}_i = \b_i \Big( \upalpha_i(x -\uprho_i)- \upnu_i \sum_{j=1}^{\n} \w_{ij}\left( y_i - y_j \right) \Big),
\end{split}
\end{align}
where $\displaystyle \b_i = \displaystyle \frac{(\a_i \Rmax + \r \s_i)}{\r^2}$, $\displaystyle \upalpha_i  =  \frac{\a_i \Rmax}{(\a_i \Rmax + \r \s_i)}$, $\displaystyle \upnu_i = \frac{\r \s_i}{(\a_i \Rmax + \r \s_i)}$ and the over-dot represents the derivative with respect to time $t$. Although the original weights $\a_i$ and $\s_i$ have different dimensions and thus are incomparable, the new weights $\upalpha_i$ and $\upnu_i$ become dimensionless and can be compared. Furthermore, $\upalpha_i+\upnu_i=1$, which means that $\upalpha_i$ and $\upnu_i$ are complements of each other. This reflects the bipolarity in social and physical dimensions as described by the theory underlying the original model \cite{mosler2003integrating}. Thus $\upalpha_i$ can be interpreted as the relevance that $i$ assigns to ecological information and $\upnu_i$ can be interpreted as the relevance that $i$ assigns to social information. The state variables $x, y_i$ and parameters $\uprho_i$, $\b_i$ also become dimensionless. Henceforth, we call $\upalpha_i$ and $\upnu_i$ the ecological and social relevance of $i$ respectively, and $\uprho_i$ the environmentalism of $i$, where $\uprho_i\leq0$ represents an extremely non-environmental individual and increasing values of $\uprho_i$ correspond to increasing levels of environmental concern. Here $\b_i$ can be interpreted as the overall susceptibility \cite{friedkin2006structural} of $i$ to change in her consumption. All subsequent analysis in this dissertation will be carried out on model \eqref{eq:ses}.

\section{Discussion}

Leon Festinger's theory on Social Comparison Processes \cite{festinger1954theory} provides the rationale for Social-Ecological Relevance (SER), which is the core theoretical concept of the model in \cite{mosler2003integrating}. SER is one-dimensional and so the preference given to social factors is the complement of that given to ecological factors, which represents the bipolarity between social and physical dimensions. The model that we put forward here also captures this concept as $\upalpha_i$ and $\upnu_i$ can be interpreted as the relative preferences given to the ecological and social factors respectively and as they are the complements of each other they depict the bipolarity inherent in the underlying social comparison process. 

It is important to note that there also exist some limitations in the model as presented here. First, the psychological parameters of the system have been assumed to be constant over time. Although in some cases a society may exhibit unvarying characteristics \cite{friedkin2006structural} for short enough time span, it is understood that in general, societies in the real-world are in a constant state of change. In Chapter \ref{chap:learning}, we study the dynamics of the environmentalism under the theory of learning in games. Furthermore, the assumption of perfect access of each consumer group, to information on the consumption of the other groups, is obviously restrictive of a true representation of reality. Future work can be directed towards relaxing this assumption, and observing the implications it may have, on the conclusions drawn from the model we put forward here. In the subsequent chapters, we explore the potential applications of the model via different mathematical frameworks in the context of NRM.

%
%
%\section{What is a Model?}
%
%\section{Psychological Models}
%
%\section{Game-theoretic Models}
%
%\section{Agent-based Models}
%
%\section{Mathematical Models}
%
%\section{Social Network Models}

\clearpage
%% This is an example first chapter.  You should put chapter/appendix that you
%% write into a separate file, and add a line \include{yourfilename} to
%% main.tex, where `yourfilename.tex' is the name of the chapter/appendix file.
%% You can process specific files by typing their names in at the 
%% \files=
%% prompt when you run the file main.tex through LaTeX.

\chapter{Exploiting Community Structure to Create Block-models of Resource Consumption}

\label{chap:lump}

Previously in Chapter \ref{chap:model}, we introduced a mathematical model of human decision making in a coupled resource-consumer setting. Here we present a mechanism whereby groups of similar consumers can be aggregated into a single unit provided that certain conditions of homogeneity hold withing the network. Such block models are studied comprehensively for similar processes in sociology \cite{borgatti2009network}, as they not only simplify the analysis but also utilize the community structure of the network to present a concise description of the society. After outlining the conditions on the structure of the population that allows us to make these simplifications, we consider societies where those conditions may not be fulfilled. For such cases, we explore an approximate block model that exibits deviations from the true model in the transient, but converges to the exact behavor in the steady state.

\section{Obtaining Simplified Models of Aggregate Consumption}

The coupled dynamics of the resource and individual consumptions is given by Equation \eqref{eq:ses} which has been reproduced below.
\begin{align}
\hspace{0pt}
\label{eq:ses2}%socioecological system
\begin{split}
	&\dot{x} =  (1-x)x - x\sum_{i=1}^{\n}  y_i,\\
	&\dot{y}_i = \b_i \Big( \upalpha_i(x -\uprho_i)- \upnu_i \sum_{j=1}^{\n} \w_{ij}\left( y_i - y_j \right) \Big),
\end{split}
\end{align}
where $ i\in \{1, \dots , \n\}$ with $\n$ being the total number of consumers. Recall from Chapter \ref{chap:model} that $\upalpha_i\in(0, 1)$ and $\upnu_i = 1- \alpha_i$ are the weights that $i$ gives to the ecological and social factors respectively. Furthermore $\w_{ij}\geq 0$ is the social influence that node $j$ exerts on node $i$ (see Figure \ref{fig:soc_inf}), with the constraint that $\sum_{j=1}^{\n} \w_{ij} = 1$ and $\w_{ii} = 0$.
\begin{figure}[h!]
	\captionsetup{font=small,width=0.75\textwidth}
%	\vspace{10pt}
	\begin{center}
		\includegraphics[width=0.4\linewidth]{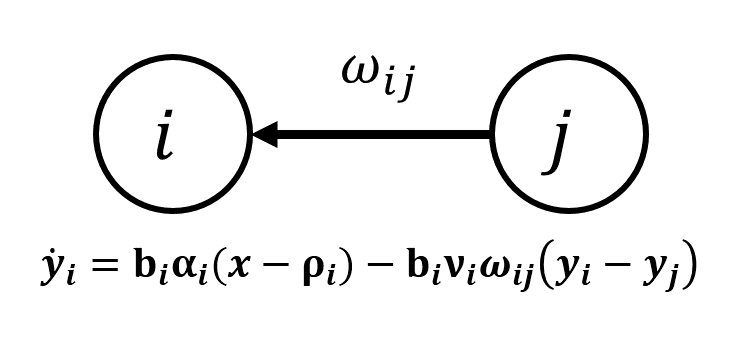}
	\end{center}
	\vspace{-20pt}
	\caption{The consumption of $i$ is influenced by $j$ through $\w_{ij}$. The overall effect is given by $\b_i \, \upnu_i \, \w_{ij}$ where $\upnu_i$ is the weight that $i$ gives to social factors in the decision making process.}
%	\vspace{-70pt}
	\label{fig:soc_inf} %social influence
\end{figure}
The $\n$ equations for the individual consumption rates are given as follows
\begin{align}
\begin{split}
	\label{eq:sys} %equation of system
	&\dot{y}_1 = \b_1 \left( \upalpha_1( x-\uprho_1) - \upnu_1\sum_{j=1}^\n \w_{1j}\left( y_1 - y_j \right) \right) \\
	&\dot{y}_2 = \b_2 \left( \upalpha_2( x-\uprho_2) - \upnu_2\sum_{j=1}^\n \w_{2j}\left( y_2 - y_j \right) \right) \\
	& \hspace{80pt} \vdots \\
	&\dot{y}_\n = \b_\n \left( \upalpha_\n( x-\uprho_\n) - \upnu_\n\sum_{j=1}^\n \w_{\n j}\left( y_n - y_j \right) \right)
\end{split}
\end{align}
The summation term can be expanded as
\begin{align}
\begin{split}
\hspace{-20pt}
	&\hspace{-20pt}\dot{y}_1 = \b_1 \left( \upalpha_1( x-\uprho_1) - \sum_{j=1}^\n \w_{1j} \, \upnu_1 \, y_1 + \w_{11} \, \upnu_1 \, y_1 + \w_{12} \, \upnu_1 \, y_2 + \cdots + \w_{1\n} \, \upnu_1 \, y_\n \right) \\
	&\hspace{-20pt}\dot{y}_2 = \b_2 \left( \upalpha_2( x-\uprho_2) - \sum_{j=1}^\n \w_{2j} \, \upnu_2 \, y_2 + \w_{21} \, \upnu_2 \, y_1 + \w_{22} \, \upnu_2 \, y_2 + \cdots + \w_{2\n} \, \nu_2 \, y_\n \right) \\
	& \hspace{-20pt}\hspace{80pt} \vdots \\
	&\hspace{-20pt}\dot{y}_\n = \b_\n \left( \upalpha_\n( x-\uprho_\n) - \sum_{j=1}^\n \w_{\n j} \, \upnu_\n \, y_\n + \w_{\n 1} \, \upnu_\n \, y_1 + \w_{\n 2} \, \upnu_2 \, y_\n + \cdots + \w_{\n\n} \, \upnu_\n \, y_\n \right)
\end{split}
\end{align}
Summing up all equations we get
\begin{align}
\begin{split}
\hspace{-20pt}
	\sum_{i=1}^\n \dot{y}_i = \sum_{i=1}^\n \b_i\upalpha_i(x-\uprho_i) - &\sum_{j=1}^\n \w_{1j} \, \b_1\upnu_1 \, y_1 - \sum_{j=1}^\n \w_{2j} \, \b_2\upnu_2 \, y_2 - \cdots - \sum_{j=1}^\n \w_{\n j} \,\b_\n\upnu_\n \, y_\n \\
	+ & \sum_{j=1}^\n \w_{j1} \, \b_j\upnu_j \, y_1 + \sum_{j=1}^\n \w_{j2} \, \b_j \upnu_j \, y_2 + \cdots + \sum_{j=1}^\n \w_{j \n} \, \b_j \upnu_j \, y_\n
\end{split}
\end{align}
the summations containing the same consumption rate term can be collected as
\begin{align}
\begin{split}
	\sum_{i=1}^\n \dot{y}_i =& \sum_{i=1}^\n \b_i\upalpha_i(x-\rho_i) -  \sum_{j=1}^\n \big( \w_{1j} \, \b_1\upnu_1 - \w_{j1} \, \b_j\upnu_j \big)y_1\\ -& \sum_{j=1}^\n \big( \w_{2j} \, \b_2\upnu_2 - \w_{j2} \, \b_j\upnu_j \big)y_2 - \cdots
	- \cdots \sum_{j=1}^\n \big( \w_{\n j} \, \b_\n\nu_\n - \w_{j\n} \, \b_j\upnu_j \big)y_\n
\end{split}
\end{align}
and can be written in a more compact manner by summing the terms on the R.H.S over a second index as
\begin{align}
\label{eq:hom1} %homogeneous 1
\begin{split}
	\sum_{i=1}^\n \dot{y}_i = \sum_{i=1}^\n \b_i\upalpha_i(x-\uprho_i) + \sum_{i=1}^\n \Big( \sum_{j=1}^\n \big( \w_{ji} \, \b_j\upnu_j - \w_{ij} \, \b_i\upnu_i \big) \Big) y_i 
\end{split}
\end{align}
In order to transform this equation into one for $\sum y_i$, the following condition must hold true
\begin{align}
\label{eq:con_hom}
\begin{split}
	\sum_{j=1}^\n \big( \w_{ji} \, \b_j\upnu_j - \w_{ij} \, \b_i\upnu_i \big) = \mathrm{c} ; \quad \forall \, i
\end{split}
\end{align}

\subsection{Influence and Leadership in the Consumer Network}

In  order to understand the condition in (\ref{eq:con_hom})  better, let us analyze the individual terms in the L.H.S of the condition
\begin{align}
\begin{split}
	\mathbf{\textcolor{black}{\sum_{j=1}^\n \w_{ji} \, \b_j \, \upnu_j} - \textcolor{black}{\sum_{j=1}^\n \w_{ij} \, \b_i \, \upnu_i} }
\end{split}
\label{eq:con_hom2} %condition homogenous
\end{align}

To develop some intuition about the individual terms, recall from Figure \ref{fig:soc_inf} that the effective weight of the edge directed from $i$ to $j$ in the consumption network is $\w_{ji} \, \b_j \, \upnu_j$. If this is so then the first term in \eqref{eq:con_hom2} can be viewed as the aggregate social influence that $i$ has on all other members of the network. Let us call this the \emph{out-influence} of $i$.  Similarly, the second term in \eqref{eq:con_hom2} can be viewed as the aggregate influence on $i$ from all other members of the network. Let us call this the $in-influence$ of $i$. The \emph{net-influence} of $i$ is then simply her out-influence minus her in-influence and can be viewed as representing $i$'s role in the network as either a leader (positive net-influence), a follower (negative net-influence) or neutral (zero net-influence). Thus (\ref{eq:con_hom}) states that for all nodes present in the network, their respective net-influence is constant.
\begin{figure}[h!]
	\captionsetup{width=0.6\textwidth}
	\begin{center}
		\includegraphics[width=0.8\linewidth]{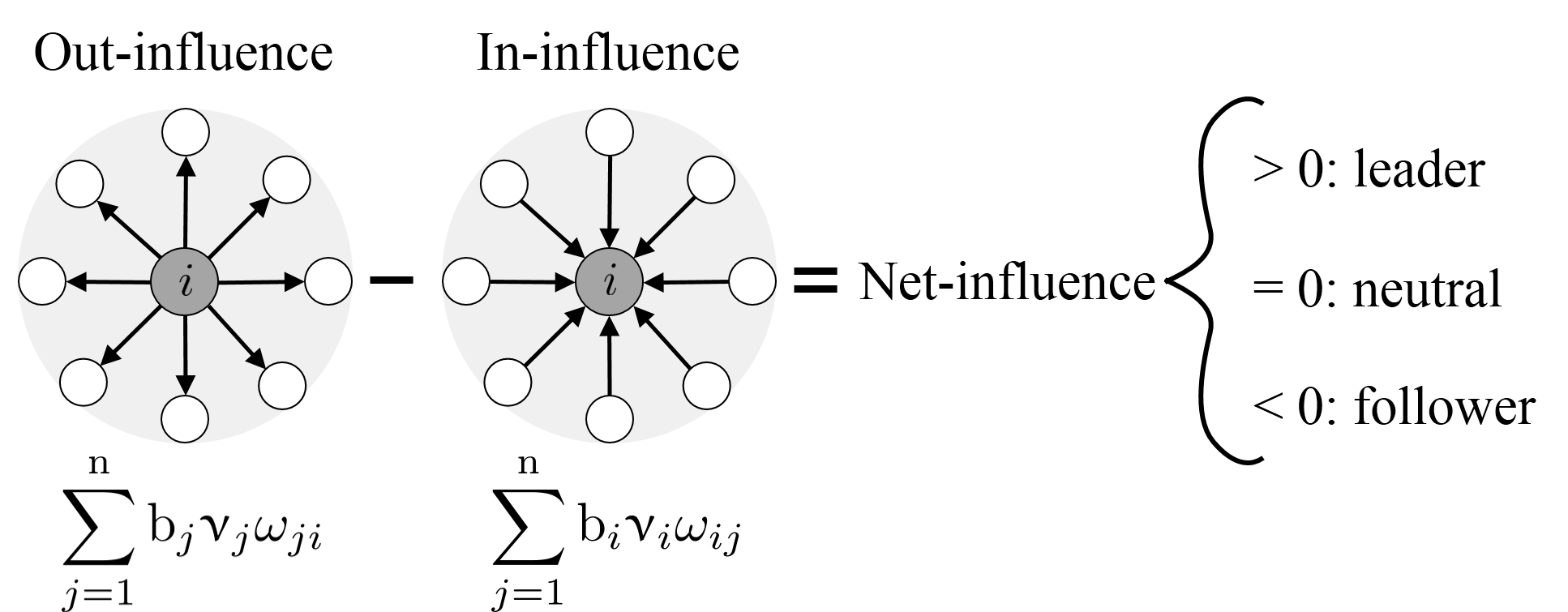}
	\end{center}
	\caption{Graphical depiction of net-influence.}
	\label{fig:net_inf} %net influence
\end{figure}
We now show that condition \eqref{eq:con_hom} can hold only for $\mathrm{c} = 0$.
\begin{lemma}
\label{th:lem_1}
	The following condition 
	\begin{align*}
		\sum_{j=1}^\n\! \big( \w_{ij} \, \b_i\upnu_i \!-\! \w_{ji} \, \b_j\upnu_j \big) = \mathrm{c},  \quad \forall \,\, i \in \{1,\dots,\n\},
	\end{align*}
	can hold only for $\mathrm{c} = 0$.
\end{lemma}
\begin{proof}
	Assuming that the condition holds true, we can expand the expression as follows
	\begin{align*}
		\sum_{j=1}^\n \w_{ij} \, \b_i\upnu_i - \sum_{j=1}^\n \w_{ji} \, \b_j\upnu_j  = \mathrm{c} .
	\end{align*}
	Summing over all $i$
	\begin{align*}
		\sum_{i=1}^\n \sum_{j=1}^\n \w_{ij} \, \b_i\upnu_i - \sum_{i=1}^\n \sum_{j=1}^\n \w_{ji} \, \b_j\upnu_j  = \n \mathrm{c}.
	\end{align*}
	Switching the index variables for the second term on the L.H.S
	\begin{align*}
		&\sum_{i=1}^\n \sum_{j=1}^\n \w_{ij} \, \b_i\upnu_i - \sum_{j=1}^\n \sum_{i=1}^\n \w_{ij} \, \b_i\upnu_i  = \n \mathrm{c} ,\\
		&\sum_{i=1}^\n \sum_{j=1}^\n \big( \w_{ij} - \w_{ij} \big)\b_i\upnu_i = \n\mathrm{c} \Rightarrow \mathrm{c} = 0,
	\end{align*}
	which concludes the proof.
\end{proof}
Thus in order for condition \eqref{eq:con_hom} to hold i.e., for a network with constant net-influence, there can exist no leaders or followers in the networks. We call such networks as \emph{self-directed} networks. With this stipulation in hand, we are now able to define our first canonical network.

\subsection{The Canonical Networks}
\label{sec:lumped}
Previously, the network of the consuming population has been seen at an individual level. Now, we explore the structural effects of the network at two additional levels of abstraction. The \emph{homogeneous consumer network} allows the abstraction of individual consumptions by viewing the society as a whole. Needless to say, this abstraction comes with a loss of information (in the form of individual consumptions) and additional constraints (explored in the following) that the network must satisfy. The \emph{semi-homogeneous} and \emph{symmetric semi-homogeneous consumer networks} abstract the individual consumptions at the group level, where it is assumed that the society consists of homogeneous communities with certain regularities prevailing in the influences across groups. In what follows, we define the homogeneous, semi-homogeneous and symmetric semi-homogeneous consumer networks and examine the regularities enforced by each on the structure of the consuming population. Together, we call these three constructs collectively as the \emph{canonical} consumer networks.

\subsubsection{The Homogeneous Consumer Network}
We understand a homogeneous network as one in which \eqref{eq:con_hom} holds, along with some other requirements stated below. Note that due to Lemma \ref{th:lem_1}, condition \eqref{eq:con_hom} can hold only for $\mathrm{c}=0$ (i.e., a self-directed network). The homogeneous network is then given by the following definition.
\begin{definition}
\label{def:hom}
	{\bf Homogeneous consumer network:} A consuming population is said to comprise a  homogeneous consumer network if the following conditions hold 
	\begin{enumerate}
		\item All agents have uniform attribution $\a_i$ and social value orientation $\s_i$ i.e., there exist $\Beta, \Alpha$ and $\Nu$ such that,
		\begin{align*} &\b_i \upalpha_i = \frac{\a_i \Rmax}{\r^2} = \Beta \, \Alpha \,\, \forall \,\, i \in \{1, \dots, \n\}\text{, \,\,and}\\
				& \b_i \upnu_i = \frac{\s_i}{\r} = \Beta \, \Nu \,\, \forall \,\, i \in \{1, \dots, \n\}. 
		\end{align*}
		\item All agents have uniform environmentalism i.e., there exists $\Rho$ such that, 
		\begin{equation*} \uprho_i = \Rho \,\, \forall \,\, i \in \{1, \dots, \n\}. \end{equation*}
		\item There are no leaders or followers in the network i.e.,
		\begin{equation*} \sum_{j=1}^\n \big( \w_{ji} \, \b_j\upnu_j - \w_{ij} \, \b_i\upnu_i \big) = 0  \quad \forall \, i \in \{1, \dots, \n\}. \end{equation*}
	\end{enumerate}
\end{definition}
Moving forward from \eqref{eq:hom1} and using Definition \ref{def:hom}, the consumption dynamics of the homogeneous network can simply be represented as
\begin{align}
\label{eq:hom_sys} %homogeneous system
\begin{split}
	&\dot{x} = (1-x)x - x Y,\\
	&\dot{Y} = \n \Beta \, \Alpha \, (x - \Rho),
\end{split}
\end{align}
where $Y = \sum_{i=1}^{\n} y_i$ represents the aggregate consumption of the network. Thus the system is now represented by a two-dimensional system which gives a macro-level view of the society as one single unit. It is important to realize here that due to this aggregation, we lose information on the individual consumptions.

\subsubsection{The Semi-homogeneous Consumer Network}

We now assume that the network has $\m$ non-overlapping consumer groups with populations given by $\n_k, k \in \{1\dots \m\}$ respectively, and $\n = \sum_{k=1}^\m \n_k$. The groups are classified on basis of their social and ecological relevances, and their environmentalisms. For clarity of notation, let us define the set $\N_k = \{k_1, k_2, \dots, k_{\n_k} \}$ to contain the indices of all consumers belonging to group $k$. Further more let $\N = \{1,\dots, \n\}$ be the set containing the indices of the whole network such that $\N_1 \cup \N_2 \cup \dots \cup \N_{\m} = \N$. Now consider the $\n_k$ equations representing the consumption of Group $k$. 
{\allowdisplaybreaks
\fontsize{9.5}{11.5}\selectfont
\begin{align}
\begin{split}
	\label{eq:sys} %equation of system
	&\dot{y}_{k_1} = \b_{k_1} \left( \upalpha_{k_1}( x-\uprho_{k_1}) - \upnu_{k_1}\sum_{j\in \N} \w_{{k_1}j}\left( y_{k_1} - y_j \right) \right) \\
	&\dot{y}_{k_2} = \b_{k_2} \left( \upalpha_{k_2}( x-\uprho_{k_2}) - \upnu_{k_2}\sum_{j\in \N} \w_{{k_2}j}\left( y_{k_2} - y_j \right) \right) \\
	& \hspace{80pt} \vdots \\
	&\dot{y}_{k_{\n_k}} = \b_{k_{\n_k}} \left( \upalpha_{k_{\n_k}}( x-\uprho_{k_{\n_k}}) - \upnu_{k_{\n_k}}\sum_{j\in \N} \w_{{k_{\n_k}}j}\left( y_{k_{n_k}} - y_j \right) \right)
\end{split}
\end{align}}
Note that we have not yet asserted the homogeneity of the group and have maintained the parameters in their original form for now. The summation term of each equation can now be expanded as
{\allowdisplaybreaks
\fontsize{9.5}{11.5}\selectfont
\begin{align}
\begin{split}
	&\begin{aligned}\dot{y}_{k_1} = \b_{k_1} \Bigg( \upalpha_{k_1}( x-\uprho_{k_1}) &- \sum_{j \in \N} \w_{{k_1}j} \, \upnu_{k_1} \, y_{k_1} + \w_{{k_1}{1}} \, \upnu_{k_1} \, y_{1} \\& \hspace{50pt}+ \w_{{k_1}{2}} \, \upnu_{k_1} \, y_{2} + \cdots + \w_{{k_1}\n} \, \upnu_{k_1} \, y_\n \Bigg) \end{aligned}\\
	&\begin{aligned}\dot{y}_{k_2} = \b_{k_2} \Bigg( \upalpha_{k_2}( x-\uprho_{k_2}) &- \sum_{j \in \N} \w_{{k_2}j} \, \upnu_{k_2} \, y_{k_2} + \w_{{k_2}{1}} \, \upnu_{k_2} \, y_{1} \\&\hspace{50pt}+ \w_{{k_2}{2}} \, \upnu_{k_2} \, y_{2} + \cdots + \w_{{k_2}\n} \, \upnu_{k_2} \, y_\n \Bigg) \end{aligned} \\
	& \hspace{80pt} \vdots \\
	&\begin{aligned} \dot{y}_{k_{n_k}} = \b_{k_{n_k}} \Bigg( \upalpha_{k_{n_k}}( x-\uprho_{k_{n_k}}) &- \sum_{j \in \N} \w_{{k_{n_k}}j} \, \upnu_{k_{n_k}} \, y_{k_{n_k}} + \w_{{k_{n_k}}{1}} \, \upnu_{k_{n_k}} \, y_{1} \\ &\hspace{40pt}+ \w_{{k_{n_k}}{2}} \, \upnu_{k_{n_k}} \, y_{2} + \cdots + \w_{{k_{n_k}}\n} \, \upnu_{k_{n_k}} \, y_\n \Bigg) \end{aligned}
\end{split}
\end{align}}
summing up all equations we get
\begin{align}
\begin{split}
	\sum_{i\in \N_k} \dot{y}_i =& \sum_{i\in \N_k} \b_i \upalpha_i (x-\uprho_i) \\- &\sum_{j \in \N} \w_{{k_1}j} \, \b_{k_1} \, \upnu_{k_1} \, y_{k_1} - \sum_{j \in \N} \w_{{k_2}j} \, \b_{k_2} \, \upnu_{k_2} \, y_{k_2} - \cdots - \sum_{j \in \N} \w_{{k_{n_k}}j} \, \b_{k_{n_k}} \upnu_{k_{n_k}} \, y_{k_{n_k}} \\
	+ & \sum_{j\in \N_k} \w_{j1} \, \b_j\upnu_j \, y_1 + \sum_{j\in \N_k} \w_{j2} \, \b_j\upnu_j \, y_2 + \cdots + \sum_{j\in \N_k} \w_{j\n} \, \b_j\upnu_j \, y_\n
\end{split}
\end{align}
the summations containing the same consumption rate term can be collected as
\begin{align}
\begin{split}
	\sum_{i\in \N_k} \dot{y}_i =& \sum_{i \in \N_k} \b_i \upalpha_i (x-\uprho_i) \\
	- &\Bigg(\sum_{j\in \N_k} \big( \w_{{k_1}j} \, \b_{k_1}\upnu_{k_1} - \w_{j{k_1}} \, \b_j\upnu_j \big) + \sum_{j \in \N \setminus \N_k} \w_{{k_1}j} \, \b_{k_1}\upnu_{k_1} \Bigg) y_{k_1}\\
	- &\Bigg(\sum_{j\in \N_k} \big( \w_{{k_2}j} \, \b_{k_2}\upnu_{k_2} - \w_{j{k_2}} \, \b_j\upnu_j \big) + \sum_{j \in \N \setminus \N_k} \w_{{k_2}j} \, \b_{k_2}\upnu_{k_2} \Bigg) y_{k_2}\\
	\\ & \hspace{90pt}\vdots \\
	- &\Bigg(\sum_{j\in \N_k} \big( \w_{{k_{n_k}}j} \, \b_{k_{n_k}}\upnu_{k_{n_k}} - \w_{j{k_{n_k}}} \, \b_j\upnu_j \big) + \sum_{j \in \N \setminus \N_k} \w_{{k_{n_k}}j} \, \b_{k_{n_k}}\upnu_{k_{n_k}} \Bigg) y_{k_{n_k}}\\
	+ & \sum_{i\in \N\setminus \N_k} \Bigg( \sum_{j\in \N_k} \w_{ji} \, \b_j\upnu_j \Bigg) y_i 
\end{split}
\end{align}
combining the terms containing the consumption rates of Group $k$ over a second index we get
\begin{align}
\begin{split}
	\sum_{i\in \N_k} \dot{y}_i =& \sum_{i \in \N_k} \b_i \upalpha_i (x-\uprho_i) \\
	- & \sum_{i\in\N_k} \Bigg(\sum_{j\in \N_k} \big( \w_{ij} \, \b_i \upnu_i - \w_{ji} \, \b_j\upnu_j \big) + \sum_{j \in \N \setminus \N_k} \w_{ij} \, \b_i\upnu_i \Bigg) y_i\\
	+ & \sum_{i\in \N\setminus \N_k} \Bigg( \sum_{j\in \N_k} \w_{ji} \, \b_j\upnu_j \Bigg) y_i 
\end{split}
\end{align}
Repeating the procedure for all the groups, we get
{\allowdisplaybreaks
\fontsize{9.5}{11.5}\selectfont
\begin{align}
\begin{split}
	\sum_{i\in \N_1} \dot{y}_i =& \sum_{i \in \N_1} \b_i \upalpha_i (x-\uprho_i) \\
	- & \sum_{i\in\N_1} \Bigg(\sum_{j\in \N_1} \big( \w_{ij} \, \b_i \upnu_i - \w_{ji} \, \b_j\upnu_j \big) + \sum_{j \in \N \setminus \N_1} \w_{ij} \, \b_i\upnu_i \Bigg) y_i\\
	+ & \sum_{i\in \N\setminus \N_1} \Bigg( \sum_{j\in \N_1} \w_{ji} \, \b_j\upnu_j \Bigg) y_i \\
	& \hspace{90pt}\vdots \\
	& \hspace{90pt}\vdots \\
	\sum_{i\in \N_\m} \dot{y}_i =& \sum_{i \in \N_\m} \b_i \upalpha_i (x-\uprho_i) \\
	- & \sum_{i\in\N_\m} \Bigg(\sum_{j\in \N_\m} \big( \w_{ij} \, \b_i \upnu_i - \w_{ji} \, \b_j\upnu_j \big) + \sum_{j \in \N \setminus \N_\m} \w_{ij} \, \b_i\upnu_i \Bigg) y_i\\
	+ & \sum_{i\in \N\setminus \N_\m} \Bigg( \sum_{j\in \N_\m} \w_{ji} \, \b_j\upnu_j \Bigg) y_i 
\end{split}
\end{align}}
Expanding the last summation in each equation with respect to each group, we get the following system
{\allowdisplaybreaks
\fontsize{9.5}{11.5}\selectfont
\begin{align}
\label{eq:sem_hom_1}
\begin{split}
	\dot{Y}_1 =& \sum_{i \in \N_1} \b_i \upalpha_i (x-\uprho_i) \\
	- & \sum_{i\in\N_1} \Bigg(\sum_{j\in \N_1} \big( \w_{ij} \, \b_i \upnu_i - \w_{ji} \, \b_j\upnu_j \big) + \sum_{j \in \N \setminus \N_1} \w_{ij} \, \b_i\upnu_i \Bigg) y_i\\
	+ & \sum_{i\in \N_2} \Bigg( \sum_{j\in \N_1} \w_{ji} \, \b_j\upnu_j \Bigg) y_i + \dots + \sum_{i\in \N_\m} \Bigg( \sum_{j\in \N_1} \w_{ji} \, \b_j\upnu_j \Bigg) y_i  \\
	& \hspace{90pt}\vdots \\
	& \hspace{90pt}\vdots \\
	\dot{Y}_\m =& \sum_{i \in \N_\m} \b_i \upalpha_i (x-\uprho_i) \\
	- & \sum_{i\in\N_\m} \Bigg(\sum_{j\in \N_\m} \big( \w_{ij} \, \b_i \upnu_i - \w_{ji} \, \b_j\upnu_j \big) + \sum_{j \in \N \setminus \N_\m} \w_{ij} \, \b_i\upnu_i \Bigg) y_i\\
	+ & \sum_{i\in \N_1} \Bigg( \sum_{j\in \N_\m} \w_{ji} \, \b_j\upnu_j \Bigg) y_i + \dots + \sum_{i\in \N_{\m-1}} \Bigg( \sum_{j\in \N_\m} \w_{ji} \, \b_j\upnu_j \Bigg) y_i  \\
\end{split}
\end{align}}
where $Y_k = \sum_{i \in \N_k} y_i$ is the aggregate consumption of Group $k$ where $k \in \{1, \dots , \m\}$. Note that the above system, though simplified, is still not expressed purely in-terms of the $Y_k$'s. It is also easy to see that applying the condition \eqref{eq:con_hom} alone is not sufficient to achieve this. In the following definition, we describe the additional requirements to do so. The resulting network is called the semi-homogeneous consumer network.

\begin{definition}
\label{def:semi-hom}
	{\bf Semi-homogeneous consumer network:} A consuming population is said to comprise a  semi-homogeneous consumer network if the following conditions hold 
	\begin{enumerate}
		\item There exist $\m$ groups in the network, each of which individually fulfill the conditions for the homogeneous consumer network according to Definition \ref{def:hom}.
		\item The total in-influence of one group on another group should be uniformly distributed across all nodes of the influenced group i.e., for all $i$ $\in$ $\N_s$, there exists $\d^-_{sr}$ such that
		\begin{equation}\label{eq:lump_cond1}\sum_{j\in \N_r} \b_i \upnu_i \w_{ij} = \Beta_s \Nu_s \sum_{j \in \N_r} \w_{ij} = \Beta_s \Nu_s \d^-_{sr},\end{equation}
		where $\Beta_s$ and $\Nu_s$ represent the common sensitivity and social relevance of Group $s$ respectively.
		\item The total out-influence of one group from another group should be uniformly distributed across all nodes of the influencing group i.e., for all $j$ $\in$ $\N_r$, there exists $\d^+_{sr}$ such that,
		\begin{equation}\label{eq:lump_cond2} \sum_{i\in \N_s} \b_i \upnu_i \w_{ij} = \Beta_s \Nu_s \sum_{i \in \N_s} \w_{ij} = \Beta_s \Nu_s \d^+_{sr}. \end{equation}
	\end{enumerate}
\end{definition}

\begin{figure}[t!]
    \begin{center}
%    \captionsetup{width=0.85\linewidth}
    \begin{subfigure}[t]{0.5\linewidth}
        \captionsetup{width=0.8\linewidth}
        \includegraphics[width = 0.9\linewidth]{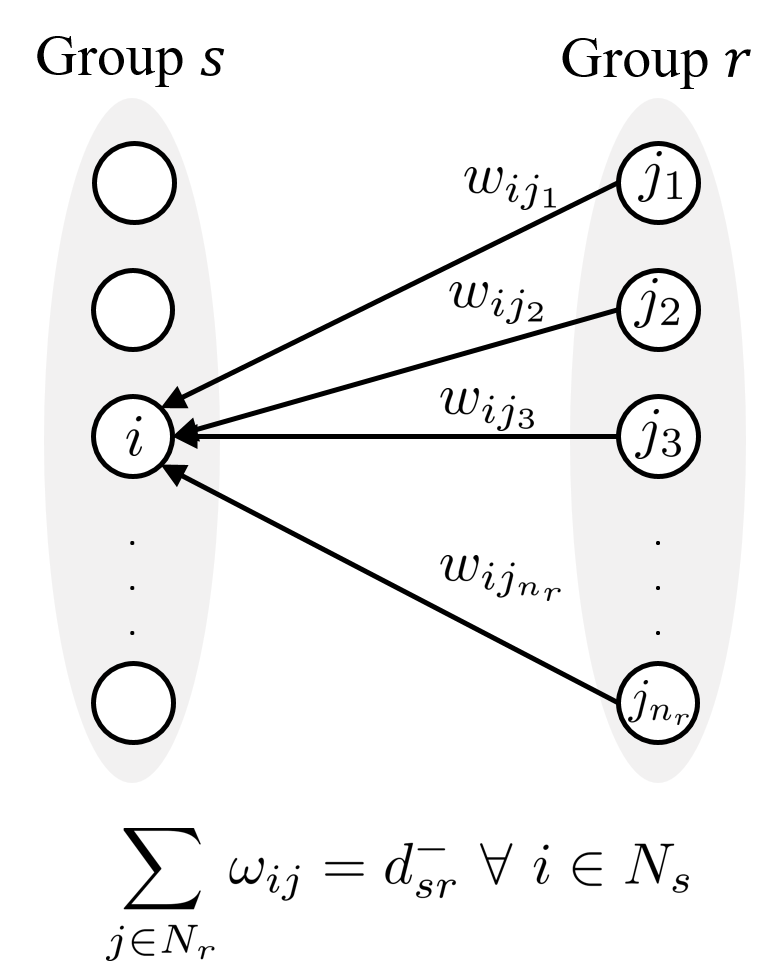}
        \caption{Pictorial representation of \eqref{eq:lump_cond1} for a general network.}
        \label{fig:cond1}
    \end{subfigure}%
    \begin{subfigure}[t]{0.5\linewidth}
        \captionsetup{width=0.8\linewidth}
        \includegraphics[width = 0.9\linewidth]{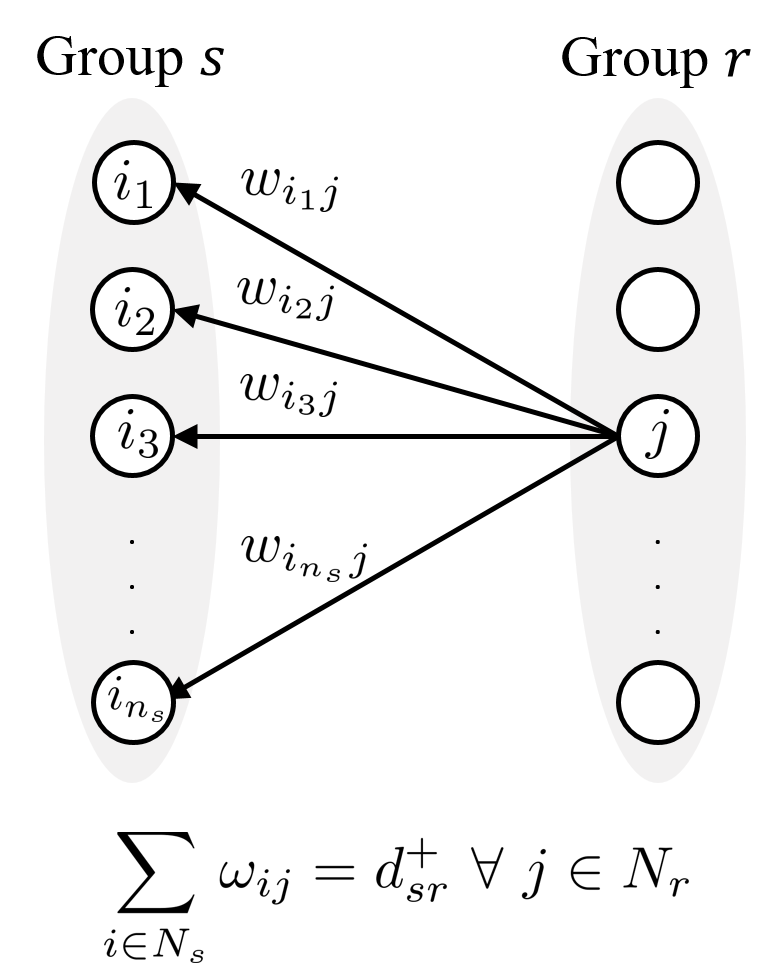}
        \caption{Pictorial representation of \eqref{eq:lump_cond2} for a general network.}
        \label{fig:cond2}
    \end{subfigure}
    \caption{Conditions on cross-influence between groups in the semi-homogeneous network }
   \end{center}
\end{figure}

Moving forward from Equation \eqref{eq:sem_hom_1} and using Definition \ref{def:semi-hom}, the consumption dynamics of the semi-homogeneous network are expressible as

\begin{align*}
\begin{split}
	&\dot{Y}_1 =  \begin{aligned} &  \n_1 \Beta_1 \Alpha_1 (x-\Rho_1) - \Beta_1 \Nu_1\Big( \n_2 \d^-_{12} + \dots + \n_\m \d^-_{1\m}\Big)   Y_1 \\ &+\,\, \n_1 \Beta_1 \Nu_1 \d^+_{12}Y_2 + \dots + \n_1 \Beta_1 \Nu_1 \d^+_{1\m} Y_\m, \end{aligned}
\\
	&\,\,\vdots \hspace{0.45 \linewidth} \vdots\\
	&\dot{Y}_\m   =    \begin{aligned} &  \n_{ \m} \Beta_{ \m} \Alpha_{ \m}  ( x - \Rho_{  \m} )  -  \Beta_\m    \Nu_{ \m}   \Big(   \n_1 \d^-_{\m 1}   +   \dots  +  \n_{\m - 1} \d^-_{\m (\m - 1)} \Big)    Y_\m \\ &+\,\, \n_\m \Beta_\m \Nu_\m \d^+_{\m1}Y_1 + \dots + \n_\m \Beta_\m \Nu_\m \d^+_{\m(\m - 1)} Y_{\m - 1}, \end{aligned}
\end{split}
\end{align*}

where $\Beta_k,\Alpha_k, \Nu_k$ and $\Rho_k$ represent the psychological characteristics of Group $k$. Note that although the system dynamics are expressed completely in-terms of the $Y_k$'s, they still do not appear as a scaled version of \eqref{eq:ses2}. The next canonical network includes a further simplification that enables this.

\subsubsection{The Symmetric Semi-Homogeneous Network}

\begin{definition}
\label{def:sym_semi-hom}
	{\bf Symmetric semi-homogeneous \\consumer network:} A consuming population is said to comprise a symmetric  semi-homogeneous consumer network if the following conditions hold 
	\begin{enumerate}
		\item There exist $\m$ homogeneous groups in the network. Together the $\m$ sub-groups fulfill the conditions for the semi-homogeneous consumer network as specified by Definition \ref{def:semi-hom}.
		\item The size of all connected subgroups are equal, the resulting implication being that the aggregated out-influences and in-influences as defined by \eqref{eq:lump_cond1} and \eqref{eq:lump_cond2} are equal for each pair of sub-groups, i.e., 
		\begin{equation}\label{eq:lump_cond3}  \d^-_{sr} = \d^{+}_{sr} = \mathrm{\W}_{sr} \text{ for all }  s,r \in \{1,\dots ,\m\}. \end{equation}
	\end{enumerate}
\end{definition}

The coupled dynamics for the resource and consumption can now be expressed for the symmetric semi-homogeneous network as follows
\begin{align}
\label{eq:sym_hom_sys}
\begin{split}
	&\dot{x} = (1-x)x - x \sum_{i=1}^\m Y_i,\\
	&\dot{Y}_i = \n_i \Beta_i \, \Alpha_i \, (x - \Rho_i) - \Beta_i\Nu_i \sum_{j=1}^\m \mathrm{\W}_{ij} (Y_i - Y_j),
\end{split}
\end{align}
where $i \in \{1,\dots ,\m\}$. Note that the homogeneous consumer network given by \eqref{eq:hom_sys} is a special case of the symmetric semi-homogeneous network with $m=1$. The utility of the block model \eqref{eq:sym_hom_sys} is that  it presents a picture of the society at a community level which can aid in informed decision making at the macro-level. The block model includes the effects of community sizes and the linkages between them, similar to a regular network except that the nodes are not individual consumers but groups of consumers with similar characteristics. Obtaining \eqref{eq:sym_hom_sys} in the same form as the original model \eqref{eq:hom_sys} shows that the block model and aggregation scheme are scalable, and also carries forward the interpretation of system variables and parameters.

\subsection{Bonding and Bridging in the Symmetric Semi-\newline homogeneous Network}

Here we discuss interpretations of the lumped social ties $\mathrm{W}_{ij}$ in the symmetric semi-homogeneous consumer network defined by \eqref{eq:sym_hom_sys}. From the definition of $\mathrm{W}_{ij}$ we have that $\sum_j \mathrm{W}_{ij} = 1$, for all $i \in \{1,\dots,\m\}$, however $\mathrm{W}_{ii} \geq 0 \,\, \forall \,\, i$, as opposed to $\w_{ii}$ which is zero by definition. $\mathrm{W}_{ii}$ is defined by
\begin{equation*}
	\mathrm{W}_{ii} = 1 - \sum_{\substack{j=1\\j\neq i}}^\m \mathrm{W}_{ij} = \sum_{s \in \N_i} \omega_{rs} \,\, \text{, for any} \,\, r \in \N_i.
\end{equation*}
Thus $\mathrm{W}_{ij}$ gives the tie-strength between members of group $i$, which represents the \emph{bonding} capital of the group. Conversely $\mathrm{W}_{ij}, i \neq j$ represents the \emph{bridging} capital between the groups. Both these quantities are complementary i.e., a group that has strong bonding will have weak bridging and vice versa which also agrees with observations of real-world communities (see \cite{easley2010networks}). 

It is also interesting to note that the aggregated model \eqref{eq:sym_hom_sys} allows the possibility of isolated nodes in the network, which did not exist in the original model. Such a node represents a group that has no bridging capital and maximum bonding capital.

\section{Aggregation Mechanisms for Non-Canonical Consumer Networks}
Here we look at aggregation mechanisms for populations that do not satisfy the requirements for the canonical networks presented in Section \ref{sec:lumped}. We first consider communities that have no leaders or followers as given by condition (iii) of Definition 1. We call the underlying network for such populations a \emph{self-directed} network. We show for such networks, that one can construct a block-model without any inaccuracies by intelligently selecting the values for the lumped parameters. We then consider aggregation when accurate information about the consumer characteristics may not be available in a self-directed network.

\label{sec:approx}
\subsection{An Aggregation Mechanism for Self-directed Networks}
Assume a self-directed network, i.e., a population of $\n$ consumers in which there are no leaders or followers. Thus the following condition holds
\begin{equation} 
\label{eq:no_lf}
	\sum_{j=1}^\n\! \big(  \w_{ji} \, \b_j\upnu_j -\w_{ij} \, \b_i\upnu_i \big) = 0  \quad \forall \, i \in \{1, \dots, \n\}. \end{equation}

Note that no assumptions have been made yet on the ecological weights $\upalpha_i$ and the environmentalisms $\uprho_i$. In what follows, we show that the aggregate consumption dynamics for such a network can be described similar to \eqref{eq:hom_sys} if the choices for the aggregate environmentalism and ecological weight are made appropriately. If \eqref{eq:no_lf} holds, then the aggregate consumption evolves as
\begin{align*}
	\dot{Y} &= \sum_{i = 1}^\n \b_i \upalpha_i (x - \uprho_i)= \sum_{i=1}^\n \b_i \upalpha_i x - \sum_{i=1}^\n \b_i \upalpha_i \uprho_i, \\
		&= \n \left( \frac{1}{\n}\sum_{i=1}^\n \b_i \upalpha_i \right) \left( x - \sum_{i=1}^\n \frac{\b_i \upalpha_i}{\,\,\,\,\text{\raisebox{-2pt}{$\sum_{j=1}^{\n}\b_j \upalpha_j $}}\,\,}  \,\,\uprho_i \right),\\[5pt]
		&= \n \,\hat{\Beta}\hat{\Alpha}\big(x - \hat{\Rho}\,\big),
\end{align*}
where the lumped parameters $\hat{\Beta}, \hat{\Alpha}$ and $\hat{\Rho}$ have been chosen such that
\begin{align}
\label{eq:hom2_cond}
	\hat{\Beta} \hat{\Alpha} = \frac{1}{\n}\sum_{i=1}^\n \b_i \upalpha_i; \quad \hat{\Rho} = \sum_{i=1}^\n \frac{\b_i \upalpha_i}{\,\,\,\,\text{\raisebox{-2pt}{$\sum_{j=1}^{\n}\b_j \upalpha_j $}}\,\,}  \,\,\uprho_i. 
\end{align}

Thus $\hat{\Beta}$ and $\hat{\Alpha}$ are chosen such that the lumped product $\hat{\Beta} \hat{\Alpha}$ is the average of the individual products $\b_i \upalpha_i$ and the lumped environmentalism $\hat{\Rho}$ is chosen as a convex combination of the individual environmentalisms $\uprho_i$. It is easy to see that this combination is equal to the equilibrium value of the resource and so $\hat{\Rho} = \bar{x}$, where $\bar{x} = \lim_{t\rightarrow \infty}x(t)$ (if the limit exists). Given this aggregation mechanism, the coupled resource and aggregate consumption dynamics are expressible as
\begin{align}
\label{eq:hom_sys2} %homogeneous system
\begin{split}
	&\dot{x} = (1-x)x - x Y,\quad \dot{Y} = \n \,\hat{\Beta} \, \hat{\Alpha} \, (x - \hat{\Rho}),
\end{split}
\end{align}
where $\hat{\Beta}, \hat{\Alpha}$ and $\hat{\Rho}$ are chosen in accordance to \eqref{eq:hom2_cond}.

\subsection{Aggregation for Self-directed Populations with \newline Unknown Characteristics}
Consider a self-directed consumer network whose dynamics are given by \eqref{eq:hom_sys2}. If information of the population characteristics are not available, one may aggregate the system as
\begin{align}
\label{eq:un_obs} %homogeneous system
\begin{split}
	&\dot{\tilde{x}} = (1-\tilde{x})\tilde{x} - \tilde{x} \tilde{Y},\quad \dot{\tilde{Y}} = \n \,\tilde{\Beta} \, \tilde{\Alpha} \, (x - \tilde{\Rho}),
\end{split}
\end{align}
where $\tilde{\Beta}$, $\tilde{\Alpha}$ and $\tilde{\Rho}$ are chosen at the aggregator's discretion. One can then define an error vector $e = [e_x \, e_Y]^T$ such that $e_x = x - \tilde{x}$ and $e_Y = Y - \tilde{Y}$. From \eqref{eq:un_obs} and \eqref{eq:hom_sys2}, the steady state error is expressible as follows
\begin{align*}
	&\bar{e}_x = \bar{x} - \tilde{\Rho}, \quad \bar{e}_Y = (1-\bar{x}) - (1-\tilde{\Rho}),
\end{align*}
which shows that if $\tilde{\Rho}$ is chosen to be equal to $\bar{x}$, then the steady state error is zero regardless of the choice of $\tilde{B}$ and $\tilde{\Alpha}$. Figure \ref{fig:compare} shows the behavior of this approximation scheme when the initial conditions are known perfectly. It can be seen that the error exists only during the transient.

\begin{figure}[t!]
    \begin{center}
    \captionsetup{width=\linewidth}
    \begin{subfigure}[t]{0.5\linewidth}
        \captionsetup{width=0.9\linewidth}
        \includegraphics[width = \linewidth]{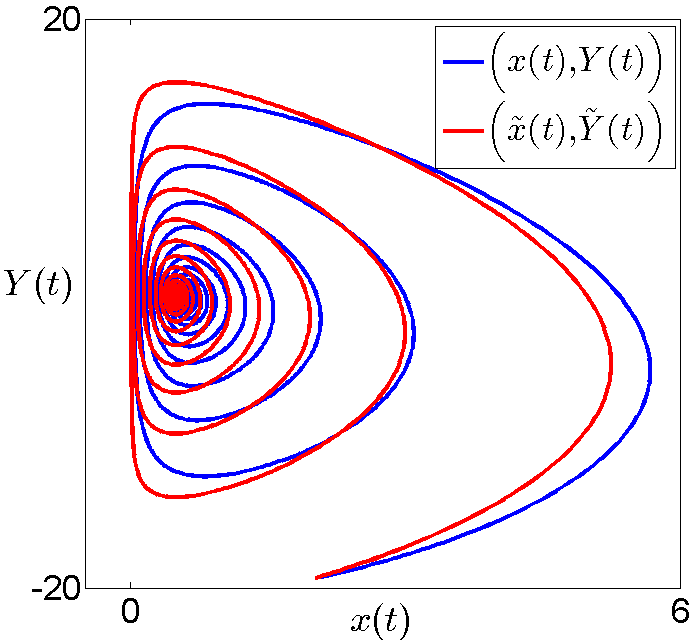}
%        \caption{ }
%        \label{fig:results_1}
    \end{subfigure}%
    \begin{subfigure}[t]{0.5\linewidth}
        \captionsetup{width=0.9\linewidth}
        \includegraphics[width = \linewidth]{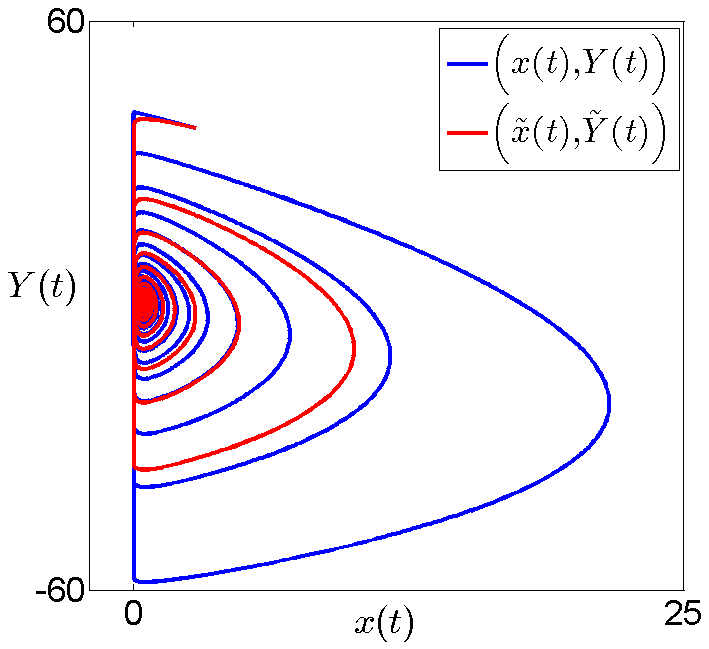}
%        \label{fig:results_2}
    \end{subfigure}
   \end{center}
\caption{Trajectories of original and approximated aggregated systems when $\tilde{\Rho}=\bar{x}$, for two different initial conditions. Here $\n = 100$, $\b_i = 1 \,\, \forall \,\, i$, $\w_{ij} = 1/\n \,\, \forall \,\, i \neq j$, $\upalpha_i \sim \mathrm{U}(0,1)$ and $\uprho_i \sim \mathrm{U}(0,1)$.}
\label{fig:compare}
\end{figure}

This suggests that for systems at steady state for which information on the resource is readily available, one may aggregate such that $\tilde{\Rho} = \bar{x}$ and choose any estimate for parameters $\tilde{\Beta}$ and $\tilde{\Alpha}$. In the event of any disturbances, it is then guaranteed that the error in aggregation will exist only during the transient.

\section{Discussion}

In this chapter we have studied the structural effects of the resource consumption model under an aggregation mechanism via homogeneous communities present in the population. The aggregation yields a concise representation of extremely complex networks, which not only makes the analysis tractable, but can also aid in informed decision making at the macro-level. While such block models are frequently employed in practice, in most cases the condition of homogeneity within groups is more of an approximation than ground reality. We have also studied an approximate aggregation scheme for such cases. Building on such techniques, it may be possible to quantify the error of such approximations so that an appropriate level of confidence may be attached to the resulting decisions. The analysis of this chapter has also revealed that the lumped parameter model is simply a scaled version of the original model. This allows us to view the agents in the network generically as groups of consumers since a single individual is simply a group with one member. 

The question on how different individual actions can be aggregated to obtain social dynamics at the societal level is still an open one, particularly in the context of natural resource governance. The problem however, has attracted the attention of the economics community for quite a while now. One of the first mathematical treatments on the subject is widely attributed to Henri Theil, who published his book ``Linear Aggregation of Economic Relations" in 1954 \cite{theil1954linear}. There he considers the relationship between macro and micro variables, where the micro variables are expressed via certain ``observation" variables through a linear relationship. Theil then goes on to determine the implied linear relationship between the macro variables  i.e., the sums of the observations and the micro variables. He also considers the errors that may result from aggregation (also see some related publications \cite{grunfeld1960aggregation, misra1969note}). A detailed account of the aggregation problem in economics which includes recent developments is given by Green in \cite{green2015aggregation}, where a distinction is made between aggregation in the context of production functions, economic relations and welfare maximization. A widely acclaimed result in the latter is given by Arrow's Impossibility Theorem \cite{maskin2014arrow} where Arrow shows that it is impossible to convert individual utilities into a single objective function that simultaneously ``satisfies" all individuals.  In this chapter, we have studied the aggregation problem for our model of natural resource consumption with a special focus on the underlying social network, an aspect that has not received much attention in the work mentioned above. The work can thus be viewed as a small step towards bridging mathematical network models \cite{bodin2011social} (that typically use only a small number of nodes) and agent based models \cite{anderies2000modeling}(that lack tractability) in the domain of natural resource governance.

With this, we conclude the modeling exercise of Part \RN{2} of this dissertation. Part \RN{3} explores the potential of the model through the application of various mathematical frameworks. As we will see, the exercise yields stimulating insights regarding the system with interesting implications and guidelines for policy making in NRM problems.

\clearpage

\part{Cybernetic Insights, Applications and Extensions}
\clearpage
%% This is an example first chapter.  You should put chapter/appendix that you
%% write into a separate file, and add a line \include{yourfilename} to
%% main.tex, where `yourfilename.tex' is the name of the chapter/appendix file.
%% You can process specific files by typing their names in at the 
%% \files=
%% prompt when you run the file main.tex through LaTeX.
\chapter{Understanding Open-loop System Behavior}

\label{chap:open}

The preceding chapters have led to a mathematical model of our assumed social-ecological system. Once such a model is obtained, there exist various methods to analyze system performance \cite{franklin1994feedback, katsuhiko1970modern}. In designing an effective control for the system, it is important to have a basic understanding of the process before the application of any feedback mechanism. This not only gives an indication of the required complexity of the control, but also provides a common basis for performance comparison of various control mechanisms. In this chapter, we undertake a time-domain analysis of the system at two different levels of abstraction: the single agent network and the dual agent network (the $\n$-agent network is the focus of Chapter \ref{chap:n}). Where applicable, we discuss various social and environmental implications of the results that the undertaken analysis yields.

\section{The Single-Agent/Homogeneous Network}

Here we consider the single agent consumer network. The single agent may be thought of as a society with a single consumer, or a homogeneous society (see Chapter \ref{chap:lump}) lumped into a single node. We reproduce the coupled resource and consumption dynamics as follows
\begin{align}
\label{eq:hom_sys1} %homogeneous system
\begin{split}
	&\dot{x} = (1-x)x - x y,\\
	&\dot{y} = \b \, \upalpha \, (x - \uprho).
\end{split}
\end{align}
Here we have dropped the indexed notation for $y_1$ and have simply used $y$ instead since there is only one equation for the social sub-system.

\subsection{Fixed-point Analysis}
Natural resources have been essential inputs to economic development, which is considered necessary to overcome world poverty. It has been argued for decades now that ``development" cannot imply infinite economic growth, as it has obvious restrictions due to the finite limits of the environment, which houses our economic system \cite{perman2003natural,meadows1972limits}. Many influential economists have championed the concept of a steady-state economy, where the question is not that of how to achieve maximum growth, but of how to realize the most attractive equilibrium. In the past, there has been extensive debate on the viability and usefulness for the practice of the concept of the steady-state economy. Some recent studies suggest to interpret the economy's steady-state as an ``unattainable-goal" \cite{kerschner2010economic} and to view its analysis as an efficient tool, which is useful to guide the design of the long-term policy for real-world economies. This gives special importance to examining the steady states of mathematical models of economic growth and resource consumption, for providing long-term strategies especially in the context of sustainability. In this spirit, we analyze the equilibrium of \eqref{eq:hom_sys1}, which is found out to be unique and given as
\begin{align}
\label{eq:hom_equi}
\begin{split}
	&\xbar = \uprho,\\
	&\ybar = 1-\uprho.
\end{split}
\end{align}
Note that according to the definition of the model in Chapter \ref{chap:model}, $\uprho \in \mathbb{R}$. However, from the form of the resource dynamics (see Equation \ref{eq:hom_sys1}), it can be easily seen that for a positive initial resource stock, the stock remains non-negative for all time i.e., if $x(0) \geq 0$, then $x(t) \geq 0 \,\, \forall \,\, t$. Thus, the equilibrium ($\xbar$, $\ybar$) exists only if $\uprho \geq 0$, otherwise no equilibrium exists. In the analysis that follows we assume that this condition holds. 

From Equation \ref{eq:hom_equi} we further note that the resource stock at equilibrium simply equals the environmentalism $\uprho$. Thus, the higher the environmentalism, the higher the equilibrium stock level. However from \ref{eq:hom_equi}, we also see that higher environmentalism reduces $\ybar$ to the point that if $\uprho > 1$ then $\ybar < 0$. This means that if the environmentalism is beyond the natural carrying capacity, the stock level will settle at that value, however it will be at the cost of a negative aggregate consumption effort for the society. In order for the society to harvest a positive resource stock at steady state, it is then necessary that $\uprho \in (0,1)$. 

\subsection{Local Behavior}
\begin{figure*}[t!]
    \centering
    \captionsetup{width=0.85\textwidth}
    \begin{subfigure}[t]{0.45\textwidth}
        \captionsetup{width=0.9\textwidth}
        \centering
        \includegraphics[trim = 30mm 70mm 30mm 70mm, clip, width = \linewidth]{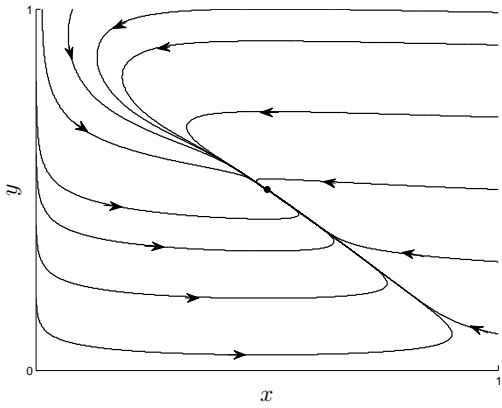}
        \caption{A stable node. Here $\b = 0.1$, $\uprho = 0.5$ and $\upalpha = 0.5$.}
        \label{fig:1agent_node}
    \end{subfigure}%
	\hspace{0.05\textwidth}
    \begin{subfigure}[t]{0.45\textwidth}
        \captionsetup{width=0.9\textwidth}
        \centering
        \includegraphics[trim = 30mm 70mm 30mm 70mm, clip, width = \linewidth]{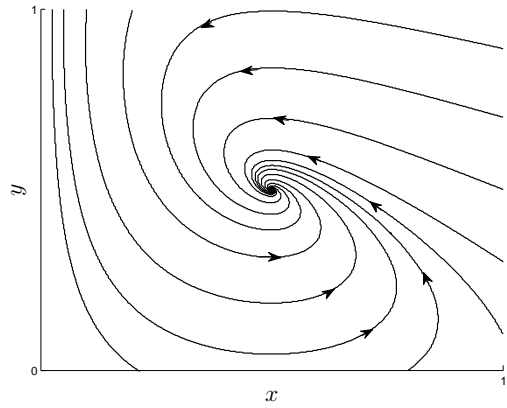}
        \caption{A stable spiral. Here $\b = 1$, $\uprho = 0.5$ and $\upalpha = 0.5$.}
        \label{fig:1agent_spiral}
    \end{subfigure}
    \begin{subfigure}[t]{0.45\textwidth}
        \captionsetup{width=0.9\textwidth}
        \centering
        \includegraphics[trim = 30mm 70mm 30mm 70mm, clip, width = \linewidth]{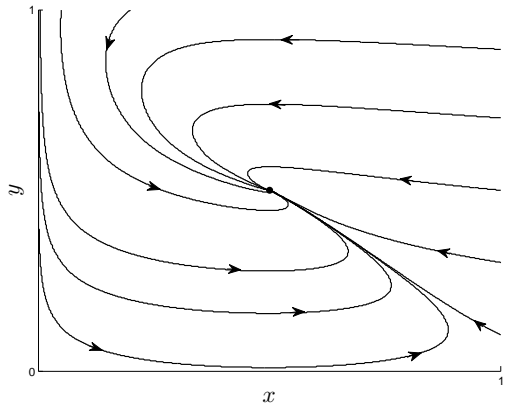}
        \caption{A stable degenerate node. Here $\b = 0.5$, $\uprho = 0.5$ and $\upalpha = 0.25$.}
        \label{fig:1agent_spiral}
    \end{subfigure}
    \caption{The nature of the fixed point ($\xbar, \ybar$), as the system parameters are varied.}
\label{fig:1agent_equi}
\end{figure*}
The linearized system around ($\xbar$, $\ybar$) is given by the matrix
\begin{align*}
	\left[ \begin{array}{cc} -\uprho & -\uprho \\ \b \upalpha & 0 \end{array} \right],
\end{align*}
whose eigenvalues are given by
\begin{align*}
	\uplambda_{1,2} = - \frac{\uprho}{2} \pm \frac{\sqrt{\uprho^2 - 4\b \upalpha \uprho}}{2}.
\end{align*}
As $\uprho > 0$, the eigenvalues have negative real part which implies that ($\xbar$, $\ybar$) is locally stable. The discriminant of the characteristic equation determines how the system approaches the steady state. This categorizes the behavior of the system into three cases (see \cite[Chapter 9]{boyce1969elementary})
\begin{itemize}
	\item \emph{Case 1):} The eigenvalues are real and distinct. In this case $\uprho>4\b \upalpha$ and the equilibrium is a stable node.
	\item \emph{Case 2):} The eigenvalues are complex conjugates. In this case $\uprho<4\b \upalpha$ and the equilibrium is a stable spiral.
	\item \emph{Case 3):} The eigenvalues are real and repeated. In this case $\uprho = 4 \b\upalpha$ and the linearized system has a single independent eigenvector $[-2\,\,1]^T$. In this case, the equilibrium is a stable degenerate node. 
\end{itemize}
This suggests that oscillatory system behavior emerges in the case of low environmental concern and high levels of sensitivity \& preference to ecological information. The same may be observed from Figure \ref{fig:1agent_osc}.
\begin{figure}[b!]
	\captionsetup{width=0.8\textwidth}
	\begin{center}
		\includegraphics[width=\linewidth]{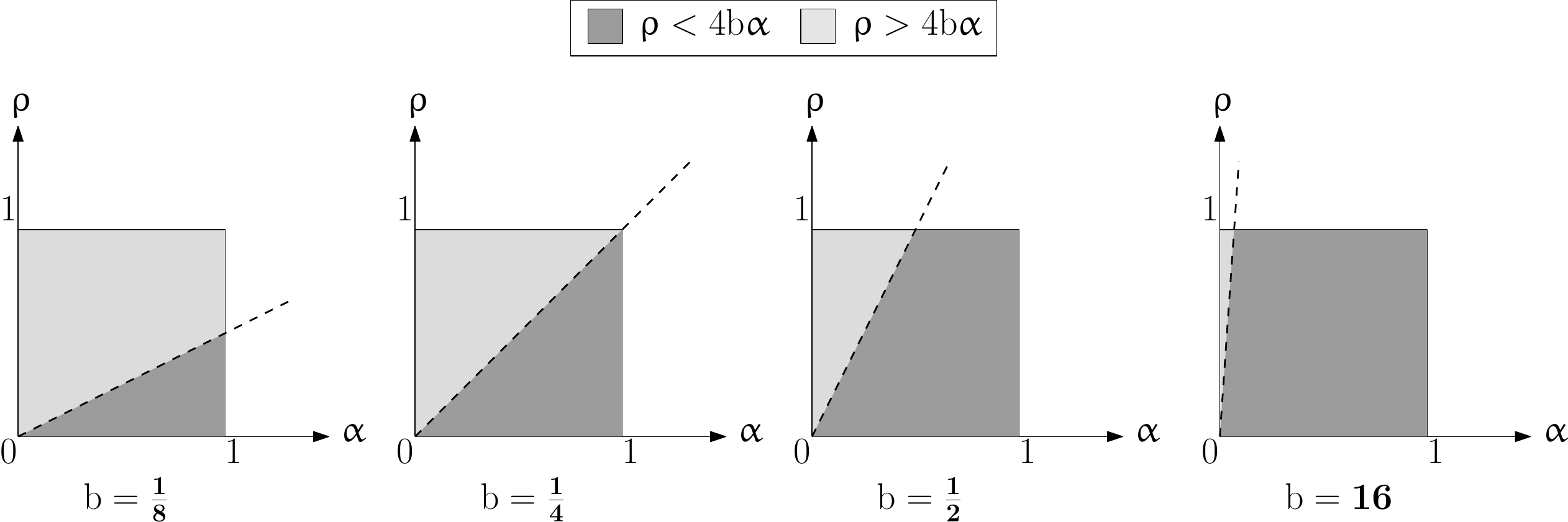}
	\end{center}
	\caption{Regions in the parameter space depecting the nature of the equilibrium for system \eqref{eq:hom_sys1}. The darkly shaded area represents the oscillatory region whereas the lightly shaded area represents the nodal region.}
	\label{fig:1agent_osc}
\end{figure}

\subsection{Global Behavior}
Above, we studied the local stability of the equilibrium which implies that there exists at least some region of the phase space, such that if the system is initiated from any point in this region the dynamics asymptotically approach the equilibrium. This region exists as a non-empty neighborhood of the equilibrium. Since there is only a single fixed point of system \eqref{eq:hom_sys1}, it is possible for this region to cover the entire two-dimensional space. However, since the system is non-linear, there exists the possibility of limit cycles \cite[Chapter 9]{boyce1969elementary}, which can be determined through a global stability analysis.

Introduce new variable $z = \ln x$. System \eqref{eq:hom_sys1} is then transformed as follows
\begin{align*}
\begin{split}
	&\dot{z} = 1 - e^{z} - y, \quad
	\dot{y} = \b \upalpha(e^{z}-\uprho).
\end{split}
\end{align*}
This system has a unique equilibrium which is given by
\begin{align*}
	&\zbar = \ln \uprho, \quad \ybar = 1 - \uprho.
\end{align*}
Introducing further new variables $p = z - \zbar$ and $q = y - \ybar$ to bring the equilibrium point to the origin, we obtain the following system
\begin{align}
\begin{split}
\label{sys_stability1}
	&\dot{p} = \uprho (1-e^p) - q,\\
	&\dot{q} = -\b\upalpha\uprho(1-e^p).\\
\end{split}
\end{align}
\begin{theorem}
	The system \eqref{sys_stability1} is globally stable with all solutions tending to zero.
\end{theorem}
\begin{proof}
	In order to prove the theorem we construct a Lyapunov function of the following form
\begin{align*}
	V = (e^z - z - 1) + \P q^2,
\end{align*}
which is positive for positive $\P$ everywhere except the origin. Thus if $\dot{V}$ is shown to be negative outside the origin for some positive $\P$,, it would imply global stability of \eqref{sys_stability1}. Differentiating $V$ along the vector field of system \eqref{sys_stability1} we get
\begin{align*}
	\dot{V} &= e^z\dot{z} - \dot{z} + 2 \P q\dot{q}.
\end{align*}
Choose $\displaystyle \P = \frac{1}{2\b\upalpha\uprho}$. Note that as $\uprho$ has been restricted to be positive, $\P>0$ so it does not violate the status of $V$ as a Lyapunov function candidate. Substituting this and \eqref{sys_stability1} in the above equation gives us
\begin{align*}
	\dot{V}= -\uprho {(e^{z}-1)}^2.	
\end{align*}
Now note that the derivative $\dot{V}<0$ outside the $q$-axis, and on this axis the vector field  of \eqref{sys_stability1} has value $(-q,0)$, which is not tangent to the axis. This implies that the system is globally stable, which completes the proof.
\end{proof}

\section{The Dual-Agent/Semi-Homogeneous Network}

\label{sec:open_dual}

We now consider consumer network with two agents. The agents may represent individual consumers or groups of consumers with homogeneous characteristics as described in Chapter~\ref{chap:lump}. The consumption dynamics for the two-agent society are given by
\begin{align}
\hspace{0pt}
\label{eq:two_com_sys}
\begin{split}
	&\dot{x} =  (1-x)x - (y_1+y_2)x,\\
	&\dot{y}_1 = \b_1 \Big( (1-\upnu_1)(x -\uprho_1)- \upnu_1\left( y_1 - y_2 \right) \Big),\\
	&\dot{y}_2 = \b_2 \Big( (1-\upnu_2)(x -\uprho_2)- \upnu_2\left( y_2 - y_1 \right) \Big),
\end{split}
\end{align}
where $\b_i$ and $\upnu_i$ have been redefined during the non-dimensionalization process to include $\w_{ij}$. Thus $\displaystyle \b_i = \frac{(\n_i\a_i \Rmax + \r \s_i\w_{ij})}{\r^2}$ and $\displaystyle \upnu_i = \frac{\r \s_i \w_{ij}}{(\n_i\a_i \Rmax + \r \s_i \w_{ij})}$.

\subsection{Fixed-point Analysis}
By solving the system given by $\dot{x}= \dot{y}_1= \dot{y}_2=0$ (Equation \ref{eq:two_com_sys}) we find that the system has a unique equilibrium which is given by
\begin{align}
\label{eq:equi-2comp}
\begin{split}
	&\xbar =   \frac{\upalpha_1 \upnu_2\uprho_1 + \upalpha_2 \upnu_1 \uprho_2}{\upalpha_2 \upnu_1 + \upalpha_1\upnu_2}, \\[10pt]
	& \ybar_1 = \frac{(1-\uprho_1)\upalpha_1\upnu_2 + (1-\uprho_2)\upalpha_2\upnu_1 - (\uprho_1-\uprho_2)\upalpha_1\upalpha_2}{2 \left( \upalpha_2 \upnu_1 + \upalpha_1\upnu_2 \right)}, \\[10pt]
	& \ybar_2 = \frac{(1-\uprho_1)\upalpha_1\upnu_2 + (1-\uprho_2)\upalpha_2\upnu_1  - (\uprho_2-\uprho_1)\upalpha_1\upalpha_2}{2 \left( \upalpha_2 \upnu_1 + \upalpha_1\upnu_2 \right)}. 
\end{split}
\end{align}
Note here that $\xbar$ is a convex combination of the $\uprho_i$'s. If the point representing this combination has a negative value, then due to \eqref{eq:two_com_sys}, $\xbar$ is not reachable if $x(0)>0$, and so, no equilibrium exists. Thus, similar to our argument for the single-agent case in the previous section, we will restrict our attention to the case $\uprho_1, \uprho_2 \in (0,1)$. This results in $\xbar$ turning out to lie between $0$ and $1$ for all admissible parameter values, which ensures a non-negative resource stock that does not exceed the natural carrying capacity at steady state. The efforts $\ybar_i$ can be either both positive, or have different signs, the interpretation of which has been discussed in Chapter \ref{chap:model}. If equilibrium \eqref{eq:equi-2comp} is stable, the society converges to a stable resource stock and stable consumption rate regardless of the initial state of the system. 

Before proceeding with the analysis, let us point out that there are two degenerate cases of the equilibrium \eqref{eq:equi-2comp}. First, if both groups have extremely low social value, i.e., if $\upnu_1=\upnu_2=0$, then an infinite number of equilibria exist if the groups are equally environmental ($\uprho_1=\uprho_2$), otherwise no equilibrium exists. Second, if both groups have extremely high social value, i.e., if $\upnu_1=\upnu_2=1$, then infinite equilibria exist for all values of $\uprho_1$ and $\uprho_2$. We ignore these pathological cases in the subsequent analyses.

\subsubsection{Comparative Statics}

While the dependence of the equilibrium on the various system parameters was relatively straightforward for the single-agent network, such is not the case for the two-agent network. Here we inspect for system \eqref{eq:two_com_sys}, how the magnitudes of the steady-state resource quantity $\xbar$ and the steady-state consumption efforts $\ybar_1$ and $\ybar_2$ change in response to changes in the parameters $\uprho_1$, $\uprho_2$, $\upnu_1$, and $\upnu_2$. We do this by conducting a comparative statics analysis \cite{perman2003natural}, which is undertaken by obtaining the first-order partial derivatives of each of the three equations in \eqref{eq:equi-2comp}, and determining whether or not an unambiguous sign can be attached to each partial derivative. Table \ref{tab:comp_stat} shows the results for the analysis. A plus sign means that the derivative is positive, a minus sign means that the derivative is negative, ``0" means that the derivative is zero and a ``?" means that no sign can be assigned unambiguously to the derivative. 

\begin{table}[t]
\centering
\caption{Comparative statics results for system \eqref{eq:two_com_sys}}
\begin{mdframed}
\vspace{10pt}
\hspace{5pt}
\begin{tabular}{c|c|c|c|c|}
	\cline{2-5}
	 & & & & \\[-1em]
	& $\upnu_1$ & $\upnu_2$ & $\rho_1$ & $\uprho_2$ \\ \cline{1-5}
	\multicolumn{1}{ |c| }{ } & & & & \\[-1em]
	\multicolumn{1}{ |c| }{$\xbar$} & 0 & 0 & + & +     \\ \cline{1-5}
	\multicolumn{1}{ |c| }{ } & & & & \\[-1em]
	\multicolumn{1}{ |c| }{$\ybar_1$} & 0 & 0 & - & ? \\ \cline{1-5}
	\multicolumn{1}{ |c| }{ } & & & & \\[-1em]
	\multicolumn{1}{ |c| }{$\ybar_2$} & 0 & 0 & ? & - \\ \cline{1-5}
	\multicolumn{5}{ c }{ } \\
	\cline{2-4}
	\multicolumn{1}{ c }{ } & \multicolumn{3}{ |c| }{$\uprho_1 = \uprho_2$} & \multicolumn{1}{ c }{ } \\ 	\cline{2-4}
\end{tabular}
\hspace{15pt}
\begin{tabular}{c|c|c|c|c|}
	\cline{2-5}
	 & & & & \\[-1em]
	& $\upnu_1$ & $\upnu_2$ & $\rho_1$ & $\uprho_2$ \\ \cline{1-5}
	 \multicolumn{1}{ |c| }{ } & & & & \\[-1em]
	\multicolumn{1}{ |c| }{$\xbar$} & - & + & + & +     \\ \cline{1-5}
	 \multicolumn{1}{ |c| }{ } & & & & \\[-1em]
	\multicolumn{1}{ |c| }{$\ybar_1$} & + & ? & - & ? \\ \cline{1-5}
	 \multicolumn{1}{ |c| }{ } & & & & \\[-1em]
	\multicolumn{1}{ |c| }{$\ybar_2$} & ? & - & ? & - \\ \cline{1-5}
	\multicolumn{5}{ c }{ } \\
	\cline{2-4}
	\multicolumn{1}{ c }{ } & \multicolumn{3}{ |c| }{$\uprho_1 > \uprho_2$} & \multicolumn{1}{ c }{ } \\ 	\cline{2-4}
\end{tabular}
\hspace{15pt}
\begin{tabular}{c|c|c|c|c|}
	\cline{2-5}
	 & & & & \\[-1em]
	& $\upnu_1$ & $\upnu_2$ & $\rho_1$ & $\uprho_2$ \\ \cline{1-5}
	 \multicolumn{1}{ |c| }{ } & & & & \\[-1em]
	\multicolumn{1}{ |c| }{$\xbar$} & + & - & + & +     \\ \cline{1-5}
	 \multicolumn{1}{ |c| }{ } & & & & \\[-1em]
	\multicolumn{1}{ |c| }{$\ybar_1$} & - & ? & - & ? \\ \cline{1-5}
	 \multicolumn{1}{ |c| }{ } & & & & \\[-1em]
	\multicolumn{1}{ |c| }{$\ybar_2$} & ? & + & ? & - \\ \cline{1-5}
	\multicolumn{5}{ c }{ } \\
	\cline{2-4}
	\multicolumn{1}{ c }{ } & \multicolumn{3}{ |c| }{$\uprho_1 < \uprho_2$} & \multicolumn{1}{ c }{ } \\ 	\cline{2-4}
\end{tabular}
\vspace{10pt}
\end{mdframed}
\label{tab:comp_stat}
\end{table}

From Table \ref{tab:comp_stat}, we see that an increase (decrease) in the levels of environmentalism $\uprho_1$ and $\uprho_2$, always increases (decreases) the steady-state resource stock $\xbar$. Thus higher environmentalism results in a larger resource stock at steady-state. In contrast, an increase in the environmentalism $\uprho_i$ of individual $i$, decreases the steady-state consumption $\ybar_i$ of $i$, which means that environmental individuals consume less at steady-state (the same can be seen in Figure \ref{fig:2agent_rho}). It can also be noted that the effect of $\uprho_i$ on $\ybar_j$, where $i\neq j$, is ambiguous and not much can be said by simply observing the first-order partial derivatives.
 
The effect of the social relevances $\upnu_1$ and $\upnu_2$ on the steady state varies according to the relative magnitudes of the environmentalisms $\uprho_1$ and $\uprho_2$. When both groups are equally environmental $\uprho_1=\uprho_2$ then the social values do not affect the steady-state. When this is not the case, the social value of the relatively environmental (non-environmental) group has a negative (positive) effect on the steady-state resource stock. Furthermore, the social value of the relatively environmental (non-environmental) group has a positive (negative) effect on the group's own steady-state consumption effort (the same can be observed in Figure \ref{fig:2agent_nu}). We see that the effects of the parameters on the relationship between the steady-state consumption efforts are not entirely revealed by this exercise, and require additional analysis to uncover, which is done in the next section. 

\begin{figure*}[t!]
    \begin{center}
    \captionsetup{width=0.9\linewidth}
    \begin{subfigure}[t]{0.5\linewidth}
        \captionsetup{width=0.9\linewidth}
        \includegraphics[width = \linewidth]{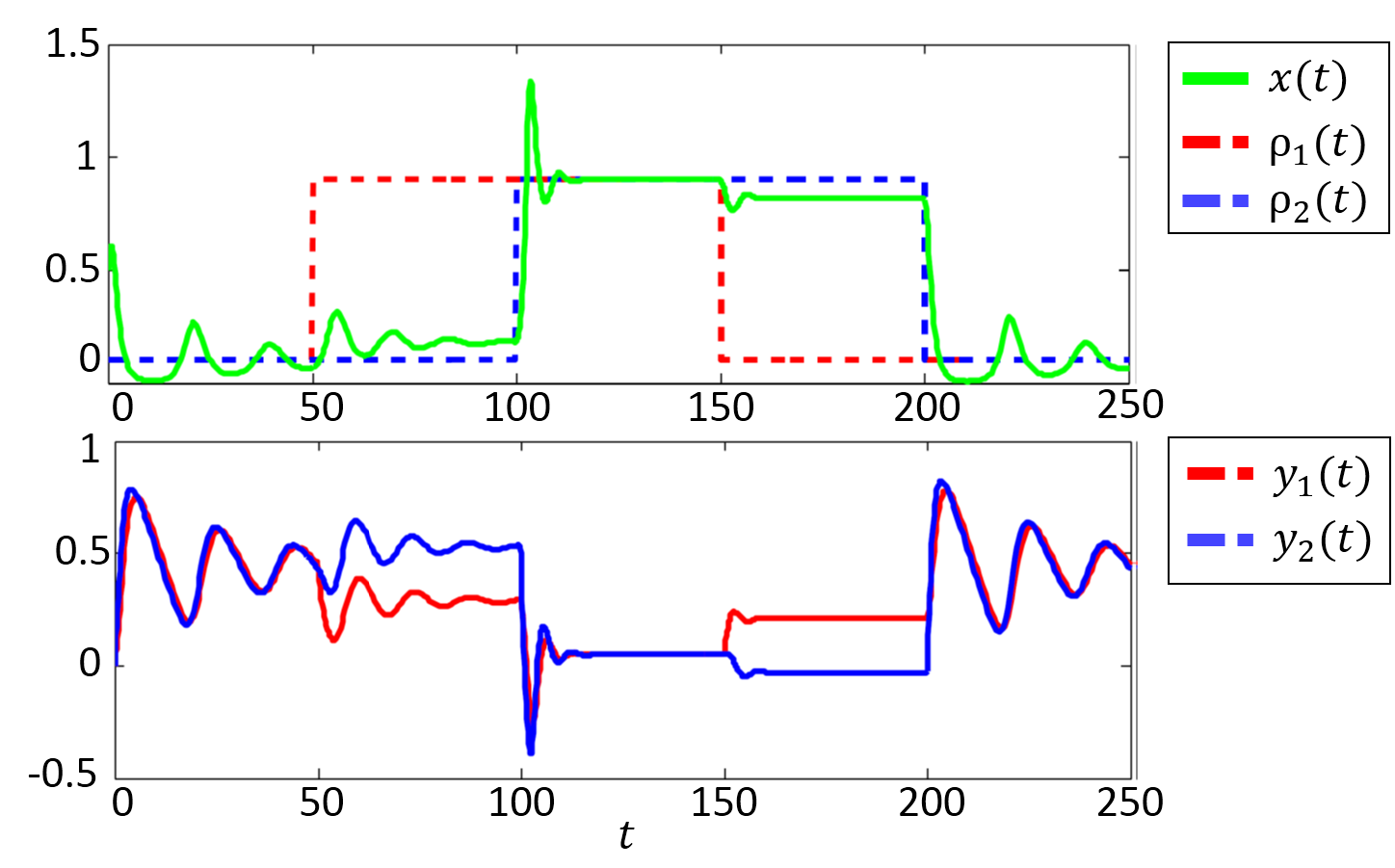}
        \caption{Change in the stock and consumption as the environmentalisms $\uprho_1$ and $\uprho_2$ are varied. It can be seen that due to group 2's higher ecological relevance, the resource follows $\uprho_2$ more closely. It can also be seen that the relatively environmental group 1 consumes less than the non-environmental group 2.}
        \label{fig:2agent_rho}
    \end{subfigure}%
    \begin{subfigure}[t]{0.5\linewidth}
        \captionsetup{width=0.9\linewidth}
        \includegraphics[width = \linewidth]{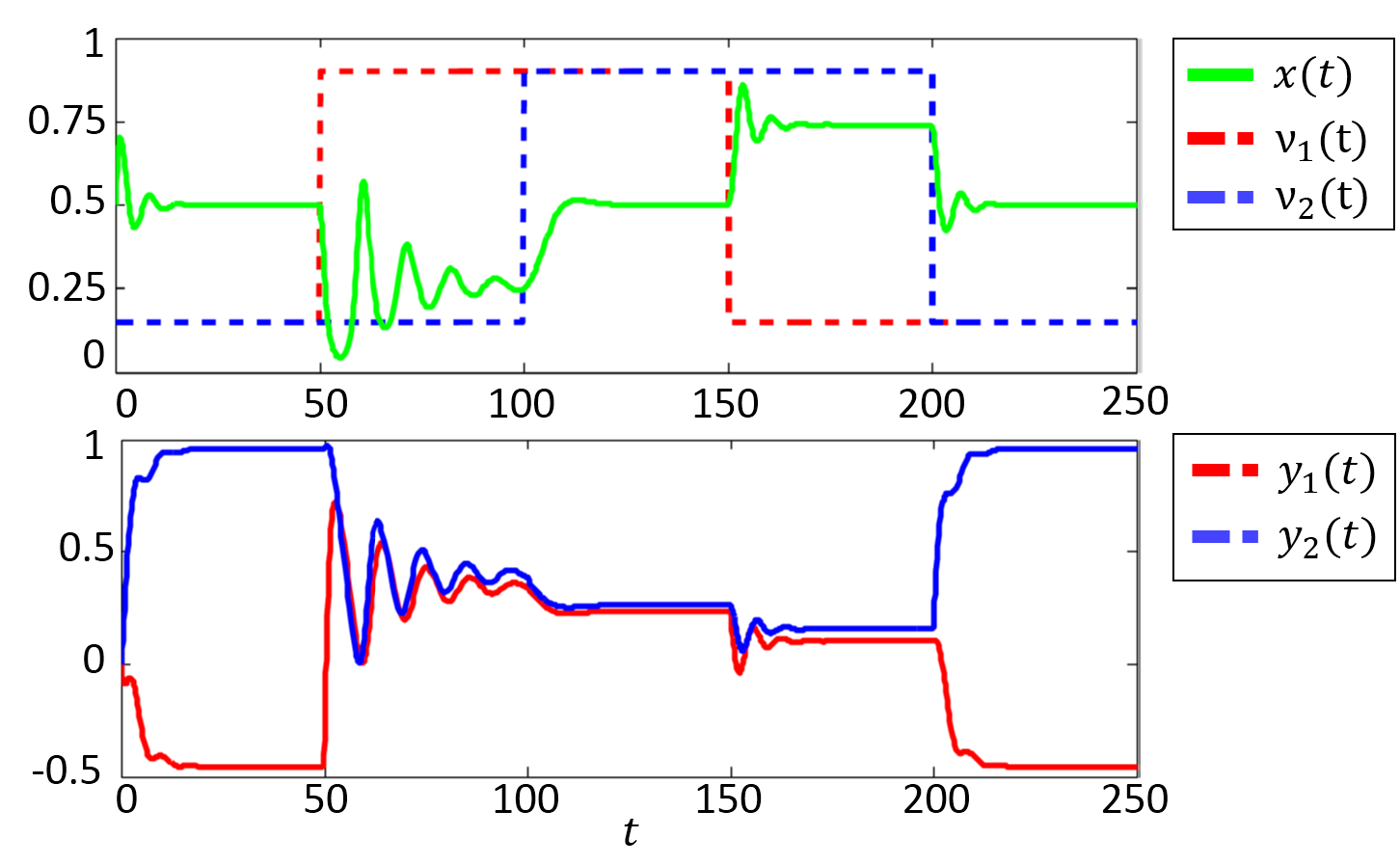}
        \caption{Change in stock and consumption as the social relevances $\upnu_1$ and $\upnu_2$ are varied. It can be seen that when the social relevances are identical, the resource converges at the average of the set-points with the non-environmental group enjoying a higher consumption. However when one of the group's social-relevance is relatively lower, the resource converges at the environmentalism of that group.}
        \label{fig:2agent_nu}
    \end{subfigure}
    \caption{Change in steady state stock and consumption efforts in the two-community model, as the population characteristics are varied. The nominal values selected for this simulation are $\b_1 = \b_2 = 1$, $\upnu_1=0.75$, $\upnu_2=0.25$, $\uprho_1 = 0.75$, $\uprho_2 = 0.25$ and $x(0) = 0.5$ and $y_1(0) = y_2(0) = 0$.}
\label{fig:2agent_openloop}
   \end{center}
\end{figure*}

\subsubsection{The Free-riding Phenomenon}

Ostrom's work \cite{ostrom1990governing} reports several real-world examples of free-riding in open-access resource settings, which is a major cause of concern for the successful governance of such resources. The free-riding phenomenon refers to situations where certain individuals (the self-reliants) sacrifice their consumption for the sustenance of the resource, while simultaneously, other individuals (the free-riders) continue to enjoy the benefits of consumption, which would not have been possible without the sacrifice of the self-reliants in the first place. Thus the free-riders enjoy a ``free lunch" that is being served in-directly by the self-reliants. System \eqref{eq:two_com_sys} captures this through two possible types of equilibrium. We label these two types as the ``self-reliant" equilibrium and the ``free-riding" equilibrium. The self-reliant equilibrium represents both consumer groups harvesting at a positive rate (both $\ybar_1$ and $\ybar_2$ of \eqref{eq:equi-2comp} are positive). The free-riding equilibrium represents one of the groups harvesting at a positive rate and the other harvesting at a negative rate (one of $\ybar_1$ and $\ybar_2$ is positive, while the other negative). Thus one of the groups exerts effort into increasing the resource quantity, while the other free-rides and enjoys a positive consumption.
\begin{figure}[h!]
	\captionsetup{width=0.8\textwidth}
	\begin{center}
		\includegraphics[width=0.9\linewidth]{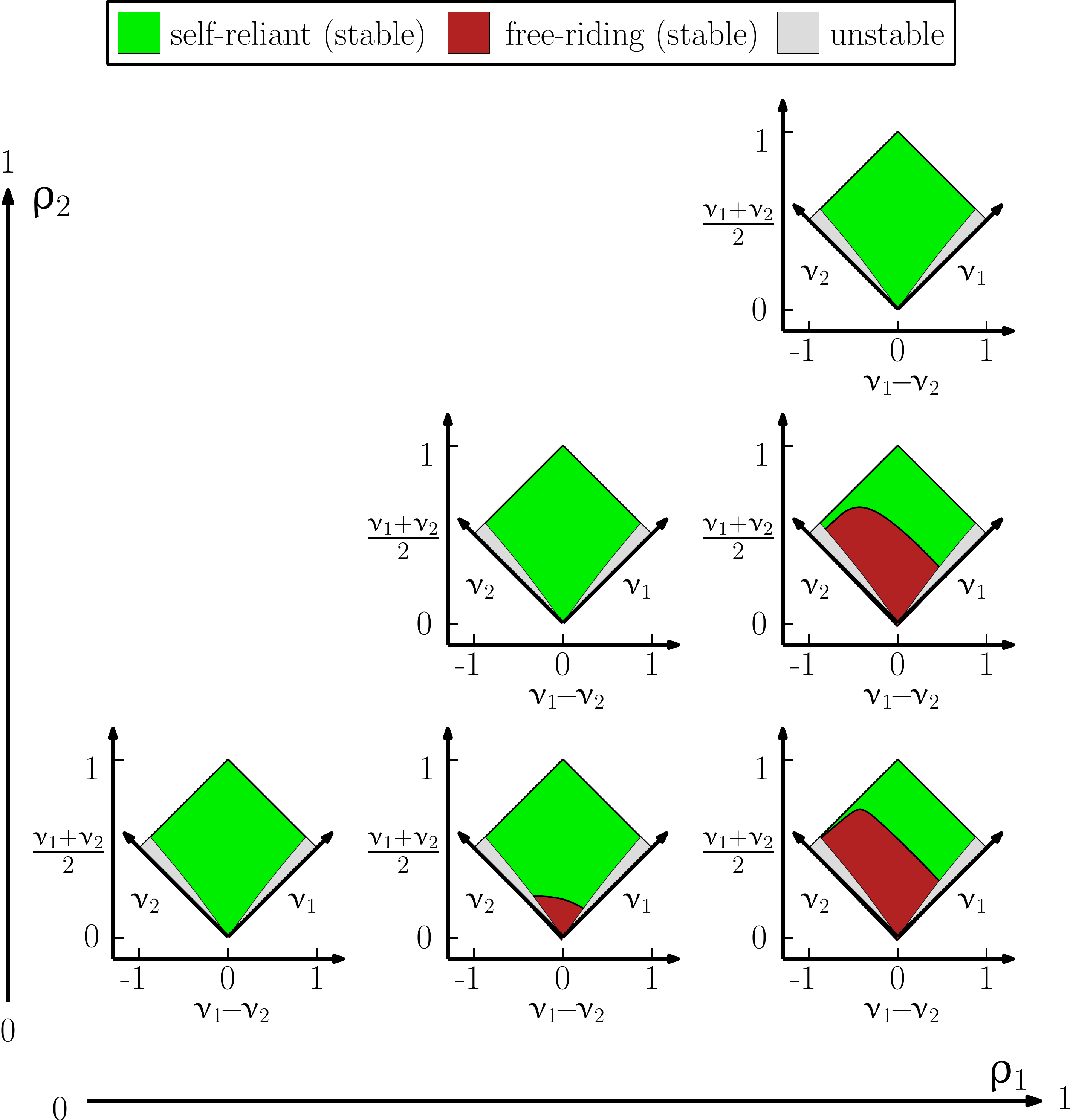}
	\end{center}
	\caption{Positivity of the equilibrium \eqref{eq:equi-2comp} in the $\upnu_1$, $\upnu_2$ plane for different values of $\uprho_1$ and $\uprho_2$. The plots are symmetric in $\uprho_1$ and $\uprho_2$and so the missing plots are simply a reflection of the existing plots about the symmetric axis. The free-riding equilibrium corresponds to $\ybar_1<0<\ybar_2$.}
	\label{fig:free_ride} %net influence
\end{figure}
An exhaustive simulation of the equilibrium consumption values in the parameter space is shown in Figure \ref{fig:free_ride}. We see that free-riding is excluded only when both groups have the same environmentalism level $\uprho_i$. When groups are different in what concerns their environmentalism, the equilibrium is self-reliant if the group corresponding to the higher level of environmentalism has a social relevance beyond a certain level. This level can be interpreted as the reluctance of the relatively environmental group to subsidize the consumption of the relatively non-environmental group. The reluctance increases with increasing levels of social relevance, and indirectly with the level of cooperativeness of the environmental group. Thus the more cooperative a group becomes (consequently possessing a higher social relevance) the more reluctant it is to subsidize the non-environmental group. This is consistent with the postulate that cooperative individuals promote equality, and favor outcomes which maximize joint benefits (contrasted with altruistic individuals, who willingly sacrifice their own benefit in order to maximize the benefit of others). Interestingly, it has been observed that cooperative individuals are also likely to have high levels of environmentalism and vice-versa \cite{de2010relationships, garling2003moderating}, which in the context of our model implies that it is really the social relevance of the groups that determines whether they end up in a self-reliant or free-riding equilibrium.

\subsection{Local Behavior}
The local stability of the equilibrium point of system \eqref{eq:two_com_sys} can be found by examining the eigenvalues of the linearized system about that point. The Jacobian matrix is given by
\begin{align*}
	\left[ \begin{array}{ccc} 1-2x-y_1-y_2 & -x & -x \\ \b_1 \upalpha_1 & -\b_1 \upnu_1 & \b_1 \upnu_1 \\ \b_2 \upalpha_2 & \b_2 \upnu_2 & -\b_2 \upnu_2 \end{array} \right].
\end{align*}
Evaluation at the equilibrium ($\bar{x}$, $\bar{y}_1$, $\bar{y}_2$) results in
\begin{align*}
	\left[ \begin{array}{ccc} -\frac{\upalpha_1\upnu_2\uprho_1+\upalpha_2\upnu_1\uprho_2}{\upalpha_2\upnu_1+\upalpha_1\upnu_2} & -\frac{\upalpha_1\upnu_2\uprho_1+\upalpha_2\upnu_1\uprho_2}{\upalpha_2\upnu_1+\upalpha_1\upnu_2} & -\frac{\upalpha_1\upnu_2\uprho_1+\upalpha_2\upnu_1\uprho_2}{\upalpha_2\upnu_1+\upalpha_1\upnu_2} \\ \b_1 \upalpha_1 & -\b_1 \upnu_1 & \b_1 \upnu_1 \\ \b_2 \upalpha_2 & \b_2 \upnu_2 & -\b_2 \upnu_2  \end{array} \right],
\end{align*}
whose eigenvalues are given by the roots of the following characteristic polynomial
\begin{align}
	\begin{split}
	\!\!\!\!p(\uplambda)=&\uplambda^3 \quad+ \frac{1}{\upalpha_2\upnu_1+\upalpha_1\nu_2} \Big((\b_1\upnu_1 +\b_2\upnu_2) (\upalpha_2\upnu_1+\upalpha_1\upnu_2) + \upalpha_2\upnu_1\uprho_2+\upalpha_1\upnu_2\uprho_1 \Big) \uplambda^2 \\
	+ &\frac{\upalpha_2\upnu_1\uprho_2+\upalpha_1\upnu_2\uprho_1}{\upalpha_2\upnu_1+\upalpha_1\upnu_2} \Big(\b_1 + \b_2\Big)\uplambda \quad + 2\b_1\b_2(\upalpha_2\upnu_1\uprho_2+\upalpha_1\upnu_2\uprho_1).
	\end{split}
\end{align}
the final expressions for the roots of this polynomial are not simple enough to work with analytically. It can be noted however that all the coefficients of $p(\uplambda)$ are positive, which is a necessary (but not sufficient) condition for all roots to be negative and hence the linearized system to be stable. Whether or not the equilibrium is indeed stable can be checked by Routh's Stability Criterion \cite{franklin1994feedback}, which for cubic polynomials of the form $ax^3 + bx^2 + cx + d$ is given by $bc > ad$. For $p(\uplambda)$, this inequality is given as
\begin{align*}
	\begin{split}
	\frac{\upalpha_2\upnu_1\uprho_2+\upalpha_1\upnu_2\rho_1}{(\upalpha_2\upnu_1+\upalpha_1\upnu_2)^2} (\b_1 + \b_2)\Big((\b_1\upnu_1 + \b_2\upnu_2) (\upalpha_2\upnu_1+\upalpha_1\upnu_2) + \upalpha_2\upnu_1\uprho_2+\upalpha_1\upnu_2\uprho_1 \Big)\\
	> 2\b_1\b_2(\upalpha_2\upnu_1\uprho_2+\upalpha_1\upnu_2\uprho_1),
	\end{split}
\end{align*}
which is simplified to
\begin{align}
\label{eq:stab_1} %stability
	\frac{\b_1 + \b_2}{(\upalpha_2\upnu_1+\upalpha_1\upnu_2)^2}\Big((\b_1\upnu_1 +\b_2\upnu_2) (\upalpha_2\upnu_1+\upalpha_1\upnu_2) + \upalpha_2\upnu_1\uprho_2+\upalpha_1\upnu_2\uprho_1 \Big) - 2b_1b_2 > 0.
\end{align}
It can be checked that \eqref{eq:stab_1} does not hold for all $\b_1, \b_2, \uprho_1, \uprho_2, \upnu_1, \upnu_2$. An example of one such point in the parameter space is $\b_1 = 0.2, \b_2 = 0.1, \uprho_1 = 0.001, \uprho_2 = 0.1, \upnu_1 = 0.01, \upnu_2 = 0.9$. On the other hand, the following inequality
\begin{align}
\label{eq:stab_1.5}
\frac{(\b_1 + \b_2)(\b_1\upnu_1 +\b_2\upnu_2)}{(\upalpha_2\upnu_1+\upalpha_1\upnu_2)}> 2\b_1\b_2,
\end{align}
clearly offers a sufficient condition for \eqref{eq:stab_1} to hold. It can be simplified to 
\begin{align}
\label{eq:stab_2}
	\mathrm{q}(b) = \upnu_1 b^2 + \big( \upnu_1(2\upnu_2-1) + \upnu_2(2\upnu_1 -1) \big) b + \upnu_2 > 0,
\end{align}
where $b = \b_1/\b_2$. Note that the right hand side of \eqref{eq:stab_2} is a quadratic polynomial in $b>0$. As $\upnu_i>0$, the graph of $\mathrm{q}(\cdot)$ points upwards and the roots are either complex conjugates or both real and positive. In what follows we examine the implications of each case in detail.
\begin{figure}[t!]
	\captionsetup{width=0.8\textwidth}
	\begin{center}
		\includegraphics[width=0.5\linewidth]{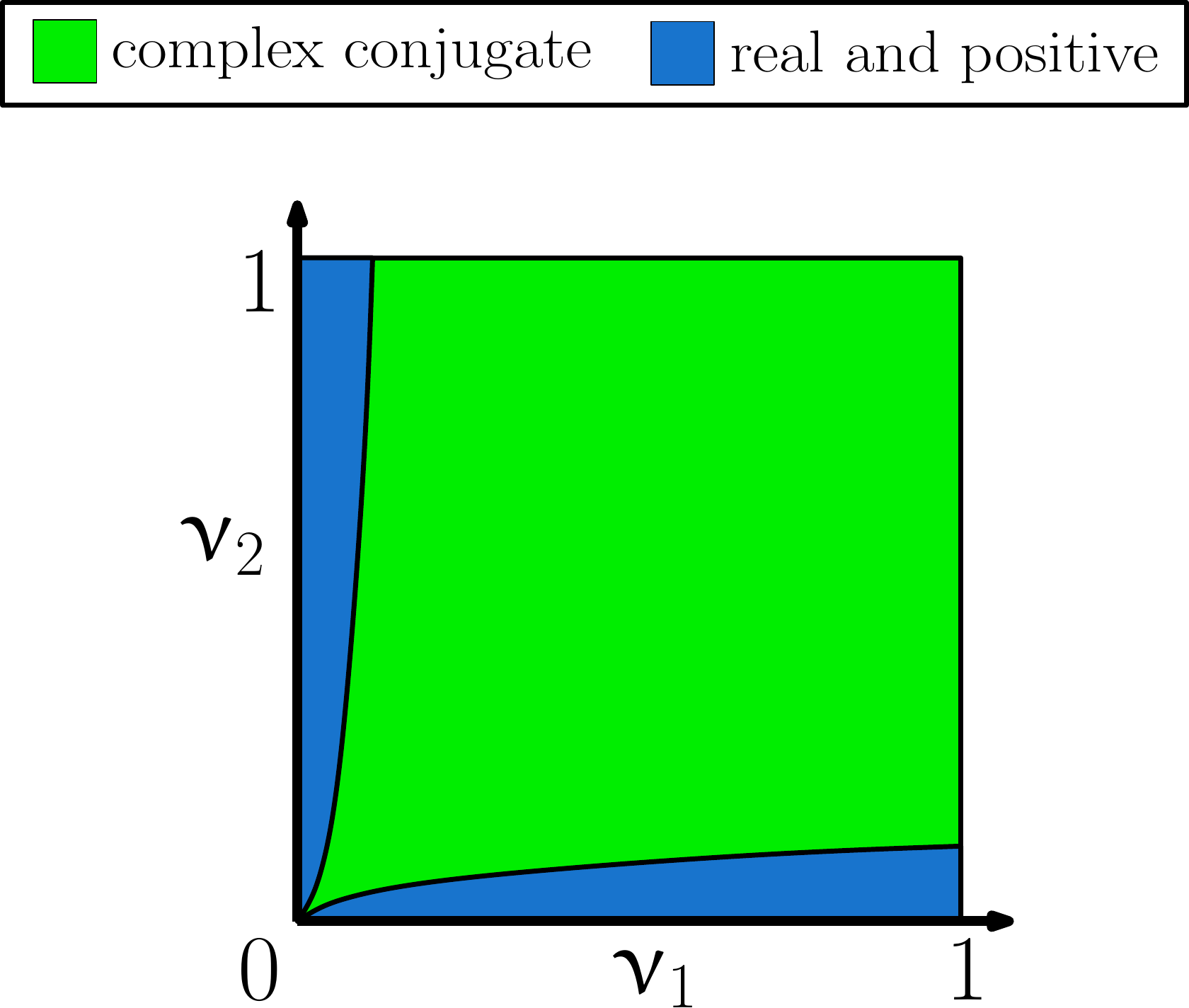}
	\end{center}
	\caption{Evaluation of the roots of $\mathrm{q}(\cdot)$.}
	\label{fig:stab_1} %net influence
\end{figure}
\begin{enumerate}
	\item Case 1: $\mathrm{q}(\cdot)$ has complex conjugate roots i.e., the discriminant is negative.
	\begin{align}
	\label{eq:stab_3}
		{\big( \upnu_1(2\upnu_2-1) + \upnu_2(2\upnu_1 -1) \big)}^2 - 4\upnu_1\upnu_2 < 0.
	\end{align}
	In this case \eqref{eq:stab_2} will always hold as the graph of the quadratic polynomial will lie above the horizontal axis. 		
	\item Case 2: $\mathrm{q}(\cdot)$ has both roots real and positive. This happens if
	\begin{align}
	\label{eq:stab_3.5}
		{\big( \upnu_1(2\upnu_2-1) + \upnu_2(2\upnu_1 -1) \big)}^2 - 4\upnu_1\upnu_2 \geq 0,
	\end{align}
	and
	\begin{align}
	\label{eq:stab_4}
		- \upnu_1(2\upnu_2-1) - \upnu_2(2\upnu_1 -1) > 0
	\end{align}
	{\setlength{\parindent}{0cm} In this case \eqref{eq:stab_2} does not hold for those $b$ which lie between the~roots of~$\mathrm{q}(\cdot)$.}
\end{enumerate}
\begin{figure}[b!]
	\captionsetup{width=0.8\textwidth}
	\begin{center}
		\includegraphics[width=0.8\linewidth]{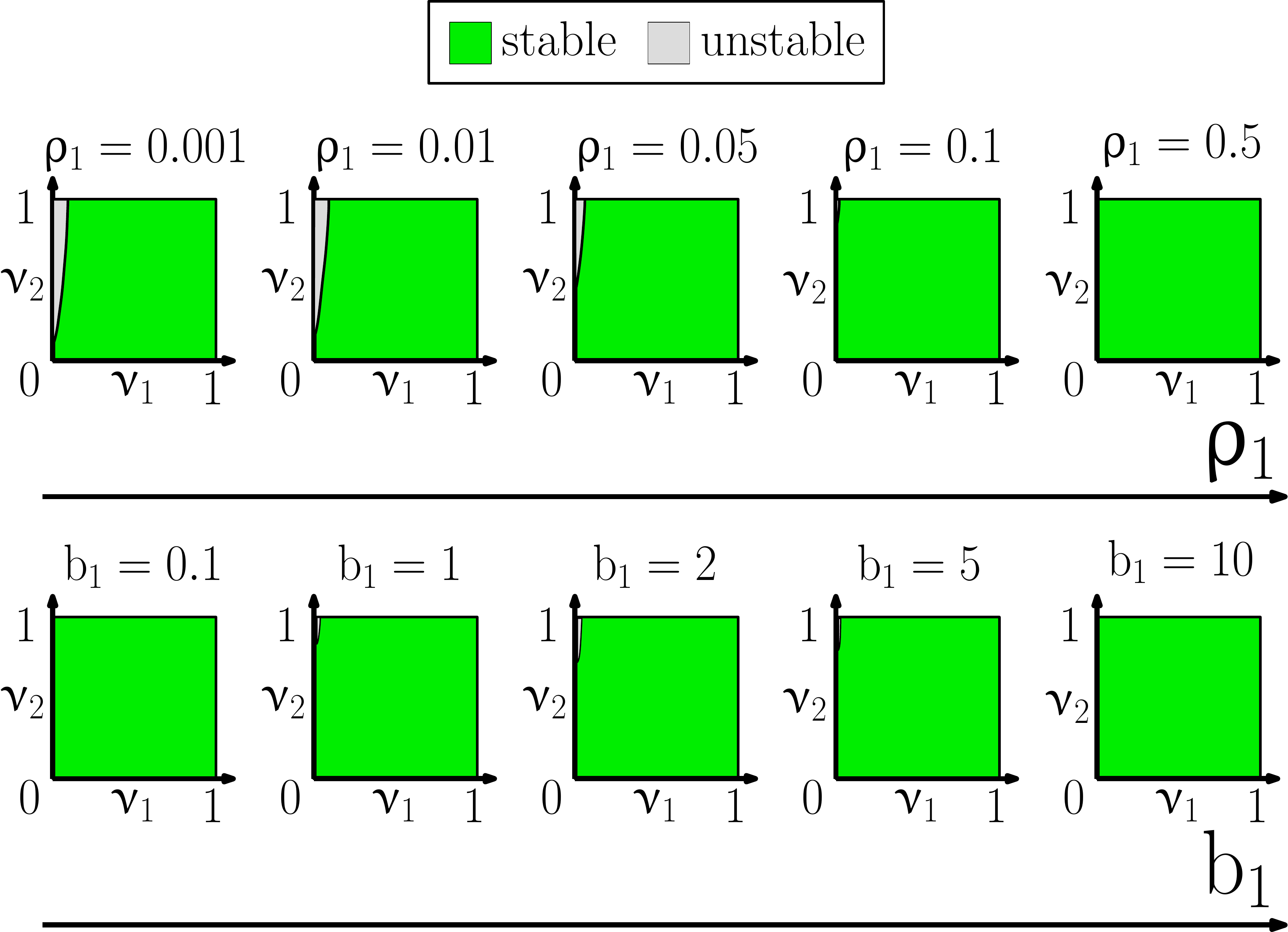}
	\end{center}
	\caption{Here we evaluate the local stability of \eqref{eq:two_com_sys}, through inequality \eqref{eq:stab_1}, as the parameters are varied. The nominal values for the parameters are $\b_1=0.2$, $\b_2=0.1$,  and $\uprho_1=\uprho_2=0.1$. The plots obtained by interchanging the indices of the groups are symmetric to the ones displayed above and thus are not shown here.}
	\label{fig:stab_2} %net influence
\end{figure}
Thus \eqref{eq:stab_3} is a sufficient condition for local stability. However if \eqref{eq:stab_3} does not hold, then the system is stable if \eqref{eq:stab_3.5} and \eqref{eq:stab_4} hold true and $b$ does not lie between the roots of $\mathrm{q}(\cdot)$. The region in the parameter space where these inequalities hold true are shown in Figure \ref{fig:stab_1}.

Figure \ref{fig:stab_2} shows a numerical simulation in the space $\upnu_1, \upnu_2 \in (0,1)$ to determine when condition \eqref{eq:stab_1} holds. We see that stability is guaranteed in a major portion of the $(\upnu_1, \upnu_2)$ space regardless of the values of the rest of the parameters. Also note that for small values of $\uprho_1, \uprho_2$, \eqref{eq:stab_1} may be approximated by \eqref{eq:stab_1.5}.

\subsection{Global Behavior}
Here we consider the global stability properties of the fixed point \eqref{eq:equi-2comp} for system \eqref{eq:two_com_sys}. We do this by transforming the system to an appropriate form and prove the validity of a Lyapunov function candidate. We proceed as follows.

 Introduce new variables $\tilde{z} = \ln x$, $\tilde{u} = y_1+y_2$ and $\tilde{w} = y_1 - y_2$. System \eqref{eq:two_com_sys} is then transformed as follows
\begin{align*}
\begin{split}
	&\dot{\tilde{z}} = 1 - e^{\tilde{z}} - \tilde{u}, \\
	&\dot{\tilde{u}} = \A(e^{\tilde{z}}-1) + \D + \b\tilde{w},\\
	&\dot{\tilde{w}} = \a(e^{\tilde{z}}-1) - \d -\B\tilde{w},
\end{split}
\end{align*}
where $\A = \b_1(\!1\!-\!\upnu_1\!) + \b_2(\!1\!-\!\upnu_2\!)$, $\a = \b_1 (\!1\!-\!\upnu_1\!) - \b_2 (\!1\!-\!\upnu_2\!)$, $\B = \b_1\upnu_1+\b_2\upnu_2$, $\b = -\b_1\upnu_1 + \b_2\upnu_2$, $\D = \b_1(\!1\!-\!\upnu_1\!)(1-\uprho_1) + \b_2(\!1\!-\!\upnu_2\!)(1-\uprho_2)$ and $\d = -\b_1(\!1\!-\!\upnu_1\!)(1-\uprho_1) + \b_2(\!1\!-\!\upnu_2\!)(1-\uprho_2)$.
It is easy to verify the following relations between the parameters which we will make use of later on
\begin{align}
\label{eq:cond}
	|\a| < \A, \quad |\b| < \B, \quad \text{and} \quad |\d| < \D.
\end{align}
Now, introduce the new variable $\hat{u} = \tilde{u} + \frac{\b}{\B} \tilde{w}$. The system is then expressible as
\begin{align*}
\begin{split}
	&\dot{\tilde{z}} = 1 - e^{\tilde{z}} - \hat{u} + \frac{\b}{\B} \tilde{w},\\
	&\dot{\hat{u}} = \left(\A+\frac{\b\a}{\B}\right)(e^{\tilde{z}}-1) + \left(\D - \frac{\d\b}{\B} \right),\\
	&\dot{\tilde{w}} = \a(e^{\tilde{z}}-1) - \d - \B\tilde{w}.\\
\end{split}
\end{align*}
This system has a unique equilibrium which is given by
\begin{align*}
	&\zo = \ln\left(1 - \frac{\D\B - \b\d}{\A\B+\b\a}\right), \,\,\, \wo = \frac{1}{\B}\left( \!-\a\frac{\D\B - \b\d}{\A\B+\b\a} - \d \right),\\  &\uo = \frac{\D\B - \b\d}{\A\B+\b\a} + \frac{\b}{\B}\wo.
\end{align*}
Introducing further new variables $z = \tilde{z} - \zo$, $u = \hat{u} - \uo$ and $w = \tilde{w} - \wo$ to bring the equilibrium point to the origin, we obtain the following system
\begin{align}
\begin{split}
\label{sys_stability}
	&\dot{z} = e^{\zo}(1\!-\!e^z) \!-\! u \!+\! \frac{\b}{\B}w,\\
	&\dot{u} = \left(\A\!+\!\frac{\b\a}{\B}\right)e^{\zo}(e^z \!-\! 1),\\
	&\dot{w} = \a e^{\zo}(e^{z}-1) - \B w.
\end{split}
\end{align}
\begin{theorem}
	The system \eqref{sys_stability} is globally stable with all solutions tending to zero if the following condition holds true
	\vspace{-10pt}\[\hspace{100pt} \B^2 > \a\b.\]
\end{theorem}
\begin{proof}
	In order to prove the theorem we construct a Lyapunov function of the following form
\begin{align*}
	V = (e^z - z - 1) + \P u^2 + \Q w^2,
\end{align*}
which is positive for positive $\P$, $\Q$ everywhere except the origin. Thus if $\dot{V}$ is shown to be negative outside the origin for some positive $\P$, $\Q$, it would imply global stability of \eqref{sys_stability}. Differentiating $V$ along the vector field of system \eqref{sys_stability} we get
\begin{align*}
	\dot{V} &= e^z\dot{z} - \dot{z} + 2 \P u\dot{u} + 2 \Q w\dot{w}.
\end{align*}
Choose $\displaystyle \P = {e^{-\zo}}/{(2\left(\A+ \b\a/\B\right))}$. Due to \eqref{eq:cond}, $\P>0$ so it does not violate the status of $V$ as a Lyapunov function candidate. Substituting this and \eqref{sys_stability} in the above equation gives us
\begin{align*}
	\dot{V}= -e^{\zo}{(e^{z}-1)}^2 \!+\! w(e^z-1)\left(\frac{\b}{\B} + 2 \Q \a e^{\zo}\right) \!-\! 2 \Q \B w^2.	
\end{align*}
Completing squares, we get
\begin{align*}
	\dot{V} =  &-{\left(e^{\zo/2}(e^z - 1) - \frac{we^{-\zo/2}}{2}\left(\frac{\b}{\B} + 2 \Q \a e^{\zo}\right)\right)}^2 \\
		&- \left(2 \Q \B - \frac{e^{-\zo}}{4}{\left(\frac{\b}{\B} + 2 \Q \a e^{\zo}\right)}^2\right)w^2.
\end{align*}
Now the condition $\displaystyle 2 \Q \B - \frac{e^{-\zo}}{4}{\left(\frac{\b}{\B} + 2 \Q \a e^{\zo}\right)}^2 > 0$ gives the derivative $\dot{V}<0$ outside the $u$-axis, and on this axis the vector field  of \eqref{sys_stability} has value $(-u,0,0)$, which is not tangent to the axis. Hence we get the desired global stability if this condition takes place, i.e., the above inequality  is true for some positive $\Q$, $\B$. This condition can be expressed as 
\begin{align*}
	\p\left(e^{\zo} \Q \right) := 4 \a^2{\left(e^{\zo} \Q \right)}^2 \!+\! \left( - 8 \B \!+\! 4\frac{\a \b}{\B} \right)\left(e^{\zo} \Q \right) \!+\! \frac{\b^2}{\B^2} < 0.
\end{align*}
Now, if $\a=0$, the inequality holds for any $\B>0$ by choosing $\Q$ to be sufficiently large. However if $\a\neq 0$ then due to $4\a^2>0$ and $\b^2/\B^2 \ge 0$, the inequality holds true for some positive $\Q$, $\B$ only if the vertex of the respective 
parabola belongs to the fourth quadrant, i.e.,
\[
\displaystyle \vspace{0pt} -\frac{1}{8 \a^2}\left( - 8 \B + 4\frac{\a\b}{\B} \right)>0, \quad \text{and} \quad -\frac{\Delta}{16\a^2}<0, 
\]
where $\Delta$\vspace{0pt} is the discriminant of $\p\left(e^{\zo}\Q\right)$. This translates to the following two conditions: \emph{1)} $\displaystyle -2\B^2 + \a \b < 0$, and \emph{2)} $\displaystyle \B^2 - \a \b > 0$, or simply
\begin{align}
\label{eq:cond_stab}
	\B^2 > \a\b,
\end{align}
which completes the proof.
\end{proof}
Substituting the values of $\a$, $\b$ and $\B$ in condition \eqref{eq:cond_stab} gives us the following
\begin{align*}
	&{(\b_1\upnu_1\!+\!\b_2\upnu_2)}^2 > \big(\b_1(\!1\!-\!\upnu_1) \!-\! \b_2(\!1\!-\!\upnu_2)\big)(-\b_1\upnu_1\!+\!\b_2\upnu_2).
\end{align*}
After some simplification and realization that $\b_1$, $\b_2$, $\upnu_1$ and $\upnu_2$ are all positive numbers, a sufficient condition for the above condition to hold is
\begin{align}
\label{eq:con_stab1}
	(\b_1-\b_2)(\b_1\upnu_1-\b_2\upnu_2) + 4\b_1\upnu_1\b_2\upnu_2 > 0.
\end{align}

Figure \ref{fig:2agent_limcyc} depicts an instance in the parameter space where the system \eqref{eq:two_com_sys} is unstable. Indeed, we see the existence of a stable limit cycle in this case.
\begin{figure*}[t!]
    \centering
    \captionsetup{width=0.85\textwidth}
    \begin{subfigure}[t]{0.5\textwidth}
        \captionsetup{width=0.9\textwidth}
        \centering
        \includegraphics[trim = 30mm 70mm 30mm 70mm, clip, width = \linewidth]{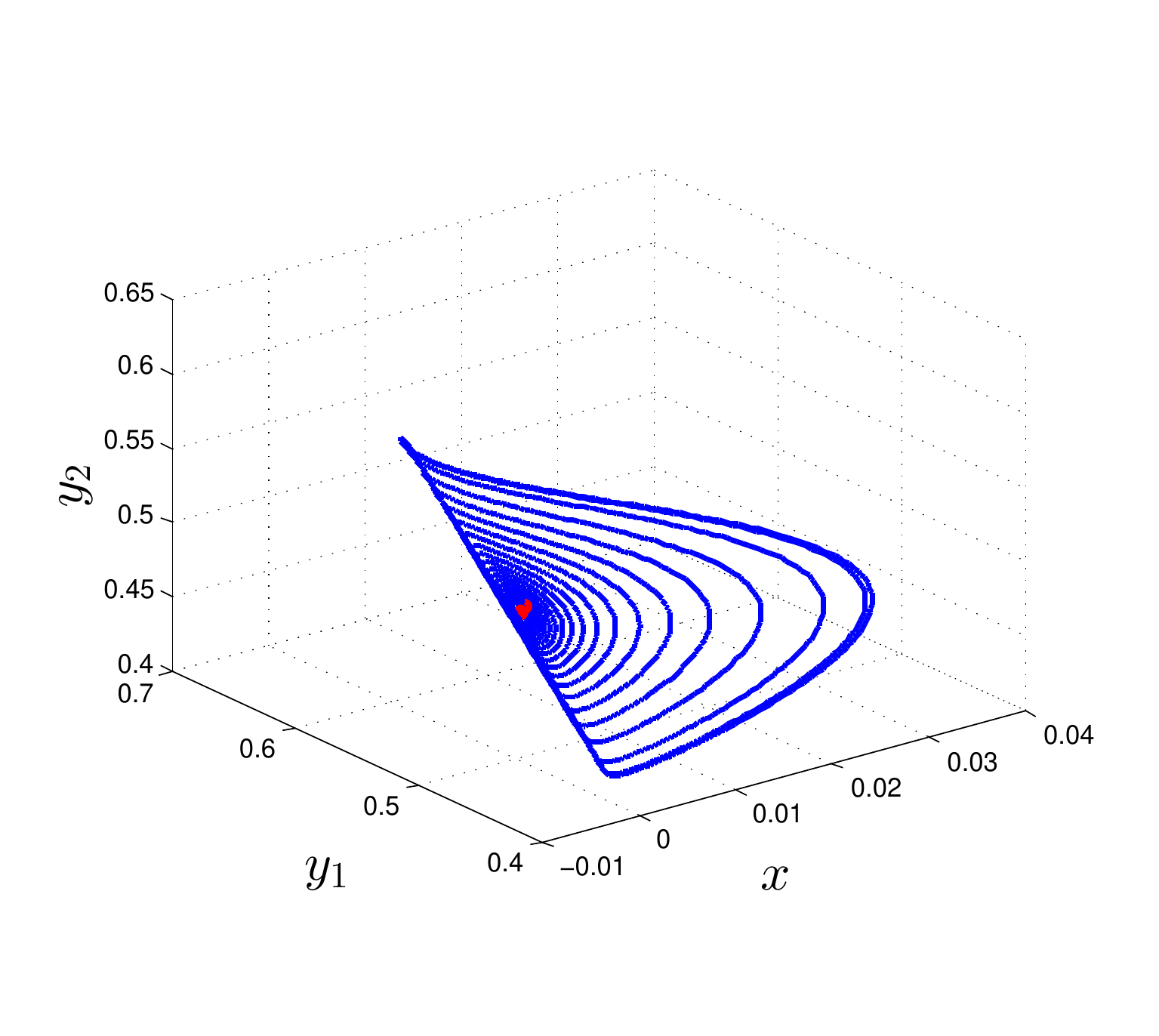}
        \caption{Starting from inside the limit cycle. Here $x(0) = 0.001$, $y_1(0) = 0.5$ and $y_2(0) = 0.5$.}
        \label{fig:2agent_limcyc_in}
    \end{subfigure}%
%	\hspace{0.05\textwidth}
    \begin{subfigure}[t]{0.5\textwidth}
        \captionsetup{width=0.9\textwidth}
        \centering
        \includegraphics[trim = 30mm 70mm 30mm 70mm, clip, width = \linewidth]{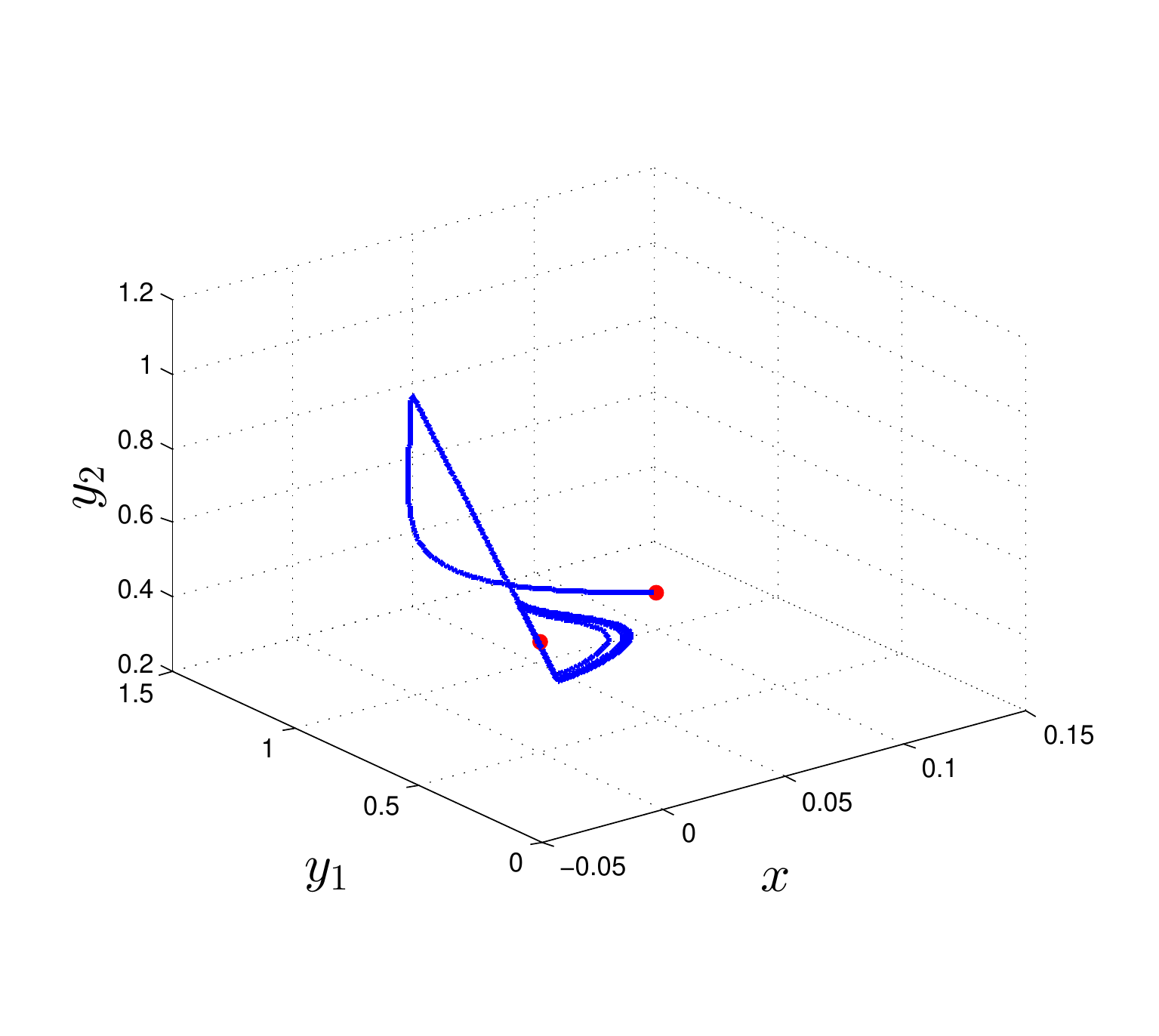}
        \caption{Starting from outside the limit cycle. Here $x(0) = 0.1$, $y_1(0) = 1$ and $y_2(0) = 0.3$.}
        \label{fig:2agent_limcyc_out}
    \end{subfigure}
    \caption{The limit cycle of system \eqref{eq:two_com_sys}. Here $\b_1 = 0.2$, $\b_2 = 0.1$, $\upnu_1=0.01$, $\upnu_2=0.9$, $\uprho_1 = 0.001$ and, $\uprho_2 = 0.1$. The red marked points indicate the fixed point and the initial point for the respective simulations.}
\label{fig:2agent_limcyc}
\end{figure*}

\section{Discussion}
This chapter presented a time-domain analysis of the social ecological system introduced in Chapters \ref{chap:model} and \ref{chap:lump}. We studied first, the single-agent network as an instance of the homogeneous consumer network. We saw how the level of environmentalism determined a trade-off between the levels of resource stock and consumption at steady state and the role of various parameters influenced the manner in which the system approached the steady state. In Chapter \ref{chap:optimal} we will continue with the same model by examining the optimal consumption paths under a defined notion of optimality. We then, in this chapter, studied the two agent network, as an instance of the symmetric semi-homogeneous network with two consumer groups. We saw that the analysis was mathematically more involved than the single agent case. The addition of another agent in the single agent network brought about the phenomenon of free-riding, something which was not possible in the single agent case. We saw how free-riding could be discouraged and also examined the local and global behavior of the system. We return to this network in Chapter \ref{chap:game} where we analyze strategic interactions between the two agents in a game-theoretic context. As indicated by the analysis in this chapter, the analysis for the $\n$-agent network is mathematically more involved, and has been deferred to Chapter \ref{chap:n} as a possible extension.

\clearpage
%% This is an example first chapter.  You should put chapter/appendix that you
%% write into a separate file, and add a line \include{yourfilename} to
%% main.tex, where `yourfilename.tex' is the name of the chapter/appendix file.
%% You can process specific files by typing their names in at the 
%% \files=
%% prompt when you run the file main.tex through LaTeX.

\chapter{Optimal Consumption for a Homogeneous Society}

\label{chap:optimal}

In this chapter we return to the most basic block model, namely the homogeneous consumer network of Chapter \ref{chap:lump} with the consuming society aggregated as one single agent. The time-domain behavior of the open-loop system has already been studied in Chapter \ref{chap:open}. We now turn our attention to the ``best" consumption pattern on behalf of the consumers. We do this by framing the system as an optimal growth model for a single resource based economy.  More precisely, the system is formulated as an infinite-horizon optimal control problem with unbounded set of control constraints and non-concave Hamiltonian. We then characterize the optimal paths for all possible parameter values and initial states by applying the appropriate version of the Pontryagin maximum principle. Our main finding is that only two qualitatively different types of behavior of sustainable optimal paths are possible depending on whether the social discount rate is greater than a specific threshold or not. An analysis of these behaviors yields a general criterion for sustainable and strongly sustainable optimal growth (w.r.t. the corresponding notions of sustainability defined herein).

\section{Considerations for Studying Optimal Behavior}

Following the first analysis conducted by Ramsey \cite{ramsey1928mathematical}, the mathematical problem of inter-temporal resource allocation\footnote{The problem of distributing the resource fairly among the present and all future generations} has attracted a significant amount of attention over the past decades, and has driven the evolution of first exogenous, and then endogenous growth theory (see \cite{acemoglu2008introduction, barro1995economic}). Employed growth models are typically identified by the production of economic output, the dynamics of the inputs of production, and the comparative mechanism of alternate consumption paths. Our framework considers a renewable resource, whose reproduction is logistic in nature, as the only input to production. The relationship of the resource with the output of the economy is explained through a Cobb-Douglas type production function with an exogenously driven knowledge stock. Alternate consumption paths are compared via a discounted utilitarian approach.  The question that we concern ourselves with for our chosen framework, is the following: what are the conditions of sustainability for optimal development?

In the context of sustainability, the discounted utilitarian approach may propose undesirable solutions in certain scenarios. For instance, discounted utilitarianism has been reported to force consumption asymptotically to zero even when sustainable paths with non-decreasing consumption are feasible \cite{asheim2010sustainability}. The Brundtland Commission defines sustainable development as development that meets the needs of the present, without compromising the ability of future generations to meet their own needs \cite{brundtland1987report}. In this spirit, we employ the notion of sustainable development, as a consumption path ensuring a non decreasing welfare for all future generations. This notion of sustainability is natural, and has also been used by various authors in their work. For instance, Valente \cite{valente2005sustainable} evaluates this notion of sustainability for an exponentially growing natural resource, and derives a condition necessary for sustainable consumption.  We extend this model by allowing the resource to grow at a declining rate (the logistic growth model). We build on the work presented previously in \cite{manzoor2014optimal} which proves the existence of an optimal path only in the case when the resource growth rate is higher than the social discount rate and admissible controls are uniformly bounded.

Our model is formulated as an infinite-horizon optimal control problem with logarithmic instantaneous utility. The problem involves unbounded controls and the non-concave Hamiltonian. These preclude direct application of the standard existence results and Arrow's sufficient conditions for optimality. We transform the original problem to an equivalent one with simplified dynamics and prove the general existence result. Then we apply a recently developed version of the maximum principle \cite{aseev2012maximum, aseev2014needle, aseev2015maximum} to our problem and describe the optimal paths for all possible parameter values and initial states in the problem. Our analysis of the Hamiltonian phase space reveals that there are only two qualitatively different types of behavior of the sustainable optimal paths in the model. In the first case the instantaneous utility is a non-decreasing function in the long run along the optimal path (we call such paths \textit{sustainable}). The second case corresponds to the situation when the optimal path is sustainable and in addition the resource stock is asymptotically nonvanishing (we call such paths \textit{strongly sustainable}). We show that a strongly sustainable equilibrium is attainable only when the resource growth rate is higher than the social discount rate. When this condition is violated, we see that the optimal resource exploitation rate asymptotically follows the Hotelling rule of optimal depletion of an exhaustible resource \cite{hotelling1931economics}. In this case optimal consumption is sustainable only if the the depletion of the resource is compensated by appropriate growth of the knowledge stock and/or decrease of the output elasticity of the resource.

\section{Formulating the Consumption Problem in an Optimal Control Setting}
\label{sec:prob}
Here we depart from the lumped parameter model of Section \ref{sec:lumped} where the aggregate consumption effort of the society is contained in the single variable $y(t)$. Thus there exists only a single agent where the agent may represent a single consumer, or a homogeneous community in the sense defined in Chapter \ref{chap:lump}. In this chapter we will frequently refer to this agent as ``society" with the understanding that the case of the single consumer is also captured by this term (a single consumer is simply a society with one member). Maintaining the previous notation, the normalized (w.r.t the carrying capacity) level of the resource stock at time $t$ is given by $x(t)>0$ and is governed by the standard model of logistic growth. The society consumes the resource by exerting some effort, the normalized (w.r.t the resource growth rate) level of which is represented by $y(t)>0$\footnote{Note that while in the original model given by \eqref{eq:ses} the efforts $y_i(t)$ may take on negative values, here we restrict $y(t)$ to be non-negative. This is because in the context of sustainable growth, we are only interested in cases where society collectively consumes positive quantities of the resource (although individual consumers may consume negative quantities).}. The dynamics of the resource stock are then given by the following equation:
\begin{align}
    \dot{x}(t) = x(t)\left(1-x(t)\right) - y(t)x(t), \hspace{20pt} y(t)\in (0,\infty).
\end{align}
The initial stock of the resource is $x(0)=x_0>0$.

We assume that the economy is a single resource economy whose output is represented by $P(t)>0$  at instant $t\geq 0$. The output  is related to the resource by a Cobb-Douglas type production function (see Appendix \ref{app:resecon}) as follows
\begin{align}\label{Y}
	P(t) = S(t){\big(y(t)x(t)\big)}^{\upbeta}, \hspace{20pt} \upbeta \in (0,1].
\end{align}
Here $S(t)>0$ represents an exogenously driven knowledge stock at time $t\geq 0$. We assume $\dot S(t)\leq\upmu S(t)$, where $\upmu\geq 0$ is a constant, and  $S(0)=S_0>0$. Furthermore, it is assumed that the whole output $P(t)$ produced at each instant $t\geq 0$ is consumed. The total discounted value of the \emph{growth} of the production function $P(.)$ along $(x(.),y(.) )$ is given by
\begin{align*}
	\tilde{J}(x(.),y(.))=\int_0^\infty e^{-\delta t} \frac{\dot{P}(t)}{P(t)} \,\, dt = -\ln P(0) + \delta \int_0^\infty e^{-\delta t} \ln P(t) \, \, dt,
\end{align*}
where $\del>0$ is the subjective discount rate. Thus maximizing the total discounted value of the growth of output $P(t)$ is equivalent to maximizing the following objective function
\begin{align*}
	J'(x(.),y(.)) &= \int_0^\infty e^{-\delta t} \ln P(t) \, \, dt,\\
	&= \int_0^\infty e^{-\delta t} \ln S(t)\,dt + \upbeta \int_0^\infty e^{-\delta t}\left[\ln x(t)+\ln y(t)\right]\,dt,
\end{align*}
Neglecting constant term $\int_0^\infty e^{-\delta t} \ln S(t)\,dt$ and scalar multiplier $\upbeta$, our final objective is thus set up as
\begin{equation*}
	J(x(\cdot),y(\cdot)) =\int_0^{\infty} e^{-\del t}\left[\ln x(t) +\ln y(t)\right].
\end{equation*}

\subsection{The Optimal Control Problem}

The above exposition leads to the following optimal control problem $(P1)$:
\begin{align}
\label{eq1}
	J(x(\cdot),y(\cdot)) &=\int_0^{\infty} e^{-\del t}\left[\ln x(t) +\ln y(t)\right]\, dt\to\mbox{max},\\
\label{eq2}
	\dot x(t) &= x(t)\left(1-x(t)\right)-y(t)x(t),\qquad x(0)=x_0,\\
\label{eq3}
	y(t)&\in (0,\infty),
\end{align}
Now, by an \emph{admissible control} in problem $(P1)$ we mean a Lebesgue measurable locally bounded function $y\colon
[0,\infty)\mapsto\mathbb{R}^1$ which satisfies the control constraint~\eqref{eq3} for all $t\geq 0$. By definition, the corresponding to $y(\cdot)$ \emph{admissible trajectory} is a (locally) absolutely continuous function $x(\cdot): [0, \infty) \mapsto\mathbb{R}^1$ which is a Caratheodory solution (see \cite{filippov2013differential}) to the Cauchy problem~\eqref{eq2} on the whole infinite time interval $[0, \infty)$. Due to the local boundedness of the admissible control $y(\cdot)$ such admissible trajectory $x(\cdot)$ always exists and is unique (see~\cite[Section~7]{filippov2013differential}). A pair $(x(\cdot), y(\cdot))$ where $y(\cdot)$ is an admissible control and $x(\cdot)$ is the corresponding admissible trajectory is called an \emph{admissible pair} in problem $(P1)$.

\subsection{Establishing the Notion of Optimality}
Due to~\eqref{eq2} for any admissible trajectory $x(\cdot)$ the following property holds:
\begin{align}
\label{eq4_1}
	x(t)\leq x_{\text{max}}=\max\{x_0,1\},\qquad t\geq 0.
\end{align}
The integral in~\eqref{eq1} is understood in improper sense, i.e.
\begin{align}
\label{eq1_1}
	J(x(\cdot),y(\cdot)) =\lim_{T \to \infty} \int_0^T \mathrm{e}^{-\del t}\left[ \ln x(t)+\ln y(t) \right] \, dt
\end{align}
if the limit exists. Now, in order to establish a notion of optimality, we must first prove that the limit in the above expression exists. For this, we will require the following intermediate result.

\begin{lemma}
\label{lem:int_cvg}
	There is a decreasing function $\omega :[0,\infty)\mapsto (0,\infty)$ such that $\omega(t)\to +0$ as $t\to\infty$ and for any admissible pair $(x(\cdot),y(\cdot))$ the following inequality holds:
	\begin{align}
	\label{eq1_2}
	        \int_T^{T'}e^{-\del t}\left[ \ln x(t)+\ln y(t)\right]\, \mathrm{d}t <\omega(T),\qquad 0\leq T< T'.
	\end{align}
\end{lemma}
\begin{proof}
	Indeed, due to inequality $\ln x < x$, $x>0$,  for arbitrary $0\leq T< T'$ we have
\begin{align*}
	\int_T^{T'}e^{-\del t}\left[ \ln x(t)+\ln y(t)\right]\, dt<\int_T^{T'}e^{-\del t}y(t)x(t)\, dt.
\end{align*}
Hence, substituting expression of $x(t)y(t)$ from~\eqref{eq2} in the inequality above we get
\begin{align*}
	\int_T^{T'}e^{-\del t}\left[\ln x(t))+\ln y(t)\right]\, dt <\int_T^{T'}e^{-\del t}\left[ x(t)\left(1-x(t)\right)-\dot{x}(t)\right]\, dt.
\end{align*}
This implies
\begin{align*}
	\int_T^{T'}&e^{-\del t}\left[\ln x(t)+\ln y(t)\right]\, dt < \int_T^{T'}e^{-\del t}\left[x(t)-\dot x(t)\right]\, dt,\\
	&\leq x_{\text{max}}\int_T^{T'}e^{-\del t}\, dt - \int_T^{T'}e^{-\del t}\dot x(t)\, dt,\\
	&=\frac{x_{\text{max}}}{\del}\left(e^{-\del T}-e^{-\del T'}\right)- e^{-\del t}x(t)\bigg|_{T}^{T'} - \del \int_T^{T'}e^{-\del t}x(t)\, dt,\\
	&<\frac{x_{\text{max}}}{\del}e^{-\del T} + e^{-\del T}x_{\text{max}} =\frac{(1+\del)x_{\text{max}}}{\del}e^{-\del T} = \omega(T),
\end{align*}
which completes the proof.
\end{proof}

Now, let us show that for any admissible pair $(x(\cdot),y(\cdot))$ the limit in~\eqref{eq1_1} exists.

\begin{lemma}
\label{lim-exist}
	For any admissible pair $(x(\cdot),y(\cdot))$ the limit in~\eqref{eq1_1} exists and is either finite or equals $-\infty$.
\end{lemma}
\begin{proof}  
Let $(x(\cdot),y(\cdot))$ be an arbitrary admissible pair.  For any $T>0$ define $J_T(x(\cdot),y(\cdot))$ as follows:
\begin{align*}
	J_T(x(\cdot),y(\cdot)) = \int_0^T \mathrm{e}^{-\del t}\left[ \ln x(t)+\ln y(t))\right]\, dt.
\end{align*}
Let $\{ \zeta_i \}_{i=1}^{\infty}$ be a sequence of positive numbers such that $\zeta_i \to \infty$ as $i\to\infty$ and
\begin{align*}
	\lim_{i\to\infty}J_{\zeta_i}(x(\cdot),y(\cdot))= \limsup_{T \to \infty}\int_0^T e^{-\del t}\left[ \ln x(t)+\ln y(t) \right]\, dt.
\end{align*}
Due to Lemma~\ref{lem:int_cvg} we have the following estimate
\begin{equation}
\label{eq1_2}
	\lim_{i\to\infty}J_{\zeta_i}(x(\cdot),y(\cdot))\leq\omega(0).
\end{equation}
Analogously, let  $\{ \tau_i \}_{i=1}^{\infty}$ be a sequence of positive numbers such that  $\tau_i \to \infty$ as $i\to\infty$
and
\begin{align*}
	\lim_{i\to\infty}J_{\tau_i}(x(\cdot),y(\cdot)) = \liminf_{T \to \infty}\int_0^T e^{-\del t}\left[ \ln x(t)+\ln y(t) \right]\, dt.
\end{align*}
Without loss of generality one can assume  that $\tau_i < \zeta_i$, $i=1,2,\dots$. Then we have
\begin{align*}
	J_{\zeta_i}(x(\cdot),y(\cdot))=J_{\tau_i}(x(\cdot),y(\cdot))+\int_{\tau_i}^{\zeta_i}e^{-\del t}\left[ \ln x(t)+\ln y(t)\right]\, dt,\quad i=1,2,\dots.
\end{align*}
Due to Lemma~\ref{lem:int_cvg} this implies
\begin{align*}
	J_{\zeta_i}(x(\cdot),y(\cdot))< J_{\tau_i}(x(\cdot),y(\cdot)) +\omega(\tau_i),\quad i=1,2,\dots.
\end{align*}
Since $\omega(\tau_i)\to 0$  as $i\to\infty$ taking the limit in the last inequality as $i\to\infty$ we get
\begin{align*}
	\limsup_{T\to\infty}\int_0^T \mathrm{e}^{-\del t}\left[\ln x(t)+\ln y(t))\right]\, dt\leq\liminf_{T\to\infty}\int_0^T \mathrm{e}^{-\del t}\left[\ln x(t)+\ln y(t))\right]\, dt.
\end{align*}
As far as the opposite inequality
\begin{align*}
	\liminf_{T\to\infty}\int_0^T \mathrm{e}^{-\del t}\left[\ln x(t)+\ln y(t))\right]\, dt\leq\limsup_{T\to\infty}\int_0^T \mathrm{e}^{-\del t}\left[\ln x(t)+\ln y(t))\right]\, dt
\end{align*}
is always true, the limit~\eqref{eq1_1} exists, and due to~\eqref{eq1_2} this limit is either finite or $-\infty$.
\end{proof}

%Now we show that the optimal admissible pair $(x_*(\cdot),y_*(\cdot))$ must correspond to a finite value of the objective $J(x_*(\cdot),y_*(\cdot))$.
%\begin{lemma}
%	\label{lem:fin_sup} %finite supremum
%	For an optimal admissible pair $(x^*(\cdot),y^*(\cdot))$ (if it exists)
%	\begin{align*}
%		J(x^*(\cdot),y^*(\cdot))= \lim_{T\to\infty}\int_0^Te^{-\del t}\left[\ln y^*(t) +\ln x (t)\right]\,\mathrm{d}t > -\infty
%	\end{align*}
%	converges to a finite number.
%\end{lemma}
%\begin{proof}
%	From Lemma \ref{lem:int_cvg} we already have that for any admissible pair, $J_T(x(\cdot),y(\cdot))$ will either converge to $-\infty$ or a finite number. Now consider the constant admissible control $\tilde y(t) \equiv 1/2$, $t\geq 0$. From the state equation \eqref{eq2} this implies that the corresponding state $\tilde x (t) \rightarrow 1/2$ as $t \rightarrow \infty$ (it is easy to see that the logistic growth dynamics are stable). Hence the functional satisfies
%\begin{align*}
%	-\infty < \int_0^\infty \mathrm{e}^{-\del t} \left( \ln \frac{1}{2} + \ln \tilde{x}(t) \right)\,dt < \infty, 	
%\end{align*}
%which indicates a finite value for the functional. Hence, if optimal admissible pair $(x^*(\cdot),y^*(\cdot))$ exists then the corresponding value $J(x^*(\cdot),y^*(\cdot))$ must also be finite.
%\end{proof}

Due to~\eqref{eq1_2} for any admissible pair $(x(\cdot), y(\cdot))$ the value $\sup_{(x(\cdot),y(\cdot))} J(x(\cdot),y(\cdot))$ is finite. This allows us to understand the optimality of an admissible pair $(x_*(\cdot), y_*(\cdot))$ in the strong sense \cite{carlson2012infinite}. By definition, an admissible pair $(x_*(\cdot), y_*(\cdot))$ is \emph{strongly optimal} (or, for brevity, simply \emph{optimal}) in the problem $(P1)$ if the functional~\eqref{eq1} takes the maximal possible value on this pair , i.e.
\begin{align*}
	J(x_*(\cdot), y_*(\cdot)) = \sup_{(x(\cdot),y(\cdot))} J(x(\cdot),y(\cdot))<\infty.
\end{align*}
Notice, that the set of control constraints in problem $(P1)$ (see~\eqref{eq3}) is nonclosed and unbounded. Due to this circumstance the standard existence theorems (see e.g. \cite{balder1983existence,cesari2012optimization}) are not applicable to problem $(P1)$ directly. Moreover, the situation is complicated here by the fact that the Hamiltonian of problem $(P1)$ is non-concave in the state variable $x$. These preclude the usage of Arrow's sufficient conditions for optimality (see \cite{carlson2012infinite}).

Our analysis below is based on application of the recently developed normal form version of the Pontryagin maximum principle \cite{pontryagin1987mathematical} for infinite-horizon optimal control problems with adjoint variable specified explicitly via the Cauchy type formula (see \cite{aseev2012maximum, aseev2014needle, aseev2015maximum}). However, such approach assumes that the optimal control exists. So, the proof of the existence of an optimal admissible pair $(x_*(\cdot), y_*(\cdot))$ in problem $(P1)$ and establishing of the corresponding version of the maximum principle will be our primary goal in the next section.

%%%%%%%%%%%%%%%%%%%%%%%%%%%%%%%%%%%%%%%%%%%%%%%%%%%%%%%%

\section{Reduction to a Problem with Linear Dynamics}

Let us transform problem $(P1)$ into a more appropriate equivalent form.

Due to~\eqref{eq2} along any admissible pair $(x(\cdot),y(\cdot))$ we have
\begin{align*}
	\frac{d}{dt}\left[e^{-\del t}\ln x(t)\right]\stackrel{\text{a.e.}}{=} -\del e^{-\del t}\ln x(t) + e^{-\del t} -e^{-\del t}\left(x(t)+ y(t)\right),\quad t>0.
\end{align*}
Integrating this equality on arbitrary time interval $[0,T]$, $T>0$, we obtain
\begin{multline*}
	\int_0^Te^{-\del t}\ln x(t)\, dt =\frac{\ln x_0-e^{-\del T}\ln
x(T)}{\del}\\
	+\frac{1}{\del^2}\left(1-e^{-\del T}\right) - \int_0^Te^{-\del t}\left(\frac{1}{\del}x(t)+ \frac{y(t)}{\del}\right)\, dt.
\end{multline*}
Hence, for any admissible pair  $(x(\cdot),y(\cdot))$ and arbitrary $T>0$ we have
\begin{multline}
\label{eq4}
	\int_0^Te^{-\del t}\left[ \ln x(t)+\ln y(t)\right]\, dt =\frac{\ln x_0-e^{-\del T}\ln x(T)}{\del} +\frac{1}{\del^2}\left(1-e^{-\del T}\right)\\
	-  \frac{1}{\del}\int_0^Te^{-\del t}x(t)\, dt + \int_0^Te^{-\del t}\left(\ln y(t) - \frac{y(t)}{\del}\right)\, dt.
\end{multline}
Here due to estimate~\eqref{eq1_2} limits of the both sides in~\eqref{eq4} as $T\to\infty$  exist and equal either a finite number or $-\infty$ simultaneously, and due to~\eqref{eq4_1} either $(i)$ $\lim_{T\to\infty}e^{-\del T}\ln x(T) =0$ or $(ii)$ $\lim\inf_{T\to\infty}e^{-\del T}\ln x(T)<0$. In the case $(i)$ passing to the limit in~\eqref{eq4} as $T\to\infty$  we get
\begin{multline}
\label{eq5}
	\int_0^{\infty}e^{-\del t}\left[ \ln x(t)+\ln y(t)\right]\, dt =\frac{\ln x_0}{\del} +\frac{1}{\del^2}\\ 
	-  \frac{1}{\del}\int_0^{\infty}e^{-\del t}x(t)\, dt + \int_0^{\infty}e^{-\del t}\left(\ln y(t) - \frac{y(t)}{\del}\right)\, dt,
\end{multline}
where both sides in~\eqref{eq5} are equal to a finite number or $-\infty$ simultaneously.

In the case $(ii)$ condition $\lim\inf_{T\to\infty}e^{-\del T}\ln x(T)<0$ implies
\begin{align*}
	\int_{0}^{\infty}e^{-\del t}\left[\ln x(t)+\ln y(t)\right]\, dt=\lim_{T\to\infty}\int_0^T e^{-\del t}\left[\ln x(t)+\ln y(t)\right]\, dt=-\infty
\end{align*}
(see \cite{aseev2016optimal} for details). Hence, in the case $(ii)$~\eqref{eq5} also holds as $-\infty=-\infty$.

Neglecting now the constant terms in the right-hand side of~\eqref{eq5} we obtain the following optimal control problem $(\tilde P1)$ which is equivalent to $(P1)$:
\begin{align}
	\tilde J(x(\cdot),y(\cdot)) &=\int_0^{\infty} e^{-\del t}\left[\ln y(t) - \frac{y(t)}{\del} - \frac{1}{\del}x(t)\right]\, dt\to\mbox{max},\\
\label{eq8}
	\dot x(t) &= x(t)\left(1-x(t)\right)-y(t)x(t),\qquad x(0)=x_0,\\
\label{eq9}
	y(t)&\in (0,\infty).
\end{align}

Further, the function  $y\mapsto ln\, y-y/\del$ is increasing on $(0,\del]$ and it reaches the global maximum at point $y_*=\del$. Hence, any optimal control $y_*(\cdot)$ in $(\tilde P1)$ (if such exists) must satisfy the inequality $y_*(t)\geq \del$ for almost all $t\geq 0$. Hence the control constraint \eqref{eq9} in $(\tilde P1)$  (and also the control constraint~\eqref{eq3} in $(P1)$) can be replaced by the control constraint $y(t)\in [\del,\infty)$. Thus we arrive to the following (equivalent) problem $(P2)$:
\begin{align}
	J(x(\cdot),y(\cdot)) &=\int_0^{\infty} e^{-\del t}\left[\ln y(t) + \ln x(t)\right]\, dt\to\mbox{max},\\
	\dot x(t) &= x(t)\left(1-x(t)\right)-y(t)x(t),\qquad x(0)=x_0,\\
\label{eq12}
	y(t)&\in [\del,\infty).
\end{align}
Here the class of admissible controls in problem $(P2)$ consists of all locally bounded functions $y(\cdot)$ satisfying the control
constraint~\eqref{eq12} for all $t\geq 0$.

To simplify dynamics in $(P2)$ let us introduce the new state variable $z(\cdot)$: $z(t)=1/x(t)$, $t\geq 0$. As can be verified directly, in terms of the state variable $z(\cdot)$ problem $(P2)$ can be rewritten as the following (equivalent) problem $(P3)$:
\begin{align}
\label{eq21}
	J(z(\cdot),y(\cdot)) &=\int_0^{\infty} e^{-\del t}\left[ \ln y(t)-\ln z(t)\right]\, dt\to\mbox{max},\\
\label{eq22}
	\dot z(t) &= \left[y(t)-1\right]z(t) + 1,\qquad z(0)=z_0=\frac{1}{x_0},\\
\label{eq23}
	y(t)&\in [\del,\infty).
\end{align}
The class of admissible controls $y(\cdot)$ in $(P3)$  consists of all measurable locally bounded functions $y\colon [0,\infty)\mapsto [\del,\infty)$.

Notice, that due to linearity of \eqref{eq22} for arbitrary admissible control $y(\cdot)$ the corresponding trajectory $x(\cdot)$ can be expressed via the Cauchy formula (see \cite{hartman1970ordinary}):
\begin{equation}
\label{eq-22}
	z(t)=z_0e^{\int_0^ty(\xi)\,d\xi-t}+e^{\int_0^ty(\xi)\,d\xi-t}\int_0^te^{-\int_0^sy(\xi)\,d\xi+s}\,ds, \qquad t\geq 0.
\end{equation}
Since the problems $(P1)$, $(P2)$ and $(P3)$ are equivalent we will focus our analysis below on problem $(P3)$ with simplified dynamics (see~\eqref{eq22}) and the closed set of control constraints (see~\eqref{eq23}).

\section{Existence of an Optimal Control and the Maximum Principle}\label{sec:exist}

The constructed problem $(P3)$ is a particular case of the following autonomous infinite-horizon optimal control problem $(P4)$ with exponential discounting:
\begin{align}
	J(z(\cdot),y(\cdot))&=\int_{0}^{\infty}e^{-\del t}g(z(t),y(t))\,dt\to\max,\\
\label{g-Cauchy}
	\dot z(t)&=f(z(t),y(t)),\qquad z(0)=z_0,\\
	y(t)&\in U.
\end{align}
Here $U$ is a nonempty closed subset of $\mathbb{R}^m$, $z_0\in G$ is an initial state, $G$ is an open convex subset of $\mathbb{R}^n$, $\del>0$ is the discount rate, and  $f : G\times U \mapsto\mathbb{R}^n$ and $g : G \times U\mapsto\mathbb{R}^m$ are given functions. The class of admissible controls in $(P4)$ consists of all measurable locally bounded functions $y\colon [0,\infty)\mapsto U$. The optimality of admissible pair $(z_*(\cdot), y_*(\cdot))$ is understood in the strong sense \cite{carlson2012infinite}.

\subsection{Proving the Existence of an Optimal Pair}

Problems of type $(P4)$ were intensively studied in last decades (see \cite{aseev2015adjoint, aseev2015boundedness, aseev2016existence, aseev2012rm, aseev2004pontryagin, aseev2007pontryagin, aseev2012maximum, aseev2014needle, aseev2015maximum}). Here we will employ the existence result and the variant of the Pontryagin maximum principle for problem $(P4)$ developed in \cite{ aseev2015boundedness, aseev2016existence} and \cite{aseev2012maximum, aseev2014needle, aseev2015maximum} respectively.

We will need to verify validity of the following conditions (see~\cite{aseev2015boundedness, aseev2016existence, aseev2012rm, aseev2012maximum, aseev2014needle, aseev2015maximum}). \vskip3mm

\noindent{\bf (A1)} {\it The functions $f(\cdot,\cdot)$ and $g(\cdot,\cdot)$ together with their partial derivatives $f_z(\cdot,\cdot)$ and $g_z(\cdot,\cdot)$ are continuous  and locally bounded on $G\times U$.}
\vskip3mm

\noindent{\bf (A2)} {\it There exists a number $\mathrm{c} > 0$ and a nonnegative integrable function $\lambda: [0,\infty)\mapsto\mathbb{R}^1$ such that for every $\zeta \in G$ with $\|\zeta - z_0 \| < \mathrm{c}$ equation $(\ref{g-Cauchy})$ with $y(\cdot) = y_*(\cdot)$ and initial condition $z(0) = \zeta$ (instead of $z(0) = z_0$) has a solution $z(\zeta;\cdot)$ on $[0,\infty)$ in $G$, and
\[
	\max_{\theta\in [z(\zeta;t),z_*(t)]} \; \Big|e^{-\del t}\langle g_z(\theta,y_*(t)),z(\zeta;t)-z_*(t)\rangle\Big| \, \stackrel{a.e.}{\leq} \, \|\zeta-z_0\|\lambda(t).
\]
Here $[z(\zeta;t),z_*(t)]$ denotes the line segment with vertices $z(\zeta;t)$ and $z_*(t)$.}
\vskip3mm

\noindent{\bf (A3)} {\it  There is a positive function $\omega(\cdot)$ decreasing on $[0,\infty)$ such that $\omega(t)\to +0$ as $t\to\infty$ and for any admissible pair $(z(\cdot),y(\cdot))$ the following estimate holds:
\[
  \int_T^{T'}e^{-\del t}g(z(t),y(t))\,dt\le\omega(T), \qquad   0\leq T\leq T'.
\]
} \vskip3mm

Obviously, condition $(A1)$ is satisfied because $f(z,y)=\left[y-1\right]z + 1$, $g(z,y)=\ln y-\ln z$, $f_z(z,y)=y-1$ and $g_z(z,y) = -1/z$, $z>0$, $y\in [\del,\infty)$, in $(P3)$.

Let us show that  $(A2)$ also holds for any admissible pair $(z_*(\cdot), y_*(\cdot))$ in $(P3)$. Set $\mathrm{c} = z_0/2$ and define the nonnegative integrable function $\lambda: [0,\infty)\mapsto \mathbb{R}^1$ as follows: $\lambda(t)=2e^{-\del t}/z_0$, $t\geq 0$. Then, as it can be seen directly, for any real $\zeta$: $|\zeta - z_0 | < \mathrm{c}$, the Cauchy problem $\eqref{eq22}$ with $y(\cdot) = y_*(\cdot)$ and the initial condition $z(0) = \zeta$ (instead of $z(0) = z_0$) has a solution $z(\zeta;\cdot)$ on $[0,\infty)$ and
\[
	\max_{\theta\in [z(\zeta;t),z_*(t)]} \; \Big|   e^{-\del t}g_z(\theta,y_*(t))\left(z(\zeta;t)- z_*(t)\right)\Big| \, \stackrel{\text{a.e.}}{\leq} \, |\zeta-z_0|\lambda(t).
\]
Hence, for any admissible pair $(z_*(\cdot),y_*(\cdot))$ condition $(A2)$ is also satisfied.

Validity of $(A3)$ for any admissible pair $(z_*(\cdot),y_*(\cdot))$ follows from \eqref{eq1_2} directly.

For an admissible pair $(z(\cdot),y(\cdot))$ consider the following linear system:
\begin{equation}\label{eq-17}
      \dot q(t) = -\left[f_z(z(t),y(t)) \right]^* q(t)=\left[-y(t)+1\right]q(t).
\end{equation}
The normalized fundamental solution  $Q(\cdot)$ to equation~\eqref{eq-17}  is defined as follows:
\begin{equation}
\label{eq-19}
	Q(t)=e^{-\int_0^t y(\xi)\,d\xi + t}, \qquad t\geq 0.
\end{equation}
Due to~\eqref{eq-22} and~\eqref{eq-19} for any admissible pair $(z(\cdot),y(\cdot))$ we have
\begin{multline*}
	\left| e^{-\del t}Q^{-1}(t)g_z(z(t),y(t))\right|\\
 	=\left|\frac{e^{-\del t}e^{\int_0^t y(\xi)\,d\xi - t}}{z_0 e^{\int_0^t y(\xi)\,d\xi-t} + e^{\int_0^t y(\xi)\,d\xi-t}\int_0^t e^{-\int_0^s y(\xi)\,d\xi+s}\,ds} \right|\leq\frac{e^{-\del t}}{z_0},\quad t\geq 0.
\end{multline*}
Hence, for any $T > 0$  the function $\psi_T:\, [0,T]\mapsto\mathbb{R}^1$ defined as
\begin{multline}
\label{psiT}
	\psi_T(t)=Q(t)\int_t^{T}e^{-\del s}Q^{-1}(s)g_z(z(s),y(s))\, ds\\
	=-e^{-\int_0^t y(\xi)\,d\xi + t}\int_t^{T}\frac{e^{\int_0^s y(\xi)\,d\xi - s}e^{-\del s}}{z(s)}\, ds,\qquad t\in [0,T],
\end{multline}
is absolutely continuous, and the function $\psi\colon [0,\infty)\mapsto\mathbb{R}^1$ defined as
\begin{multline}
\label{psi}
	\psi(t)=Q(t)\int_t^{\infty}e^{-\del s}Q^{-1}(s)g_z(z(s),y(s))\, ds\\
	=-e^{-\int_0^t y(\xi)\,d\xi + t}\int_t^{\infty}\frac{e^{\int_0^s y(\xi)\,d\xi - s}e^{-\del s}}{z(s)}\, ds,\qquad t\geq 0,
\end{multline}
is locally absolutely continuous.

Define the normal form Hamilton-Pontryagin function $\mathcal{H}:[0,\infty)\times (0,\infty)\times [\del,\infty)\times \mathbb{R}^1\mapsto \mathbb{R}^1$ and the normal-form Hamiltonian $H :[0,\infty)\times (0,\infty)\times \mathbb{R}^1\mapsto \mathbb{R}^1$ for problem $(P3)$ in the standard way:
\[
     \mathcal{H}(t,z,y,\psi) = \psi f(z,y) +e^{-\del t}g(z,y) =\psi[(y-1)z+1] +e^{-\del t}[ \ln y - \ln z],
\]
\[
	H(t,z,\psi) =\sup_{y\geq\del}\mathcal{H}(t,z,y,\psi),
\]
\[
	t\in [0,\infty),\quad z\in (0,\infty),\quad y\in [\del,\infty),\quad  \psi\in \mathbb{R}^1.
\]

\begin{theorem}
\label{thm-exist}
	There is an optimal admissible control $y_*(\cdot)$ in problem $(P3)$. Moreover, for any optimal admissible pair $(z_*(\cdot),y_*(\cdot))$ we have
	\begin{equation}\label{eq-20_1}
		y_*(t)\stackrel{\text{a.e.}}{\leq}\left(1+\frac{1}{z_*(t)}\right)(1+\del),\qquad t\geq 0.
	\end{equation}
\end{theorem}
\begin{proof}  
	Let us show that there is a continuous function $M\colon [0,\infty)\mapsto\mathbb{R}^1$, $M(t)\geq 0$, $t\geq 0$, and a function $\phi\colon\, [0,\infty)\mapsto\mathbb{R}^1$, $\phi(t)>0$, $t\geq 0$, $\lim_{t\to\infty}\left(\phi (t)/t\right)=0$, such that for any admissible pair $(z(\cdot),y(\cdot))$, satisfying on a set $\mathfrak{M}\subset[0,\infty)$, $\meas\mathfrak{M}>0$, to inequality $y(t)>M(t)$, for all $t\in\mathfrak{M}$  we have
	\begin{equation}
	\label{eq-20}
		\inf_{T:\, T-\phi(T)\geq t}\left\{\sup_{y\in [\del,M(t)]}\,\mathcal{H}(t,z(t),y,\psi_T(t))- \mathcal{H}(t,z(t),y(t),\psi_T(t))\right\}>0,
	\end{equation}
	where the function $\psi_T(\cdot)$ is defined on $[0,T]$, $T>0$, by equality \eqref{psiT}.

	Let $(z(\cdot),y(\cdot))$ be an arbitrary admissible pair in $(P3)$. Then due to \eqref{eq-22} and \eqref{eq-19}, for any $T>0$ and arbitrary $t\in [0,T]$ we get (see~\eqref{psiT})
	\begin{multline}
	\label{eq19}
		-z(t)\psi_T(t) = \left[z_0 + \int_0^te^{-\int_0^s y(\xi)\,d\xi+s}\,ds\right]\int_t^{T}\frac{e^{-\del s}}{z_0 + \int_0^s e^{-\int_0^\tau y(\xi)\,d\xi + \tau}\, d\tau}\, ds\\
		\geq z_0\int_t^{T}\frac{e^{-\del s}}{z_0 + \int_0^s e^{\tau}\, d\tau}\, ds \geq \frac{z_0e^{-(1+\del)t}}{(z_0+1)(1+\del)}\left[1-e^{-(1+\del)(T-t)}\right].
	\end{multline}

	For a $\phi >0$  define the function $M_{\phi}\colon [0,\infty)\mapsto\mathbb{R}^1$ by equality
	\begin{equation}
	\label{M-delta}
		M_\phi(t)=\frac{(z_0+1)(1+\del)}{z_0\left[1-e^{-(1+\del)\phi}\right]}e^{t}+\frac{1}{\phi},\qquad t\geq 0.
	\end{equation}
	Then for any $T$: $T-\phi>t$ and arbitrary $(z(\cdot),y(\cdot))$ the function $y\mapsto \mathcal{H}(t,z(t),y,\psi_T(t))$ reaches its maximal value on $[\del,\infty)$ at the point (see~\eqref{eq19})
	\begin{equation}
	\label{eq-21}
		y_T(t) = -\frac{e^{-\del t}}{z(t)\psi_T(t)}\leq\frac{(z_0+1)(1+\del)}{z_0\left[1-e^{-(1+\del)(T-t)}\right]}e^{t}\leq M_{\phi}(t)-\frac{1}{\phi}.
	\end{equation}

	Now, set $\phi(t)\equiv\phi$ and $M(t)\equiv M_{\phi}(t)$, $t\geq 0$. Let $(z(\cdot),y(\cdot))$ be an admissible pair such that inequality $y(t)>M_{\phi}(t)$ holds on a set $\mathfrak{M}\subset[0,\infty)$, $\meas\mathfrak{M}>0$. For arbitrary $t\in \mathfrak{M}$ define the function $\Phi\colon [t+\phi,\infty)\mapsto\mathbb{R}^1$ as follows
	\begin{multline*}
		\Phi(T)=\sup_{y\in [\del, M(t)]}\mathcal{H}(t,z(t),y,\psi_T(t))-\mathcal{H}(t,z(t),y(t),\psi_T(t))\\
		=\psi_T(t)y_T(t)z(t) +e^{-\del t}\ln y_T(t)-\left[\psi_T(t)y(t)z(t) + e^{-\del t}\ln y(t)\right],\qquad T\geq t+\phi.
	\end{multline*}
	Due to~\eqref{eq-21} we have
	\begin{multline*}
		\Phi(T)=-e^{-\del t}+e^{-\del t}\left[ -\del t -\ln (-\psi_T(t)) -\ln z(t)\right]\\
		-\left[\psi_T(t)y(t)z(t) + e^{-\del t}\ln y(t)\right],\qquad T\geq t+\phi.
	\end{multline*}
	Hence, due to~\eqref{psiT} and~\eqref{eq-21}  for a.e. $T\geq t+\phi$ we get
	\begin{multline*}
		\frac{d}{dT}\Phi(T)=-\frac{e^{-\del t}}{\psi_T(t)}\frac{d}{dT}\left[\psi_T(t)\right] -y(t)z(t)\frac{d}{dT}\left[\psi_T(t)\right]\\
		=z(t)\frac{d}{dT}\left[\psi_T(t)\right]\left[ \frac{e^{-\del t}}{-\psi_T(t)z(t)} -y(t)\right]= z(t)\frac{d}{dT}\left[\psi_T(t)\right](y_T(t)-y(t))>0.
	\end{multline*}
	Hence,
	\begin{multline*}
		\inf_{T>0:\, t\leq T-\phi}\left\{\sup_{y\in [\del,M(t)]}\, \mathcal{H}(t,z(t),y,\psi_T(t))- \mathcal{H}(t,z(t),y(t),\psi_T(t))\right\}\\
		=\inf_{T>0:\, t\leq T-\phi}\Phi(T) =\Phi(t+\phi)>0.
	\end{multline*}
	Thus, for any $t\in \mathfrak{M}$ inequality~\eqref{eq-20} is proved.

	Since the instantaneous utility in~\eqref{eq21} is concave in $y$, the system~\eqref{eq22} is affine in $y$,  the set $U$ is closed (see~\eqref{eq23}), conditions $(A1)$ and $(A3)$ are satisfied, and  since $(A2)$ also holds for any admissible pair $(z_*(\cdot),y_*(\cdot))$ in $(P3)$, all conditions of the existence result in \cite{aseev2016existence} are fulfilled (see also \cite[Theorem~1]{aseev2015boundedness}). Hence, there is an optimal admissible control $y_*(\cdot)$ in $(P3)$ and, moreover, $y_*(t)\stackrel{a.e.}{\leq} M_{\phi}(t)$, $t\geq 0$. Passing to a limit in this inequality as $\phi\to\infty$  we get (see~\eqref{M-delta})
	\begin{equation}
	\label{u_*-bound}
		y_*(t)\stackrel{\text{a.e.}}{\leq}\left(1+\frac{1}{z_0}\right)(1+\del)e^{ t},\qquad t\geq 0.
	\end{equation}

	Further, it is easy to see that for any $\tau>0$ the pair $(\tilde z_*(\cdot),\tilde y_*(\cdot))$ defined as $\tilde z_*(t)=z_*(t+\tau)$, $\tilde y_*(t)=y_*(t+\tau)$, $t\geq 0$, is an optimal admissible pair in the problem $(P3)$ taken with initial condition $z(0)=z_*(\tau)$. Hence, using the same arguments as above we get the following inequality for $(\tilde z_*(\cdot),\tilde y_*(\cdot))$ (see~\eqref{u_*-bound}):
	\[
		\tilde y_*(t)\stackrel{\text{a.e.}}{\leq}\left(1+\frac{1}{\tilde z_*(0)}\right)(1+\del)e^{t},\qquad t\geq 0.
	\]
	Hence, for arbitrary fixed $\tau>0$ we have
	\[
		y_*(t)=\tilde y_*(t-\tau)\stackrel{\text{a.e.}}{\leq}\left(1+\frac{1}{z_*(\tau)}\right)(1+\del)e^{t-\tau},\qquad t\geq\tau.
	\]
	Due to arbitrariness of $\tau>0$ this implies~\eqref{eq-20_1}. 
\end{proof}

\subsection{Application of the Maximum Principle}

\begin{theorem}
\label{thm2}
	Let $(z_*(\cdot),y_*(\cdot))$ be an optimal admissible pair in problem $(P3)$. Then the function $\psi: [0,\infty)\mapsto\mathbb{R}^1$ defined for pair $(z_*(\cdot),y_*(\cdot))$ by formula~\eqref{psi} is (locally) absolutely continuous and satisfies the conditions of the normal form maximum principle, i.e. $\psi(\cdot)$ is a solution of the adjoint system
	\begin{equation}
	\label{dpsi}
		\dot\psi (t) =  - \mathcal{H}_z\left(z_*(t), y_*(t),\psi(t)\right),
	\end{equation}
	and the maximum condition holds:
	\begin{equation} 
	\label{max-psi}
		\mathcal{H}(z_*(t),y_*(t),\psi(t))\stackrel{\text{a.e.}}{=} H(z_*(t),\psi(t)).
	\end{equation}
\end{theorem}

\begin{proof}
	As it already have been shown above condition $(A1)$ is satisfied and $(A2)$ holds for any admissible pair $(z_*(\cdot),y_*(\cdot))$ in $(P3)$. Hence, due to the variant of the maximum principle developed in~\cite{aseev2012maximum, aseev2014needle, aseev2015maximum} the function $\psi: [0,\infty)\mapsto\mathbb{R}^1$ defined for pair $(z_*(\cdot),y_*(\cdot))$ by formula~\eqref{psi} satisfies the conditions~\eqref{dpsi} and~\eqref{max-psi}. \end{proof}

Notice, that the Cauchy type formula~\eqref{psi} implies (see~\eqref{eq-22} and \eqref{eq-19})
\begin{multline}
\label{psi-}
	\psi(t)=-e^{-\int_0^t y_*(\xi)\, d\xi + t}\int_t^{\infty}\frac{e^{-\del\tau}e^{\int_0^{\tau}y_*(\xi)\, d\xi - \tau}}{e^{\int_0^{\tau}y_*(\xi)\, d\xi - \tau}\left[z_0+ \int_0^{\tau}e^{-\int_0^{\theta}y_*(\xi)\, d\xi + \theta}\, d\theta \right]}\, d\tau\\
	> -\frac{e^{-\int_0^t y_*(\xi)\, d\xi + t}}{z_0+ \int_0^{t}e^{-\int_0^{\theta} y_*(\xi)\, d\xi + \theta}\, d\theta}\int_t^{\infty}e^{-\del\tau}\, d\tau = -\frac{e^{-\del t}}{\del z_*(t)},\qquad t\geq 0.
\end{multline}
Thus, due to~\eqref{psi}the following condition holds:
\begin{equation}
\label{g-transv-psi}
	0< -\psi(t)z_*(t)<\frac{e^{-\del t}}{\del},\qquad t\geq 0.
\end{equation}

Note also, that due to~\cite[Corollary to Theorem~3]{aseev2015adjoint} formula~\eqref{psi} implies the following stationarity condition for the Hamiltonian (see~\cite{aseev2007pontryagin, michel1982transversality})):
\begin{equation}
\label{stat}
	H(t,z_*(t),\psi(t))=\del\int_t^{\infty}e^{-\del s}g(z_*(s),y_*(s))\, ds,\qquad t\geq 0.
\end{equation}

It can be shown directly that if an admissible pair (not necessary optimal) $(z(\cdot),y(\cdot))$ together with an adjoint variable $\psi(\cdot)$ satisfies the core conditions~\eqref{dpsi} and~\eqref{max-psi} of the maximum principle and $\lim_{t\to\infty}H(t,z(t),\psi(t))=0$ then condition~\eqref{stat} holds for the  triple  $(z(\cdot),y(\cdot),\psi(\cdot))$  as well (see~\cite[Section 3]{aseev2007pontryagin}).

Further, due to the maximum condition~\eqref{max-psi} for a.e. $t\geq 0$ we have
\[
	y_*(t) ={\arg\max}_{y\in [\del,\infty)}\left[\psi(t)z_*(t)y+e^{-\del t}\ln y\right].
\]
This implies (see~\eqref{g-transv-psi})
\begin{equation}
\label{eq:max_prncpl_1}
	y_*(t)\stackrel{a.e.}{=}-\frac{e^{-\del t}}{\psi(t)z_*(t)}>\del,\qquad t\in[0,\infty).
\end{equation}

\subsection{The Hamiltonian System}

Substituting the formula \eqref{eq:max_prncpl_1} for $y_*(\cdot)$ in~\eqref{eq22} and in~\eqref{dpsi} due to Theorem~\ref{thm2} we get that any optimal trajectory $z_*(\cdot)$ together with the corresponding adjoint variable $\psi(\cdot)$  must satisfy to the Hamiltonian system of the maximum principle:
\begin{align}
\label{ham-psi}
\begin{split}
	&\dot z(t) =  -z(t)-\frac{e^{-\del t}}{\psi(t)} + 1,\\
	&\dot\psi(t) =\psi (t)+\frac{2e^{-\del t}}{z(t)}.
\end{split}
\end{align}
Moreover, estimate~\eqref{g-transv-psi} and condition~\eqref{stat} must hold as well. In the terms of the current value adjoint variable $\lambda(\cdot)$, $\lambda(t)=e^{\del t}\psi(t)$,  $t\geq 0$, one can rewrite system~\eqref{ham-psi}  as follows:
\begin{align}
\label{ham-lamb}
\begin{split}
	&\dot z(t) =  -z(t)-\frac{1}{\lambda(t)} + 1,\\
	&\dot\lambda(t) =(\del+1)\lambda (t)+\frac{2}{z(t)}.
\end{split}
\end{align}
In terms of variable  $\lambda(\cdot)$ estimate~\eqref{g-transv-psi} takes the following form:
\begin{equation}
\label{g-transv-lamb}
	0< -\lambda(t)z_*(t)<\frac{1}{\del},\qquad t\geq 0.
\end{equation}

Accordingly, the optimal control $y_*(\cdot)$ can be expressed as follows (see~\eqref{eq:max_prncpl_1}):
\begin{equation}
\label{eq:max_prncpl}
	y_*(t)\stackrel{a.e.}{=}-\frac{1}{\lambda(t)z_*(t)},\qquad t\geq 0.
\end{equation}

Define the normal form current value Hamiltonian $M : (0,\infty)\times \mathbb{R}^1\mapsto \mathbb{R}^1$ for problem $(P3)$ in the standard way (see~\cite[Section~3]{aseev2007pontryagin}):
\begin{equation}
\label{current-H}
	M(z,\lambda)=e^{\del t}H(t,z,\psi),\quad z\in (0,\infty),\quad \lambda\in \mathbb{R}^1.
\end{equation}
Then in the current value terms the stationarity condition~\eqref{stat} takes the form
\begin{equation}
\label{stat-lamb}
	M(z_*(t),\lambda(t))=\del e^{\del t}\int_t^{\infty}e^{-\del s}g(z_*(s),y_*(s))\, ds,\qquad t\geq 0.
\end{equation}

In the next section we will analyze the system~\eqref{ham-lamb} coupled with the estimate~\eqref{g-transv-lamb} and the stationarity condition~\eqref{stat-lamb}. We will show that there are only two qualitatively different types of behavior of the optimal paths that are possible. If $\del<1$ then the optimal path asymptotically approaches an optimal nonvanishing steady state while the corresponding optimal control tends to $(1+\del)/2$ as $t\to\infty$. If $\del \geq 1$ then the optimal path $z_*(\cdot)$ goes to infinity, while the corresponding optimal control $y_*(\cdot)$  tends to $\del$ as $t\to\infty$, i.e. asymptotically it follows the Hotelling rule of optimal depletion of an exhaustible resource~\cite{hotelling1931economics}.
%%%%%%%%%%%%%%%%%%%%%%%%%%%%%%%%%%%%%%%%%%%%%%%%%%%%%%%%%%%%%%%%

\section{Analysis of the Hamiltonian system}\label{sec:ham}

Due to Theorem~\ref{thm2} it is sufficient to analyze the behavior of system~\eqref{ham-lamb} only in the open set $\Gamma=\left\{\, (z,\lambda)\colon z>0,\lambda<0\right\}$ in the phase plane ${\mathbb R}^2$.

Let us introduce functions $v_1\colon\, (1,\infty)\mapsto (-\infty,0)$ and $v_2\colon\, (0,\infty)\mapsto (-\infty,0)$ as follows 
\[
	v_1(z)=\frac{1}{1-z},\quad z\in \left(1,\infty\right),\qquad v_2(z)=-\frac{2}{(\del+1)z},\quad z\in (0,\infty).
\]
Due to~\eqref{ham-lamb} the curves $\gamma_1 =\{(z,\lambda)\colon\, \lambda=v_1(z), z\in (1,\infty)\}$ and $\gamma_2 =\{ (z,\lambda)\colon\, \lambda=v_2(z), z\in (0,\infty)\}$ are the nullclines at which the derivatives of variables $z(\cdot)$ and $\lambda(\cdot)$ vanish respectively.

Two qualitatively different cases are possible: $(i)$ $\del < 1$ and $(ii)$ $\del \geq 1$.

\subsection{The Sustainable Case}
Here we consider the case where $\del < 1$. In this case the nullclines $\gamma_1$ and $\gamma_2$ have a unique intersection point $(\hat z,\hat\lambda)$ which is a unique equilibrium of system~\eqref{ham-lamb} in $\Gamma$:
\begin{equation} 
\label{eq_pt}
	\hat z=\frac{2}{1-\del},\qquad \hat\lambda=\frac{\del-1}{\del+1}.
\end{equation}
The corresponding equilibrium control $\hat y(\cdot)$ is
\begin{equation}
\label{eq-u}
	\hat y(t)\equiv\hat y=\frac{\del+1}{2},\qquad t\geq 0.
\end{equation}

The eigenvalues of the system linearized around the equilibrium are given by
\begin{align*}
	\sigma_{1,2} = \frac{\del}{2} \pm \frac{1}{2}\sqrt{2 - \del^2},
\end{align*}
which are real and distinct with opposite signs when $\del<1$. Hence, by the Grobman-Hartman theorem  in a neighborhood $\Omega$ of the equilibrium state $(\hat z,\hat\lambda)$ the system~\eqref{ham-lamb} is of saddle type (see~\cite[Chapter 9]{hartman1970ordinary}).

The nullclines $\gamma_1$ and  $\gamma_2$ divide the set $\Gamma$ in four open regions:
{
\allowdisplaybreaks
\setlength{\mathindent}{0cm}
\begin{align*}
	&\Gamma_{-,-}=\Big\{ (z,\lambda)\in\Gamma\colon\, \lambda< v_1(z), 1<z\leq\hat z\Big\}\bigcup\Big\{(z,\lambda)\in\Gamma\colon\, \lambda< v_2(z), \hat z<z<\infty\Big\},\\
	&\Gamma_{+,-} = \Big\{  (z,\lambda) \in \Gamma\colon\, \lambda <  v_2(z), 0 < z \leq 1 \Big\} \bigcup \Big\{ (z,\lambda) \in \Gamma\colon\, v_1(z) < \lambda <  v_2(z),1 < z < \hat z \Big\},\\
	&\Gamma_{+,+} = \Big\{  (z,\lambda) \in \Gamma\colon\, v_2(z) < \lambda < 0, 0 < z \leq \hat z\Big\} \bigcup \Big\{ (z,\lambda) \in \Gamma\colon\, v_1(z) < \lambda < 0, \hat z < z < \infty \Big\},\\
	&\Gamma_{-,+}=\Big\{(z,\lambda)\in\Gamma\colon\,v_2(z)<\lambda<v_1(z), z>\hat z\Big\}.
\end{align*}
}
Any solution $(z(\cdot),\lambda(\cdot))$ of~\eqref{ham-lamb} in $\Gamma$ has  definite signs of derivatives of its $(z,\lambda)$-coordinates in the sets $\Gamma_{-.-}$, $\Gamma_{-.+}$, $\Gamma_{+,+}$, and $\Gamma_{-,+}$. These signs are indicated by the corresponding subscripts.

The behavior of the flows is shown in Figure~\ref{fig:phase_sust} through the phase portrait.

\begin{figure*}[h!]
    \centering
        \captionsetup{width=0.5\textwidth}
        \centering
        \includegraphics[width = 0.6\linewidth]{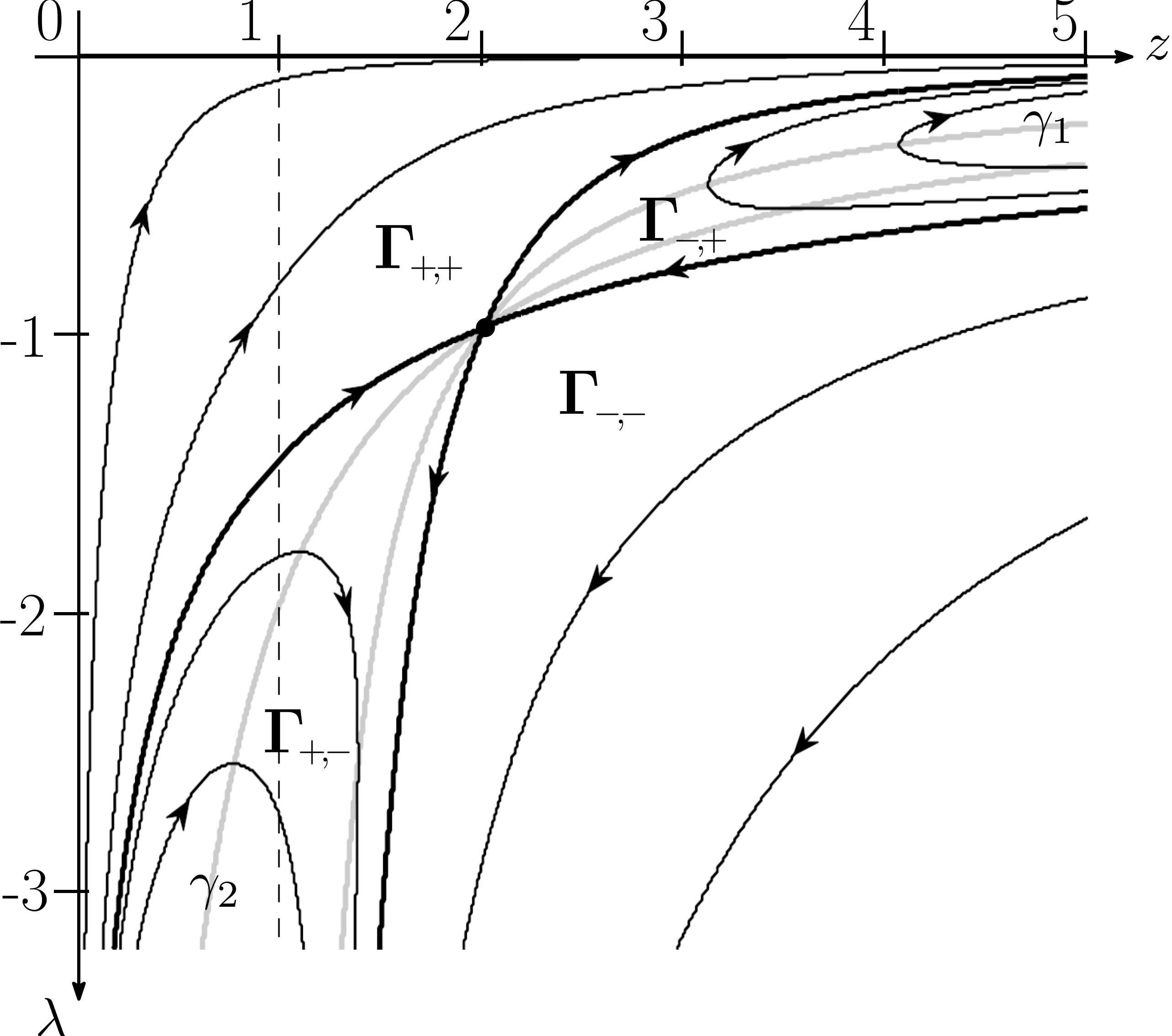}
        \caption{Phase portrait of \eqref{ham-lamb} around ($\hat{z},\hat{\lambda}$), where $\del = 0.01$}
        \label{fig:phase_sust}
\end{figure*}

For any initial state $(\xi,\beta)\in\Gamma$ there is a unique solution $(z_{\xi,\beta}(\cdot),\lambda_{\xi,\beta}(\cdot))$ of the system~\eqref{ham-lamb} satisfying initial conditions $z(0)=\xi$, $\lambda(0)=\beta$, and due to the standard extension result this solution is defined on some maximal time interval $[0,T_{\xi,\beta})$ in $\Gamma$ where $0<T_{\xi,\beta}\leq\infty$ (see~\cite[Chapter 2]{hartman1970ordinary}).

Let us consider behaviors of solutions $(z_{\xi,\beta}(\cdot),\lambda_{\xi,\beta}(\cdot))$ of system~\eqref{ham-lamb} in $\Gamma$ for all possible initial states $(\xi,\beta)\in\Gamma$ as $t\to T_{\xi,\beta}$. The standard analysis of system~\eqref{ham-lamb} shows that only three types of behavior of solutions $(z_{\xi,\beta}(\cdot),\lambda_{\xi,\beta}(\cdot))$ of~\eqref{ham-lamb} in $\Gamma$ as $t\to T_{\xi,\beta}$  are possible: \vskip3mm

{\bf 1)} $(z_{\xi,\beta}(t),\lambda_{\xi,\beta}(t))\in\Gamma_{-,-}$ or $(z_{\xi,\beta}(t),\lambda_{\xi,\beta}(t))\in\Gamma_{+,-}$ for all sufficiently large times  $t$. In this case $T_{\xi,\beta}=\infty$ and $\lim_{t\to\infty} \lambda_{\xi,\beta} (t)= -\infty$ while $\lim_{t\to\infty} z_{\xi,\beta} (t)= 1$. Due to Theorem~\ref{thm2} such asymptotic behavior does not correspond to an optimal path because it contradicts the necessary condition~\eqref{g-transv-lamb}. \vskip3mm

{\bf 2)} $(z_{\xi,\beta}(t),\lambda_{\xi,\beta} (t))\in\Gamma_{+,+}$ for all sufficiently large times $t$. In this case  $\lim_{t\to T_{\xi,\beta}}z_{\xi,\beta}(t)=\infty$ and $\lim_{t\to T_{\xi,\beta}}\lambda_{\xi,\beta}(t)=0$. If $(z_{\xi,\beta}(\cdot),\lambda_{\xi,\beta} (\cdot))$ corresponds to an optimal pair $(z_*(\cdot),y_*(\cdot))$ in $(P3)$  then due to Theorem~\ref{thm2} $z_*(\cdot)\equiv z_{\xi,\beta}(\cdot)$, $T_{\xi,\beta}=\infty$, $\lim_{t\to\infty}z_*(t)=\infty$,\\ and  $\lim_{t\to\infty}\lambda_{\xi,\beta}(t)~=~0$. Set $\lambda_*(\cdot)\equiv\lambda_{\xi,\beta}(\cdot)$ in this case and define the function $\phi_*(\cdot)$ by equality $\phi_*(t)=\lambda_*(t)z_*(t)$, $t\in [0,\infty)$.

By direct differentiation for a.e. $t\in [0,\infty)$ we get (see~\eqref{ham-lamb})
\[
	\dot\phi_*(t)\stackrel{\text{a.e.}}{=}(\del+1)\lambda_*(t)z_*(t)+2- \lambda(t)z_*(t)-1+\lambda_*(t)=\del\phi_*(t)+1+\lambda_*(t).
\]
Hence,
\begin{equation}
\label{phi_*b}
	\phi_*(t)= e^{\del t}\left[\phi_*(0)+\int_0^te^{-\del s}\left(1+\lambda_*(s)\right)\, ds\right],\qquad t\in [0,\infty).
\end{equation}
Since $\lim_{t\to\infty}\lambda_*(t)=0$ the improper integral $\int_0^\infty e^{-\del s}\left( 1+\lambda_*(s)\right)\, ds$ converges, and due to~\eqref{g-transv-lamb} we have $0>\phi_*(t)=\lambda_*(t)z_*(t)>-1/\del$ for all $t>0$. Due to~\eqref{phi_*b} this implies
\[
	\phi_*(0)= -\int_0^\infty e^{-\del s}\left(1+\lambda_*(s)\right)\, ds=-\frac{1}{\del} -\int_0^\infty e^{-\del s}\lambda_*(s)\, ds.
\]
Substituting this expression for $\phi_*(0)$ in~\eqref{phi_*b} we get
\[
	\phi_*(t)= -\frac{1}{\del}-e^{\del t}\int_t^{\infty}e^{-\del s}\lambda_*(s)\, ds,\qquad t\in [0,\infty).
\]
Due to the L'Hospital rule we have
\[
	\lim_{t\to\infty}e^{\del t}\int_t^{\infty}e^{-\del s}\lambda_*(s)\, ds =\lim_{t\to\infty}\frac{\int_t^{\infty}e^{-\del s}\lambda_*(s)\, ds}{e^{-\del t}}=\lim_{t\to\infty}\frac{\lambda_*(t)}{\del} =0.
\]
Hence,
\[
	\lim_{t\to\infty} y_*(t)=\lim_{t\to\infty}\frac{-1}{\lambda_*(t)z_*(t)}=\lim_{t\to\infty}\frac{-1}{\phi_*(t)}=\del.
\]
But due to the system~\eqref{ham-lamb} and the inequality $\del < 1$ this implies $\lim_{t\to\infty}z_*(t)\leq 1 <\infty$ that contradicts the equality $\lim_{t\to\infty}z_*(t)=\infty$. So, all these trajectories of~\eqref{ham-lamb} are the blow up ones. Thus, there are not any trajectories of~\eqref{ham-lamb} that correspond to optimal admissible pairs due to Theorem~\ref{thm2} in the case $2)$. \vskip3mm

{\bf 3)} $\lim_{t\to\infty}(z(t),\lambda(t))=(\hat z,\hat\lambda)$ as $t\to\infty$. In this case, since the equilibrium ($\hat{z},\hat{\lambda}$) is of saddle type, there are only two trajectories of~\eqref{ham-lamb} which tend to the equilibrium point $(\hat z,\hat\lambda)$ asymptotically as $t\to\infty$ and lying on the stable manifold of ($\hat{z},\hat{\lambda}$). One such trajectory $(z_1(\cdot),\lambda_1(\cdot))$ approaches the point $(\hat z,\hat\lambda)$ from the left from the set $\Gamma_{+,+}$ (we call this trajectory {\it the left equilibrium trajectory}), while the second trajectory $(z_2(\cdot),\lambda_2(\cdot))$ approaches the point $(\hat z,\hat\lambda)$ from the right from the set $\Gamma_{-,-}$ (we call this trajectory {\it the right equilibrium trajectory}). It is easy to see that both these trajectories are fit to estimate~\eqref{g-transv-lamb} and  stationarity condition~\eqref{stat-lamb}. Hence, $(z_1(\cdot),\lambda_1(\cdot))$, $(z_2(\cdot),\lambda_2(\cdot))$ and the stationary trajectory $(\hat z(\cdot),\hat\lambda(\cdot))$, $\hat z(\cdot)\equiv \hat z$, $\hat\lambda(\cdot)\equiv\hat\lambda$, $t\geq 0$, are unique trajectories of~\eqref{ham-lamb} which can correspond to the optimal pairs in problem $(P3)$ due to Theorem~\ref{thm2}.

Due to Theorem~\ref{thm-exist} for any initial state $z_0>0$ an optimal control $y_*(\cdot)$ in problem $(P3)$ exists. Hence, for any initial state $\xi\in (0,\hat z)$ there is a unique $\beta<0$ such that the corresponding trajectory $(z_{\xi,\beta}(\cdot),\lambda_{\xi,\beta}(\cdot))$ coincides (up to a shift in time) with the left equilibrium trajectory $(z_1(\cdot),\lambda_1(\cdot))$ on time interval $[0,\infty)$. Analogously, for any initial state $\xi >\hat z$ there is a unique $\beta<0$ such that the corresponding trajectory $(z_{\xi,\beta}(\cdot),,\lambda_{\xi,\beta}(\cdot))$ coincides (up to a shift in time) with the right equilibrium trajectory $(z_2(\cdot),\lambda_2(\cdot))$ on $[0,\infty)$. The corresponding optimal control is defined uniquely by~\eqref{eq:max_prncpl}. Thus, for any initial state $z_0>0$ the corresponding optimal pair $(z_*(\cdot),y_*(\cdot))$ in $(P3)$ is unique, and due to Theorem~\ref{thm2} the corresponding current value adjoint variable $\lambda_*(\cdot)$ is also unique.

\subsection{The Unsustainable Case}

Now, consider the case when $\del \geq 1$.  In this case $v_2(z)>v_1(z)$ for all $z>1$ and hence the nullclines $\gamma_1$ and $\gamma_2$ do not intersect in $\Gamma$. Accordingly, the system~\eqref{ham-lamb} does not have an equilibrium point in $\Gamma$.

The nullclines $\gamma_1$ and $\gamma_2$ divide the set $\Gamma$ in three open regions:
{
\allowdisplaybreaks
\setlength{\mathindent}{0cm}
\begin{align*}
	&\hat\Gamma_{-,-}=\Big\{ (z,\lambda)\in\Gamma\colon\, \lambda< v_1(z), z >1\Big\},\\
	&\hat\Gamma_{+,-} = \Big\{  (z,\lambda) \in \Gamma\colon\, \lambda< v_2(z), 0<z\leq 1 \Big\}  \bigcup \Big\{ (z,\lambda) \in \Gamma\colon\, v_1(z)<\lambda< v_2(z), z>1  \Big\},\\
	&\hat\Gamma_{+,+} = \Big\{  (z,\lambda) \in \Gamma\colon\,v_2(z) < \lambda < 0, 0 < z \leq \hat z \Big\} \bigcup \Big\{ (z,\lambda) \in \Gamma\colon\, v_1(z) < \lambda < 0, \hat z < z < \infty \Big\}.
\end{align*}
}
The behavior of the flows is shown in Figure~\ref{fig:phase_unsust} through the phase portrait.

\begin{figure*}[h!]
        \captionsetup{width=0.8\textwidth}
        \centering
        \includegraphics[width = 0.6\linewidth]{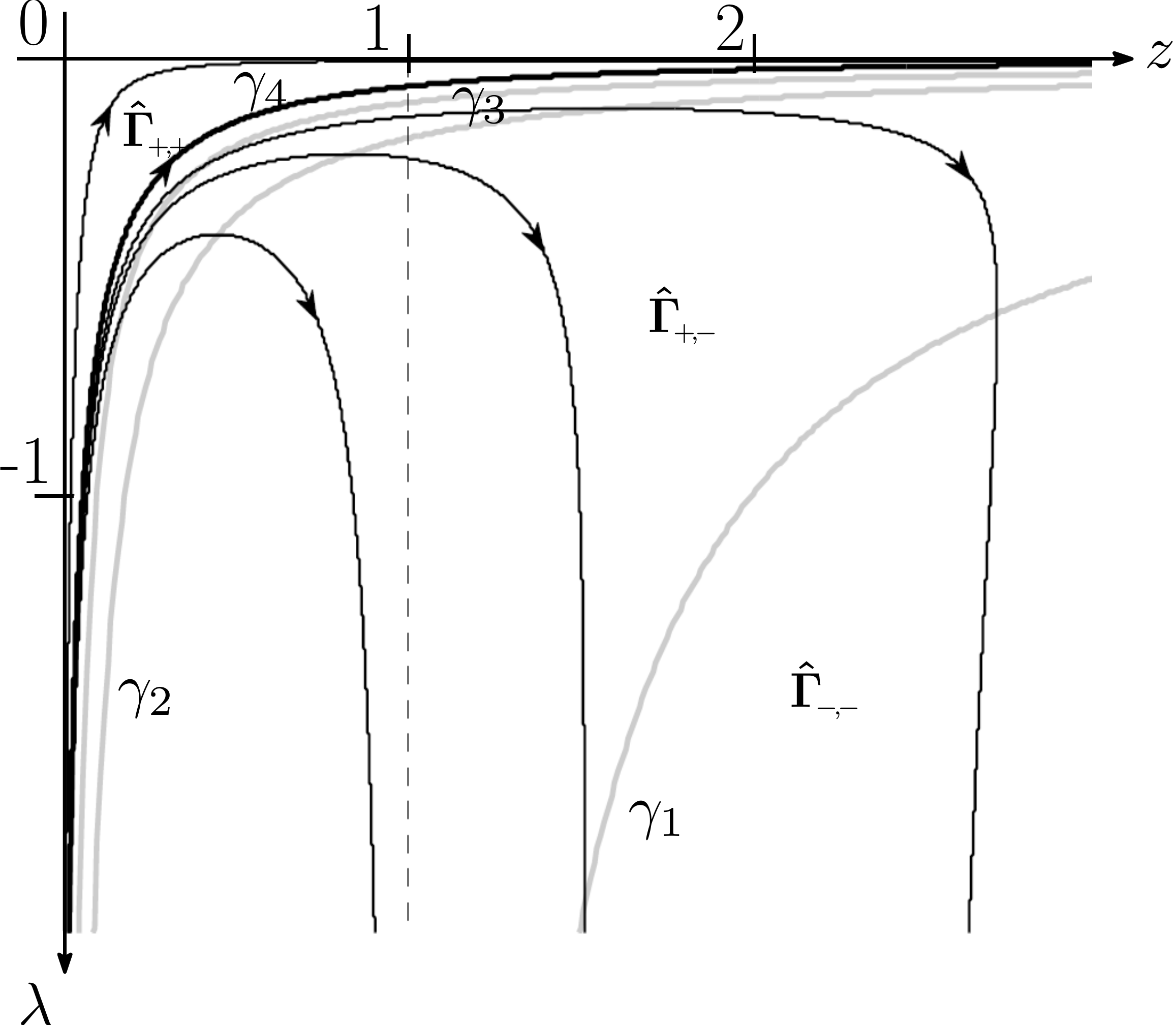}
        \caption{Phase portrait of \eqref{ham-lamb} in the case $\del = 10 \geq 1$.}
        \label{fig:phase_unsust}
\end{figure*}

Any solution $(z(\cdot),\lambda(\cdot))$ of~\eqref{ham-lamb} in $\Gamma$ has the definite signs of derivatives of its $(z,\lambda)$ coordinates in each set $\hat\Gamma_{-.-}$, $\hat\Gamma_{+,+}$, and $\hat\Gamma_{-,+}$ as indicated by the subscripts.

The standard analysis of the behaviors of solutions $(z(\cdot),\lambda(\cdot))$ of system~\eqref{ham-lamb} in each of sets $\hat\Gamma_{-.-}$, $\hat\Gamma_{+.-}$ and $\Gamma_{+,+}$ shows that there are only two types of asymptotic behavior of solutions $(z(\cdot),\lambda(\cdot))$ of~\eqref{ham-lamb} that are possible:

{\bf 1)} $\lim_{t\to\infty}z(t)=1$, $\lim_{t\to\infty}\lambda (t) =-\infty$. In this case $(z(t),\lambda (t))\in\hat\Gamma_{-,-}$ for all sufficiently large times $t\geq 0$. Due to Theorem~\ref{thm2} such asymptotic behavior does not correspond to an optimal admissible pair because in this case $\lim_{t\to\infty}\lambda(t)x(t)=-\infty$ that contradicts condition~\eqref{g-transv-lamb}. Thus this case can be eliminated from the consideration.

{\bf 2)} $\lim_{t\to\infty}z(t)=\infty$, $\lim_{t\to\infty}\lambda (t) =0$. In this case $(z(t),\lambda (t))\in\hat\Gamma_{+,+}$  for all $t\geq 0$. Since the case 1) can be eliminated from the consideration, we conclude that the case 2) is the only one that can be realized for an optimal admissible pair $(z_*(\cdot),y_*(\cdot))$ (which exists) in $(P3)$  due to the maximum principle (Theorem~\eqref{thm2}).

Let us consider behavior of trajectory $(z_*(\cdot),\lambda_*(\cdot))$ of system~\eqref{ham-lamb} that corresponds to the optimal pair $(z_*(\cdot),y_*(\cdot))$ in the set $\hat\Gamma_{+,+}$ in more details.

As in the subcase $(b)$ of case $(i)$ above, define the function $\phi_*(\cdot)$ as follows:
\[
	\phi_*(t)=\lambda_*(t)z_*(t),\qquad t\in [0,\infty).
\]
Repeating the calculations presented in the subcase $(b)$ of case $(i)$ we get
\[
	\phi_*(t)= -\frac{1}{\del}-e^{\del t}\int_t^{\infty}e^{-\del s}\lambda_*(s)\, ds,\qquad t\in [0,\infty).
\]
As in the subcase $(b)$ of case $(i)$ above,  due to the L'Hospital rule this implies
\[
	\lim_{t\to\infty}e^{\del t}\int_t^{\infty}e^{-\del s}\lambda_*(s)\, ds =\lim_{t\to\infty}\frac{\int_t^{\infty}e^{-\del s}\lambda_*(s)\, ds}{e^{-\del t}}=\lim_{t\to\infty}\frac{\lambda_*(t)}{\del} =0.
\]
Hence,
\[
	\lim_{t\to\infty}y_*(t)=\lim_{t\to\infty}\frac{-1}{\lambda_*(t)y_*(t)}=\lim_{t\to\infty}\frac{-1}{\phi_*(t)}=\del.
\]
Thus, asymptotically, any optimal admissible control $y_*(\cdot)$ satisfies the Hotelling rule~\cite{hotelling1931economics} of optimal depletion of an exhaustible resource in the case $\del \geq 1$.

Now let us show that the optimal control $y_*(\cdot)$ is defined  uniquely by Theorem~\ref{thm2} in the case $\del \geq 1$.

Define the function $v_3\colon (0,\infty)\mapsto\mathbb{R}^1$ and the curve $\gamma_3\subset\Gamma $ as follows:
\[
	v_3(z)=-\frac{1}{\del z},\quad z\in (0,\infty),\qquad \gamma_3=\left\{(z,\lambda)\colon \lambda=v_3(z), z\in (0,\infty)\right\}.
\]

It is easy to see that $v_3(z)\geq v_2(z)$ for all $z>0$  and $v_3(z)>v_1(z)$ for all  $z>1$  in the case $\del \geq 1$.  Hence, the curve $\gamma_3$ is located not below $\gamma_2$ and strictly above $\gamma_1$ in $\hat\Gamma_{+,+}$ (see Figure~\ref{fig:phase_unsust}). Notice that if $\del=1$ then $\gamma_3$ coincides with $\gamma_2$ while if $\del > 1$ then $\gamma_3$ lies strictly above $\gamma_2$ in $\hat\Gamma_{+,+}$. It can be demonstrated directly that any trajectory $(z(\cdot),\lambda(\cdot))$ of system~\eqref{ham-lamb} can intersect curve $\gamma_3$ only one time and only in the upward direction.

Due to~\eqref{g-transv-lamb} a trajectory $(z_*(\cdot),\lambda_*(\cdot))$ of system~\eqref{ham-lamb} that corresponds to the optimal pair $(z_*(\cdot),y_*(\cdot))$ lies strictly above $\gamma_3$. Since the system~\eqref{ham-lamb} is autonomous by virtue of the theorem on uniqueness of a solution of first-order ordinary differential equation (see~\cite[Chapter 3]{hartman1970ordinary}) trajectories of system~\eqref{ham-lamb} that lie above $\gamma_3$ do not intersect the curve $\gamma_4=\{(z,\lambda)\colon z=z_*(t), \lambda=\lambda_*(t), t\geq 0\}$ which is the graph of the trajectory $(z_*(\cdot),\lambda_*(\cdot))$.

Further, trajectory $(z_*(\cdot),\lambda_*(\cdot))$ is defined on infinite time interval $[0,\infty)$. This implies that all trajectories $\left(z_{z_0,\beta}(\cdot),\lambda_{z_0,\beta}(\cdot)\right)$, $\beta\in (-1/(\del z_0),\lambda_*(0))$, are also defined on the whole infinite time interval $[0,\infty)$, i.e. $T_{{z_0,\beta}}=\infty$ for all $\beta\in (-1/(\del z_0),\lambda_*(0))$. Thus, we have proved that there is a nonempty set (a continuum) of trajectories $\left\{(z_{z_0,\beta}(\cdot),\lambda_{z_0,\beta}(\cdot))\right\}$, $\beta\in (-1/(\del z_0),\lambda_*(0))$, $t\in [0,\infty)$, of system~\eqref{ham-lamb} lying strictly between the curves $\gamma_3$ and $\gamma_4$. All these trajectories are defined on the whole infinite time interval $[0,\infty)$ and, hence, all of them correspond to some admissible pairs $\left\{(z_{z_0,\beta}(\cdot),y_{z_0,\beta}(\cdot))\right\}$. Since these trajectories are located above $\gamma_3$ they satisfy also the estimate~\eqref{g-transv-lamb}.

Consider the current value Hamiltonian $M(\cdot,\cdot)$ for $(z,\lambda)$ lying above $\gamma_3$ in $\hat\Gamma_{+,+}$ (see~\eqref{current-H}):
\begin{multline}
\label{Ham-M}
	M(z,\lambda)=\sup_{y\geq\del}\left\{y\lambda z+\ln y\right\}+(1-z)\lambda -\ln z\\
	= -1-\ln (-\lambda z)+(1-z)\lambda -\ln z,\qquad -\frac{1}{\del z}<\lambda<0.
\end{multline}
For any trajectory $(z_{z_0,\beta}(\cdot),\lambda_{z_0,\beta}(\cdot))$  of system~\eqref{ham-lamb} lying above $\gamma_3$ in $\hat\Gamma_{+,+}$ we have
\[
	z_{z_0,\beta}(t)\geq e^{(\del-1)t}z_0,\quad t\geq 0.
\]
On the other hand for any trajectory $(z_{z_0,\beta}(\cdot),\lambda_{z_0,\beta}(\cdot))$ of system~\eqref{ham-lamb}  lying between $\gamma_3$ and $\gamma_4$ in $\hat\Gamma_{+,+}$ we have
\[
	\frac{1}{2(1+\del)}< -\lambda_{z_0,\beta}(t)z_{z_0,\beta}(t)<\frac{1}{\del}\quad\mbox{if}\quad z_{z_0,\beta}(t)>1.
\]
These imply that for any trajectory $(z_{z_0,\beta}(\cdot),\lambda_{z_0,\beta}(\cdot))$  of  system~\eqref{ham-lamb} lying between $\gamma_3$ and $\gamma_4$ in $\hat\Gamma_{+,+}$ and for corresponding adjoint variable $\psi_{z_0,\beta}(\cdot)$, $\psi_{z_0,\beta}(t) =e^{-\del t}\lambda_{z_0,\beta}(t)$, $t\geq 0$, we have
\[
	\lim_{t\to\infty}H(t,z_{z_0,\beta}(t),\psi_{z_0,\beta}(t))=\lim_{t\to\infty}\left\{e^{-\del t}M(z_{z_0,\beta}(t),\lambda_{z_0,\beta}(t))\right\}=0.
\]
Hence, for any such trajectory $(z_{z_0,\beta}(\cdot),\lambda_{z_0,\beta}(\cdot))$  of  system~\eqref{ham-lamb} we have (see~\eqref{stat-lamb})
\[
	M(z_{z_0,\beta}(t),\lambda_{z_0,\beta}(t))=\del e^{\del t}\int_t^{\infty}e^{-\del s}g(z_{z_0,\beta}(t),\lambda_{z_0,\beta}(t))\, ds,\qquad t\geq 0.
\]
Let $y_{z_0,\beta}(\cdot)$ be the control corresponding to $z_{z_0,\beta}(\cdot)$, i.e. $y_{z_0,\beta}(t)=-1/(z_{z_0,\beta}(t)\lambda_{z_0,\beta}(t))$. Then taking in the last equality $t=0$ we get
\[
	J(z_{z_0,\beta}(\cdot),y_{z_0,\beta}(\cdot))=\int_0^{\infty}e^{-\del s}g(z_{z_0,\beta}(t),\lambda_{z_0,\beta}(t))\, ds=\frac{1}{\del}M(z_{z_0,\beta}(0),\lambda_{z_0,\beta}(0)).
\]
For any $t\geq 0$ function $M(z_*(t),\cdot)$ (see~\eqref{Ham-M}) increases on $\left\{\lambda\colon -1/(\del z_*(t))<\lambda< 0\right\}$. Hence, $M(z_*(t),\cdot)$ reaches its maximal value in $\lambda$ on the set $\left\{ \lambda\colon -1/(\del z)<\lambda\leq\lambda_*(t)\right\}$ at the point $\lambda_*(t)$ that correspond to the optimal path $z_*(\cdot)$. Thus, all trajectories \\$(z_{z_0,\beta}(\cdot),\lambda_{z_0,\beta}(\cdot))$  of  system~\eqref{ham-lamb} lying between $\gamma_3$ and $\gamma_4$ in $\hat\Gamma_{+,+}$ do not correspond to optimal admissible pairs in $(P3)$.

From this we can also conclude that all trajectories $(z(\cdot),\lambda(\cdot))$  of system~\eqref{ham-lamb} lying above $\gamma_4$ also do not correspond to optimal admissible pars in $(P3)$. Indeed, if such trajectory $(z(\cdot),\lambda(\cdot))$ corresponds to an optimal pair $(z(\cdot),y(\cdot))$  in $(P3)$ then it must satisfy to condition~\eqref{stat-lamb}. But in this case we have $\lambda(0)>\lambda_*(0)$ and
\[
	J(z(\cdot),y(\cdot)) =\frac{1}{\del}M(z_0,\lambda(0))= \frac{1}{\del}M(z_0,\lambda_*(0)) = J(z_*(\cdot),\lambda_*(\cdot)),
\]
that contradicts the fact that function  $M(z_0,\cdot)$ increases on $\left\{ \lambda\colon -1/(\del z)<\lambda< 0\right\}$. Thus, for any initial state $z_0$ there is a unique optimal pair $(z_*(\cdot),y_*(\cdot))$ in $(P3)$ in the case $\del \geq 1$. 

In the next section we discuss the issue of sustainability  of optimal paths for different values of the parameters  in the model.

\section{Discussion: A Generalized Notion of Sustainability}\label{sec:discus}

Following Solow~\cite{solow1956contribution} we assume that the knowledge stock $S(\cdot)$ grows exponentially, i.e. $S(t)=S_0e^{\upmu t}$, $t\geq 0$, where $\upmu\geq 0$ and $S_0>0$ are constants.

Similar to Valente~\cite{valente2005sustainable} we say that an admissible pair $(x(\cdot),y(\cdot))$ is {\it sustainable} in our model if the corresponding instantaneous utility function $t\mapsto\ln P(t)$, $t\geq 0$, non-decreases in the long run, i.e.
\[
	\lim_{T\to\infty} \inf_{t\geq T}\frac{d}{dt}\ln P(t) =\lim_{T\to\infty}\inf_{t\geq T}\frac{\dot{P}(t)}{P(t)}\geq 0.
\]
Substituting $P(t)=S(t)\left(y(t)x(t)\right)^\upbeta$, $S(t)=S_0e^{\upmu t}$, $t\geq 0$,  (see~\eqref{Y}) we get the following characterization of sustainability of an admissible pair $(x(\cdot),y(\cdot))$:
\begin{equation}
\label{Y-sust}
	\frac{\mu}{\upbeta}+\lim_{T\to\infty}\inf_{t\geq T}\left[\frac{\dot y(t)}{y(t)}+\frac{\dot x(t)}{x(t)}\right]\geq 0.
\end{equation}

We  call an admissible pair $(x(\cdot),y(\cdot))$ {\it strongly sustainable} if it is sustainable and, moreover, the resource stock $x(\cdot)$ is non-vanishing in the long run, i.e.
\begin{equation}
\label{strong_Y-sust}
	\lim_{T\to\infty}\inf_{t\geq T}x(t) = x_{\infty}>0.
\end{equation}

Consider case $(i)$ when $\del < 1$. In this case due to Theorem~\ref{thm-synth} (see Chapter \ref{chap:learning}) there is a unique optimal equilibrium pair in the problem (see~\eqref{eq_pt} and~\eqref{eq-u}): $\hat u(t)\equiv\hat y=(1+\del)/2$, $\hat x(t)\equiv \hat x=(1-\del)/2>0$, $t\geq 0$, and for any initial state $x_0$ the corresponding optimal path $x_*(\cdot)$ approaches asymptotically to the optimal equilibrium state  $\hat x$  while the corresponding optimal exploitation rate $y_*(\cdot)$ approaches asymptotically to the optimal equilibrium value $\hat y$. Hence, both conditions~\eqref{Y-sust} and~\eqref{strong_Y-sust} are satisfied. Thus the optimal admissible pair $(x_*(\cdot),y_*(\cdot))$ is strongly sustainable in this case.

Consider case $(ii)$ when $\del \geq 1$. In this case due to Theorem~\ref{thm-synth} (see Chapter \ref{chap:learning}) for any initial state $x_0$ the corresponding optimal control $y_*(\cdot)$ asymptotically satisfies the Hotelling rule of optimal depletion of an exhaustible resource~\cite{hotelling1931economics}, i.e. $\lim_{t\to\infty}y_*(t)=\del$, and $\lim_{t\to\infty}\dot y_*(t)/y_*(t)=0$. The corresponding optimal path $x_*(\cdot)$ is asymptotically vanishing, and
\[
	\lim_{t\to\infty}\dot x_*(t)/x_*(t)=\lim_{t\to\infty}\left( 1-y_*(t) -x_*(t)\right)=1-\del.
 \]
Hence, in the case $(ii)$ the sustainability condition~\eqref{Y-sust} takes the following form: \vspace{-2pt}
\begin{equation}
\label{Y-sust(ii)}
	\frac{\upmu}{\upbeta} \geq \del - 1.
\end{equation}
It is interesting to point our here that in the case $\upbeta=1$ condition~\eqref{Y-sust(ii)} coincides with Valente's necessary condition for sustainability in his capital-resource model with a renewable resource growing exponentially ~\cite{valente2005sustainable}. This takes place if the resource dynamics are considered in their original form as given by~\eqref{eq:ec_mod} (see the concluding section of~\cite{aseev2016optimal}) as opposed to the non-dimensionalized dynamics considered here.  

Since in the case $(i)$  the inequality~\eqref{Y-sust(ii)} holds automatically we conclude that~\eqref{Y-sust(ii)} is a necessary and sufficient condition (a criterion) for sustainability of the optimal pair $(x_*(\cdot),y_*(\cdot))$ in our model while the stronger inequality
\[
	\del < 1
\]
gives a criterion for its strong sustainability.
\vskip3mm

The criterion \eqref{Y-sust(ii)} gives the following guidelines for sustainable optimal growth: \emph{1)} Increase ratio of growth rate of knowledge stock $\upmu$ to output elasticity $\upbeta$; and \emph{2)} Decrease social discount $\del$ i.e., plan long term. The sustainability criterion \eqref{Y-sust(ii)} gives a relationship between the state of technology (depicted by $\upbeta$),  accumulation of knowledge (depicted by $\upmu$) and  foresight of the social planner (depicted by $\del$). According to the guideline \emph{1)} above, it is the \emph{ratio} between $\upmu$ and $\upbeta$ that matters and not the individual quantities.

In the next chapter, we return our focus to the two-agent symmetric semi-homogeneous society. Examining the dual-agent network in the context of optimal control would lead to the analytically more complex methods of multivariable optimal control \cite{ostertag2011mono}. More importantly, the framework would not allow us to study the interactions between the consumer groups. We instead examine the dual-agent network via game theory which permits a more accurate representation of decentralized control mechanisms, where every agent optimizes her own objective. Moreover, the theory allows us to explicitly study the strategic interactions between the agents, something that was not present in the single agent model considered in this chapter.

\clearpage
%% This is an example first chapter.  You should put chapter/appendix that you
%% write into a separate file, and add a line \include{yourfilename} to
%% main.tex, where `yourfilename.tex' is the name of the chapter/appendix file.
%% You can process specific files by typing their names in at the 
%% \files=
%% prompt when you run the file main.tex through LaTeX.

\chapter{Strategic Interactions in the Game-theoretic Framework}

\label{chap:game}

The tragedy of the commons can be seen as one particular instance of the more general problem of eliciting cooperation between individuals, when there exists a temptation for individuals to defect. Hardin \cite{hardin1968tragedy} was one of the first who proposed to employ centralized mechanisms to coerce a society in order to achieve collectively optimal outcomes in such situations. He further asserted that any decentralized mechanism, regulated from within the society, must necessarily fail in achieving those outcomes. As discussed in Chapter~\ref{chap:res_gov}, this assertion was challenged by Ostrom \cite{ostrom1990governing}, who reported several real-world examples of successful institutions, conceived and regulated by the communities themselves, and used a game theoretic framework to conceive a new theory of collective action. Typical game theoretic formalizations, which have been used to depict tragedies, include the prisoner's dilemma, the assurance game and the snowdrift game. These games have been studied in both discrete and continuous action spaces \cite{doebeli2005models}. 

In this chapter we return to the dual-agent network whose open-loop characteristics were studied in Section \ref{sec:open_dual}. We formulate the system as static-two player game with the pay-offs defined as the resource harvested at steady state and the strategies defined as the psychological attributes of the consumers. We then develop a notion of ``tragedy" as representing the difference between the optimal and rational outcomes of the consumption game. The tragicness is studied for both continuous and discrete strategy spaces and its dependence on various system parameters is also explored. In the process we highlight various characteristics of non-tragic societies and draw attention to trends that may overcome tragic situations.

\section{The Continuous Consumption Game}
\label{sec:ct_game}
Here we formalize the socio-ecological system in \eqref{eq:two_com_sys} as a static, one-shot two-player game, in which two consumer groups (players) have access to a single common good (natural resource). Each player $i \in \{1,2\}$ chooses her level of environmentalism $\uprho_i$ and her social relevance $\upnu_i$ (note that $\upalpha_i=1-\upnu_i$ and so $\upalpha_i$ is determined by the choice of $\upnu_i$). Thus the strategy set for each individual i is given as $\tilde{\mathcal{S}}_i=\{\uprho_i, \upnu_i\}$. The payoff $\tilde{\uppi}_i$ that each individual receives is equal to the amount of resource $\xbar \ybar_i$ that individual harvests at steady state, where ($\xbar$, $\ybar_1$, $\ybar_2$) are given by \eqref{eq:equi-2comp}. Note that $\b_i$ does not affect the equilibrium and thus is not included in the strategy set $\tilde{\mathcal{S}}_i$. The game is then defined as a 3-tuple $\tilde{\mathcal{G}}=\langle \mathcal{I},(\tilde{\mathcal{S}}_i ),(\tilde{\uppi}_i )\rangle$ where $\mathcal{I}=\{1,2\}$ denotes the set of players, $\tilde{\mathcal{S}}_i= [0,1]\times[0, 1]$ ; $i\in \mathcal{I}$ is the strategy space for $i$ and $\tilde{\uppi}_i:\tilde{\mathcal{S}}_1 \times \tilde{\mathcal{S}}_2 \rightarrow \mathbb{R}$; $i\in \mathcal{I}$ is the payoff function for consumer $i$. 

\subsection{Deriving the Nash Equilibrium}

For the two-player game $\tilde{\mathcal{G}}$ defined above, the pay-off functions are given as $\tilde{\uppi}_1 (\uprho_1,\upnu_1,\uprho_2,\uprho_2) = \xbar\ybar_1$ and $\tilde{\uppi}_2 (\uprho_1,\upnu_1,\uprho_2,\uprho_2) = \xbar\ybar_2$, where
\begin{align*}
\begin{split}
	&\xbar(\uprho_1,\upnu_1,\uprho_2,\uprho_2) =   \frac{(1-\upnu_1) \upnu_2}{\upnu_1 \!+\! \upnu_2 - 2 \upnu_1 \upnu_2} \uprho_1 + \frac{(1-\upnu_2) \upnu_1}{\upnu_1 + \upnu_2 - 2 \upnu_1 \upnu_2} \uprho_2, \\
	& \ybar_1(\uprho_1,\upnu_1,\uprho_2,\uprho_2) = \frac{1}{2} - \frac{1-\upnu_1}{2\left( \upnu_1 \!\!+\!\! \upnu_2 \!-\!\! 2 \upnu_1 \upnu_2 \right)} \uprho_1 \!+\! \frac{(1\!-\!\upnu_2) \left(1 \!- \!2\upnu_1 \right)}{2\left( \upnu_1 \!\!+\!\! \upnu_2 \!-\!\! 2 \upnu_1 \upnu_2 \right)} \uprho_2, \\
	& \ybar_2(\uprho_1,\upnu_1,\uprho_2,\uprho_2) = \frac{1}{2} + \frac{(1\!-\!\upnu_1) \left(1 \!- \!2\upnu_2 \right)}{2\left( \upnu_1 \!\!+\!\! \upnu_2 \!-\!\! 2 \upnu_1 \upnu_2 \right)} \uprho_1 \!-\! \frac{1-\upnu_2}{2\left( \upnu_1 \!\!+\!\! \upnu_2 \!-\!\! 2 \upnu_1 \upnu_2 \right)} \uprho_2. 
\end{split}
\end{align*}
where we have used the formula $\upalpha_i + \upnu_i = 1$ to replace $\upalpha_i$. The ``best response" of player $i$ is the strategy ($\tilde{\uprho}_i$,$\tilde{\upnu}_i$) that maximizes $\tilde{\uppi}_i$ for a fixed strategy of the other player $j\neq i$. The best response functions are thus given as
{\setlength{\mathindent}{0cm}
\begin{align}
\label{eq:br}
	(\tilde{\uprho}_i,\tilde{\upnu}_i) = \argmax_{(\uprho_i,\upnu_i)}\tilde{\uppi}_i (\uprho_i,\upnu_i,\uprho_j,\uprho_j).
\end{align}}
By calculating the partial derivatives and putting them to zero, we find that for a fixed pair ($\uprho_j,\upnu_j)$, $\tilde{\uppi}_i (\cdot,\cdot,\uprho_j,\upnu_j)$ is maximized not at a single point but along the following curve
\begin{align*}
	\tilde{\uprho}_i = \left( \frac{(1-\upnu_j)(\uprho_j-\upnu_j(1-\uprho_j))}{2\upnu_j} \right) \frac{1}{\tilde{\upnu}_i - 1} + \frac{\uprho_j + \upnu_j(1-\uprho_j)(2\upnu_j-1)}{2 \upnu_j}; \quad i,j \in \mathcal{I}; i \neq j
\end{align*}
A realization of this curve is shown in Figure \ref{fig:nash}.
\begin{figure}[t!]
	\captionsetup{width=0.8\textwidth}
	\begin{center}
		\includegraphics[width=\linewidth]{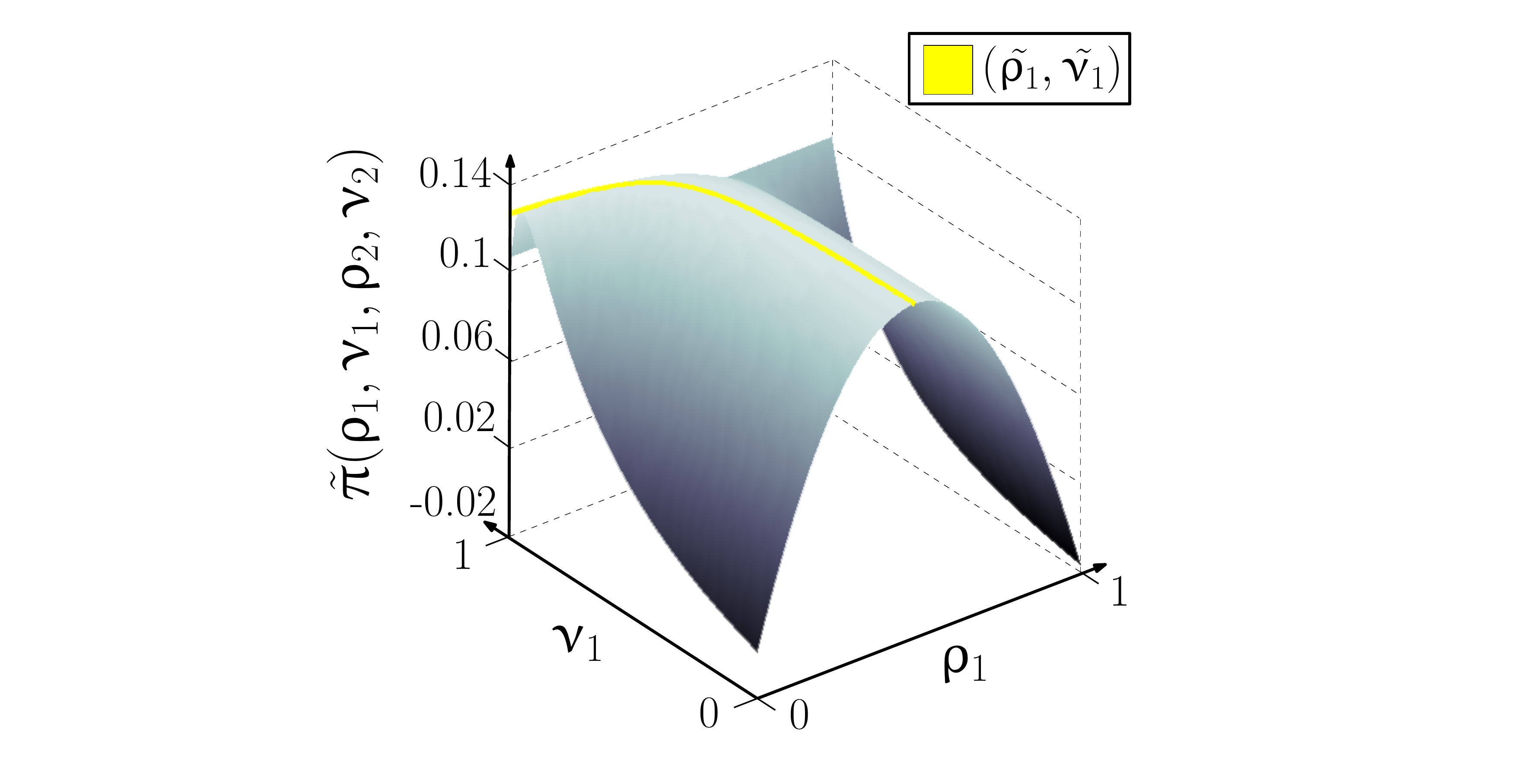}
	\end{center}
	\caption{An example of the payoff function of one player, given the other player's strategy. The yellow curve represents the best response. Here $\uprho_2 = 0.7$ and $\upnu_2 = 0.9$.}
	\label{fig:nash}
\end{figure}
The Nash equilibrium $(\uprho_1^\#,\upnu_1^\#,\uprho_2^\#,\upnu_2^\#)$, consists of all possible outcomes such that the strategy of each player is a best response to the other player's strategy. Thus at the Nash equilibrium \\$(\uprho_i^\#, \upnu_i^\#) = \argmax_{(\uprho_i,\upnu_i)}\tilde{\uppi}_i (\uprho_i,\upnu_i,\uprho_j^\#,\uprho_j^\#)$ and $(\uprho_j^\#, \upnu_j^\#) = \argmax_{(\uprho_j,\upnu_j)}\tilde{\uppi}_j (\uprho_i^\#,\upnu_i^\#,\uprho_j,\uprho_j)$. Solving this condition algebraically gives the following description of the Nash equilibria
\begin{align}
\label{eq:nash1}
	\uprho_i^\# = \frac{\upnu_i ^\#(3 \upnu_j^\# - \upnu_i^\# - 2 \upnu_i^\# \upnu_j^\# )}{(1-\upnu_i^\# )(\upnu_i^\#+\upnu_j^\#+2\upnu_i^\# \upnu_j^\# )} \quad   i,j =1,2, i \neq j.
\end{align}
which shows that there exist an infinite number of Nash equilibria in this formulation of the consumption game. For the purpose of the following analysis, we redefine the strategy set $\tilde{\mathcal{S}}$ to include only the environmentalisms ($\uprho_1$, $\uprho_2$) and treat ($\upnu_1$, $\upnu_2$) as given parameters instead of decision variables.This is because excluding the $\upnu_i$'s from the strategy set results in exactly the same Nash equilibrium as given above; so treating the $\upnu_i$ as exogenous parameters simplifies the further analysis and does not qualitatively affect the resulting findings.

Hence we redefine the consumption game as the 3-tuple $\mathcal{G} = \langle \mathcal{I}, \mathcal{S}_i, \uppi_i \rangle$, where $I = \{1,2\}$ denotes the set of players, $\mathcal{S}_i$, $i \in \mathcal{I}$ is the indexed strategy space for group $i$, and $\uppi_i: \mathcal{S}_1 \times \mathcal{S}_2 \rightarrow \mathbb{R}$, $i \in \mathcal{I}$ is the indexed payoff function for $i$, defined on $\mathcal{S}_i$. Each group is free to choose its level of environmentalism $\uprho_i \in \mathbb{R}$. The strategy set for each group $i$ is given as $\mathcal{S}_i=\{\uprho_i\}$. The payoff $\uppi_i(\uprho_i,\uprho_j)$ that each group receives is equal to the amount of resource $\xbar\ybar_i(\uprho_i,\uprho_j)$ the group $i$ harvests at steady state.

\begin{figure}[t]
	\captionsetup{width=0.8\textwidth}
	\begin{center}
		\includegraphics[width=\linewidth]{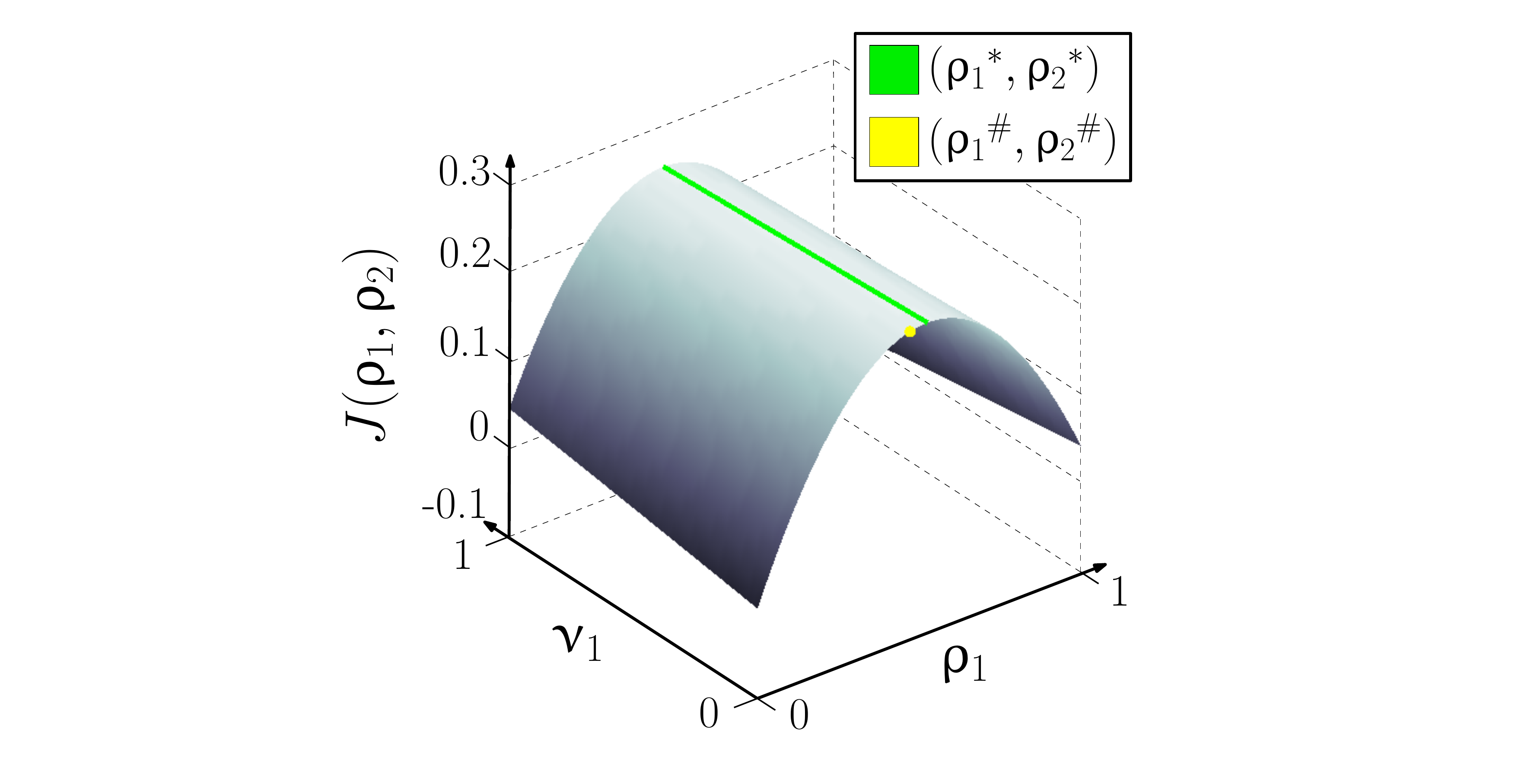}
	\end{center}
	\caption{An example of the social welfare function for $\upnu_1 = 0.3$ and $\upnu_2 = 0.9$. The green curve represents the collective optimal whereas the yellow point represents the Nash equilibrium.}
	\label{fig:optimal}
\end{figure}

\subsection{Collective Optimality and the Notion of Tragicness}

A property inherent to any representation of the tragedy of the commons is that individually rational behavior leads to outcomes that are collectively suboptimal. We thus compare the environmentalism of each group at the Nash equilibrium with one that maximizes the utilitarian welfare function of total consumption $J(\uprho_1,\uprho_2) = \uppi_1(\uprho_1,\uprho_2) + \uppi_2(\uprho_1,\uprho_2)$. We find that the latter optimum $(\uprho_1^*, \uprho_2^*)$ is not unique, and is given by the following curve
\begin{align} 
 2\upnu_2(1-\upnu_1)\uprho_1^* + 2\upnu_1(1-\upnu_2)\uprho_2^* - \upnu_1(1-\upnu_2) - \upnu_2(1-\upnu_1) = 0,
\end{align}
A single realization of this curve is shown in Figure \ref{fig:optimal}. with the Nash equilibrium situated at some distance from it (naturally determined by parameters $\upnu_1$ and $\upnu_2$). We call this distance the tragicness of the consumption game, which, more specifically, is defined here as the length of the shortest line joining the Nash equilibrium and a point on the optimal curve (see Figure \ref{fig:tragic}). Thus in a non-tragic game, the Nash equilibrium would lie exactly on the welfare optimal curve (resulting in tragicness equal to zero). All other games in which the Nash equilibrium does not lie on this curve are classified as tragic games, where the magnitude of the distance between the Nash equilibrium and the Pareto-optimal solutions is quantified through their tragicness. The concept of tragicness, which we introduce here, resembles the ``price of anarchy" \cite{koutsoupias2009worst}, which is a ratio between the cost of the worst possible Nash equilibrium and the optimum of a social welfare function as measure of effectiveness of the system. Another similar concept is the ``price of stability" \cite{anshelevich2008price}, which is a ratio between the cost of the best possible Nash equilibrium and the optimal solution. Both concepts have been applied widely in computer network design.
\begin{figure}[t!]
	\captionsetup{width=0.8\textwidth}
	\begin{center}
		\includegraphics[width=0.5\linewidth]{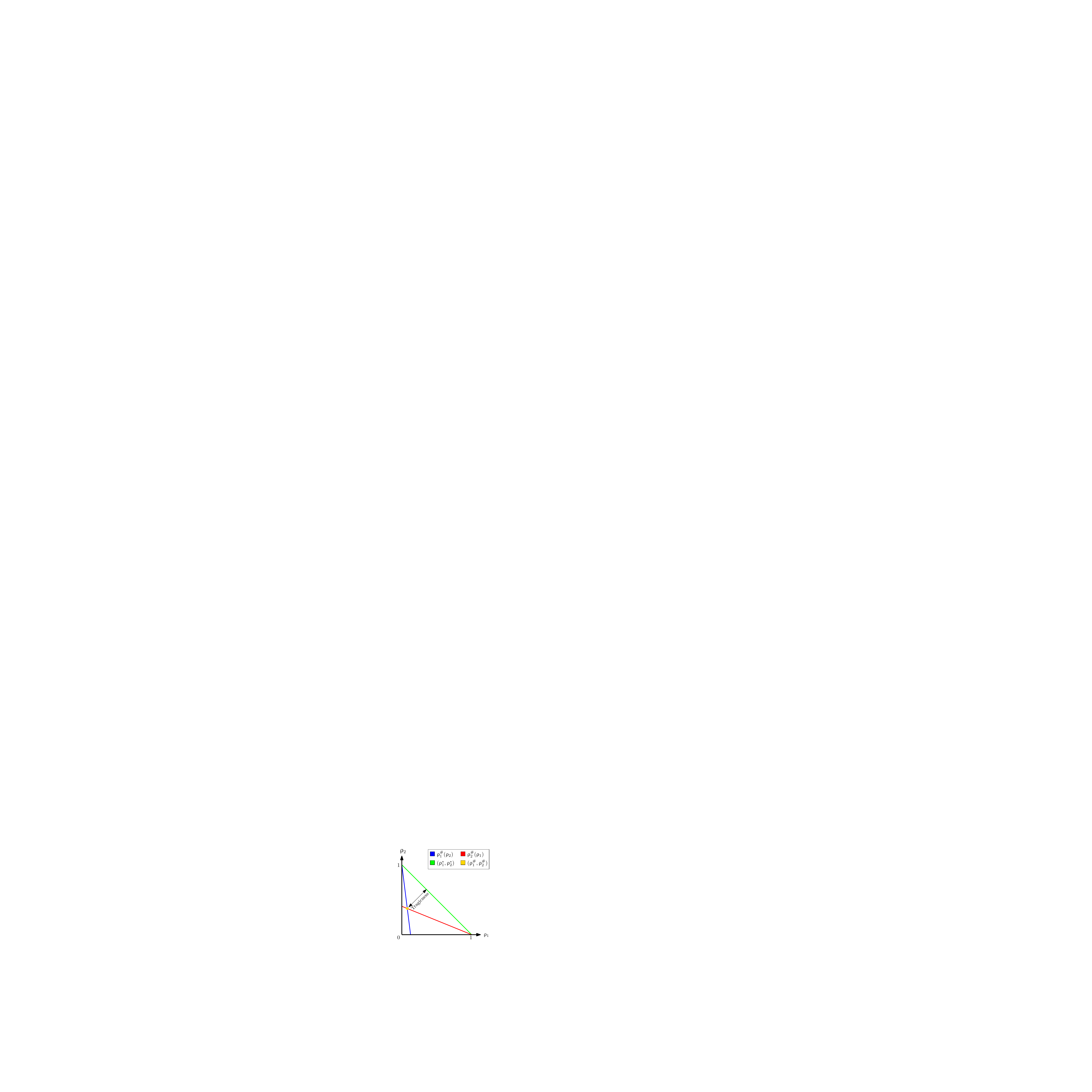}
	\end{center}
	\caption{An illustration of the definition of the game tragicness.}
	\label{fig:tragic}
\end{figure}

\subsection{The Consumption Game and Social-Ecological Relevance}
\label{sec:game_results}

Here we present insights gained by observing how the characteristics of the consumption game depend on the model parameters and what it entails. 

\subsubsection{The anti-tragic role of social relevance}
Cooperative behavior plays a key role in avoiding tragedies and increasing the steady state resource stock. In our model, this effect can be observed through Figure \ref{fig:game}(a)-(c). Figure \ref{fig:game}(a) shows that the tragicness of the consumption game decreases as we move away from the origin along any of the $\upnu_1$ or $\upnu_2$ axes. Furthermore, the tragicness declines with an increase in the average (over two groups) level of social relevance. Declining of the tragicness can indeed be seen as growing of coordination and cooperation (through increasing social relevance) between two groups. Games that are less tragic correspond to affluence in the steady state resource stock, which in-turn also allows for a higher resource consumption at steady state. Thus the same trends can be observed for the steady state resource stock and total consumption rate in Figure \ref{fig:game}(b) and (c) respectively. This shows that high social relevance not only decreases the game's tragicness, but also favors higher steady state resource stocks and consumption rates.
 
In the limit case of $\upnu_1=\upnu_2=1$, zero tragicness is achieved over the infinitely many equilibria, and it is interesting to note that in this most affluent case, the maximum resource stock $\xbar=0.5$ and maximum consumption rate $\xbar(\ybar_1+\ybar_2)=0.25$.
\begin{figure}[b!]
	\captionsetup{width=0.9\textwidth}
	\begin{center}
		\includegraphics[width=0.9\linewidth]{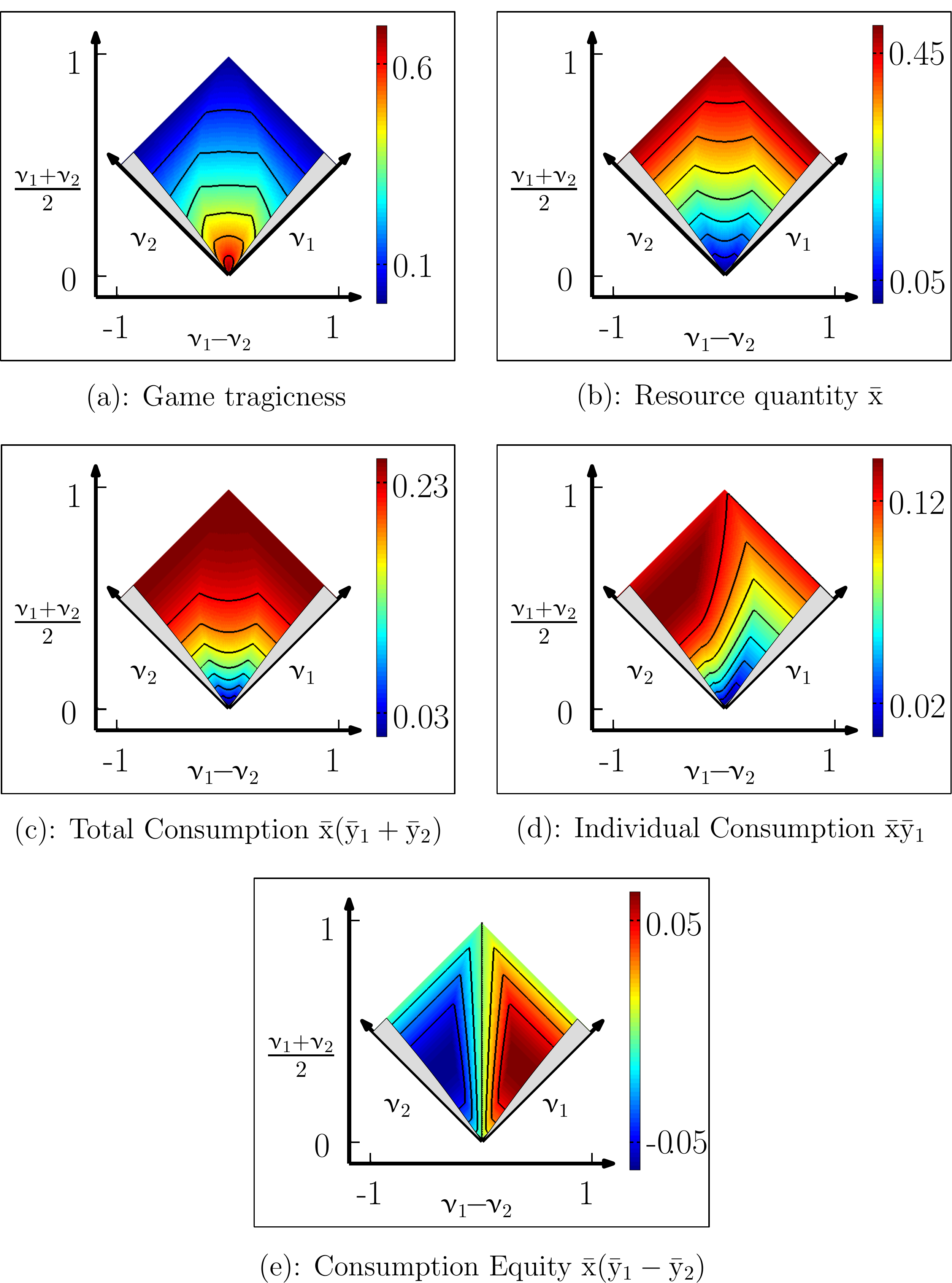}
	\end{center}
	\caption{Plots for different characteristics of the consumption game at Nash equilibrium, shown as the difference and average level of the social relevances $\upnu_1$, $\upnu_2$ are varied. For completeness the $\upnu_1$ and $\upnu_2$ axes have also been shown in each plot. The gray areas mark regions for which stability is not guaranteed. The plots for both groups are symmetric to each other, and so the conclusions drawn here for group 1 also apply to group 2.}
	\label{fig:game}
\end{figure}

\subsubsection{A positive effect of incongruity in social relevance}
The effect of incongruity or asymmetry in the social relevances on the game tragicness can also be observed in Figure \ref{fig:game}(a). Moving away from the line of zero asymmetry (the horizontal line corresponding to $\upnu_1 = \upnu_2$) horizontally in either direction results in a lower tragicness level. This effect of asymmetry on the tragicness is less pronounced for small deviations in asymmetry, which can be seen from the decreased slope of the level curves when they are near the line of zero asymmetry. When the asymmetry is large, this effect is more significant, which can be seen from the sudden slanting of the level curves when further away from zero asymmetry. Thus strong incongruity or heterogeneity in social relevance corresponds to less tragic games. 

Incongruity has a similar effect on the resource stock and total consumption rates, which can be seen in Figure \ref{fig:game}(b) and Figure \ref{fig:game}(c) respectively. In general, an increase in incongruity increases both. Note that the discontinuity of the level curves, which was present in Figure \ref{fig:game}(a) is also present here at exactly the same points in the parameter space. However the behavior of the curves is different in the center. They are now concave, which means that small deviations in asymmetry are detrimental to the resource stock and the consumed quantity. However, large deviations in asymmetry are beneficial. 

\subsubsection{Higher social relevance increases individual consumption}
In general, higher social relevance increases individual group consumption. Figure \ref{fig:game}(d) depicts how individual consumption of a group changes as the social relevances are varied. The level curves in Figure \ref{fig:game}(d) show three different behaviors. Near the symmetric line where $\upnu_1=\upnu_2$, the consumption of the focal group increases with respect to the social relevance of both groups. However, this increase is more pronounced with respect to changes in its own social relevance, than that of the other group. Variation in the other group's social relevance does not affect the group's consumption, when the social relevance of the other group is low. Conversely, for low values of its own social relevance, the group's consumption is affected only by changes in the social relevance of the other group.
 
Also note that a group's individual consumption increases monotonically with the other group's social relevance only when its own social relevance is high. When this is not the case, a mode is encountered in the region where the other group's social relevance is high enough (note the plateau in the upper right area of Figure \ref{fig:game}(d)).
 
Furthermore, a group with a higher social relevance consumes more, than a group with a lower social relevance. Figure \ref{fig:game}(e) depicts the difference in consumption $\xbar(\ybar_1-\ybar_2)$ which always stays positive to the right of the symmetric line ($\upnu_1=\upnu_2$) and always negative to the left of this line. Difference in consumption is exactly zero when social relevance for both consumer groups are the same. 

\subsubsection{Equity in social relevance induces equity in consumption}
Figure \ref{fig:game}(e) shows the behavior of the equity measure $\xbar(\ybar_1-\ybar_2)$ as a function of the average of the social relevances and their asymmetry. When the social relevances are close to equal (near the center of the plot), the isoclines are nearly vertical which shows that the average quantity of the social relevances does not affect equity in consumption. However, the equity is affected significantly by the asymmetry in social relevances. This can be traced by starting in the middle of the plot and observing the difference in consumption as we move horizontally in either direction. The difference in consumption is increased as we increase asymmetry until a peak is reached after which it decreases slightly. This shows that equity in consumption is the highest, when the social relevances are close to each other in value.

\section{The Discrete Consumption Game}
\label{sec:dt_game}

In Section \ref{sec:ct_game}, we formulated the interactions between two consumer groups as a continuous game. The utilities were assumed to be the amount of resource each group harvested at steady state. The variables $\upnu_i$ and $\uprho_i$ appeared to be natural choices for the action variables available to each group. However, we found that the game had an infinite number of rational outcomes, expressed in terms of the action variables in Equation \eqref{eq:nash1}. Consequently, we redefined the strategy set to include only one of the variables ($\uprho_i$) and the other ($\upnu_i$) to be retained as an exogenously specified parameter. This led to a single equilibrium specified in terms of the $\upnu_i$'s (thus essentially infinite in number if $\upnu_i$'s are allowed to vary as before) which resulted in a less tedious representation of the consumption game. 

Here we consider an alternative formulation of the assumed socio-ecological setting as a two player game with discrete actions spaces. As we will see, the formulation allows us to retain both psychological variables in the strategy set. After defining a parallel concept to tragicness in the continuous case, we then examine the dependence of the nature of the game on different variables.
 
We assume that the strategy of player $i$, $\{\uprho_i, \upnu_i\} \in \{ \{\uprho_H, \upnu_H\}, \{\uprho_L, \upnu_L\} \}$ where $\uprho_L < \uprho_H$ and $\upnu_L < \upnu_H$. Hence, strategies are a set of two variables where the cooperate strategy is represented by $\{ \uprho_H, \upnu_H \}$ and the defect strategy is represented by $\{ \uprho_L, \upnu_L \}$. For a two player game, the payoff matrix is as given in Table \ref{tab:pay_disc}.
%\vspace{-10pt}
\begin{table}[!h]
	\centering
	\begin{tabular}{ >{\centering}m{10pt} | >{\centering}m{100pt} | >{\centering}m{100pt} | m{0pt}}
		\multicolumn{1}{c}{ }& \multicolumn{1}{c}{ C} & \multicolumn{1}{c}{D}\\
		\cline{2-3}
		$C$ & $\xbar\ybar_1(\uprho_H,\upnu_H,\uprho_H,\upnu_H)$,\\ 
		\hspace{-4pt}$\xbar\ybar_2(\uprho_H,\upnu_H,\uprho_H,\upnu_H)$ 
		& $\xbar\ybar_1(\uprho_H,\upnu_H,\uprho_L,\upnu_L)$,\\ 
		\hspace{-4pt}$\xbar\ybar_2(\uprho_H,\upnu_H,\uprho_L,\upnu_L)$ &\\[25pt]
		\cline{2-3}
		$D$ & $\xbar\ybar_1(\uprho_L,\upnu_L,\uprho_H,\upnu_H)$,\\ 
		\hspace{-4pt}$\xbar\ybar_2(\uprho_L,\upnu_L,\uprho_H,\upnu_H)$ 
		& $\xbar\ybar_1(\uprho_L,\upnu_L,\uprho_L,\upnu_L)$,\\ 
		\hspace{-4pt}$\xbar\ybar_2(\uprho_L,\upnu_L,\uprho_L,\upnu_L)$ &\\[25pt]
		\cline{2-3}
	\end{tabular}
	\captionsetup{font=normal,width=0.8\textwidth}
	\caption{Structure of the discrete consumption game. $\xbar$, $\ybar_1$ and $\ybar_2$ are given by Equation \ref{eq:equi-2comp}.}
	\label{tab:pay_disc} %payoff discrete
\end{table}
%\vspace{-10pt}

Note that the nature of the game is dependent only on the four parameters $\uprho_H, \uprho_L, \upnu_H, \upnu_L$. By varying these parameters over their entire range along with the constraints $\uprho_L < \uprho_H$ and $\upnu_L < \upnu_H$, we encounter a variety of games which we classify as either tragic or non-tragic.

\subsection{Tragic and non-tragic games}
We define a tragic game as one where at least one nash equilibrium is sub-pareto optimal (see \cite{perman2003natural} for the notion of pareto optimality). Three such games are encountered while varying the parameters and are shown in Figure \ref{tab:tragic}. All of these games are commonly found in literature as models of cooperation~\cite{doebeli2005models}.

\begin{figure}[!htb]
\begin{center}
	\begin{minipage}[t]{.25\linewidth}
		\caption*{\large{Type 1}}
		\centering
		\begin{tabular}[t]{ >{\centering}m{10pt} | >{\centering}m{15pt} | >{\centering}m{15pt} | m{0pt}}
			\multicolumn{1}{c}{ }& \multicolumn{1}{c}{ C} & \multicolumn{1}{c}{D}\\
			\cline{2-3}
			$C$ & 3,3 & 1,4 &\\[15pt]
			\cline{2-3}
			$D$ & 4,1 & \cellcolor{gray}2,2 &\\[15pt]
			\cline{2-3}
		\end{tabular}\\[10pt]
		$-$ prisoner's dilemma
	\end{minipage}
	\begin{minipage}[t]{.25\linewidth}
		\caption*{\large{Type 2}}
		\centering{
		\begin{tabular}[t]{ >{\centering}m{10pt} | >{\centering}m{15pt} | >{\centering}m{15pt} | m{0pt}}
			\multicolumn{1}{c}{ }& \multicolumn{1}{c}{ C} & \multicolumn{1}{c}{D}\\
			\cline{2-3}
			$C$ & \cellcolor{gray} 4,4 & 1,3 &\\[15pt]
			\cline{2-3}
			$D$ & 3,1 & \cellcolor{gray}2,2 &\\[15pt]
			\cline{2-3}
		\end{tabular}\\[10pt]}
		\begin{tabular}{l}
		$-$ stag-hunt\\
		$-$ trust dilemma\\
		$-$ assurance dilemma
		\end{tabular}
	\end{minipage}
	\begin{minipage}[t]{.25\linewidth}
		\caption*{\large{Type 3}}
		\centering
		\begin{tabular}[t]{ >{\centering}m{10pt} | >{\centering}m{15pt} | >{\centering}m{15pt} | m{0pt}}
			\multicolumn{1}{c}{ }& \multicolumn{1}{c}{ C} & \multicolumn{1}{c}{D}\\
			\cline{2-3}
			$C$ & 3,3 & \cellcolor{gray} 2,4 &\\[15pt]
			\cline{2-3}
			$D$ & \cellcolor{gray}4,2 & 1,1 &\\[15pt]
			\cline{2-3}
		\end{tabular}\\[10pt]
		\begin{tabular}{l}
		$-$ chicken\\
		$-$ hawk-dove\\
		$-$ snowdrift
		\end{tabular}
	\end{minipage}
\captionsetup{font=normal,width=0.8\textwidth}
\caption{The three tragic realizations of the consumption game along with the labels commonly associated with them. Nash equilibria are shaded in gray. The values of the payoffs are hypothetical and have been chosen to illustrate the relative magnitudes of the payoffs for each outcome.}
\label{tab:tragic}
\end{center}
\end{figure}

We define non-tragic games as games that are not tragic. In these games all nash equilibria are also pareto optimal. All such games encountered in the consumption game are given in Figure \ref{tab:non-tragic}.
\begin{figure}[!htb]
\begin{center}
	\begin{minipage}[t]{.25\linewidth}
		\caption*{\large{Type 4}}
		\centering
		\begin{tabular}{ >{\centering}m{10pt} | >{\centering}m{15pt} | >{\centering}m{15pt} | m{0pt}}
			\multicolumn{1}{c}{ }& \multicolumn{1}{c}{ C} & \multicolumn{1}{c}{D}\\
			\cline{2-3}
			$C$ & 2,2 & 1,4 &\\[15pt]
			\cline{2-3}
			$D$ & 4,1 & \cellcolor{gray}3,3 &\\[15pt]
			\cline{2-3}
		\end{tabular}
	\end{minipage}
	\begin{minipage}[t]{.25\linewidth}
		\caption*{\large{Type 5}}
		\centering
		\begin{tabular}{ >{\centering}m{10pt} | >{\centering}m{15pt} | >{\centering}m{15pt} | m{0pt}}
			\multicolumn{1}{c}{ }& \multicolumn{1}{c}{ C} & \multicolumn{1}{c}{D}\\
			\cline{2-3}
			$C$ & 1,1 & 2,4 &\\[15pt]
			\cline{2-3}
			$D$ & 4,2 & \cellcolor{gray}3,3 &\\[15pt]
			\cline{2-3}
		\end{tabular}
	\end{minipage}
	\begin{minipage}[t]{.25\linewidth}
		\caption*{\large{Type 6 }}
		\centering
		\begin{tabular}{ >{\centering}m{10pt} | >{\centering}m{15pt} | >{\centering}m{15pt} | m{0pt}}
			\multicolumn{1}{c}{ }& \multicolumn{1}{c}{ C} & \multicolumn{1}{c}{D}\\
			\cline{2-3}
			$C$ &2,2 & \cellcolor{gray}3,4 &\\[15pt]
			\cline{2-3}
			$D$ & \cellcolor{gray}4,3 & 1,1&\\[15pt]
			\cline{2-3}
		\end{tabular}
	\end{minipage}\vspace{20pt}

	\begin{minipage}[t]{.25\linewidth}
		\caption*{\large{Type 7 }}
		\centering
		\begin{tabular}{ >{\centering}m{10pt} | >{\centering}m{15pt} | >{\centering}m{15pt} | m{0pt}}
			\multicolumn{1}{c}{ }& \multicolumn{1}{c}{ C} & \multicolumn{1}{c}{D}\\
			\cline{2-3}
			$C$ & \cellcolor{gray}4,4 & 2,3 &\\[15pt]
			\cline{2-3}
			$D$ & 3,2 &1,1 &\\[15pt]
			\cline{2-3}
		\end{tabular}
	\end{minipage}
	\begin{minipage}[t]{.25\linewidth}
		\caption*{\large{Type 8 }}
		\centering
		\begin{tabular}{ >{\centering}m{10pt} | >{\centering}m{15pt} | >{\centering}m{15pt} | m{0pt}}
			\multicolumn{1}{c}{ }& \multicolumn{1}{c}{ C} & \multicolumn{1}{c}{D}\\
			\cline{2-3}
			$C$ & 1,1 & \cellcolor{gray} 3,4 &\\[15pt]
			\cline{2-3}
			$D$ & \cellcolor{gray}4,3 & 2,2 &\\[15pt]
			\cline{2-3}
		\end{tabular}
	\end{minipage}
	\begin{minipage}[t]{.25\linewidth}
		\caption*{\large{Type 9}}
		\centering
		\begin{tabular}{ >{\centering}m{10pt} | >{\centering}m{15pt} | >{\centering}m{15pt} | m{0pt}}
			\multicolumn{1}{c}{ }& \multicolumn{1}{c}{ C} & \multicolumn{1}{c}{D}\\
			\cline{2-3}
			$C$ &1,1 & 2,3 &\\[15pt]
			\cline{2-3}
			$D$ & 3,2 & \cellcolor{gray}4,4&\\[15pt]
			\cline{2-3}
		\end{tabular}
	\end{minipage}
\captionsetup{font=normal,width=0.6\textwidth}
\caption{Non-tragic realizations of the consumption game.}
\label{tab:non-tragic}
\end{center}
\end{figure}

\subsection{The Influence of Social Attributes on the Game Nature}
Figure \ref{fig:disc_type} shows the type of the consumption game that arises as the four parameters are varied. The plot covers only those regions where $\uprho_L < \uprho_H$, $\upnu_L < \upnu_H$ and parameter values for which stability of the equilibrium is guaranteed.
\begin{figure}[h!]
	\captionsetup{font=normal,width=0.8\textwidth}
	\begin{center}
		\includegraphics[trim = 0mm 0mm 0mm 13mm, clip=true,  width=0.85\linewidth]{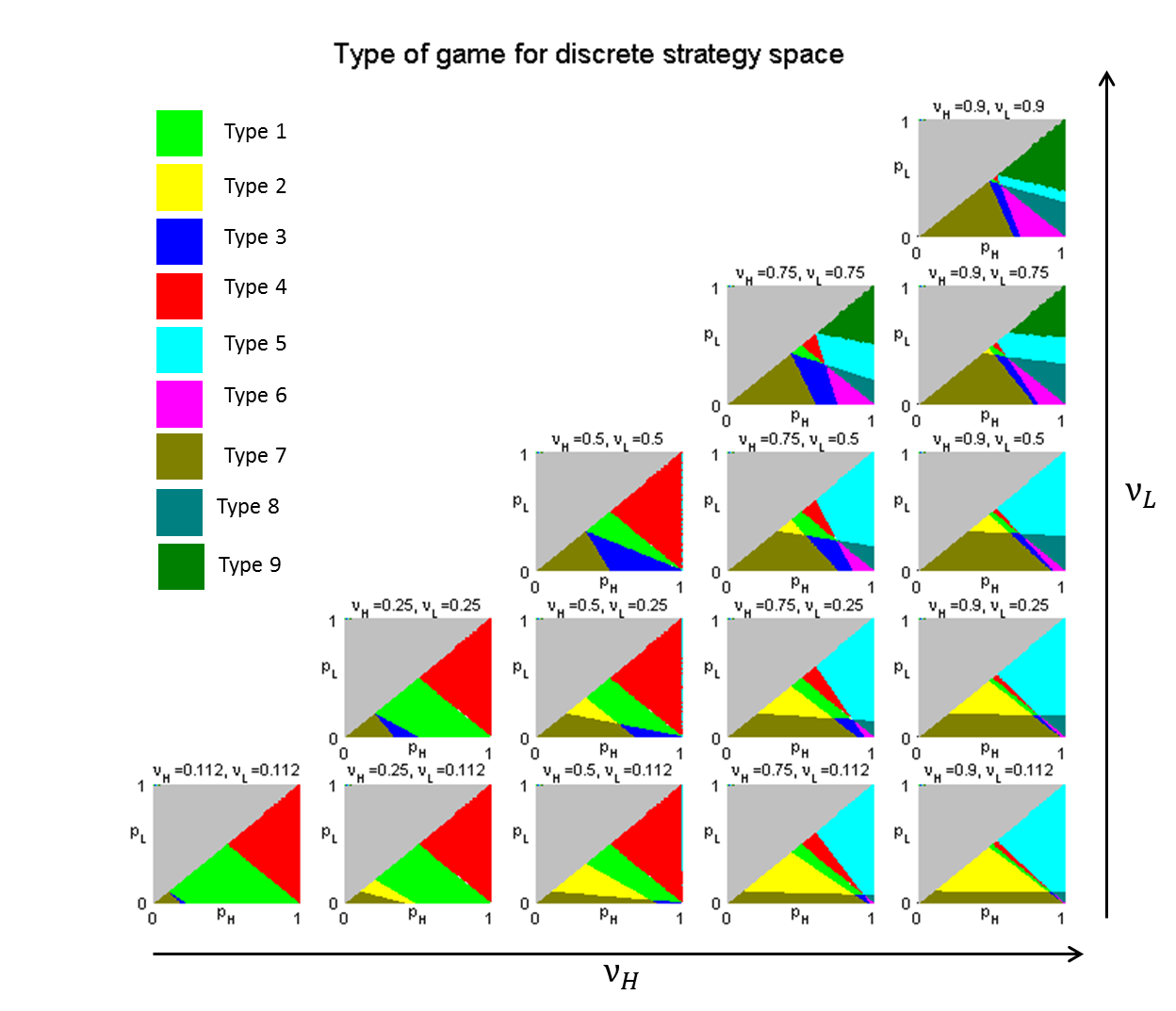}
	\end{center}
	\vspace{-25pt}
	\caption{Plot depicting the type of game that each realization of the consumption game belongs to as the parameters are varied.}
	\label{fig:disc_type} %discrete tragic
\end{figure}

Figure \ref{fig:disc_trag} shows the nature of the consumption game evaluated for the four dimensional parameter space and classified as either tragic (red) or non-tragic (green). 
\begin{figure}[h!]
	\captionsetup{font=normal,width=0.8\textwidth}
	\begin{center}
		\includegraphics[trim = 0mm 0mm 0mm 13mm, clip=true,  width=0.85\linewidth]{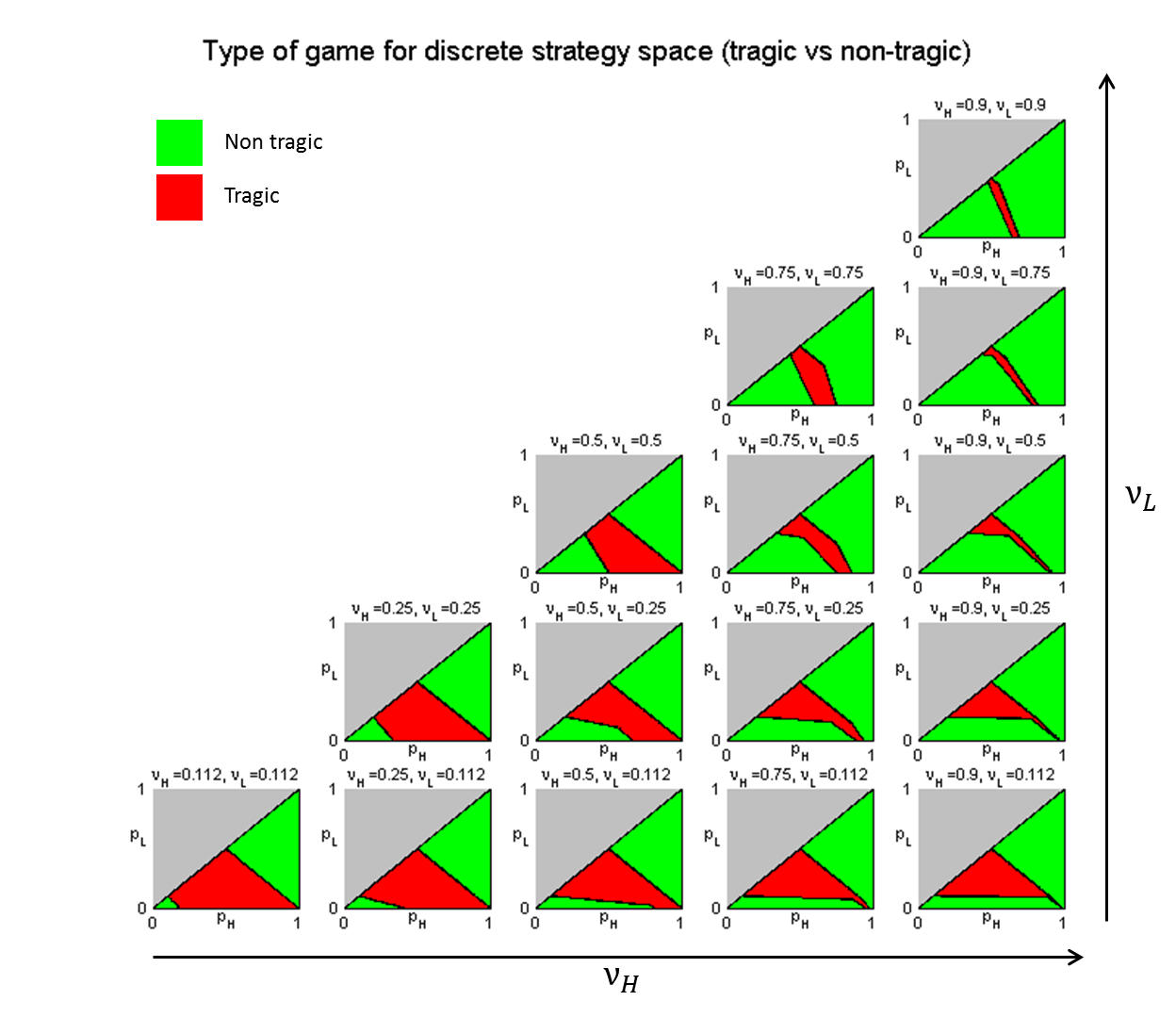}
	\end{center}
	\vspace{-25pt}
	\caption{Tragic and non-tragic regions for the consumption game.}
	\label{fig:disc_trag} %discrete tragic
\end{figure}

Another classification of the games exists depending on whether the rational outcome for both players to defect, both players to cooperate or one player to defect and one to cooperate. The distinction between these 3 classes of games is notable in Figure \ref{fig:disc_res} which shows the quantity of the resource at the Nash equilibrium.
\begin{figure}[h!]
	\captionsetup{font=normal,width=0.8\textwidth}
	\begin{center}
		\includegraphics[trim = 0mm 0mm 0mm 13mm, clip=true,  width=0.85\linewidth]{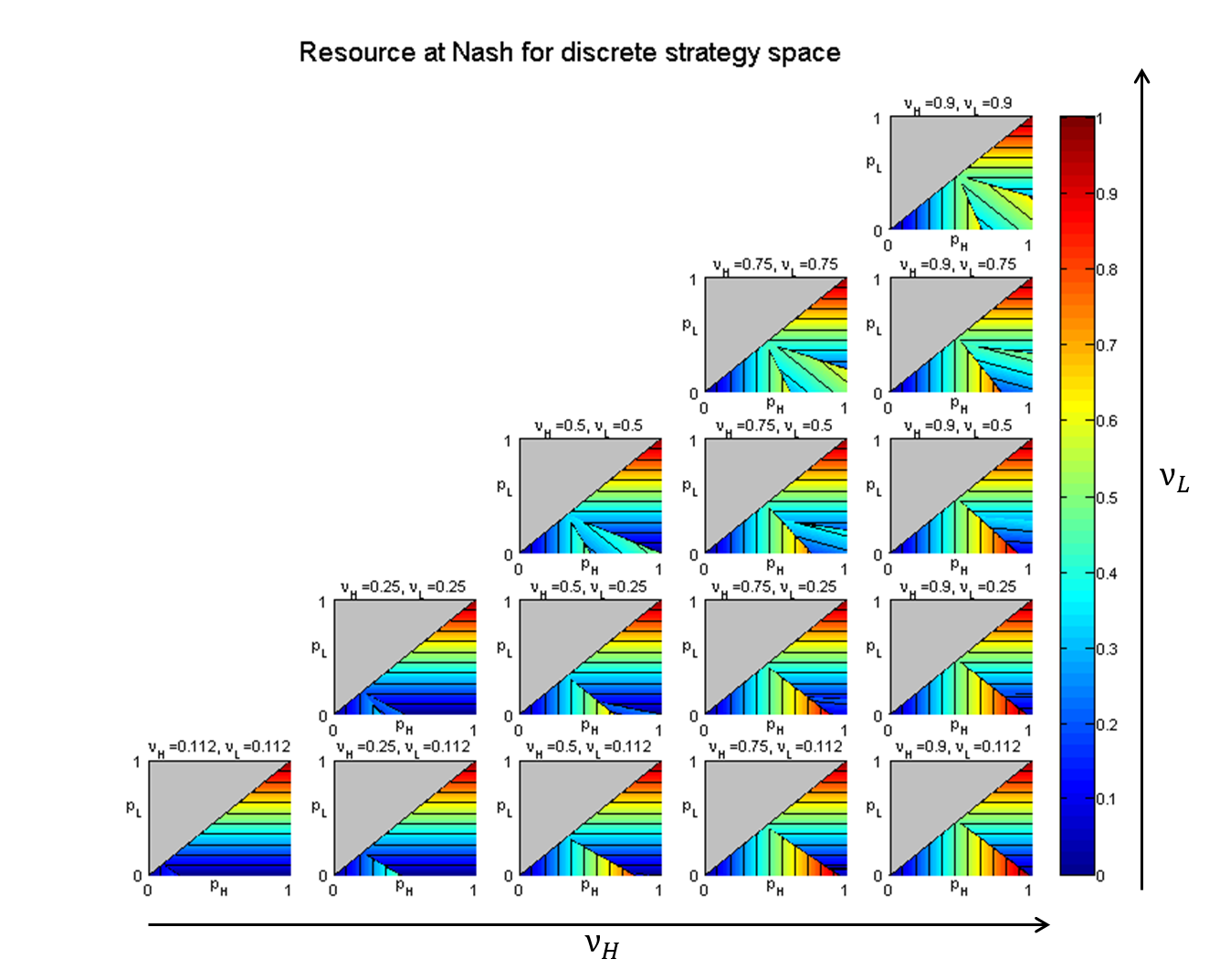}
	\end{center}
	\vspace{-25pt}
	\caption{Resource quantity at the Nash equilibrium for the discrete consumption game.}
	\label{fig:disc_res} %discrete tragic
\end{figure}

By observing the trends in the simulations, the following interpretations can be drawn. Note that in what follows, we label consumers with high values of $\uprho_i$ as environmentalists and consumers with high values of $\upnu_i$ as pro-socialists.

\subsubsection{High degrees of environmentalism can avoid tragedies}
From Figure \ref{fig:disc_trag} it can be seen that if the lower and upper levels of environmentalism ($\uprho_L$ and $\uprho_H$ respectively) are greater than the maximum quantity of the resource ($\uprho_L + \uprho_H > 1$) then tragedies are avoided irrespective of the level of socialism. From Figure \ref{fig:disc_type} it can be seen that in this case the rational outcome consists of at least one player defecting.

\subsubsection{High social ties compensate for low environmentalism in avoiding tragedies}
From Figure \ref{fig:disc_trag} we see that if the upper and lower thresholds of $\upnu_i$ are high, then tragedies are avoided even in the case of low environmentalism. From Figure \ref{fig:disc_type}, this corresponds to the Type 7 game where the rational outcome is for both individuals to cooperate. 

\subsubsection{Low environmentalism and strong social ties promote cooperation}
In all plots in Figure \ref{fig:disc_type}, it can be seen that tragedies are avoided in the lower left region of the axes where both the upper and lower thresholds of $\uprho_i$ are low (areas with low environmentalism). This region grows as $\upnu_L$ and $\upnu_H$ are increased (areas with high social weights). From Figure \ref{fig:disc_type} we see that this corresponds to the Type 7 game which is the only game with the unique rational outcome of both players cooperating. Thus these conditions promote cooperation.

\subsubsection{Small values of the lower environmental threshold promotes cooperation in the presence of strong asymmetry in the social threshold levels.}
This again refers to the Type 7 game which corresponds to bottom green region in the plots of Figure \ref{fig:disc_trag}. We observe that if there is a strong asymmetry between $\upnu_L$ and $\upnu_H$ (bottom right plot) then the Type 7 game is realized if the lower threshold $\uprho_L$ is kept small enough. Note that even though this relates to a condition on $\uprho_L$, the rational outcome corresponds to both players assuming the upper environmental threshold $\uprho_H$.  

\subsubsection{High environmentalism and strong social ties are beneficial to the resource}
This interpretation follows by an examination of Figure \ref{fig:disc_res}. Based on the trends of the contours in each plot, three distinct regions are observable. The lower left region corresponds to those games in which the rational outcome is for both players to cooperate (see Figure \ref{fig:disc_type}), or to choose the upper thresholds $\uprho_H$ and $\upnu_H$. Here the resource level remains unchanged by variation in $\uprho_L$ and $\upnu_L$ as these values are not realized by any player in the outcome of the game. Thus in this region the resource increases by increasing $\uprho_H$ and $\upnu_H$. The upper right region corresponds to games in which the outcome is for both players to defect and so the resource quantity increases for higher levels of $\uprho_L$ and $\upnu_L$ and remains invariant to changes in $\uprho_H$ and $\upnu_H$. The third intermediate region corresponds to games in which the Nash equilibrium corresponds to one player defecting and the other cooperating. In this region, the resource is increased by increase in both $\uprho_H, \upnu_H$ and $\uprho_L, \upnu_L$.

\section{Discussion}
Opinions on the effect of heterogeneity on successful resource management are highly variable \cite{varughese2001contested, bardhan2002unequal, chand2015production} and range from negative to positive to no effect at all. Our model presents mixed results regarding heterogeneity in general, which suggests that heterogeneity is important in different ways across different factors. Note that heterogeneity may be manifested in either of the parameters $\uprho_i$ and $\upnu_i$, which are both different in nature. Hence, while our model does not provide a straight-forward verdict on the role of heterogeneity, it does contribute to the debate by indicating, that heterogeneity may need to be exposed to a classification, which is deeper than the conventional classification of economic, non-economic and socio-cultural heterogeneity to correctly establish its correlation with successful resource management. 

An important limitation of the model already discussed in Chapter \ref{chap:model} is the static nature of the psychological attributes of the consumers. Individuals in the real world are in a constant state of change and the human psychology is seldom static. We study the effects of relaxing this assumption by allow the consumers from the consumption game to vary their characteristics via the theory of learning in games, an exercise undertaken in the material of the next chapter.

\clearpage
%% This is an example first chapter.  You should put chapter/appendix that you
%% write into a separate file, and add a line \include{yourfilename} to
%% main.tex, where `yourfilename.tex' is the name of the chapter/appendix file.
%% You can process specific files by typing their names in at the 
%% \files=
%% prompt when you run the file main.tex through LaTeX.

\chapter{Closing the Feedback Loop}

\label{chap:learning}

In this chapter we return to the single and dual-agent networks to study the effects of certain feedback schemes, selected in the context of the analysis conducted in previous chapters. In Section \ref{sec:feed1}, we synthesize an optimal feedback law for the single agent homogeneous network on the basis of the analysis conducted in Chapter \ref{chap:optimal}. We find that the optimal feedback law is given as the solution to a non-linear differential equation whose solution is difficult to obtain analytically. Instead we demonstrate how the control law can be obtained numerically and present simulations of the optimally controlled process that exhibit the expected behavior. Later in Section \ref{sec:learn}, we return to the two-player consumption game of Chapter \ref{chap:game}. We allow the players to change their strategy in direction of the best response to the other player's current strategy. Such a tactic known commonly as the best response learning scheme, is a basic variant of fictitious play learning. We find that under the selected learning scheme, the players not only converge to the Nash equilibrium but do so in a way that avoids free-riding behavior, a phenomenon that occurred in the open-loop system of Chapter \ref{chap:open}.

\section{The Optimal Feedback Law}
\label{sec:feed1}
Recall the Hamiltonian system of the OCP (P3) in Chapter \ref{chap:optimal}, reproduced below 
\begin{align}
\label{ham-lamb2}
\begin{split}
	&\dot z(t) =  -z(t)-\frac{1}{\lambda(t)} + 1,\\
	&\dot\lambda(t) =(\del+1)\lambda (t)+\frac{2}{z(t)}.
\end{split}
\end{align}
In Chapter \ref{chap:optimal} we categorized the long-term behavior of the solutions of the above system into two cases: the sustainable case ($\del < 1$) and the unsustainable case ($\del \geq 1$). Consider first, the sustainable case. Figure \ref{fig:sust_fp} shows the optimal region in the phase space of the Hamiltonian system in this scenario.
\begin{figure}[h]
	\captionsetup{font=normal,width=0.8\textwidth}
	\begin{center}
		\includegraphics[width=0.6\linewidth]{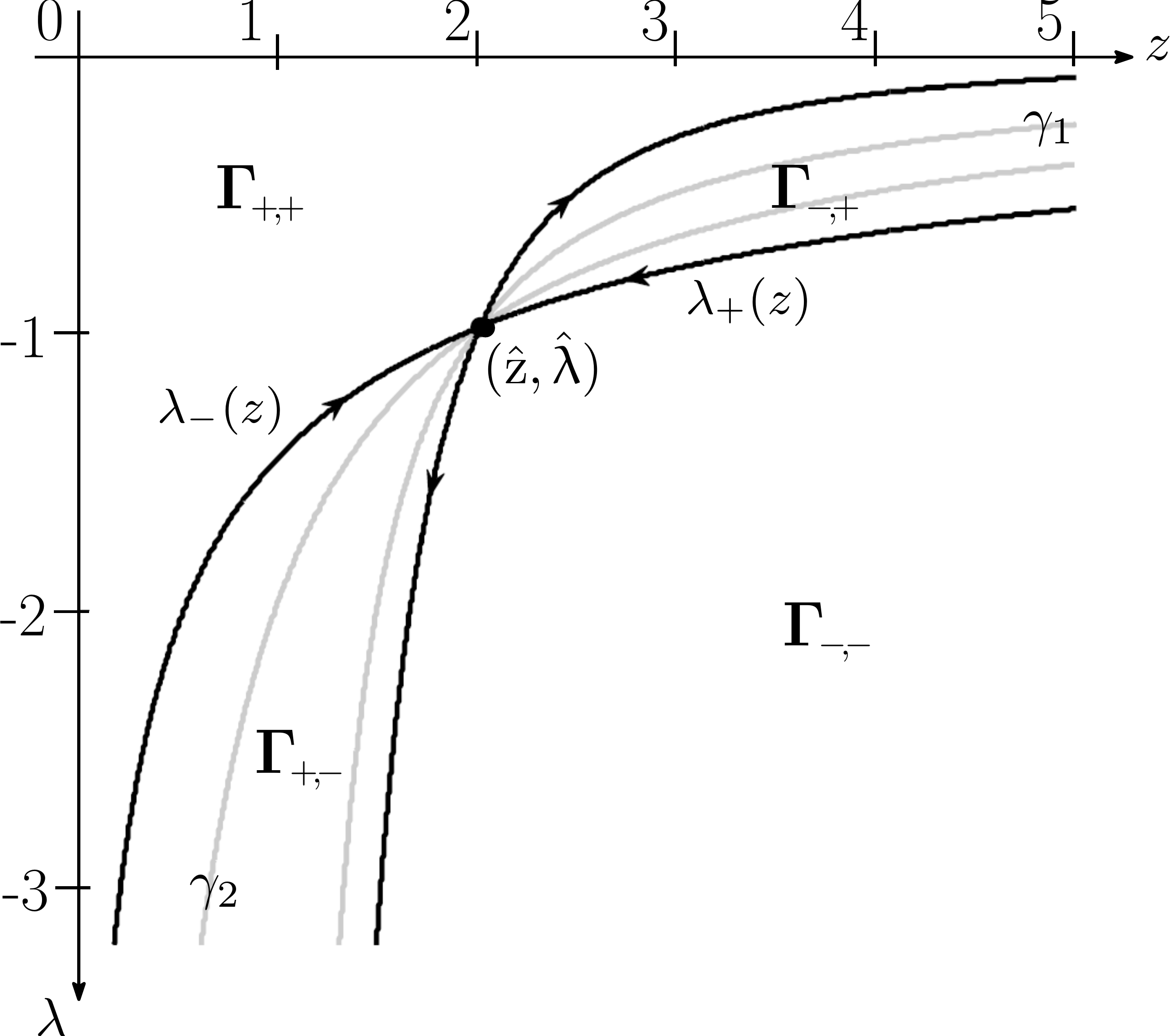}
	\end{center}
	\caption{The saddle point $(\hat(z),\hat{\uplambda})$ and regions of the stable manifold $\lambda_-(z)$ \& $\lambda_+(z)$. Here $\del = 0.01<1$}
	\label{fig:sust_fp} %discrete tragic
\end{figure}
Recall the trajectory $(z_1(\cdot),\lambda_1(\cdot))$ that approaches the point $(\hat z,\hat\lambda)$ from the left from the set $\Gamma_{+,+}$. We called this trajectory the left equilibrium trajectory, while the second trajectory $(z_2(\cdot),\lambda_2(\cdot))$ which approaches the point $(\hat z,\hat\lambda)$ from the right from the set $\Gamma_{-,-}$ was called the {\it the right equilibrium trajectory}. These trajectories are represented as the curves $\lambda_-(z)$ and $\lambda_+(z)$ respectively in the phase space. 

Note that to the left of the point $(\hat z,\hat\lambda)$ in the set $\Gamma_{+,+}$, the function $z_1(\cdot)$ increases. Therefore, while $(z_1(\cdot),\lambda_1(\cdot))$ lies in $\Gamma_{+,+}$, the time can be uniquely expressed in terms of the first coordinate of the trajectory $(z_1(\cdot),\lambda_1(\cdot))$ as a smooth function $t=t_1(z)$, $z\in (0,\hat z)$. Changing the time variable $t=t_1(z)$ on interval $(0,\hat z)$, we find that the function $\lambda_{-}(z)=\lambda_1(t_1(z))$, $z\in (0,\hat z)$, is a solution to the following differential equation on $(0,\hat z)$:
\begin{align}
\label{eq:feed_lambda}
	\frac{d\lambda (z)}{dz}= \frac{d\lambda (t_1(z))}{dt} \times \frac{dt_1(z)}{dz} = \frac{\lambda(z) \left( (\del+1)\lambda(z)z + 2\right)}{z \left( -\lambda(z)z -1 + \lambda(z)\right)}
\end{align}
with the boundary condition
\begin{equation}
\label{lamb-}
	\lim_{z\to\hat z-0}\lambda(z) =\hat\lambda.
\end{equation}
Obviously, the curve $\lambda_{-}=\left\{(z,\lambda)\colon \lambda=\lambda_{-}(z), z\in (0,\hat z)\right\}$ corresponds to the region of the stable manifold of $(\hat z,\hat\lambda)$ where $z<\hat{z}$.

Analogously, to the right of the point $(\hat z,\hat\lambda)$ in the set $\Gamma_{-,-}$, while $(z_2(\cdot),\lambda_2(\cdot))$ lies in $\Gamma_{-,-}$, the function $z_2(\cdot)$ decreases.  Hence, the time can be uniquely expressed in terms of the first coordinate of the trajectory $(z_2(\cdot),\lambda_2(\cdot))$ as a smooth function $t=t_2(z)$, $z\in (\hat z,\infty)$. Changing the time variable $t=t_2(z)$ on interval $(\hat z,\infty)$, we find that the function $\lambda_{+}(z)=\lambda_2(t_2(z))$, $z>\hat z$, is a solution to the differential equation~\eqref{eq:feed_lambda} on $(\hat z,\infty))$ with the boundary condition
\begin{equation}
\label{lamb+} 
	\lim_{z\to\hat z+0}\lambda(z) =\hat\lambda.
\end{equation}
As above, the curve $\lambda_{+}=\left\{(z,\lambda)\colon \lambda=\lambda_{+}(z), z\in (\hat z,\infty)\right\}$ corresponds to the region of the stable manifold of $(\hat z,\hat\lambda)$ where $z>\hat{z}$.

Using solutions $\lambda_{-}(\cdot)$ and $\lambda_{+}(\cdot)$ of differential equation~\eqref{eq:feed_lambda} along with~\eqref{eq:max_prncpl} we can get an expression for the optimal feedback law as follows
\begin{align*}
	y_*(z)=\begin{cases}
	-\frac{1}{\lambda_{-}(z)z}, & \text{if $z<\hat{z}$},\\
	\hskip5mm \frac{\del+1}{2}, & \text{if $z=\hat{z}$},\\
	-\frac{1}{\lambda_{+}(z)z}, & \text{if $z>\hat{z}$}.
    \end{cases}
\end{align*}
\vskip3mm

\begin{figure*}[t!]
    \centering
    \captionsetup{width=0.85\textwidth}
    \begin{subfigure}[t]{0.5\textwidth}
        \captionsetup{width=0.9\textwidth}
        \centering
        \includegraphics[trim = 40mm 80mm 40mm 80mm, clip, width = \linewidth]{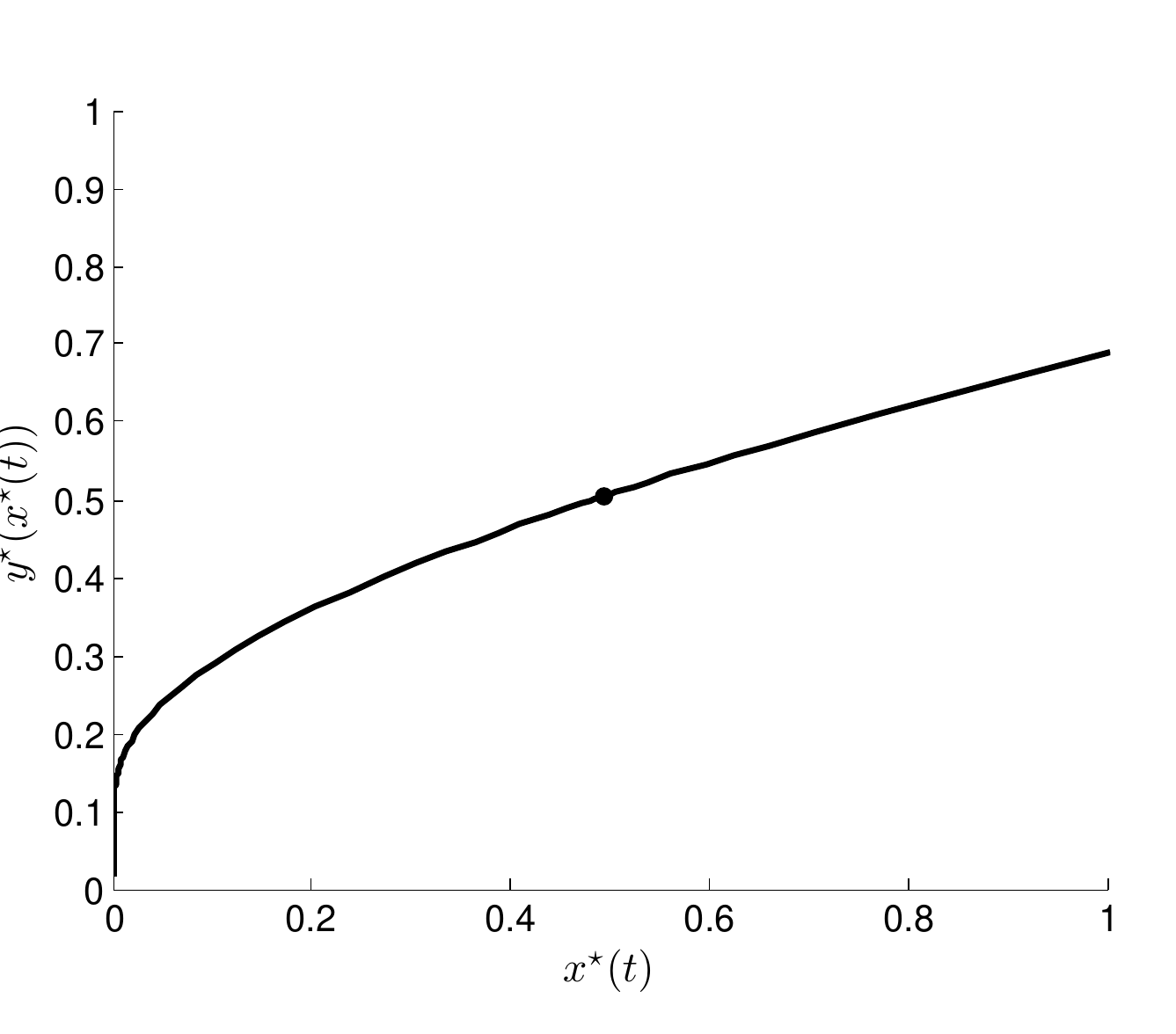}
        \caption{The optimal feedback law obtained by numerically solving \eqref{eq:feed_lambda}.}
        \label{fig:feedback_sust}
    \end{subfigure}%
    \begin{subfigure}[t]{0.5\textwidth}
        \captionsetup{width=0.9\textwidth}
        \centering
        \includegraphics[trim = 40mm 80mm 40mm 80mm, clip, width = \linewidth]{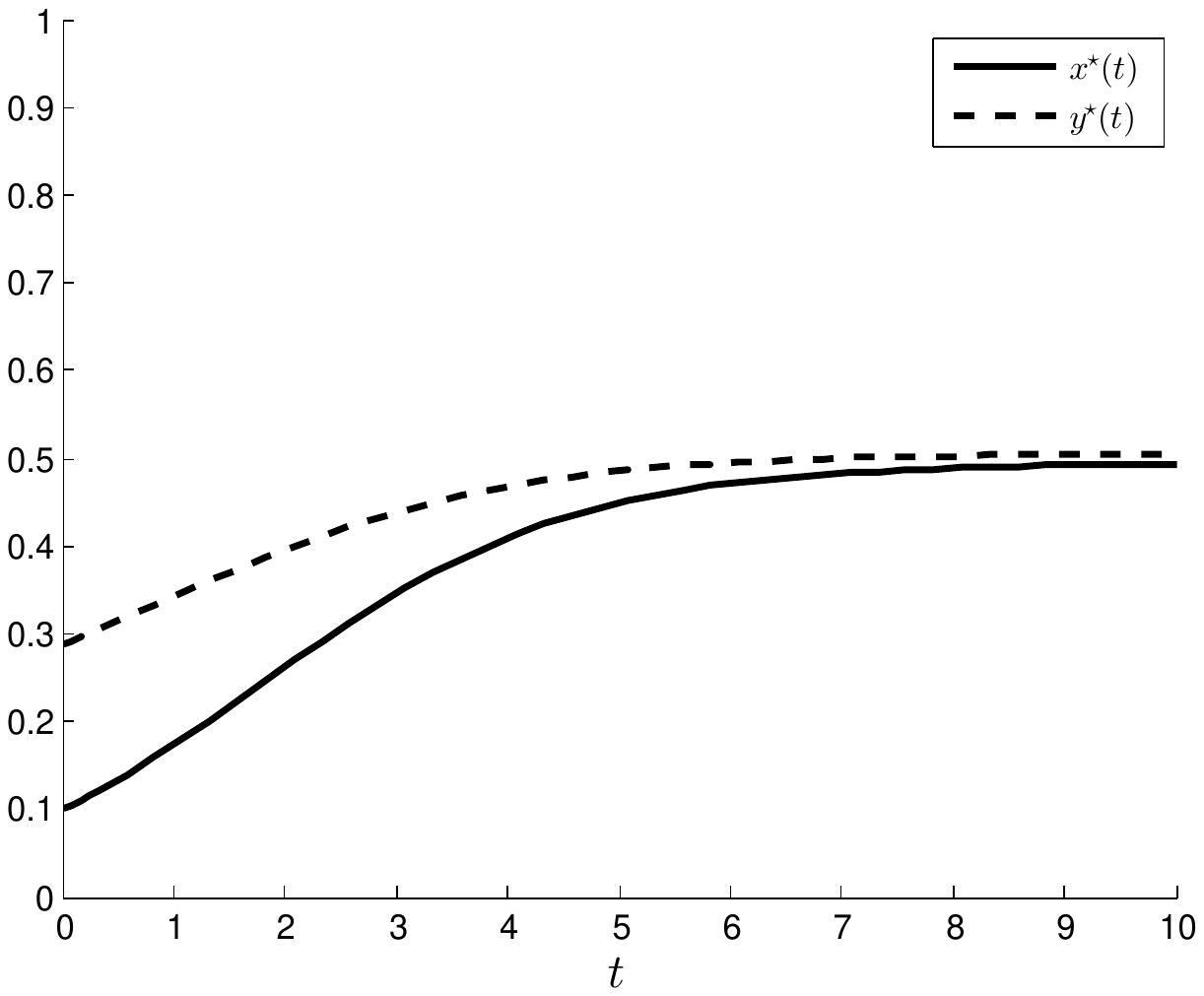}
        \caption{Solution with the optimal feedback law. Here the initial stock $x(0)$ = 0.1.}
        \label{fig:sol_sust}
    \end{subfigure}
    \caption{The optimal feedback law and a representative solution for the case $\del = 0.01 < 1$.}
\label{fig:results_sust}
\end{figure*}

An analytical solution to  nonlinear differential equation~\eqref{eq:feed_lambda} is difficult to obtain. However, it is possible to solve numerically. A graphical depiction of the feedback law in the original variables obtained by numerically solving the above ODE can be seen in Figure~\ref{fig:feedback_sust}. A representative solution of ($P1$) incorporating this feedback law is also shown in Figure~\ref{fig:sol_sust}. The trajectories show convergence of the stock and consumption to a steady state equilibrium. \vskip3mm

Now consider the unsustainable case $\del \geq 1$. We already saw in Chapter \ref{chap:optimal} that there exists a unique optimal pair $(z_*(\cdot), y_*(\cdot))$ for this case. Figure \ref{fig:phase_unsust_fp} shows the optimal trajectory along with the nullclines and the respective partitions of the phase space.

\begin{figure*}[h!]
        \captionsetup{width=0.8\textwidth}
        \centering
        \includegraphics[width = 0.6\linewidth]{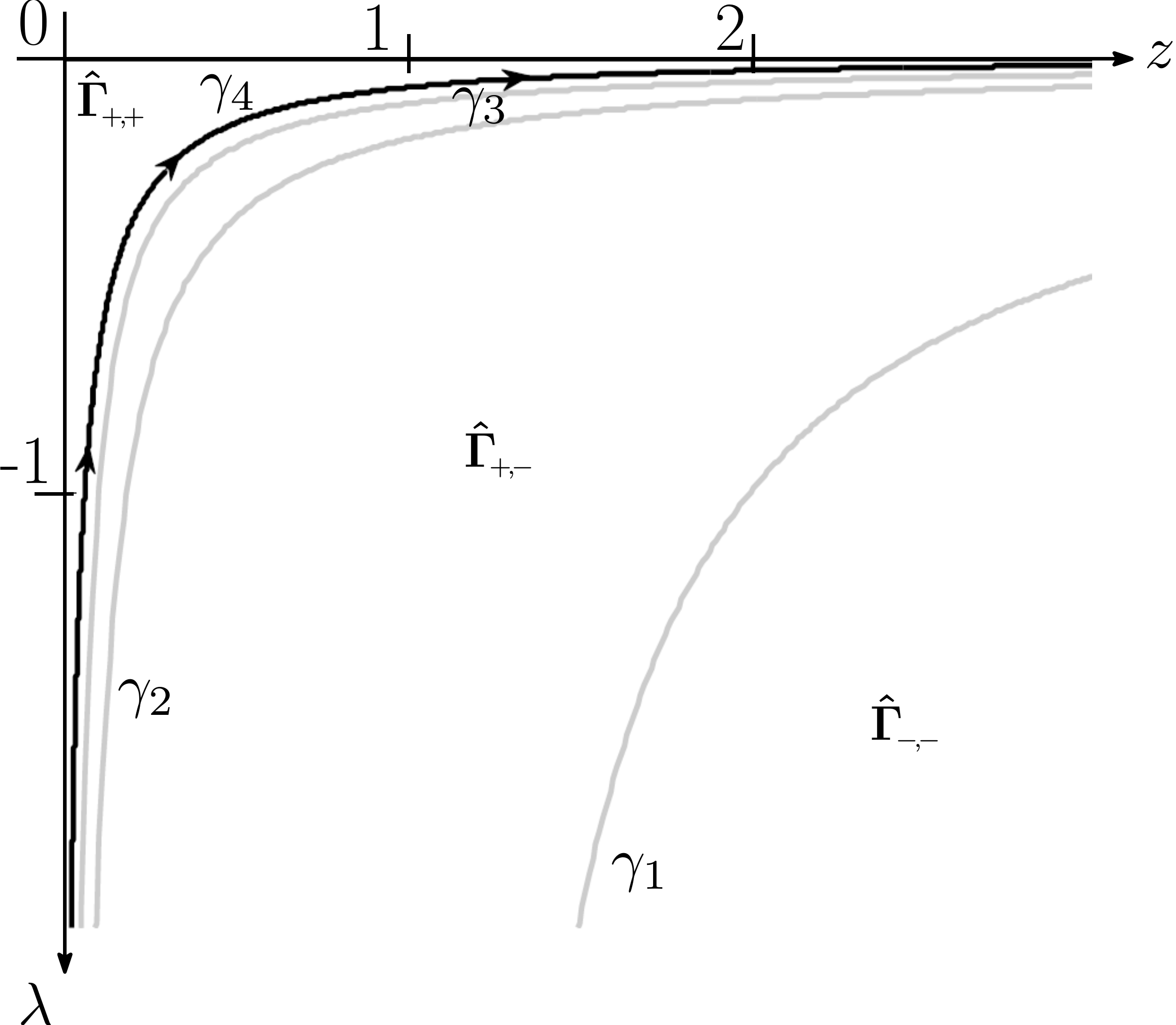}
        \caption{Optimal solution for \eqref{ham-lamb} in the case $\del = 10 \geq 1$. The optimal pair $(z_*(\cdot), y_*(\cdot))$ corresponds to curve $\gamma_4$ in the phase space. }
        \label{fig:phase_unsust_fp}
\end{figure*}

Now, since $(z_*(\cdot), y_*(\cdot))$ is unique, as a consequence the corresponding current value adjoint variable $\lambda_*(\cdot)$ is also defined uniquely as the maximal negative solution to equation (see~\eqref{ham-lamb})
\begin{equation}
\label{lamb0}
	\dot\lambda (t)=(\del+1)\lambda (t) +\frac{2}{z_*(t)}
\end{equation}
on the whole infinite time interval $[0,\infty)$. Note that the  function $z_*(\cdot)$ increases on $[0,\infty)$. Therefore, the time can be uniquely expressed as a smooth function $t=t_*(z)$, $z\in (0,\infty)$. Changing the time variable $t=t_*(z)$, we find that the function $\lambda_0(z)=\lambda_*(t_*(z))$ is a solution to the differential equation~\eqref{eq:feed_lambda} on the infinite interval $(0,\infty)$, similar to the case where $\del < 1$.

Using solution $\lambda_0(\cdot)$  of differential equation~\eqref{eq:feed_lambda}  along with~\eqref{eq:max_prncpl} we can get an expression for the optimal feedback law as follows
\begin{align*}
	y_*(z)=-\frac{1}{\lambda_0(z)z},\qquad z>0.
\end{align*}
Thus, to find the optimal feedback, we must determine for an initial state $z_0>0$ the corresponding initial state $\lambda_0<0$ such that solution  $(z_*(\cdot),\lambda_*(\cdot))$ of system~\eqref{ham-lamb} with initial conditions $z(0)=z_0$ and $\lambda(0)=\lambda_0$ exists on $[0,\infty)$ and $\lambda_*(\cdot)$ is the maximal negative function among all such solutions. Note that numerically finding this solution is highly unstable since any trajectory even infinitesimally below it is not optimal and any trajectory even infinitesimally above it is not defined on the whole time interval (and so cannot be optimal as well). We thus do not present a representative solution for this case, as we did for the sustainable case in Figure \ref{fig:results_sust}. Obtaining this solution would require an intricate numerical procedure which is beyond the scope of this dissertation.

Complementing the analysis of this section with that of Chapter \ref{chap:optimal}, let us summarize the main findings in the following theorem.

\begin{theorem}
\label{thm-synth}
For any initial state $z_0>0$ there is a unique optimal admissible pair $(z_*(\cdot),y_*(\cdot))$ in problem $(P3)$, and there is a unique adjoint variable $\psi(\cdot)$  that corresponds to $(z_*(\cdot),y_*(\cdot))$ due to the maximum principle (Theorem~\ref{thm2}).

If $\del < 1$ then there is a unique equilibrium $(\hat z,\hat\lambda)$ (see~\eqref{eq_pt}) in the corresponding current value Hamiltonian system~\eqref{ham-lamb} and the optimal synthesis is defined as follows
\begin{align*}
	y_*(x)=\begin{cases}
	-\frac{1}{\lambda_{-}(z)z}, & \text{if $z<\hat{z}$},\\
	\hskip5mm \frac{1+\del}{2}, & \text{if $z=\hat{z}$},\\
	-\frac{1}{\lambda_{+}(z)z}, & \text{if $z>\hat{z}$},
	\end{cases}
\end{align*}
where $\lambda_-(\cdot)$ and $\lambda_+(\cdot)$ are the unique solutions of~\eqref{eq:feed_lambda} that satisfy the boundary conditions~\eqref{lamb-} and~\eqref{lamb+} respectively. In this case optimal path $z_*(\cdot)$ is either decreasing, or increasing on $[0,\infty)$, or $z_*(t)\equiv\hat z$, $t\geq 0$, depending on the initial state $z_0$. For any optimal admissible pair $(z_*(\cdot),y_*(\cdot))$ we have $\lim_{t\to\infty}z_*(t)=\hat z$ and $\lim_{t\to\infty}y_*(t)=\hat y$ (see~\eqref{eq-u}).

If $\del \geq 1$ then for any initial state $z_0$ the corresponding optimal path $z_*(\cdot)$ in problem $(P3)$ is an increasing function, $\lim_{t\to\infty}z_*(t)=\infty$,  and the corresponding optimal control $y_*(\cdot)$ satisfies asymptotically to the Hotelling rule of optimal depletion of an exhaustible resource, i.e. $\lim_{t\to\infty}y_*(t)=\del$. The corresponding current value adjoint variable  $\lambda_*(\cdot)$ is defined uniquely as the maximal negative solution to equation~\eqref{lamb0} on $[0,\infty)$. The optimal synthesis is defined as
\[
	y_*(x)=-\frac{1}{\lambda_0(z)z},\qquad z>0,
\]
where $\lambda_0(z)=\lambda_*(t_*(z))$ is the corresponding solution of~\eqref{eq:feed_lambda}.
\end{theorem}

\section{Learning and Adaptation}

\label{sec:learn}

In Chapter \ref{chap:open}, Section \ref{sec:open_dual} we analyzed the open-loop characteristics of the two-agent society and observed the steady-state properties of the system. Issues such as stability of the system and the emergence of free-riding behavior were discussed. For the stable region of the parameter space, we where able to study the strategic interactions between the two groups as a two-player game, formulated at steady state. In Chapter \ref{chap:game}, we studied both continuous and discrete time versions of the game. A focus was placed on the nature of the game, more specifically, the difference in the rational outcome of the game and the collectively optimal outcome. We captured this notion via the game tragicness, and elaborated the notion for both continuous and discrete versions of the game. However, during the analysis of the Nash equilibrium, the issue of how the agents actually reach the Nash equilibrium in the consumption game was not discussed. We address this question here through the theory of learning in games \cite{fudenberg1998theory} where the players change their strategies dynamically according to a particular learning scheme. 

A frequently studied family of learning dynamics involves players constructing a belief of each other's future strategies, based on the past strategies played, and playing a best response to the predicted strategy. This form of learning dynamics is commonly known as fictitious play \cite{brown1951iterative} and is perhaps one of the most extensively studied model of learning to date \cite{brenner2006agent}. In this section we extend the game-theoretic construction to a continuous-time repeated game and give a simple demonstration of how the agents can approach the Nash equilibrium by employing basic best response dynamics (an elementary version of fictitious play). This is done by examining the stability of the basic system coupled with the learning adjustment process. Note that there exists no single result giving the learning strategy that always achieves the Nash equilibrium in general two-player games. Rather, rigorous proofs have been given only for specific combinations of games and learning schemes \cite{marden2012game} (A summary of learning in games from a control-theoretic perspective is given in \cite{marden2009cooperative}). We find that the resulting equilibrium is sustainable in the sense that the consumers do not exhibit any free-riding behavior, with the notion of free-riding being adopted from Chapter~\ref{chap:open}.

\subsection{Best-response Dynamics}

As stated above, after previously defining the consumption game and deriving the Nash equilibrium in Chapter \ref{chap:game} we consider here how the players are able to achieve that equilibrium. We do this by allowing the players to adjust their strategies under fictitious play learning. We employ fictitious play in its most basic form, commonly known as Best-response Dynamics, the Best-reply Process or Cournot Adjustment. Best-response learning is generally studied with caution as it is known to encounter several problems when applied as a model of learning in real-world scenarios \cite{brenner2006agent}. Nonetheless it remains one of the oldest, most familiar and perhaps one of the most extensively studied models of dynamic learning \cite{fudenberg1998theory}. We select it here due to its relative simplicity and its ability to reveal certain dynamic aspects of system behavior that may recur when more sophisticated models of learning are applied.

As we will see in what follows, the dynamics of the two-community system under best-response dynamics are such that it indeed results in the actions converging to the Nash equilibrium. This shows that the equilibrium is reachable by the agents in a distributed manner, however we must note that we consider only one particular learning scheme and that the Nash equilibrium may be reachable through schemes other than the one considered here. Under these dynamics, the system propagates as follows.
\begin{subequations}
\label{eq:br_dyn}
\begin{align}
\label{eq:br_dyn1}
\begin{split}
	&\dot{x} = (1-x)x - (y_1 + y_2)x,\\
	&\dot{y}_1 =  \b_1 (1-\upnu_1) (x-\rho_1)  -  \b_1 \upnu_1 (y_1 - y_2),\\
	&\dot{y}_2 =  \b_2 (1-\upnu_2) (x-\rho_2) -  \b_2 \upnu_2 (y_2 - y_1),
\end{split}
\end{align}
\vspace{-10pt}
\begin{align}
\label{eq:br_dyn2}
\begin{split}
	&\dot{\rho}_1 = \mathrm{BR}_1(\rho_2) - \rho_1,\quad
	\dot{\rho}_2 = \mathrm{BR}_2(\rho_1) - \rho_2,
\end{split}
\end{align}
\end{subequations}
\begin{figure}[t!]
	\begin{center} 
		\includegraphics[width=0.75\linewidth]{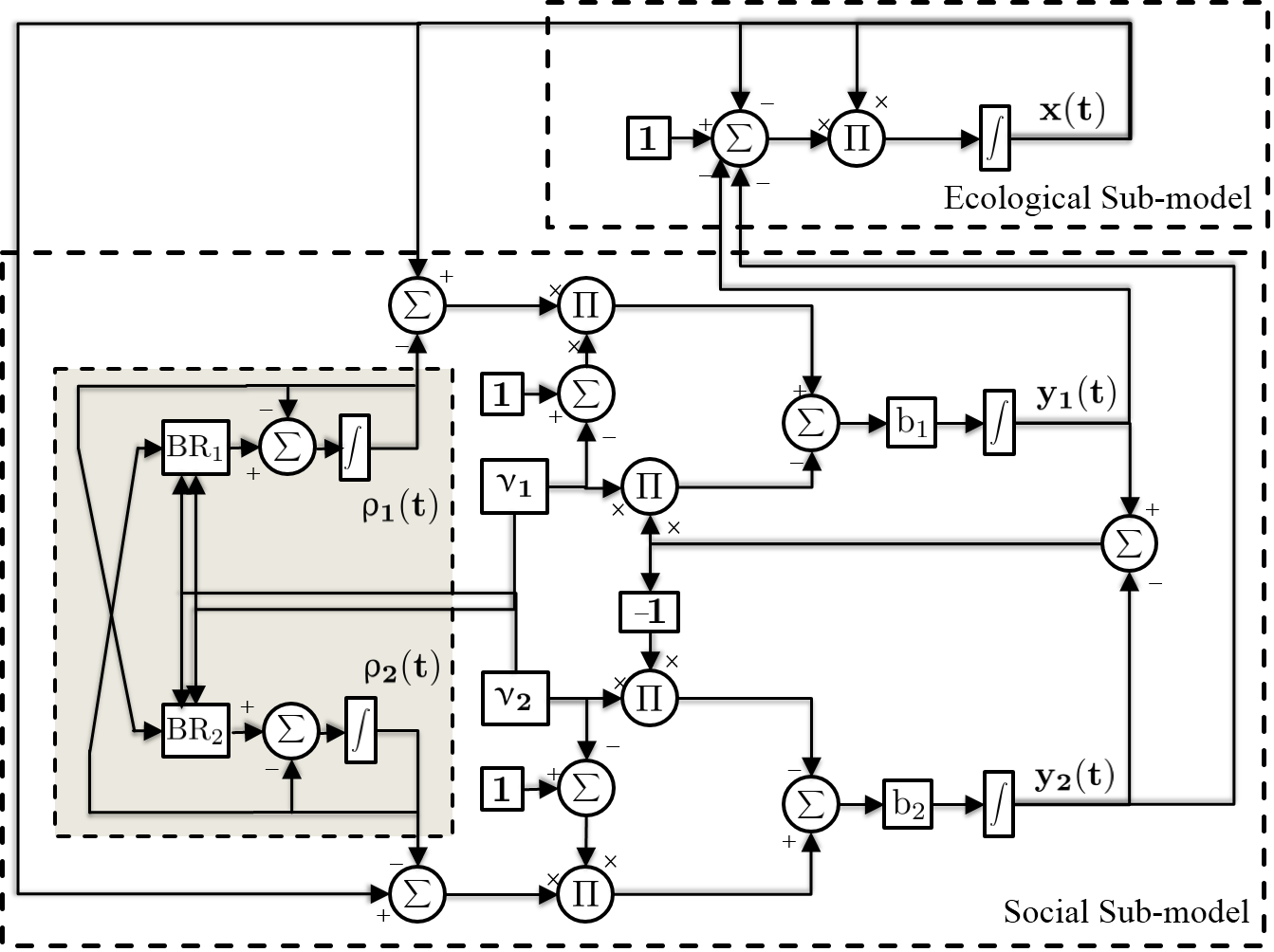}
	\end{center}
	\caption{Block diagram for \eqref{eq:br_dyn1} coupled with the learning dynamics of \eqref{eq:br_dyn2}. The learning process is part of the social component of the model and is highlighted by the shaded area.}
\label{fig:block_dia}
\end{figure}
where $\mathrm{BR}_i = \tilde{\uprho}_i$ is given by \eqref{eq:br} \footnote{Such a formulation of resource-based games while ignoring the transient and focusing on the steady state is not new (see \cite{ostrom2015governing} for similar examples). One must bear in mind however, that the conclusions drawn under such assumptions are relevant only for large time-scales and thus caution must be practiced while interpreting the results.}. Figure~\ref{fig:block_dia} depicts the dynamics of the overall system coupled with the learning process.

Under the dynamics given by \eqref{eq:br_dyn2}, the groups' strategies $\uprho_i$ are updated in direction of the best-response to each other's actions. Thus by definition, if an equilibrium $(\xbar,\ybar_1,\ybar_2,\bar{\uprho}_1,\bar{\uprho}_2)$ of the above system exists it must be the Nash equilibrium, since at steady state $\bar{\uprho}_1 = \mathrm{BR}_1(\bar{\uprho}_2)$ and $\bar{\uprho}_2 = \mathrm{BR}_2(\bar{\uprho}_1)$. We find upon solving \eqref{eq:br_dyn} for the equilibrium, that it is given by
\begin{align}
\begin{split}
\label{eq:learn_eq}
	&\xbar \!=\! \frac{2 \upnu_1\upnu_2}{\upnu_1\!+\!\upnu_2\!+\!2\upnu_1\upnu_2},\,\,
	\ybar_1 \!=\! \frac{\upnu_1}{\upnu_1\!+\!\upnu_2\!+\!2\upnu_1\upnu_2}, \\
	&\ybar_2 \!=\! \frac{\upnu_2}{\upnu_1\!+\!\upnu_2\!+\!2\upnu_1\upnu_2},\,
	\bar{\uprho}_1 \!=\! \frac{\upnu_1(3\upnu_2\!-\!\upnu_1\!-\!2\upnu_1\upnu_2)}{(1\!-\!\upnu_1)(\upnu_1\!+\!\upnu_2\!+\!2\upnu_1\upnu_2)},\\
	&\bar{\uprho}_2 \!=\! \frac{\upnu_2(3\upnu_1\!-\!\upnu_2\!-\!2\upnu_2\upnu_1)}{(1\!-\!\upnu_2)(\upnu_2\!+\!\upnu_1\!+\!2\upnu_2\upnu_1)},
\end{split}
\end{align}
where $\bar{\uprho}_1$ and $\bar{\uprho}_2$ indeed constitute the Nash equilibrium as given in \eqref{eq:nash1}. Note that in \eqref{eq:learn_eq} both consumptions are positive i.e., $\ybar_1 > 0, \ybar_2 > 0$ which constitutes a self-reliant equilibrium as per the notion defined in Chapter \ref{chap:open}. This is in contrast to the equilibrium \eqref{eq:equi-2comp} of the open-loop system, which permits the qualitatively different steady state  $\ybar_i > 0, \ybar_j < 0$ (the free-riding phenomenon). Thus in a society following best-response dynamics, the consuming groups do not exhibit free-riding behavior as opposed to the open-loop system \eqref{eq:two_com_sys} where the strategies are specified as exogenous parameters. 
%Figure 3 shows simulations of the system \ref{eq:br_dyn} under two different choices of the parameters $\upnu_1$ and $\upnu_2$. As expected, both result in a self-reliant steady state.\\

\subsection{Convergence to the Nash Equilibrium}
In order to find whether the Nash equilibrium is reachable or not, we need to explore the stability properties of system \eqref{eq:br_dyn}. The global stability of \eqref{eq:br_dyn1} has already been discussed in Chapter~\ref{chap:open}, Section~\ref{sec:open_dual} for constant $\rho_i$'s. Since the subsystem \eqref{eq:br_dyn2} does not depend on \eqref{eq:br_dyn1}, the stability of the overall system \eqref{eq:br_dyn} can then be deduced by analyzing \eqref{eq:br_dyn2} separately. More precisely, if $\rho_1$ and $\rho_2$ are guaranteed to reach a fixed point, then they may be treated as constants in \eqref{eq:br_dyn1} which is known to be globally stable if inequality \eqref{eq:con_stab1} holds. Thus, once \eqref{eq:br_dyn2} reaches its steady-state, \eqref{eq:br_dyn1} will also converge to its equilibrium (provided \eqref{eq:con_stab1} holds) regardless of where the system was when the steady-state was achieved. Since \eqref{eq:br_dyn2} is linear in $\rho_1$ and $\rho_2$, the stability can be deduced by analyzing the eigenvalues, which are given as 
\begin{align}
\label{eq:eigs}
	\uplambda_{1,2} = -1 \pm \sqrt{ \frac{  \left(\upnu_1 - \upnu_2\right)^2+4 \upnu_1^2 \upnu_2^2}{4 \upnu_1 \upnu_2}}.
\end{align}
%\begin{align*}
%	\lambda_{1,2} = \frac{-2 \tilde{\Nu}_1 \tilde{\Nu}_2\!\pm\! \sqrt{2 \tilde{\Nu}_1^2 \tilde{\Nu}_2^2\!-\!\tilde{\Nu}_1^3 \tilde{\Nu}_2\!-\!\tilde{\Nu}_1 \tilde{\Nu}_2^3\!+\!4 \tilde{\Nu}_1^3 \tilde{\Nu}_2^3}}{2 \tilde{\Nu}_1 \tilde{\Nu}_2},
%\end{align*}
It is obvious from \eqref{eq:eigs}, that both $\uplambda_1$ and $\uplambda_2$ are real. However, $\uplambda_1$ may be positive which corresponds to unstable learning dynamics. It can be seen from Figure \ref{fig:eigs} that the learning dynamics are stable in a major area of the parameter space. Combining this information with the stability condition for \eqref{eq:two_com_sys}, the overall criterion for stability of \eqref{eq:br_dyn} constitutes the following two conditions 
\begin{subequations}
\begin{align}
\label{eq:cond1}
	&(\b_1-\b_2)(\b_1\upnu_1-\b_2\upnu_2) + 4\b_1\upnu_1\b_2\upnu_2 > 0,
\end{align}\vspace{-15pt}
\begin{align}
\label{eq:cond2}
	&\left(\upnu_1 - \upnu_2\right)^2+4 \upnu_1^2 \upnu_2^2 - 4 \upnu_1 \upnu_2 < 0.
\end{align}
\end{subequations}
This shows that the Nash Equilibrium is reachable via best-response dynamics, if \eqref{eq:cond1} and \eqref{eq:cond2} hold collectively. 

%\begin{figure}[b!]
%	\begin{center} 
%		\includegraphics[width=0.65\linewidth]{br_stability.png}
%	\end{center}
%\vspace{-5pt}
%	\caption{The positivity of $\uplambda_1$ as given by \eqref{eq:eigs}.}
%\label{fig:eigs}
%\vspace{-15pt}
%\end{figure}
\subsection{Overcoming Free-riding Behavior}
Figures \ref{fig:graph_1} and \ref{fig:graph_2} illustrate the learning process for two different points in the parameter space. Due to \eqref{eq:learn_eq}, when $\upnu_1 = \upnu_2$ (Figure \ref{fig:graph_1}) then the resource at steady state $\xbar$ equals $\bar{\uprho}_1 = \bar{\uprho}_2$. However when $\upnu_1 \neq \upnu_2$ (Figure \ref{fig:graph_2}), the resource $x$ follows the group with lower social relevance $\upnu_i$. The case depicted in Figure \ref{fig:graph_2} is particularly interesting and merits additional interpretation. As seen from the long-term behavior of the graphs, and from the mathematical form of the equilibrium \eqref{eq:learn_eq}, higher social-relevance $\upnu_i$ implies a lower level of environmentalism $\rho_i$ and a higher steady state consumption $\ybar_i$. This is in contrast to what the steady-state of the open-loop system \eqref{eq:two_com_sys} predicts i.e., cooperative individuals are usually associated with higher levels of environmental concern, and also consume less than non-cooperative individuals on average. Indeed this is what has been observed through different social-psychological studies as well \cite{manzoor2016game}. The closed-loop system \eqref{eq:br_dyn} however predicts the behavior of the consumers \emph{if} they act rationally according to the game defined by \eqref{eq:br_dyn}. This highlights the relevance of game theory to distributed control problems, whereby the selected payoffs and learning scheme may be treated as design principles in the bigger problem of social control \cite{marden2012game, manzoor2014optimal} in order to obtain desired behavior. Indeed as we find here, the selected scheme eliminates the phenomenon of free-riding, which poses a significant challenge in the effective governance of natural resources \cite{ostrom2015governing}. 
\begin{figure*}[t]
\vspace{25pt}
\begin{center}
\captionsetup{width=\linewidth}
\begin{subfigure}[t]{0.25\linewidth}
	\captionsetup{width=0.9\linewidth}
	\includegraphics[width=\linewidth]{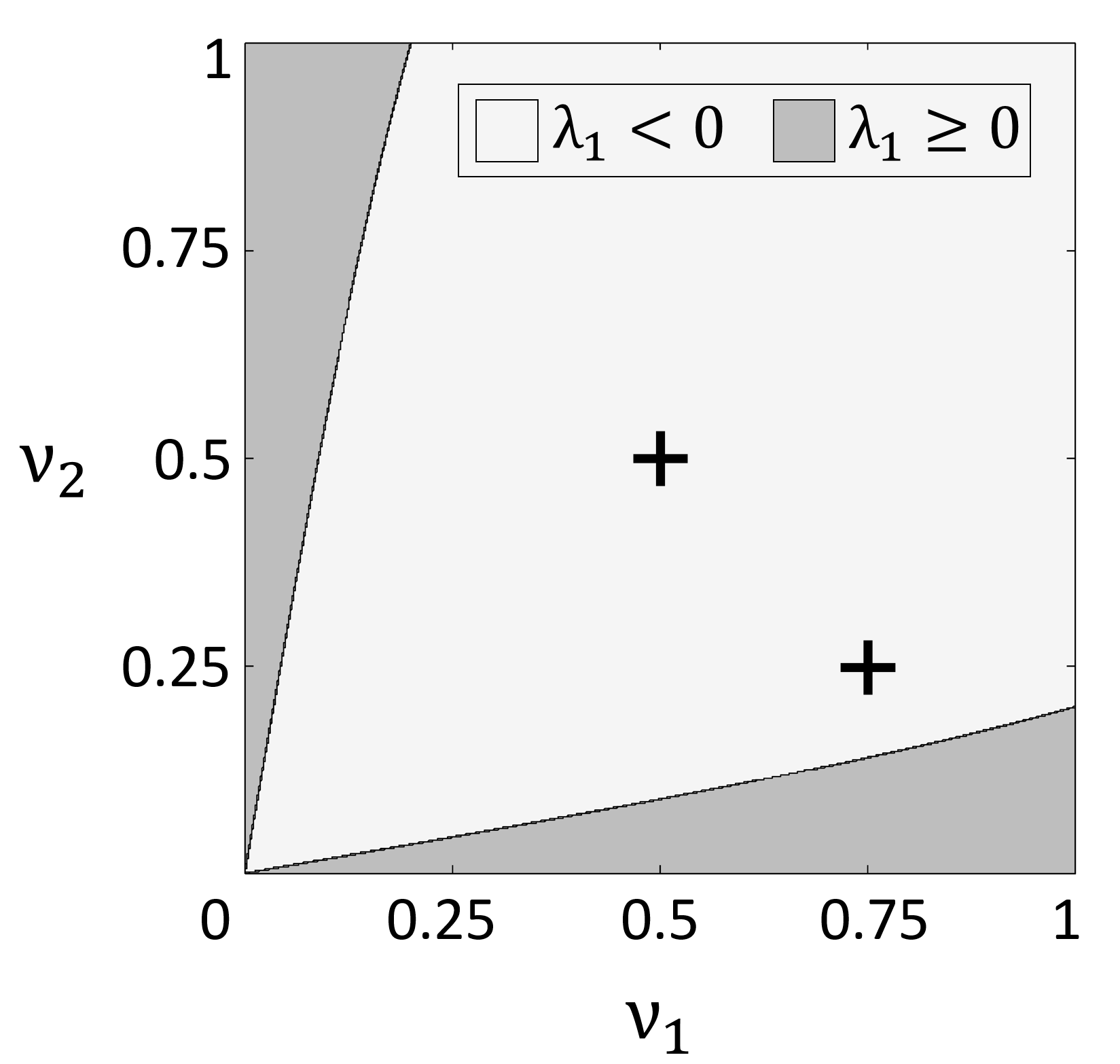}
	\vspace{-15pt}
	\caption{The positivity of $\uplambda_1$ as given by \eqref{eq:eigs}.}
	\label{fig:eigs}
\end{subfigure}
%\hspace{0.0025\linewidth}
\begin{subfigure}[t]{0.37\linewidth}
	\captionsetup{width=0.9\linewidth}
	\includegraphics[width = \linewidth]{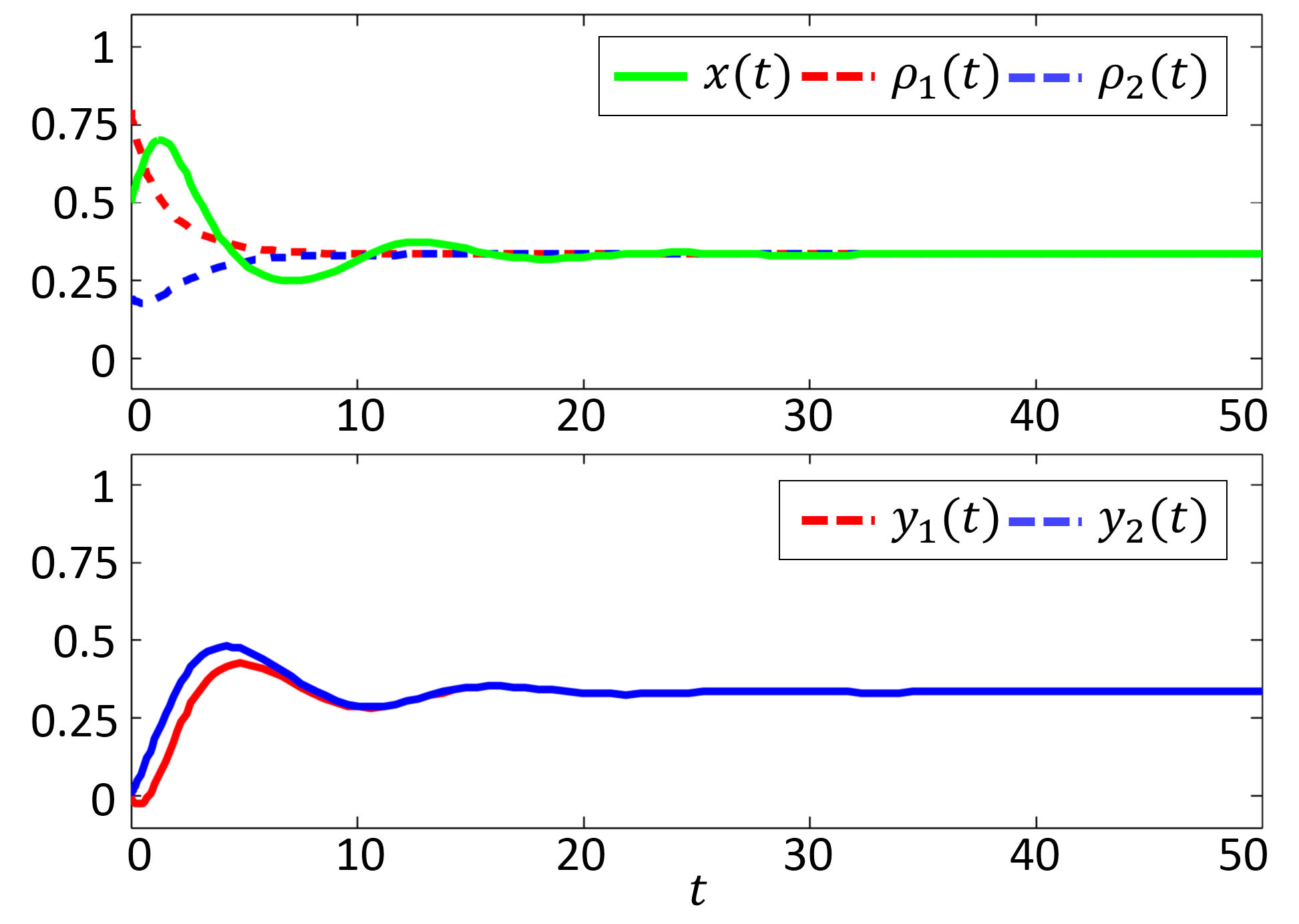}
	\vspace{-15pt}
	\caption{Learning process for $\upnu_1 = \upnu_2 = 0.5$.}
	\label{fig:graph_1}
\end{subfigure}%
%\hspace{0.0025\linewidth}
\begin{subfigure}[t]{0.37\linewidth}
	\captionsetup{width=0.9\linewidth}
	\includegraphics[width = \linewidth]{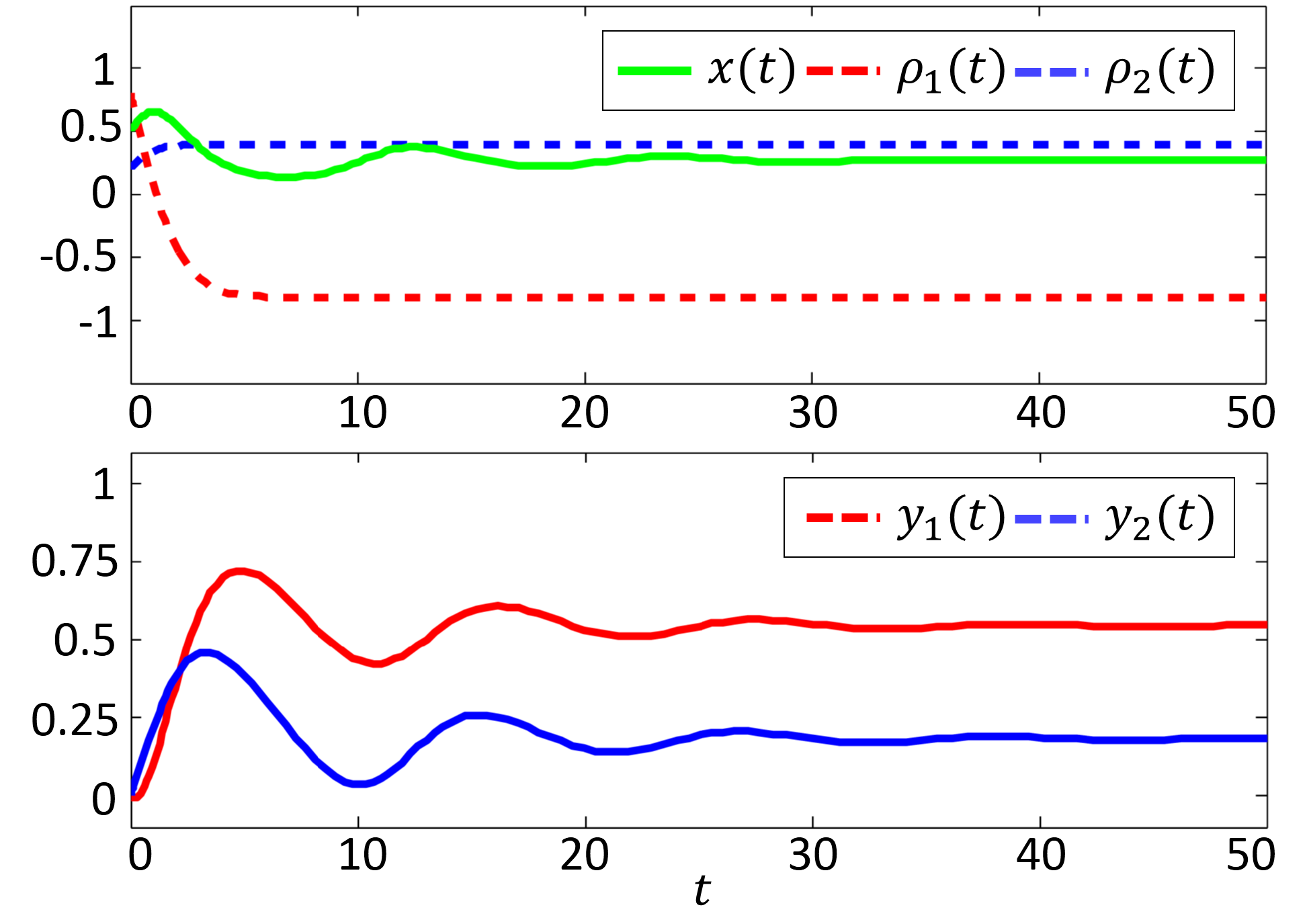}
	\vspace{-15pt}
	\caption{Learning process for $\upnu_1 = 0.75$ and $\upnu_2 = 0.25$.}
	\label{fig:graph_2}
\end{subfigure}
\vspace{-2pt}
\caption{\emph{left:} The region of stability for \eqref{eq:br_dyn2}. The cross-hairs correspond to points selected for the subsequent simulations. \emph{center:} The learning process for identical $\upnu_i$'s. \emph{right:} The learning process for different $\upnu_i$'s. In both simulations $\b_1 = \b_2 = 1$, $x(0) = 0.5$, $y_1(0) = y_2(0) = 0$, $\uprho_1(0) = 0.8$ and $\uprho_2(0) = 0.2$. Note that since $\b_1 = \b_2$, \eqref{eq:cond1} holds true for all $\upnu_1. \upnu_2$.}
\end{center}
\label{fig:results}
\vspace{-15pt}
\end{figure*}

\section{Discussion}

In this chapter we have studied the basic consumer networks under feedback mechanisms inspired from the analysis conducted in preceding chapters. Firstly, we devised the optimal feedback law for the single-group network for both sustainable and unsustainable cases. We also discussed how to obtain it numerically for both cases and demonstrated the results through a numerical simulation. Secondly, we applied the theory of learning in games to study the effects of the consumers adapting their characteristics in the two-agent network. The relevance of the Nash equilibrium has been demonstrated by observing the dynamics of agents employing best-response learning dynamics. We hope this serves to instigate inquiry into other complex learning schemes that may also converge to the Nash equilibrium. We also found that while the consumers may indulge in free-riding in the open-loop system, the adoption of best-response learning eliminates this phenomenon. These findings may potentially be used to inform policy making for natural resource management, especially since the study is based on one of the few fully justified mathematical models (at least to our knowledge) of socio-ecological couplings.

\clearpage
\chapter{Moving Towards Generically Structured Networks}

\label{chap:n}

The analysis of Chapters \ref{chap:open} -- \ref{chap:learning} has focused on the basic single and dual agent consumer networks. We now consider the full $\n$-agent network as an open problem. Due to the large dimension and non-linear nature of the system, replicating the analysis of the previous chapters appears not to be straight forward. Here we present the fixed point calculations for two structured $\n$-agent populations using an inductive rather than constructive approach. The exercise highlights the difficulties to be expected for analyzing the behavior of the generically structured network. In the end, we indicate how the network topology may influence the system dynamics via the ``Laplacian" of the influence network.

\section{Static Characteristics for Structured Networks}

\subsection{Well-mixed population}

Here we assume that the consuming population is well-mixed, i.e., all consumers are connected  and exert a uniform influence on each other. Thus the network constitutes a complete graph and the dynamics are given by
\begin{subequations}
\label{well_mixed_system}
\begin{align}
	\label{well_mixed_x}
	&\dot{x} = (1-x)x - \sum_{i=1}^{\n} y_i \, x \\
	\label{well_mixed_yi}
	&\dot{y}_i = \b_i \left( (1-\upnu_i) (x-\uprho_i) + \frac{\upnu_i}{\n-1}\sum_{j=1}^\n  y_j - y_i \right)
\end{align}
\end{subequations}
where $ i\in \{1, \dots , \n\}$.  We represent any possible equilibria by $(\xbar,\ybar_1, \dots, \ybar_\n)$.\\

First let us note the following assumptions that we will invoke at various points in the analysis that follows

\vskip1mm \noindent{\bf (A1)} {\bf Maximally social society} - {\it  All consumers carry the maximum possible social weight i.e.,\\ $\upnu_1 = \upnu_2 = \dots = \upnu_\n=1$} \vskip1mm
\vskip1mm \noindent{\bf (A2)} {\bf Minimally social society} - {\it  All consumers carry the minimum possible social weight i.e., \\ $\upnu_1 = \upnu_2 = \dots = \upnu_\n=0$} \vskip1mm
\vskip1mm \noindent{\bf (A3)} {\bf Maximally environmental society} - {\it  All consumers manifest the maximum possible level of environmentalism i.e., $\uprho_1 = \uprho_2 = \dots = \uprho_\n=1$} \vskip1mm
\vskip1mm \noindent{\bf (A4)} {\bf Minimally environmental society} - {\it  All consumers manifest the minimum possible level of environmentalism i.e., $\uprho_1 = \uprho_2 = \dots = \uprho_\n=0$} \vskip1mm

Before presenting the equilibrium for system \eqref{well_mixed_system}, we require the following intermediate results

\begin{lemma}
\label{lemmax0_1}
	If (A1) holds, there exist an infinite number of equilibria for system \eqref{well_mixed_system}.
\end{lemma}
\begin{proof}
	Substituting the assertion of $(A1)$ in \eqref{well_mixed_yi} and solving for $\dot{y}_i=0$ gives the following linear system
\begin{align*}
	(\n-1)\ybar_i - \sum_{j\neq i} y_j = 0
\end{align*}
In matrix notation, the row reduced echelon form for this system is given by
\begin{align*}
	\left[\begin{array}{ccccc} 1 & 0 & \dots & 0 & -1 \\ 0 & 1 & \dots & 0 & -1 \\ \vdots & \vdots & \ddots & \vdots & \vdots \\ 0 & 0 & \dots & 1 & -1 \\ 0 & 0 & \dots & 0 & 0 \end{array} \right] \left[\begin{array}{c} \ybar_1 \\ \ybar_2 \\ \vdots \\ \ybar_{\n\!-\!1} \\ \ybar_\n \end{array} \right] = \left[ \begin{array}{c} 0 \\ 0 \\ \vdots \\ 0 \\  0 \end{array} \right]
\end{align*}
From here, the solution for the consumption efforts is given by $\ybar_1 = \ybar_2 = \dots \ybar_\n = \ybar$ where $\ybar \in \mathbb{R}$. Substituting this result in \eqref{well_mixed_x} we see that either $\xbar = 0$ or $\xbar = 1 -  \n y$. Note that due to \eqref{well_mixed_x}, $\xbar = 1-\n y$ constitutes a valid equilibrium only if $y\leq 1/\n$. This completes the proof.
\end{proof}

\begin{lemma}
\label{lemma_well_mixed_A2}
	If (A2) holds, there exist an infinite number of equilibria for system \eqref{well_mixed_system} only when $\uprho_1 = \uprho_2 = \dots = \uprho_\n = \uprho$ (where $\uprho \in [0,1]$). Otherwise no equilibrium exists.
\end{lemma}
\begin{proof}
	Substituting the assertion of $(A2)$ in \eqref{well_mixed_system} gives us the following system
\begin{align*}
	&\dot{x} = (1-x)x - \sum_{i=1}^{\n} y_i \, x \\
	&\dot{y}_i = \b_i \left( x-\uprho_i\right)
\end{align*}

It is easy to see that for the case $\xbar=0$, the dynamics for $y_i$ are given by $\dot{y}_i = -\b_i \uprho_i$. Thus an equilibrium can exist in this case only if $\uprho_i = 0 \,\,\forall \,\, i$. In fact, the equilibrium would consist of the single state $\xbar = \ybar_1 = \dots = \ybar_\n = 0$. For the case $\xbar = 1 - \sum_{i=1}^\n \ybar_i$, we obtain the following system for the steady state efforts
\begin{align*}
	\sum_{i=1}^\n \ybar_i = 1-\uprho_i 
\end{align*}
This system is solvable only when $\uprho_i = \uprho \,\, \forall \,\, i$. In this case the steady state effort is given by $\xbar = \uprho$ and $\sum_{i=1}^\n \ybar_i = 1-\uprho$. Otherwise no equilibrium exists.

Thus, if  $\uprho_1 = \uprho_2 = \dots = \uprho_\n = \uprho$, then the steady state constitutes all points where  $\sum_{i=1}^\n \ybar_i = 1-\uprho$ with $\xbar = \uprho$. In the event that $\uprho = 0$ there also exists an additional equilibrium consisting of the single state $\xbar = \ybar_1 = \dots = \ybar_\n = 0$. This completes the proof.
\end{proof}

\begin{lemma}
\label{lemma_xbar0}
	In the event that none of (A1), (A2) or (A4) hold true, there can exist no equilibrium for system \eqref{well_mixed_system} with $\xbar = 0$.
\end{lemma}
\begin{proof}
	Assuming $\xbar=0$ and solving \eqref{well_mixed_yi} for $\dot{y}_i = 0$ leads to the following system of equations
\begin{align*}
	\left[\begin{array}{cccc} -(\n-1) & 1 & \dots & 1 \\ 1 & -(\n-1) & \dots & 1 \\ \vdots & \vdots & \ddots & \vdots \\ 1 & 1 & \dots & -(\n-1) \end{array} \right] \left[\begin{array}{c} \ybar_1 \\ \ybar_2 \\ \vdots \\ \ybar_\n \end{array} \right] = (\n-1)\left[ \begin{array}{c} \frac{(1-\upnu_1)\uprho_1}{\upnu_1} \\ \frac{(1-\upnu_2)\uprho_2}{\upnu_2} \\ \vdots \\ \frac{(1-\upnu_\n)\uprho_\n}{\upnu_\n} \end{array} \right]
\end{align*}
It is easy to check that the rank of the coefficient matrix is equal to $\n-1$. Adding all rows to the last one, gives us the following equivalent form of the equation
\begin{align*}
	\left[\begin{array}{ccccc} -(\n-1) & 1 & \dots & 1 & 1 \\ 1 & -(\n-1) & \dots & 1 & 1 \\ \vdots & \vdots & \ddots & \vdots & \vdots \\ 1 & 1 & \dots & -(\n-1) & 1 \\ 0 & 0 & \dots & 0 & 0 \end{array} \right] \left[\begin{array}{c} \ybar_1 \\ \ybar_2 \\ \vdots \\ \ybar_{\n-1} \\ \ybar_\n \end{array} \right] = (\n-1)\left[ \begin{array}{c} \frac{(1-\upnu_1)\uprho_1}{\upnu_1} \\ \frac{(1-\upnu_2)\uprho_2}{\upnu_2} \\ \vdots \\ \frac{(1-\upnu_{\n\!-\!1})\uprho_{\n\!-\!1}}{\upnu_{\n\!-\!1}} \\ \sum_{i=1}^\n \frac{(1-\upnu_i)\uprho_i}{\upnu_i} \end{array} \right]
\end{align*}
which means that a solution can exist only if $\sum_{i=1}^\n \frac{(1-\upnu_i)\uprho_i}{\upnu_i}=0$. Since none of (A1),(A2) or (A4) are true,  $\frac{(1-\upnu_i)\uprho_i}{\upnu_i}>0 \,\, \forall\,\, i$ and so, $\sum_{i=1}^\n \frac{(1-\upnu_i)\uprho_i}{\upnu_i}\neq 0$. Thus no equilibrium exists in this scenario.
\end{proof}

\begin{lemma}
\label{lemma_uniqueness}
	In the event that none of (A1), (A2), (A3) or (A4) hold true, there is a unique equilibrium for system \eqref{well_mixed_system} with $\xbar = 1-\sum_{i=1}^\n \ybar_i$. This equilibrium will exist only if there are at least $\n-1$ consumers with positive social weight ($\upnu_i > 0$)
\end{lemma}
\begin{proof}
	We know from \eqref{well_mixed_x} that $\dot{x} = 0$ if $x + \sum_{i=1}^\n y_i = 1$. Substituting this in \eqref{well_mixed_yi} and solving for $\dot{y} = 0$, gives us the following system of equations
\begin{align*}
	(\n-1)\ybar_i + (\n-1-\n\upnu_i)\sum_{j\neq i} \ybar_j = (\n-1)(1-\upnu_i)(1-\uprho_i)
\end{align*}
In matrix form
{
\setlength{\mathindent}{0cm}
\fontsize{10}{12}\selectfont
\begin{align*}
	\left[\begin{array}{cccc} (\n-1) & (\n-1-\n\upnu_1) & \dots & (\n-1-\n\upnu_1) \\ (\n-1-\n\upnu_2) & (\n-1) & \dots & (\n-1-\n\upnu_2) \\ \vdots & \vdots & \ddots & \vdots \\ (\n-1-\n\upnu_n) & (\n-1-\n\upnu_\n) & \dots & (\n-1) \end{array} \right] \left[\begin{array}{c} \ybar_1 \\ \ybar_2 \\ \vdots \\ \ybar_\n \end{array} \right] = (\n-1)\left[ \begin{array}{c} (1-\upnu_1)(1-\uprho_1) \\ (1-\upnu_2)(1-\uprho_2) \\ \vdots \\ (1-\upnu_\n)(1-\uprho_\n) \end{array} \right]
\end{align*}}
Since neither of $(A1)$ nor $(A3)$ hold true, the vector on the RHS does not equate to zero. Thus for the purpose of this proof, it suffices to show that the coefficient matrix for the above equation is full rank. We do this by observing the determinant, which is given as
{
\setlength{\mathindent}{0cm}
\begin{align*}
	\left|\begin{array}{cccc} (\n-1) & (\n-1-\n\upnu_1) & \dots & (\n-1-\n\upnu_1) \\ (\n-1-\n\upnu_2) & (\n-1) & \dots & (\n-1-\n\upnu_2) \\ \vdots & \vdots & \ddots & \vdots \\ (\n-1-\n\upnu_\n) & (\n-1-\n\upnu_\n) & \dots & (\n-1) \end{array} \right| = \n\left( \n\prod_{j}{\upnu_j} - \sum_{j} \left( \prod_{k\neq j}{\upnu_k}\right) \right)
\end{align*}}
It can be seen that in order for the determinant to be non-zero, there cannot be more than one consumer with social weight $\upnu_i$ equal to zero. In this event, the coefficient matrix above will always be full rank, and thus a unique equilibrium will always exist for system \eqref{well_mixed_system}. This concludes the proof.
\end{proof}

In light of the new condition obtained by Lemma \ref{lemma_uniqueness}, we define the following assumption

\vskip3mm \noindent{\bf (A5)} {\bf Sufficiently social society} - {\it  There exists no more than one consumer in the population who carries social weight $\upnu_i$ equal to 0.} \vskip3mm

\begin{corollary}
	If $(A5)$ does not hold true, then there exists no equilibrium for system \eqref{well_mixed_equilibrium}
\end{corollary}
\begin{proof}
	The corollary follows directly from Lemma \ref{lemma_xbar0} and Lemma \ref{lemma_uniqueness} .
\end{proof}

Now that all the pathological cases have been delineated, we are ready to complete the specification of the equilibrium for system \eqref{well_mixed_system}.

\begin{lemma}
\label{lemma_well_mixed_equilibrium}
	In the event that (A5) holds, and none of (A1)-(A4) hold true, the equilibrium for system \eqref{well_mixed_system} is given by
{
\setlength{\mathindent}{0cm}
\fontsize{9.5}{11.5}\selectfont
\begin{subequations}
\label{well_mixed_equilibrium}
\begin{align}
\label{well_mixed_equilibrium_x}
	&\xbar = \frac{\displaystyle \sum_{i}\left(\uprho_i(\upnu_i-1) \prod_{j\neq i} \upnu_j\right)}{\displaystyle \n\prod_{j}{\upnu_j} - \sum_{j} \left( \prod_{k\neq j}{\upnu_k}\right) }\\
\label{well_mixed_equilibrium_y}
	&\ybar_i = \frac{1}{\n} + \frac{\displaystyle \uprho_i(\upnu_i-1) \left( (1-\n)\sum_{j\neq i} \left( \prod_{k \neq i,j} \upnu_k \right) + \n(\n-2)\prod_{j\neq i} \upnu_j \right) - (1-\n+\n \upnu_i)\sum_{j\neq i}\left( \uprho_j (\upnu_j-1)\prod_{k\neq i,j}\upnu_k \right)}{\displaystyle \n\left( \n\prod_{j}{\upnu_j} - \sum_{j} \left( \prod_{k\neq j}{\upnu_k}\right) \right)}
\end{align}
\end{subequations}}
where $i,j,k \in \{1,\dots,\n\}$. 
\end{lemma}
\begin{proof}
	We begin by substituting the expression given by \eqref{well_mixed_equilibrium} in \eqref{well_mixed_x}. Since $\xbar \neq 0$, in order for $\dot{x}$ to be zero, it must be true that $\xbar + \sum \ybar_i=1$. Summing equations \eqref{well_mixed_equilibrium_x} and \eqref{well_mixed_equilibrium_y} over all $i$ gives us

{ \allowdisplaybreaks
\fontsize{9.5}{11.5}\selectfont
\begin{align*}
	\hspace{-20pt}&\xbar + \sum_{i=1}^\n \ybar_i \\
	\hspace{-20pt}&= \n\left(\frac{1}{\n}\right) + \sum_{i=1}^\n \frac{\displaystyle \n\prod_{j\neq i} \upnu_j \!+\! (1\!-\!\n)\!\sum_{j\neq i} \!\!\left( \prod_{k \neq i,j} \upnu_k \!\!\right) \!+\! \n(\n\!-\!2)\prod_{j\neq i} \upnu_j \!-\! \sum_{j\neq i}(1\!-\!\n\!+\!\n \upnu_j)\prod_{k\neq i,j}\upnu_k}{\displaystyle \n\left( \n\prod_{j}{\upnu_j} - \sum_{j} \left( \prod_{k\neq j}{\upnu_k}\right) \right)}\uprho_i(\upnu_i\!-\!1)\\
	\hspace{-20pt}&= 1 + \sum_{i=1}^\n \frac{\displaystyle \n\prod_{j\neq i} \upnu_j \!+\! \n(\n\!-\!2)\prod_{j\neq i} \upnu_j \!-\! \n\sum_{j\neq i}\left(\prod_{k\neq i}\upnu_k\right) \!+\! (1\!-\!\n)\!\sum_{j\neq i} \!\!\left( \prod_{k \neq i,j} \upnu_k \!\!\right)\!-\! \sum_{j\neq i}(1\!-\!\n)\prod_{k\neq i,j}\upnu_k}{\displaystyle \n\left( \n\prod_{j}{\upnu_j} - \sum_{j} \left( \prod_{k\neq j}{\upnu_k}\right) \right)}\uprho_i(\upnu_i\!-\!1)\\
	\hspace{-20pt}&= 1 + \sum_{i=1}^\n \frac{\displaystyle \n(\n\!-\!1)\prod_{j\neq i} \upnu_j \!-\! \n(\n\!-\!1)\prod_{k\neq i}\upnu_k \!+\! (1\!-\!\n)\!\sum_{j\neq i} \!\!\left( \prod_{k \neq i,j} \upnu_k \!\!\right)\!-\! (1\!-\!\n)\sum_{j\neq i}\left(\prod_{k\neq i,j}\upnu_k\right)}{\displaystyle \n\left( \n\prod_{j}{\upnu_j} - \sum_{j} \left( \prod_{k\neq j}{\upnu_k}\right) \right)}\uprho_i(\upnu_i\!-\!1)\\
	\hspace{-20pt}&= 1 + \sum_{i=1}^\n \frac{\displaystyle (\n(\n\!-\!1)\!-\!\n(\n\!-\!1))\prod_{j\neq i} \upnu_j \!+\! (1\!-\!\n\!-\!(1\!-\!\n))\!\sum_{j\neq i} \!\!\left( \prod_{k \neq i,j} \upnu_k \!\!\right)}{\displaystyle \n\left( \n\prod_{j}{\upnu_j} - \sum_{j} \left( \prod_{k\neq j}{\upnu_k}\right) \right)}\uprho_i(\upnu_i\!-\!1)\\
	\hspace{-20pt}&= 1 + \sum_{i=1}^\n \frac{\displaystyle (0)\prod_{j\neq i} \upnu_j \!+\! (0)\!\sum_{j\neq i} \!\!\left( \prod_{k \neq i,j} \upnu_k \!\!\right)}{\displaystyle \n\left( \n\prod_{j}{\upnu_j} - \sum_{j} \left( \prod_{k\neq j}{\upnu_k}\right) \right)}\uprho_i(\upnu_i\!-\!1)\\
	\hspace{-20pt}&= 1
\end{align*} 
}
which shows that $\dot{x}$ indeed vanishes at the point given by \eqref{well_mixed_equilibrium}. Next we need to show that $\dot{y}_i$ also vanishes at \eqref{well_mixed_equilibrium}. Since we have just shown that $\xbar + \sum \ybar_i = 1$, we can use this relation to simplify \eqref{well_mixed_yi}, prior to substitution of \eqref{well_mixed_equilibrium}. We write \eqref{well_mixed_yi} as follows
\begin{align*}
	\dot{y}_i &=  \b_i \left( (1-\upnu_i) (x-\uprho_i)  - \upnu_i y_i + \frac{\upnu_i}{\n-1}\sum_{j\neq i}  y_j \right)\\
	&= \b_i \left( (1-\upnu_i) (x-\uprho_i) - \frac{\upnu_i}{\n-1}y_i - \upnu_i y_i + \frac{\upnu_i}{\n-1}\sum_{j=1}^\n  y_j \right)\\
	&= \b_i \left( (1-\upnu_i) (x-\uprho_i) - \frac{\n \upnu_i}{\n-1}y_i + \frac{\upnu_i}{\n-1}\sum_{j=1}^\n  y_j \right)
\end{align*}
While evaluating at \eqref{well_mixed_equilibrium}, since $\xbar + \sum \ybar_i = 1$, we can substitute $\sum_{j=1}^\n \ybar_i = 1-\xbar$. Making this substitution, we get
\begin{align*}
	\left.\dot{y}_i \right|_{(\xbar,\ybar_i)}&=  \b_i \left( (1-\upnu_i) (\xbar-\uprho_i) - \frac{\n \upnu_i}{\n-1}\ybar_i + \frac{\upnu_i}{\n-1}(1-\xbar) \right) 
\end{align*}
rearranging the terms gives us
\begin{align*}
	\left.\dot{y}_i \right|_{(\xbar,\ybar_i)}&=  \b_i \left( \left( 1 - \frac{\n \upnu_i}{\n-1} \right)\xbar - \frac{\n\upnu_i}{\n-1}\ybar_i - (1-\upnu_i)\uprho_i + \frac{\upnu_i}{\n-1} \right) 
\end{align*}
now substituting $\xbar$ and $\ybar_i$ from \eqref{well_mixed_equilibrium}
{
\setlength{\mathindent}{0cm}
\fontsize{9.5}{11.5}\selectfont
\begin{align*}
	\left.\dot{y}_i \right|_{(\xbar,\ybar_i)} &= \frac{\displaystyle \left( 1 - \frac{\n \upnu_i}{\n-1} \right) \sum_{j}\left(\uprho_j(\upnu_j-1) \prod_{k\neq j} \upnu_k\right)}{\displaystyle \n\prod_{j}{\upnu_j} - \sum_{j} \left( \prod_{k\neq j}{\upnu_k}\right) } - \frac{\n\upnu_i}{\n-1}\left(\frac{1}{\n}\right) + \frac{\displaystyle \uprho_i\upnu_i(\upnu_i-1) \left(\sum_{j\neq i} \left( \prod_{k \neq i,j} \upnu_k \right)\right)}{\displaystyle  \n\prod_{j}{\upnu_j} - \sum_{j} \left( \prod_{k\neq j}{\upnu_k} \right)}\\
	&- \frac{\displaystyle \uprho_i\upnu_i(\upnu_i-1) \left( \n(\n-2)\prod_{j\neq i} \upnu_j \right)}{\displaystyle (\n-1)\left( \n\prod_{j}{\upnu_j} - \sum_{j} \left( \prod_{k\neq j}{\upnu_k}\right) \right)} +  \frac{\displaystyle \upnu_i(1-\n+\n \upnu_i)\sum_{j\neq i}\left( \uprho_j (\upnu_j-1)\prod_{k\neq i,j}\upnu_k \right)}{\displaystyle (\n-1)\left(\n\prod_{j}{\upnu_j} - \sum_{j} \left( \prod_{k\neq j}{\upnu_k}\right)\right)} \\
	& - (1-\upnu_i)\uprho_i + \frac{\upnu_i}{\n-1}
\end{align*}}
simplifying
{
\allowdisplaybreaks
\setlength{\mathindent}{0cm}
\fontsize{9.5}{11.5}\selectfont
\begin{align*}
\left.\dot{y}_i \right|_{(\xbar,\ybar_i)}&\begin{aligned}[t]
	= &\frac{\displaystyle \left( 1 - \frac{\n \upnu_i}{\n-1} \right) \sum_{j}\left(\uprho_j(\upnu_j-1) \prod_{k\neq j} \upnu_k\right)}{\displaystyle \n\prod_{j}{\upnu_j} - \sum_{j} \left( \prod_{k\neq j}{\upnu_k}\right) } + \frac{\displaystyle \uprho_i(\upnu_i-1) \left(\sum_{j\neq i} \left( \prod_{k \neq j} \upnu_k \right)\right)}{\displaystyle  \n\prod_{j}{\upnu_j} - \sum_{j} \left( \prod_{k\neq j}{\upnu_k} \right)}\\
	&- \frac{\displaystyle \uprho_i(\upnu_i-1) \left( \n(\n-2)\prod_{j} \upnu_j \right)}{\displaystyle (\n-1)\left( \n\prod_{j}{\upnu_j} - \sum_{j} \left( \prod_{k\neq j}{\upnu_k}\right) \right)} +  \frac{\displaystyle(1-\n+\n \upnu_i)\sum_{j\neq i}\left( \uprho_j (\upnu_j-1)\prod_{k\neq j}\upnu_k \right)}{\displaystyle (\n-1)\left(\n\prod_{j}{\upnu_j} - \sum_{j} \left( \prod_{k\neq j}{\upnu_k}\right)\right)} \\
	& - (1-\upnu_i)\uprho_i 
\end{aligned}\\
&\begin{aligned}
	= &\frac{\displaystyle \left( \n-1 - \n\upnu_i \right) \sum_{j\neq i}\left(\uprho_j(\upnu_j-1) \prod_{k\neq j} \upnu_k\right)}{\displaystyle (\n-1)\left(\n\prod_{j}{\upnu_j} - \sum_{j} \left( \prod_{k\neq j}{\upnu_k}\right)\right) } +  \frac{\displaystyle(1-\n+\n \upnu_i)\sum_{j\neq i}\left( \uprho_j (\upnu_j-1)\prod_{k\neq j}\upnu_k \right)}{\displaystyle (\n-1)\left(\n\prod_{j}{\upnu_j} - \sum_{j} \left( \prod_{k\neq j}{\upnu_k}\right)\right)} \\
	& + \frac{\displaystyle \uprho_i(\upnu_i-1) \left(\sum_{j\neq i} \left( \prod_{k \neq j} \upnu_k \right)\right)}{\displaystyle  \n\prod_{j}{\upnu_j} - \sum_{j} \left( \prod_{k\neq j}{\upnu_k} \right)} - \frac{\displaystyle \uprho_i(\upnu_i-1) \left( \n(\n-2)\prod_{j} \upnu_j \right)}{\displaystyle (\n-1)\left( \n\prod_{j}{\upnu_j} - \sum_{j} \left( \prod_{k\neq j}{\upnu_k}\right) \right)} \\
	& + \frac{\displaystyle \uprho_i(\upnu_i-1)\left( \n-1 - \n\upnu_i \right) \prod_{k\neq i} \upnu_k}{\displaystyle (\n-1)\left(\n\prod_{j}{\upnu_j} - \sum_{j} \left( \prod_{k\neq j}{\upnu_k}\right)\right) } - (1-\upnu_i)\uprho_i 
\end{aligned}\\
&\begin{aligned}
	= &\frac{\displaystyle \left( \n-1 - \n\upnu_i - (\n-1+\n\upnu_i) \right) \sum_{j\neq i}\left(\uprho_j(\upnu_j-1) \prod_{k\neq j} \upnu_k\right)}{\displaystyle (\n-1)\left(\n\prod_{j}{\upnu_j} - \sum_{j} \left( \prod_{k\neq j}{\upnu_k}\right)\right) } - (1-\upnu_i)\uprho_i \\
	& + \frac{\displaystyle \uprho_i(\upnu_i-1) \left((\n-1)\sum_{j\neq i} \left( \prod_{k \neq j} \upnu_k \right) - \n(\n-2)\prod_{j} \upnu_j + ( \n-1 - \n\upnu_i) \prod_{k\neq i} \upnu_k \right)}{\displaystyle (\n-1)\left( \n\prod_{j}{\upnu_j} - \sum_{j} \left( \prod_{k\neq j}{\upnu_k}\right) \right)} 
\end{aligned}\\
&\begin{aligned}
	= &\frac{\displaystyle \left(0\right) \sum_{j\neq i}\left(\uprho_j(\upnu_j-1) \prod_{k\neq j} \upnu_k\right)}{\displaystyle (\n-1)\left(\n\prod_{j}{\upnu_j} - \sum_{j} \left( \prod_{k\neq j}{\upnu_k}\right)\right) } - (1-\upnu_i)\uprho_i \\
	& + \frac{\displaystyle \uprho_i(\upnu_i-1) \left((\n-1)\sum_{j\neq i} \left( \prod_{k \neq j} \upnu_k \right) - \n(\n-2)\prod_{j} \upnu_j + ( \n-1)\prod_{k\neq i} \upnu_k - \n\prod_{j} \upnu_j \right)}{\displaystyle (\n-1)\left( \n\prod_{j}{\upnu_j} - \sum_{j} \left( \prod_{k\neq j}{\upnu_k}\right) \right)} 
\end{aligned}\\
&\begin{aligned}
	= 0 - (1-\upnu_i)\uprho_i + \frac{\displaystyle \uprho_i(\upnu_i-1) \left((\n-1)\sum_{j} \left( \prod_{k \neq j} \upnu_k \right) - \n(\n-1)\prod_{j} \upnu_j \right)}{\displaystyle (\n-1)\left( \n\prod_{j}{\upnu_j} - \sum_{j} \left( \prod_{k\neq j}{\upnu_k}\right) \right)} 
\end{aligned}\\
&\begin{aligned}
	= &(\upnu_i-1)\uprho_i - \frac{\displaystyle \uprho_i(\upnu_i-1)(\n-1)\left( \n\prod_{j}{\upnu_j} - \sum_{j} \left( \prod_{k\neq j}{\upnu_k}\right) \right)}{\displaystyle (\n-1)\left( \n\prod_{j}{\upnu_j} - \sum_{j} \left( \prod_{k\neq j}{\upnu_k}\right) \right)} 
	\\&= (\upnu_i-1)\uprho_i - (\upnu_i-1)\uprho_i = 0
\end{aligned}\\
\end{align*}}
Thus both $\dot{x}$ and $\dot{y}_i$ vanish at \eqref{well_mixed_equilibrium}, which shows that it is indeed the equilibrium for system \eqref{well_mixed_system}. This completes the proof.
\end{proof}

\subsection{Star population}

Now we assume that the population is structured so that there exists a consumer $i_*$ that is connected to every other consumer, and that no other connections exist. Thus the network constitutes a star graph and the dynamics are given by
\begin{subequations}
\label{star_system}
\begin{align}
	\label{star_x}
	&\dot{x} = (1-x)x - \sum_{i=1}^{\n} y_i \, x \\
	\label{star_y1}
	&\dot{y}_1 = \b_1 \left( (1-\upnu_1) (x-\uprho_1) + \frac{\upnu_1}{\n-1}\sum_{j=1}^\n  y_j - y_1 \right)\\
	\label{star_yi}
	&\dot{y}_i = \b_i \left( (1-\upnu_i) (x-\uprho_i) + \upnu_i (y_1 - y_i) \right) \hspace{20pt} i \neq 1
\end{align}
\end{subequations}
where $ i\in \{1, \dots , \n\}$, and the indices are ordered so that $i_*=1$. We represent any possible equilibria by $(\xbar,\ybar_1, \dots, \ybar_\n)$.\\
Before presenting the equilibrium for system \ref{star_system}, we require the following intermediate results
\begin{lemma}
\label{lemma_star_x0_1}
	If (A1) holds, there exist two equilibria for system \eqref{star_system}. The first equilibrium $(\xbar,\ybar_1, \dots, \ybar_\n)$ is given by $(1,0,\dots, 0)$ and the second equilibrium is given by $(0,0,\dots,0)$.
\end{lemma}
\begin{proof}
	Substituting the assertion of $(A1)$ in \eqref{star_y1} and \eqref{star_yi}, and solving for $\dot{y}_i=0$ gives the following linear system
\begin{align*}
	&(\n-1)\ybar_1 - \sum_{j=2}^\n y_j = 0\\
	&\ybar_i - \ybar_1 = 0 \hspace{20pt} i \neq 1
\end{align*}
it is easy to see that $\ybar_1 = \dots \ybar_\n = 0$ is the only solution to the above equations. Substituting this in \eqref{star_x} and solving for $\dot{x}=0$ yields $\xbar = 1$ or $\xbar = 0$. Thus $(1,0,\dots, 0)$ and $(0,0,\dots,0)$ are the two equilibria for this case.
\end{proof}

\begin{lemma}
	If (A2) holds, there exist an infinite number of equilibria for system \eqref{star_system} only when $\uprho_1 = \uprho_2 = \dots = \uprho_\n = \uprho$ (where $\uprho \in [0,1]$). Otherwise no equilibrium exists.
\end{lemma}
\begin{proof}
	Substituting the assertion of $(A2)$ in \eqref{star_system} gives us the following system
\begin{align*}
	&\dot{x} = (1-x)x - \sum_{i=1}^{\n} y_i \, x \\
	&\dot{y}_i = \b_i \left( x-\uprho_i\right)
\end{align*}

Following the same argument as that presented in proof of Lemma \ref{lemma_well_mixed_A2}, we see that if  $\uprho_1 = \uprho_2 = \dots = \uprho_\n = \uprho$, then the steady state constitutes all points where  $\sum_{i=1}^\n \ybar_i = 1-\uprho$ with $\xbar = \uprho$. In the event that $\uprho = 0$ there also exists an additional equilibrium consisting of the single state $\xbar = \ybar_1 = \dots = \ybar_\n = 0$. This completes the proof.
\end{proof}

\begin{lemma}
\label{lemma_star_xbar0}
	In the event that none of (A1), (A2) or (A4) hold true, there can exist no equilibrium for system \eqref{star_system} with $\xbar = 0$.
\end{lemma}
\begin{proof}
	Assuming $\xbar=0$ and solving \eqref{star_y1} \& \eqref{star_yi} for $\dot{y}_i = 0$ leads to the following system of equations
\begin{align*}
	\left[\begin{array}{ccccc} -(\n-1) & 1 & 1 & \dots & 1 \\ -1 & 1 & 0 & \dots & 0 \\ -1 & 0 & 1 & \dots & 0 \\ \vdots & \vdots & \vdots & \ddots & \vdots \\ -1 & 0 & 0 & \dots & 1 \end{array} \right] \left[\begin{array}{c} \ybar_1 \\ \ybar_2 \\ \ybar_3 \\ \vdots \\ \ybar_\n \end{array} \right] = \left[ \begin{array}{c} \frac{-(\n-1)(1-\upnu_1)\uprho_1}{\upnu_1} \\ \frac{(1-\upnu_2)\uprho_2}{\upnu_2} \\ \frac{(1-\upnu_3)\uprho_3}{\upnu_3} \\ \vdots \\ \frac{(1-\upnu_\n)\uprho_\n}{\upnu_\n} \end{array} \right]
\end{align*}
It is easy to check that the rank of the coefficient matrix is equal to $\n-1$. Subtracting all subsequent rows from the first row, gives the following equivalent form of the equation
\begin{align*}
	\left[\begin{array}{ccccc} 0 & 0 & 0 & \dots & 0 \\ -1 & 1 & 0 & \dots & 0 \\ -1 & 0 & 1 & \dots & 0 \\ \vdots & \vdots & \vdots & \ddots & \vdots \\ -1 & 0 & 0 & \dots & 1 \end{array} \right] \left[\begin{array}{c} \ybar_1 \\ \ybar_2 \\ \ybar_3 \\ \vdots \\ \ybar_\n \end{array} \right] = \left[ \begin{array}{c} \frac{-(\n-1)(1-\upnu_1)\uprho_1}{\upnu_1} - \sum_{j=2}^\n \frac{(1-\upnu_j)\uprho_j}{\upnu_j} \\ \frac{(1-\upnu_2)\uprho_2}{\upnu_2} \\ \frac{(1-\upnu_3)\uprho_3}{\upnu_3} \\ \vdots \\ \frac{(1-\upnu_\n)\uprho_\n}{\upnu_\n} \end{array} \right]
\end{align*}
which means that a solution can exist only if $\frac{(\n-1)(1-\upnu_1)\uprho_1}{\upnu_1} + \sum_{j=2}^\n \frac{(1-\upnu_j)\uprho_j}{\upnu_j}=0$. Since none of (A1),(A2) or (A4) are true,  $\frac{(1-\upnu_i)\uprho_i}{\upnu_i}>0 \,\, \forall\,\, i$ and so, $\frac{(\n-1)(1-\upnu_1)\uprho_1}{\upnu_1} + \sum_{j=2}^\n \frac{(1-\upnu_j)\uprho_j}{\upnu_j}\neq 0$. Thus no equilibrium exists in this scenario.
\end{proof}

\begin{lemma}
\label{lemma_star_equilibrium}
	If none of (A1),(A2) or (A4) hold true, the equilibrium for system \eqref{star_system} is given by
{\allowdisplaybreaks
\setlength{\mathindent}{0cm}
\fontsize{9.5}{11.5}\selectfont
\begin{subequations}
\label{star_equilibrium}
\begin{align}
\label{star_equilibrium_x}
	&\xbar = \frac{\displaystyle (\n-1)\uprho_1(\upnu_1-1) \prod_{j\neq 1} \upnu_j + \sum_{i}\left(\uprho_i(\upnu_i-1) \prod_{j\neq i} \upnu_j\right)}{\displaystyle 2(\n-1)\prod_{j}{\upnu_j} - \sum_{j\neq 1} \left( \prod_{k\neq j}{\upnu_k}\right) - (\n-1)\prod_{k\neq 1}{\upnu_k} }\\
\label{star_equilibrium_y1}
	&\ybar_1 = \frac{1}{\n} + \frac{\displaystyle (\n\!-\!1)\uprho_1(\upnu_1\!-\!1)\left( \!-\!\sum_{j\neq 1} \left( \prod_{k\neq 1,j} \upnu_k \right) \!+\! (\n\!-\!2) \prod_{k\neq 1} \upnu_k \right) \!-\! (1\!-\!\n\!+\!\n \upnu_1)\sum_{j\neq 1}\left( \uprho_j (\upnu_j\!-\!1)\prod_{k\neq 1,j}\upnu_k \right)}{\n\left(\displaystyle 2(\n-1)\prod_{j}{\upnu_j} - \sum_{j\neq 1} \left( \prod_{k\neq j}{\upnu_k}\right) - (\n-1)\prod_{k\neq 1}{\upnu_k}\right)}\\
	&\begin{aligned}\ybar_i = \frac{1}{\n} - &\frac{ \displaystyle (\n-1)(-1+\upnu_1)\left( \sum_{j\neq i} \left( \prod_{k\neq 1, k \neq j} \upnu_k \right) - (\n-1)\prod_{k\neq 1, k \neq i} \upnu_k + 2 \prod_{j \neq 1} \upnu_k \right)\uprho_1}{\n\left(\displaystyle 2(\n-1)\prod_{j}{\upnu_j} - \sum_{j\neq 1} \left( \prod_{k\neq j}{\upnu_k}\right) - (\n-1)\prod_{k\neq 1}{\upnu_k}\right)} \\+ &\frac{ \displaystyle (-1+\upnu_i)\left( -\n \sum_{j\neq 1} \left( \prod_{k \neq i,j} \upnu_k \right) - {(\n-1)}^2 \prod_{k\neq i,1} \upnu_k + 2\n(\n-2)\prod_{k \neq i} \upnu_k \right) \uprho_i}{\n\left(\displaystyle 2(\n-1)\prod_{j}{\upnu_j} - \sum_{j\neq 1} \left( \prod_{k\neq j}{\upnu_k}\right) - (\n-1)\prod_{k\neq 1}{\upnu_k}\right)} \\ - &\frac{ \displaystyle \sum_{j\neq i, j \neq 1} (-1+\upnu_j) \left( -\n \prod_{k \neq i,j} \upnu_k - (\n-1) \prod_{k \neq 1,j} \upnu_k + 2\n \prod_{k \neq j} \upnu_k \right) \uprho_j}{\n\left(\displaystyle 2(\n-1)\prod_{j}{\upnu_j} - \sum_{j\neq 1} \left( \prod_{k\neq j}{\upnu_k}\right) - (\n-1)\prod_{k\neq 1}{\upnu_k}\right)}\end{aligned}
\end{align}
\end{subequations}
}
where $i \in \{2,\dots,\n\}$ and $j,k \in \{1,\dots,\n\}$. 
\end{lemma}
\begin{proof}
	The proof follows by substituting the expressions of $\xbar$, $\ybar_1$ and $\ybar_i$, $i \in \{2,\dots,\n\}$ from \eqref{star_equilibrium} in \eqref{star_system}. Following the same steps as in Lemma \ref{lemma_well_mixed_equilibrium}, it can be found that the derivatives $\dot{x}$, $\dot{y}_1$ and $\dot{y}_i$ $i \in \{2,\dots,\n\}$ indeed equate to zero. Since none of (A1),(A2) or (A4) are true, \eqref{star_equilibrium_x} is the unique equilibrium for system \eqref{star_system}. 
\end{proof}

\section{System Dynamics and the Role of the Laplacian}
Consider the original system of the consumer network reproduced as follows from \eqref{eq:ses}
\begin{align}
\label{eq:ses10}%socioecological system
\begin{split}
	&\dot{x} =  (1-x)x - x\sum_{i=1}^{\n}  y_i,\\
	&\dot{y}_i = \b_i \Big( \upalpha_i(x -\uprho_i)- \upnu_i \sum_{j=1}^{\n} \w_{ij}\left( y_i - y_j \right) \Big).
\end{split}
\end{align}
Recall from Chapter \ref{chap:lump} that the influence that $j$ has on $i$'s consumption is given by $\b_i \upnu_i \w_{ij}$. Let us represent this influence by $\upgamma_{ij} = \b_i \upnu_i \w_{ij}$. The system may now be written as
\begin{align}
\label{eq:ses20}%socioecological system
\begin{split}
	&\dot{x} =  (1-x)x - x\sum_{i=1}^{\n}  y_i,\\
	&\dot{y}_i = \b_i \upalpha_i(x -\uprho_i) - \sum_{j=1}^{\n} \upgamma_{ij}\left( y_i - y_j \right),
\end{split}
\end{align}
where the weights $\gamma_{ij}$ directly define the \emph{influence network} which is different from the social network that the $\w_{ij}$'s represent. The consumption dynamics can now be written in compact form as
{ \setlength{\mathindent}{0 cm}
\fontsize{10}{11.2}\selectfont
\begin{align}
	\label{eq:net_mat} %network matrix
	\begin{split}
	\left[\!\!\begin{array}{c} \dot{y}_1 \\ \vdots \\ \dot{y}_\n \end{array}\!\! \right] = \left[\!\! \begin{array}{c} \b_1 \upalpha_1 \\ \vdots \\ \b_\n \upalpha_\n \end{array}\! \!\right] x - \left[\! \!\begin{array}{cccc} d^-_1 & -\upgamma_{12} & \dots & -\upgamma_{1\n} \\ -\upgamma_{21} & d^-_2 & \dots & -\upgamma_{2\n} \\ \vdots & \vdots & \ddots & \vdots \\ -\upgamma_{\n 1} & -\upgamma_{\n 2} & \dots & d^-_\n \end{array}\! \!\right] \left[\!\!\begin{array}{c} y_1 \\ \vdots \\ y_\n \end{array}\!\! \right] - \left[\!\! \begin{array}{c} \b_1 \upalpha_1 \uprho_1 \\ \vdots \\ \b_\n \upalpha_\n \uprho_\n \end{array}\!\! \right]
	\end{split}
\end{align}}
%\begin{align}
%	\label{eq:net_mat} %network matrix
%	\begin{split}
%	\left[\!\!\begin{array}{c} \dot{y}_1 \\ \vdots \\ \dot{y}_\n \end{array}\!\! \right] = \left[\!\! \begin{array}{c} \b_1 \upalpha_1 \\ \vdots \\ \b_\n \upalpha_\n \end{array}\! \!\right] x - \left[\!\! \begin{array}{ccc} \upnu_1 & & \\ & \ddots & \\ & & \upnu_\n \end{array}\! \!\right] \left[\! \!\begin{array}{cccc} d^-_1 & -\w_{12} & \dots & -\w_{1\n} \\ -\w_{21} & d^-_2 & \dots & -\w_{2\n} \\ \vdots & \vdots & \ddots & \vdots \\ -\w_{\n 1} & -\w_{\n 2} & \dots & d^-_\n \end{array}\! \!\right] \left[\!\!\begin{array}{c} y_1 \\ \vdots \\ y_\n \end{array}\!\! \right] - \left[\!\! \begin{array}{c} \b_1 \upalpha_1 \uprho_1 \\ \vdots \\ \b_\n \upalpha_\n \uprho_\n \end{array}\!\! \right]
%	\end{split}
%\end{align}}
where $d^-_i = \sum_{j=1}^\n \upgamma_{ij}$ is the in-degree of node $i$ in the influence network. We may represent this system as
\begin{align}
	\dot{\vec{y}} =  \vec{\mathrm{a}} \, x - \mathcal{L}\,\vec{y} - \vec{\mathrm{c}}
\end{align}
\begin{figure}[t!]
	\captionsetup{font=normal,width=0.8\textwidth}
%	\vspace{-25pt}
	\begin{center}
		\includegraphics[width=0.6\linewidth]{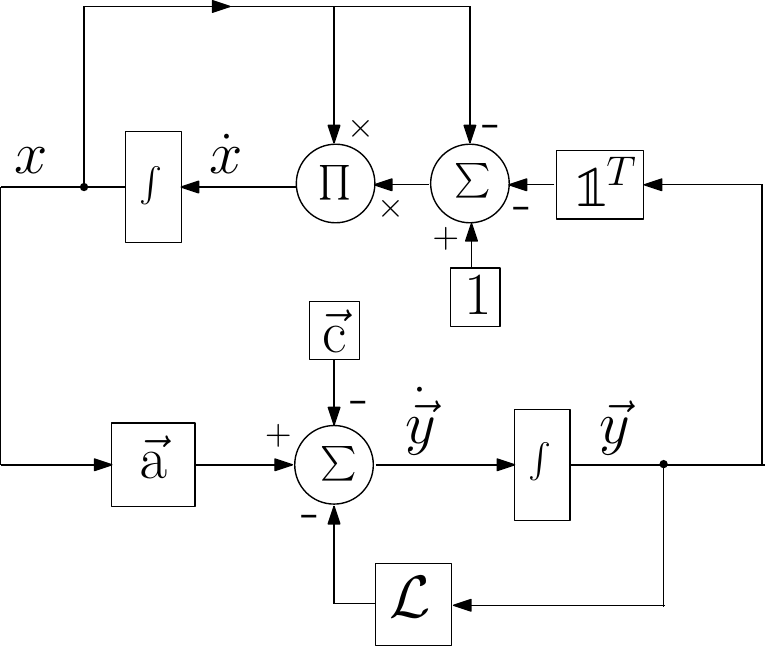}
	\end{center}
%	\vspace{-20pt}
	\caption{Block diagram showing the coupled resource and consumption dynamics.}
%	\vspace{-10pt}
	\label{fig:blk_dia} %block diagram
\end{figure}
where the definitions of the matrices can be seen from \eqref{eq:net_mat}. $\mathcal{L}$ is the \emph{Laplacian} matrix \cite{chung1997spectral}. The coupled resource and consumption dynamics can now be represented as follows
\begin{align}
	\label{cons_sys} %consumption system
\begin{split}
	&\dot{x} = f(x,\vec{y}) =  \left[ 1-x - \mathbbm{1}^T \vec{y} \right]x\\
	&\dot{\vec{y}} =  \vec{\mathrm{a}} \, x - \mathcal{L}\,\vec{y} - \vec{\mathrm{c}}
\end{split}
\end{align}
It is worth noting here that the matrix $\mathcal{L}$ contains all information about the network topology. This is a system in which the network topology (in the form of $\mathcal{L}$) appears explicitly in the dynamical model. Thus if there exists a relation between the poles of \eqref{cons_sys} and the eigenvalues of $\mathcal{L}$ then it may be translated into a relation between the network topology and overall behavior of the system. Appendix \ref{app:laplacian} gives a short demonstration of how network structure effects the eigenvalues of the Laplacian. 

%\begin{align}
%\label{eq:sys_res} %system response
%\begin{split}
%	C(t) &= \exp\left( (\mathcal{W}\mathcal{L}-\mathcal{B})t\right)C(0) + \int_0^t \exp\left( (\mathcal{W}\mathcal{L} -\mathcal{B})(t-\tau)\right)\mathcal{A}R(\tau) \, \text{d}\tau \\
%	&= \exp(-\mathcal{B}t)\exp(\mathcal{W}\mathcal{L}t)C(0) + \int_0^t \exp(-\mathcal{B}(t-\tau))\exp(\mathcal{W}\mathcal{L} (t-\tau)) \, \mathcal{A}R(\tau) \,\,\, \text{d}\tau
%\end{split}
%\end{align}

%Thus the stability of the system is dictated by the eigenvalues of $\mathcal{L}$. All information of the network is encoded in this matrix and so it can be used to relate the structure of the network to the behavior of the system.

\bookmarksetup{startatroot}

\chapter{Epilogue}

\label{chap:conc}

This chapter is intended as a commentary on the work that has been presented in this thesis. More specifically, the material in this chapter reinforces the arguments of the first part of the dissertation, in hindsight of the more technical subject matter presented thereafter. After revisiting the cybernetic ideology and its potential with respect to the control of socio-ecological systems, we provide a summary of conclusions for the dissertation. We then explore future avenues to expand the work and also contemplate the prospects offered by a few systems-theoretic tools developed by the controls faction of the engineering community.

\section{Revisiting the Cybernetic Picture}

Norbert Weiner, through his discourse \cite{wiener1961cybernetics, wiener1954human}, intended for cybernetics to embrace universal principles, independent of system domain, that were applicable to both man-made and living systems \cite{skyttner2005general}. He also succeeded in demonstrating that the cybernetic principles could be applied fruitfully, at least in theory, to all systems, independent of context. Unfortunately, Weiner's time came and went, and slowly, the face of cybernetics was lost under its mainstream descendants of control and communication. Indeed in the present era, as Thomas Rid puts it \cite{rid2016rise}, the prefix ``cyber" is commonly slapped without comprehension, in front of other terms (as in ``cyberspace" or ``cyberwar" ) to make them sound more technologically intensive, more cutting edge, more compelling -- and sometimes more ironic.

The principles of cybernetics are now becoming increasingly relevant due to the ever-growing involvement of human beings in technological applications. While the population bomb is constantly testing the sustaining boundaries of our planet, mankind has come to realize the intricate relationship between our economic system, the environment and technology. One cannot be changed independently of the others and so, an integrative approach is required to seize the holy grail of sustainable development. In this dissertation, we have argued that the principles of cybernetics may provide the necessary language for such an approach, at least for the effective governance of our natural resources, an absolute necessity for sustainability. 

Before framing the resource governance problem, it is important to comprehend the criticality of the issue and to realize that it is possible to address it in a distributed and scalable manner. This has been our focus in Chapter \ref{chap:res_gov}.  In Chapters \ref{chap:sa} and \ref{chap:cyber} we discussed the relevance of the cybernetic approach to natural resource governance from two different perspectives. Firstly, scientists already working on global issues in coupled human and natural systems have employed the systems theoretical approach in various forms with different levels of abstraction and in varying contexts. However there is much potential in employing the cybernetic tradition which is still largely untapped, especially due to the prospects that technological solutions now offer regarding various environmental challenges of global concern. Secondly, engineers have actively been applying the cybernetic principles in the form of feedback control and communication to a diverse range of applications which have recently begun to take on more human-centric implementations. The age of ``smarts" is quickly being ushered in, with the dream of smart cities and the internet of everything seemingly within grasp. We argue that the cybernetic method offers an integrative language in which the engineers are able to view complex CHANS through the same lens that they are so accustomed to view their machines through, and thus account for human and environmental factors more effectively.

\begin{figure}[t!]
	\captionsetup{font=normal,width=0.9\textwidth}
	\begin{center}
		\includegraphics[width=\linewidth]{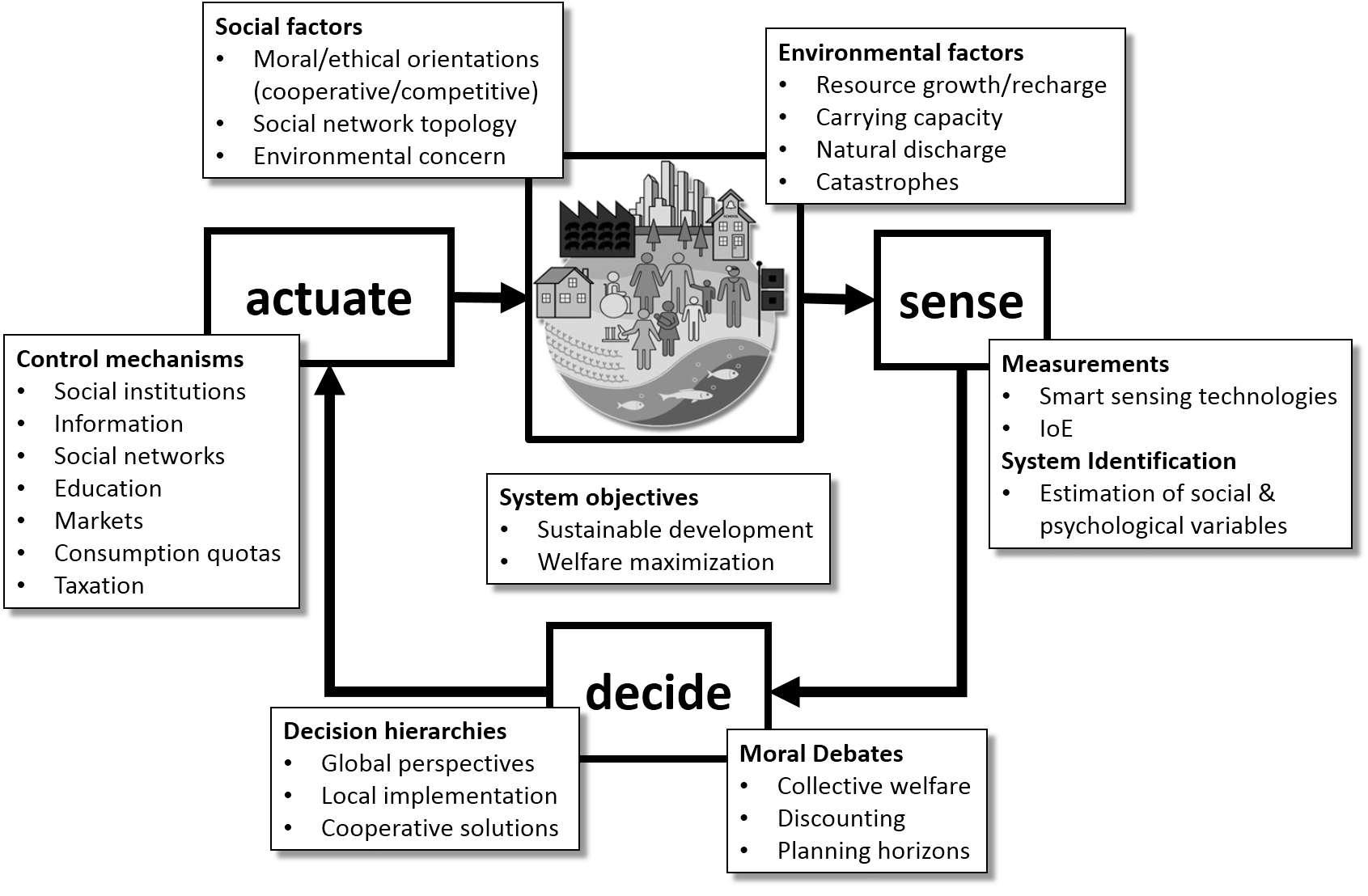}
	\end{center}
	\caption{The cybernetic picture revisited. This figure has first appeared in Chapter \ref{chap:cyber} and has been reproduced here for the sake of readability.} 
	\label{fig:cybernet_pic2} 
\end{figure}

The idea of viewing natural systems as machines is not new. Wescoat for instance, describes in \cite{wescoat2013reconstructing} how the British engineers took inspiration from the steam engine while laying out the design of the irrigation system in colonial India. However, while the physical principles of matter and energy have largely guided this ``mechanistic" world view \cite{hammond2005philosophical} in prior times, the cybernetic world view includes information as a new fundamental pillar, the relevance of which has grown immensely in recent times. Although we have not talked explicitly about the role of information in cybernetic systems, its immense significance is undeniable \cite[Part III]{mobus2015principles} and has perhaps already been grasped by the reader. Consider the control loop depicted in Figure \ref{fig:cybernet_pic2}. Each arrow in the figure represents a flow of information in one sense or another. The sensors send information about the state of the system to the ``controller" or decision maker. Based on the information received from the sensors and the ultimate objective of the system, the decision maker then sends information about the correcting action to the actuator. The action of the actuator and the stimulus received by the sensor also have interpretations as information being conveyed to and received from the system respectively (see \cite[Chapters III, VII]{wiener1961cybernetics} and \cite{corning2007control}).

It is also worth noting that the concepts of cybernetics are not new for the ecological and social sciences communities. In 1979 James Lovelock introduced the Gaia\footnote{The term Gaia is a reference to the Greek Earth Goddess.} hypothesis as a new perspective of the ecosphere \cite{lovelock1979new}. Gaia is envisioned as a complex entity that includes all of Earth's biosphere, atmosphere, oceans and soil. The aggregate of these components constitutes a complex feedback system whose sole objective is to sustain an environment that is optimal for life. Although the Gaia hypothesis has been met with much controversy \cite{schneider2004scientists} it has served as a precursor to the more rigorous discipline of geo-cybernetics \cite{schellnhuber1998geocybernetics}, that has much overlap with the Earth Systems Analysis framework referred to in Table \ref{tab:chans_fw}, and envisions control strategies at a planetary level. 
%Lovelock suggests many phenomena as evidence supporting the hypothesis. For instance, life first appeared on earth more than three billion years ago. It is agreed upon that the climate of the Earth in general has changed little eversince. However, the temperature of the Sun has changed significantly which indicates the presence of some regulating mechanism that keeps the climate at a level that is hospitable to life.

On the other hand, the cybernetics ideology has also seeped into the social sciences, giving birth to the discipline of socio-cybernetics \cite{mancilla2013introduction} that is concerned with the application of cybernetic models to the social and human sciences. Anton \cite{anton2015sociocybernetics} identifies over 15 research directions for socio cybernetics which include the cybernetics of education, social accountability, cyberspace and the internet, social networks, democratization, management of complex firms, trust and problem solving. In particular, the application of cybernetics to management and organization goes back more than half a century \cite{beer1960cybernetics} and management cybernetics is, by now a well studied and established domain.

Even though cybernetics has been applied separately to systems originating from both ecological and social sciences, natural resource systems pose a unique challenge for the application of cybernetic models as they span across both social and ecological disciplines. The addition of a third aspect in the form of technology and its recent accentuation in CHANS makes the challenge extremely relevant to the controls community in particular. As we have seen in the content of the dissertation, the application of cybernetics to NRM involves the study of human psychology, the exercise of ecological modeling and an understanding of social interactions, all of which are beyond the conventional domain of the control sciences. Although we have argued in support of the relevance of cybernetic methods to NRM for the social and ecological communities, the main focus of the dissertation is geared towards the engineering community.  We have argued that the technical systems envisioned for social and environmental control cannot be successfully realized if the human aspect is not properly accounted for. We hope that through this dissertation we have been able to depict, at least to some extent, how this may be achieved.

\section{Summary of Main Findings}
Here we present a summary of the main findings for the technical part of the dissertation. The language of cybernetics translates all processes into mathematical models. Whether the models are conceptual or data-driven, they are always understood to lie at the heart of the regulating mechanism \cite{conant1970every}. In Chapter \ref{chap:model}, we presented a mathematical model of consumer decision making, developed from social psychological research in consumer behavior. The model is agent-based with $n$ equations with $n$ consumers and one additional equation for the resource dynamics. Understandably, the dimension of such a system makes it less tractable, especially due to the non-linearity of the resource dynamics, for large consuming populations. In Chapter \ref{chap:lump} we addressed this problem by introducing block-models of resource consumption where the dynamics describe not the consumption of individuals, but the aggregate consumption of tightly knitted communities with uniform characteristics. We found under certain structural assumptions, that the lumped parameter model has the same form as the original model and thus in the subsequent analysis, we are able to understand a consuming agent as either comprising a single individual, or groups of individuals with homogeneous characteristics. Referring to the cybernetic picture of Figure \ref{fig:cybernet_pic2}, this pertains to understanding the process, how its dynamics are dependent on the various social and ecological factors and the conditions under which the selected model of is valid.

Chapters \ref{chap:open} -- \ref{chap:n} presented various applications and extensions of the developed model. Chapter \ref{chap:open} discussed the open-loop characteristics of the system at different levels of abstraction. We saw how a mathematically straight-forward exercise such as steady-state analysis was able to yield profound insights on phenomena such as free-riding behavior. We found that high preference to social information and homogeneity in the levels of environmental concern promoted self-reliant behavior. Where there was asymmetry in the environmentalism, the relatively environmental group showed a tendency to subsidize the consumption of the other group. This tendency decreased with increasing levels of preference to social behavior. We then studied the society as an aggregated single entity in Chapter \ref{chap:optimal} and studied favorable consumption patterns via the application of optimal control theory. This required us to define a mathematical notion of collective welfare and sustainable development, which may require considerable debate while interpreting the results for real-world settings (refer to the ``decide" block in Figure \ref{fig:cybernet_pic2}). We found the optimal consumption path to be sustainable if the discounting of future consumption was below a certain level, which corresponds to a long-term planning policy in a real-world setting. We also devised the optimal feedback law in both the sustainable and unsustainable cases. In order to analyze the interactions between groups and how they might strategically be exploited by the consuming agents, we analyzed the most basic network (the diad) in a game-theoretic framework in Chapter \ref{chap:game}. The positive role of environmental concern and preference to social information was explored along with the mixed role of heterogeneity in the consuming society, prompting for a deeper categorization of heterogeneity beyond conventional classifications. We developed a notion of tragedy in the consumption game and saw that high preference to social information and high levels of environmentalism were able to reduce the game ``tragicness".  In Chapter \ref{chap:learning}, we returned to the single and dual agent networks of Chapters \ref{chap:optimal} \& \ref{chap:game} and discussed the effects of feedback, via the optimal feedback law for the single agent case, and via the theory of learning in games for the dual agent case where we saw how the selected learning scheme discouraged free-riding behavior in the long-run. We also saw in the formulated scenario that the agents were able to reach the Nash equilibrium via simple best response dynamics. Finally, Chapter \ref{chap:n} intended as a possible extension, included calculations for the equilibrium of the full $\n$ agent network for two different types of structured populations, namely the complete network and the star network. In the end it was discussed that the topology of the network may possibly play a significant role in determining the overall behavior of the system via the Laplacian. If such an exercise is undertaken, it may give an indication as to what type of network structures are ``sustainable" in our model of natural resource consumption.

\section{Open Problems and Future Avenues}
Here we describe a few possibilities for further research. Note that there exist some natural extensions to the work of this dissertation that have already been discussed in the preceding chapters. These mostly pertain to the relaxation of assumptions and the consideration of additional effects both from within and outside the system. However the discussion below is related to certain aspects not covered in the dissertation, that we feel must be covered in the envisioned cybernetic foundation for natural resource governance. 

\subsection{A Rigorous Definition of Sustainability}
Given the state of the current landscape in sustainability science, aspiring for a rigorous, agreed-upon definition of sustainability is perhaps too optimistic at this point. However consciously or unconsciously, it remains the underlying goal of any management scheme for natural resource systems. The definition of sustainable development given by the Brundtland Commission \cite{brundtland1987report} is as follows : ``sustainable development is development that meets the needs of the present without compromising the ability of future generations to meet their own needs". While this definition of sustainability has been agreed upon and adopted by many nations, a clear, rigorous definition is still not evident. The different perspectives on sustainability coming from various disciplines further aggravates the issue. For instance, in biology sustainability means avoiding extinction and learning to survive and reproduce. In economics sustainability means avoiding major disruptions and collapses, hedging against instabilities and discontinuities. In ecology sustainability describes the capacity of ecosystems to maintain their essential functions and processes, and retain their biodiversity in full measure over the long-term. Costanza and Patten \cite{costanza1995defining} attempt to decipher the common aspects of these various definitions and argue that a sustainable system is essentially \emph{a system that persists}. 

Unfortunately, the definition of a sustainable system as one that persists poses more questions than it answers. The foremost of them being: what characteristics of the system must persist, and for how long? It is easy for control engineers to confuse sustainability with stability. However, a stable system is not necessarily a sustainable one, at least on some scale. The evolution of systems dictates that in order for the system to survive as a whole, smaller, sub-optimal parts must adapt or vanquish. The cybernetic framework relies much heavily on mathematical models, and any notion of sustainability applied within the framework must be mathematically rigorous as well. There have already been some attempts in this regard (see for instance \cite{chichilnisky1996axiomatic}, \cite{cabezas2002towards}, \cite{valente2005sustainable}), however they are mostly domain specific and have not qualified as definitions that are agreed upon across disciplines. Indeed we admit that the definition of sustainability provided in Chapter \ref{chap:optimal} of this dissertation also falls within this category. Nevertheless, the cybernetically inclined  still yearn for a ``sustainability function" similar to the Lyapunov functions that they rely most heavily on for characterizing long-term behavior of their systems. 

\subsection{Formulating Social Variables in the Context of Feedback Control}
During the discourse of the dissertation, we have largely overlooked the ``sense" and ``actuate" blocks of the control loop from Figure \ref{fig:cybernet_pic2}. However these blocks hold immense importance in any realization of feedback control. Before any mathematically intricate issues are brought up, let us consider the following: how do we sense? and how do we actuate? In Figure \ref{fig:cybernet_pic2} we have listed several mechanisms by which control may be exerted on the consuming population of a natural resource. While it is certainly hard to disagree that mechanisms such as education, taxation and information dissemination can indeed influence consumer behavior, the more relevant issue to us right now is how to incorporate these strategies in a formal control law. For instance, what would education look like as a control variable in system \eqref{eq:ses}? What would be the mathematical nature of that variable, what would its units be and how would we be able to justifiably associate any numerics with it? Clearly we are nowhere near to furnishing an answer to these questions than we are to reaching a quantifiable notion of sustainable development. It seems befitting that Stankovic \cite{stankovic2014research} lists the incorporation of social models into the formal methodology of feedback control as one of the major challenges for the control of CPS.

Let us now consider the ``sense" block of Figure \ref{fig:cybernet_pic2}. In all the mathematical models developed in the dissertation, it has conveniently been assumed that the state of the system is accessible for the decision maker (whether the decision making happens in a centralized or distributed manner). Sensing technologies have greatly evolved over time \cite{soloman2009sensors}, with smart sensing technologies offering opportunities for control that were not thought to be possible a few decades ago. However, we are still not quite there yet when it comes to sensing cognitive variables and social influences at least for the particular control loop envisioned in Figure \ref{fig:cybernet_pic2}. Not only this, but formulation of the mathematical nature of the variables to be sensed would pose the same challenges that have been discussed above regarding the control variables.

On the other hand, implementing the ``decide" block associated with the decision maker requires intense debates of moral and ethical nature. The difficulty of framing the objective of sustainability has already been discussed in the previous section. Since these issues have been relevant long before the recent advent of control and sensing technology, they have at least been thoroughly discussed for a long time now by the political scientists and economists \cite{perman2003natural}. 

\subsection{The Issues of Controllability and Observability}
%Controllability and observability \cite{franklin1994feedback} are two major concepts of control science, introduced by R.E. Kalman \cite{kalman1960contributions} in 1960. While we have not explicitly mentioned it in the technical content in the thesis, controllability is, in some sense, related to the notion of optimality defined in Chapter \ref{chap:optimal} (the optimal solution must lie in the controllable subspace of the system). The basic theory of controllability and observability are well established both for linear and non-linear systems \cite{khalil1996noninear}. Since the mathematical definitions of both concepts are a slightly technical for the sake of reproducing here, they can roughly be defined as follows
%
%\emph{Controllability:} In order to attain whatever we wish to attain from a dynamical system under a control input, the system must be controllable.
%
%\emph{Observability:} In order to gain an accurate picture of the internal state of a dynamical system under observation, the system must be observable.
%
%Let us consider how these concepts relate to the dynamical system representing our natural resource system. 

In the previous section, we have deliberated over the formulation of the various control mechanisms listed in Figure \ref{fig:cybernet_pic2}. Assuming that this has been done successfully, the next question that arises: is it possible to steer the system to the required state using the selected control scheme? For instance, is educating the population enough to ensure sustainability or is there a need for increased taxation? In order to control the consumption pattern of the entire network, is there a need to control every consumer individually or is there a possibility of targeting only the most influential nodes? The set of states that are attainable using one control mechanism may be different than that of another, which can reveal the viability of one mechanism over the rest. This is the type of information that the controllability \cite{franklin1994feedback} of the system can reveal.

Observability \cite{franklin1994feedback} on the other hand relates to issues in sensing. Consider the synthesis of the feedback control law in Section \ref{sec:feed1}, which requires complete information about the state of the resource being available to the decision maker. Similarly, the the adjustment process from Section \ref{sec:learn} assumes that each consumer has complete knowledge about the other consumer's information preferences. Assuming that there existed a sensor to convey this information, an $n$-agent network would require knowledge of the characteristics of all $n$ agents. Installing a sensor for each consumer in large populations would not only be an extremely tedious exercise, but would also be extremely expensive both economically and computationally. In such a scenario, the tool of observability can help identify exactly which of the nodes need to be sensed/observed in order for the control law to serve its purpose effectively. It can also be used to explore the feasibility across different types of sensing mechanisms. 

The concepts of controllability and observability are not new \cite{kalman1960contributions} and there exist elegant results for linear systems, non-linear systems \cite{khalil1996noninear} and even for complex networks \cite{liu2011controllability, liu2013observability} which have the potential to yield important insights for our resource-consumer network.

\section{Final Remark}

In this chapter we have presented a concluding note to the work conducted as part of this dissertation by summarizing the ideology and main findings of the research along with a few considerations for further work. As a final remark we wish to emphasize that the analysis presented in this dissertation, is in no way meant to give a complete account of the cybernetics of natural resources. Indeed such an exposition would require several dissertations to say the least. Rather, it is intended to portray the potential that the cybernetic perspective holds for the resource governance problem, and global environmental issues in general. One must also bear in mind that at the cybernetic foundation of every synthesized regulating mechanism, lies a mathematical model. However, no model (including the one presented herein) is perfect, thus implying that the models and their conclusions are only as good as their underlying assumptions. 

It would appear from this dissertation (and almost all logical treatments of NRM for that matter) that sustainable management of our natural resources while simultaneously catering for the needs of a growing world population, is an extremely daunting task to say the least. Perhaps as science and human knowledge progress, mankind will be able to dig itself out of the current environmental crisis. For now, we must stay focused on doing our part and not lose faith in our planet. The frame of mind that should guide our environmental research is perhaps best summarized in the following note that Donella Meadows had taped on her office door \cite{meadows2001reflection} while she was still alive,

\finishchapquote{``Even if I knew the world would end tomorrow, I would plant a tree today"}{Anonymous}{ }
%
%
%\section{Limitations and Shortcomings}
%
%\section{Promising Directions for Future Work}

\cleardoublepage
\singlespacing

\addtocontents{toc}{\protect\addvspace{2.25em}}
\bookmarksetup{startatroot}
%% This defines the bibliography file (main.bib) and the bibliography style.
%% If you want to create a bibliography file by hand, change the contents of
%% this file to a `thebibliography' environment.  For more information 
%% see section 4.3 of the LaTeX manual.
\begin{singlespace}
\bibliography{main}
\bibliographystyle{plain}
\end{singlespace}

\addcontentsline{toc}{chapter}{Bibliography}

\cleardoublepage

\pdfbookmark[-1]{Appendices}{appendix}
\appendix
%\appendix
\chapter{Game Theory Primer}
\label{app:games}

%\section*{Games in Normal Form}

This appendix gives a very brief introduction to basic game theoretical concepts. It is in no way meant to serve as an introduction to game theory in general for which we refer the interested reader to some excellent texts on the subject \cite{osborne1994course, myerson2013game, bacsar1998dynamic}. The extent of the following material has been selected solely for the reader to be able to grasp the game theoretic terminology used in this dissertation.

\section{Definitions}
A game $\mathcal{G}$ is defined as a three-tuple $\mathcal{G}=\langle \mathcal{I},(\mathcal{S}_i ),(\uppi_i )\rangle$ where
\begin{itemize}
	\item $\mathcal{I} \coloneqq \{1, \dots, n\}$ is the set of \emph{players}, where $n$ is the total number of players.
	\item $\mathcal{S}_i$ is the \emph{strategy} space for player $i \in \mathcal{I}$. The elements of this set, also called the \emph{actions} may either be discrete or continuous in nature. An element from the set $\mathcal{S}_1 \times \dots \times \mathcal{S}_n$ is called a \emph{strategy profile}.
	\item $\uppi_i : \mathcal{S}_1 \times \dots \times \mathcal{S}_n \rightarrow \mathbbm{R}$ is the payoff function for player $i$.
\end{itemize}
Games are commonly represented as a table. In Figure \ref{fig:pris_dil_app} we reproduce the Prisoner's dilemma game from Chapter \ref{chap:res_gov}. Here there are two players so $n = 2$. The rows of the table correspond to the strategy of Player 1, whereas the columns correspond to the strategies of Player 2. In some writings the players are also referred to as the row player and the column player. Each player has two possible strategies and so, according to Figure~\ref{fig:pris_dil_app}, $\mathcal{S}_i = \{\mathrm{C}, \mathrm{D} \}$, where $\C$ and $\D$ represent the ``cooperate" and ``defect" strategy respectively\footnote{where the ``cooperate" strategy corresponds to pleading innocent}. The entries in the table depict the payoffs of each player as an ordered pair where the first entry gives the payoff of the first player and the second entry gives the payoff of the second player respectively, if each player plays the strategy given by that row and column of the table. Thus from Figure~\ref{fig:pris_dil_app}, if both the players cooperate, this corresponds to the upper left corner of the table, from which we see that in this case both the players receive payoff $\uppi_1(\C,\C) = \uppi_2(\C,\C) = -1$.
\begin{figure*}[!htb]
\begin{center}
	\centering
	\begin{tabular}[t]{ >{\centering}m{10pt} | >{\centering}m{30pt} | >{\centering}m{30pt} | m{0pt}}
		\multicolumn{1}{c}{ }& \multicolumn{1}{c}{ C} & \multicolumn{1}{c}{D}\\
		\cline{2-3}
		C & -1,-1 & -10,0 &\\[30pt]
		\cline{2-3}
		D & 0,-10 & \cellcolor{gray}-5,-5 &\\[30pt]
		\cline{2-3}
	\end{tabular}\\[10pt]
\end{center}
\caption{The prisoner's dilemma game}
\label{fig:pris_dil_app}
\end{figure*}
The story of the prisoner's dilemma game is usually narrated as follows -- both the players are in separate rooms of a police station where they are being questioned about a crime they have committed together. Since the police do not have enough evidence to evict them, each prisoner is given the following proposition. If the prisoner confesses to the crime and rats his partner out (the ``defect" strategy) but the other prisoner does not confess (the ``cooperate" strategy), then he would be freed but his partner would get a sentence of 10 years. However if his partner also confesses to the crime then both would serve a 5 year sentence. If neither of them confesses, then both of them will be evicted for a lighter crime and serve a sentence of 1 year. The payoffs in Figure \ref{fig:pris_dil_app} are based on the number of years the prisoners would have to serve for each possible outcome of the game. 

\section{Dominant Strategies}

Given a game in normal form, it is reasonable to expect certain strategies to be played more often than the others. These strategies are usually the ``dominant" strategies. A strategy is said to be dominant for player $i$ if it is guaranteed to result in the highest possible payoff no matter what strategy the other players choose. More precisely, a strategy $s'_i \in \mathcal{S}_i$ is weakly dominant if $\uppi(s'_i,s_{-i}) \geq \uppi(s_i, s_{-i})$ for all $s_i \in \mathcal{S}_i$ and $s_{-i} \in \mathcal{S}_{-i}$. $s'_i$ is then strongly dominant (or simply dominant) $\uppi(s'_i,s_{-i}) > \uppi(s_i, s_{-i})$ for all $s_i \in \mathcal{S}_i$ and $s_{-i} \in \mathcal{S}_{-i}$.

In the prisoner's dilemma, it is straight forward to see which strategy is the dominating one. If the column player plays $\C$, then the row player gets the highest payoff by playing $\D$. Now, if the column player plays $\D$, the row player will still ensure the highest payoff for himself by playing $\D$. In both cases, $\D$ gives the highest payoff and so, $\D$ is the dominant strategy for Player 1. Similar reasoning leads to the realization that $\D$ is also the dominant strategy for Player 2.

\section{Best Response}

A strategy $\tilde{s}_i \in \mathcal{S}_i$ is said to be the best response to a fixed strategy profile $s_{-i} \in \mathcal{S}_{-i}$, on behalf of the remaining players if $\uppi(\tilde{s}_i,s_{-i}) \geq \uppi(s_i, s_{-i})$ for all $s_i \in \mathcal{S}_i$. If the best response is unique then it is also called a strict best response.

Note that the best response is a slightly weaker concept than dominance. While a dominant strategy might not exist for a particular player, the best response of a player to a particular strategy of the other players always exists. A dominant strategy, if it exists, is then the common best response to \emph{all} strategies of the opposing players.

\section{Nash Equilibrium}

A strategy profile $s^\# = \in \mathcal{S}_1, \times, \mathcal{S}_n$ is called a Nash Equilibrium if $s^\#_i$ is a best response to $s^\#_{-i}$ for every $i \in \mathcal{I}$. More precisely, $s^\# = (s^\#_1, \dots, s^\#_n)$ constitutes a Nash Equilibrium if $\uppi(s^\#_i,s^\#_{-i}) \geq \uppi(s_i,s^\#{-i})$ for all $s_i \in \mathcal{S}_i$ and all $i \in \mathcal{I}$. While this definition appears to be similar to that of dominance, it is important to note the subtlety that the Nash Equilibrium requires, for every $i$, that $s^\#_i$ be the best response to only those actions of the other players that are being played at the equilibrium, and not to the whole set of possible actions.

The Nash Equilibrium also has associated with it a notion of stability. If all players play the Nash Equilibrium strategy, then no player can gain a higher payoff by unilaterally changing her strategy. The Nash Equilibrium for the Prisoner's dilemma game is given by $\{\D,\D\}$ and is highlighted in gray in Figure \ref{fig:pris_dil_app}. This highlights the ``dilemma" of the game. While the players can clearly get the collectively best outcome if both cooperate, rational reasoning leads them both to defect, resulting in a worse outcome for both. The strategies considered here are \emph{pure} strategies i.e., the player commits to playing a single strategy and plays it with certainty. There is also the possibility of a player assigning probabilities to each of the actions such that each strategy has a certain probability of being payed. For instance, in the prisoner's dilemma game, one player may decide to toss a coin and play $\C$ if head comes up and play $\D$ otherwise, thus assigning a probability of 0.5 to each action. The act of assigning probabilities to the strategies is also considered a valid strategy and is usually referred to as a \emph{mixed} strategy. It is known that the Nash equilibrium for finite games does not always exist in the form of pure strategies, however if mixed strategies are also considered, then the Nash Equilibrium always exists. 

\section{Games with Continuous Strategies: The Cournot Game}

The concepts defined for games with discrete strategies carry forward to games with continuous strategies. We illustrate this through the Cournot game whose story is described as follows. Imagine two firms, each producing an identical good for the same market. The two firms 1 \& 2 are the players and so $\mathcal{I} = \{1,2\}$. The strategy for each firm $i$ is the quantity of good $q_i$ it decides to produce and so $\mathcal{S}_i = \{q_i\}$ constitutes a continuous strategy space.    
\begin{figure}[t!]
	\captionsetup{font=normal,width=\textwidth}
%	\vspace{-25pt}
	\begin{center}
		\includegraphics[width=0.45\linewidth]{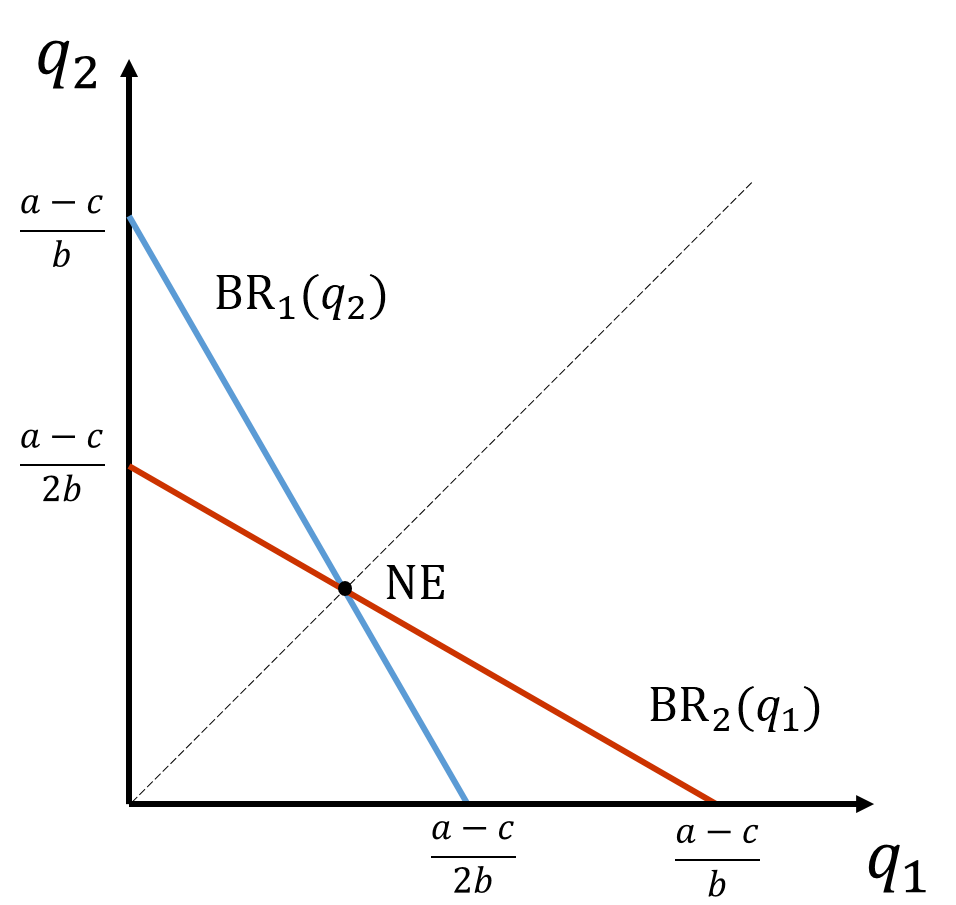}
	\end{center}
%	\vspace{-20pt}
	\caption{Illustration of the Cournot Game.}
%	\vspace{-10pt}
	\label{fig:cournot} %laplacian eigenvalues
\end{figure}
The payoffs for each firm are given by the profit each firm earns by selling the good. The cost to produce the good is given by $c>0$. The good is sold at the market-clearing price which is given by the inverse demand curve $p = a - b(q_1 + q_2$ where $p$ denotes the price and $a,b >0$ are constant parameters. The payoff for each firm is then given by
\begin{align*}
	\uppi_i(q_1,q_2) = \Big(a-b(q_1+q_2)\Big)\,q_i - c\,q_i \quad i \in \{1,2\},
\end{align*}
The best response is given as $\displaystyle \tilde{q}_i = \mathrm{BR}_i(q_j) = \max_{q_i} \uppi_i(q_1,q_2) \quad i \neq j$, whose solution gives us
\begin{align*}
	\tilde{q}_1 = \mathrm{BR}_1(q_2) = \frac{a-c}{2b} - \frac{q_2}{2},\quad \tilde{q}_2 = \mathrm{BR}_2(q_1) = \frac{a-c}{2b} - \frac{q_1}{2}.
\end{align*}
The Nash Equilibrium, $(q^\#1,q^\#_2)$ is the strategy profile such that $q^\#_1 = \mathrm{BR}_1(q^\#_2)$ and $q^\#_2 = \mathrm{BR}_2(q^\#_1)$. This is given by the solution to the following system
\begin{align*}
	q^\#_1 = \frac{a-c}{2b} - \frac{q^\#_2}{2},\quad 	q^\#_2 = \frac{a-c}{2b} - \frac{q^\#_1}{2}.
\end{align*}
Solving the above algebraically gives us the Nash Equilibrium
\begin{align*}
	q^\#_1 = q^\#_2 = \frac{a-c}{3b}.
\end{align*}
Figure \ref{fig:cournot} gives a graphical depiction of the Equilibrium in the Cournot game. It is possible to show that in this equilibrium of the Cournot game, the firms produce too much quantity of the good than they would need at the point of joint profit maximization. As such, the Cournot game is also used as a model for scenarios where individually optimum behavior does not result in the collectively optimal outcome.

The concept of the Nash Equilibrium is named after John Nash who gave the first existence proof for finite games. There also exist other equilibrium concepts that can be used to make descriptive or predictive reasoning about the outcome of the game. Moreover, there exist various classifications of games based on whether the strategies are played simultaneously or sequentially, whether the games are played just once or repeatedly, whether the players are allowed to form coalitions or not, whether or not the players have complete access to information on the parameters of the game, and so on. Regarding the totality of the game theoretic literature, the material discussed here is just a tiny tip of the proverbial iceberg, presented only with the intention of enabling the reader to grasp the contents of the dissertation more effectively.

\chapter{Basic Concepts in Economic Growth Theory}
\label{app:resecon}

	In this appendix we provide a small compilation of a few basic concepts frequently referred to in the resource economics and growth theory literature. The material is intended to serve as a small supplement for the reader to better comprehend the exposition of Chapter \ref{chap:optimal}.  

\section{Early Developments in Growth Theory}

Economic growth theory is separated into exogenous and endogenous growth. The major ingredients of an exogenous growth problem consist of a production function relating factors of production to output whose level determines the current state of the economy. This is coupled with the dynamics of the production factors usually as differential equations of capital and labor. It is assumed that all parameters of the model are fixed from some mechanism outside our system, thus the term exogenous growth. Equilibrium conditions are found and policy prescriptions are made on the base of measures that have a positive influence on steady state growth.

The first influential model of exogenous growth is the Harrod-Domar model \cite{harrod1939essay} introduced in 1939. This model assumed a linear production function with capital as the only input to production. It concludes that the amount of reinvestment of the capital determines the long term growth. Solow and Swan introduced the Solow-Swan model\cite{solow1956contribution} in 1956 as an extension to the Harrod-Domar model which relaxed the fixed capital to labor ratio assumption by taking labor as a separate input to production. Later on the Mankiw-Romer-Weil model \cite{mankiw1990contribution} incorporated human capital as a third input to production thus distinguishing between physical capital and human capital.

The starting point of endogenous growth theory is considered by most to be the seminal paper by Ramsey \cite{ramsey1928mathematical}. The potential of Ramsey's work was not realized until 1965 when Cass and Koopmans adapted it to form the Ramsey-Cass-Koopmans (RCK) model \cite{cass1965optimum}. In problems of endogenous growth, parameters assumed to be determined by outside forces are taken so as to maximize a certain criterion for social welfare represented by an objective function. The underlying machinery used in such problems is that of optimal control with Pontryagin's Maximum Principle \cite{aseev2007pontryagin} used as the basic optimization apparatus. Earlier in 1961 Phelps introduced the golden savings rule \cite{phelps1961golden} which was later found out to be a special case of the RCK model. 

Endogenous growth problems rely heavily on objective functions which specify aggregate and intertemporal preferences on individual utility. The normative debate on which preference scheme is ethically correct is the foundation of the field of welfare economics \cite{perman2003natural}. Despite criticism, discounting the utility of future generations is the most practiced scheme by growth theorists as it allows for the intertemporal objective to converge in the limit. In this regard, Chichilnisky developed an axiomatic criterion for optimal growth \cite{chichilnisky1996axiomatic} which later on was also used to develop the green golden rule \cite{beltratti1993sustainable} of resource consumption.

\section{A Compilation of Fundamental Notions}

\subsection{The production function and returns to scale}

The term returns to scale arises in the context of a firm's production function to describe what happens to the scale of production when all inputs to production are subject to a given change in scale. Formally, a production function $F(K,L)$ is defined to have
\begin{itemize}
	\item constant returns to scale if $F(a\,K, a\,L) = a\,F(K,L)$
	\item increasing returns to scale if $F(a\,K, a\,L) > a\,F(K,L)$
	\item decreasing returns to scale if $F(a\,K, a\,L) < a\,F(K,L)$
\end{itemize}
where $K$ and $L$ are factors of production, usually used to represent capital and labor respectively, and $a>0$ is some constant.

\subsection{The production function and technological progress}

A production function in which technological progress enters as labour-augmenting is defined as \emph{Harrod-neutral}
\begin{align*} 
	Y(t) = f\big(K(t), A(t)L(t)\big)
\end{align*}
where Y(t) is the economic output and A(t) represents the level of technology at time $t$. If technological progress is capital augmenting of the form
\begin{align*} 
	Y(t) = f\big(A(t)K(t), L(t)\big)
\end{align*}
it is defined as \emph{Solow-neutral}. Finally, if technological progress simply multiplies the production function by an increasing scale factor as
\begin{align*} 
	Y(t) = A(t)\, f\big(K(t), L(t)\big)
\end{align*}
it is defined as \emph{Hicks-neutral}

\subsection{A Note on Cobb-Douglass production functions}
The Cobb-Douglass production function is a particular form of production function which represents the technological relationship between a single good called the output, and different factors of productions called the inputs. It was developed and tested against statistical evidence by Charles Cobb and Paul Douglas during 1927-1947. In its most standard form, for production of a single good with two factors of production
\begin{align*} 
	Y = A\,L^{\beta}\,K^{\alpha}
\end{align*}
where $Y$ is the total output, $L$ is the labor input, $K$ is the capital input, $A$ is the total factor productivity and $\alpha, \beta$ are the output elasticities of capital and labor respectively. $\alpha$ and $\beta$ are constants which are determined by the current level of technology. Furthermore they control the returns to scale of the production as follows
\begin{itemize}
	\item If $\alpha+\beta=1$ then $Y$ has constant returns to scale
	\item If $\alpha + \beta > 1$ then $Y$ has increasing returns to scale
	\item If $\alpha+ \beta < 1$ then $Y$ has decreasing returns to scale
\end{itemize}

\subsection{Continuity, differentiability, positive and diminishing marginal products, and constant returns to scale}
The first common assumption on the production function in economic growth theory is as follows: the production function $F: \mathbb{R}^2_+ \rightarrow \mathbb{R}_+$ is twice differentiable in $K$ and $L$, and satisfies
\begin{align*}
	\frac{\partial F(K,L)}{\partial K} > 0,& \hspace{20pt} \frac{\partial F(K,L)}{\partial L} > 0\\
	\frac{\partial^2 F(K,L)}{\partial K^2} < 0,& \hspace{20pt} \frac{\partial^2 F(K,L)}{\partial L^2} < 0
\end{align*}
Moreover, $F$ exhibits constant returns to scale in $K$ and $L$.

\subsection{The Inada conditions}
The second common assumption on the production function in economic growth theory is a collection of conditions called the Inada conditions. They are given as follows: The production function $F: \mathbb{R}^2_+ \rightarrow \mathbb{R}_+$ satisfies \begin{align*}
	\lim_{K\rightarrow 0}\frac{\partial F(K,L)}{\partial K} = \infty,& \hspace{20pt} \lim_{K\rightarrow \infty}\frac{\partial F(K,L)}{\partial K} = 0 \hspace{20pt} \forall \,\, L>0\\
	\lim_{L\rightarrow 0}\frac{\partial F(K,L)}{\partial L} = \infty,& \hspace{20pt} \lim_{L\rightarrow \infty}\frac{\partial F(K,L)}{\partial L} = 0 \hspace{20pt} \forall \,\, K>0
\end{align*}

\subsection{Some formal examples of social welfare functions}
The Social Welfare Function (SWF) is a mechanism to aggregate the utilities of different individuals to determine which set of utilities are socially favorable. Assume that there are $N$ individuals in a society and denote the utility of each individual as $U_i$ where $i\,\,\in\,\{1, \dots , N\}$. Following this notation, some formal examples of SWFs are given below
\begin{itemize}
	\item Benthamite Utilitarianism

	{\centering
	 $ \displaystyle
	    	\begin{aligned} 
			W(U_1,\dots,U_N) = \sum_{i=1}^{N} a_i\,U_i \hspace{10pt} \text{where } a_i>0\,\,\forall i
    		\end{aligned}$
	\par}
	\item Egalitarian

	{\centering
	 $ \displaystyle
	    	\begin{aligned} 
			W(U_1,\dots,U_N) = \sum_{i=1}^{N} - \lambda \sum_{i=1}^{N} \left( U_i - \min_{i} U_i \right) \hspace{10pt} \text{where } 0<\lambda<1
    		\end{aligned}$
	\par}
	\item Rawlsian

	{\centering
	 $ \displaystyle
	    	\begin{aligned} 
			W(U_1,\dots,U_N) = \min_{i} U_i
    		\end{aligned}$
	\par}	
\end{itemize}
Kenneth Arrow argued that any ``reasonable" social welfare function must satisfy 6 basic axioms and then demonstrated via his \emph{impossibility theorem} \cite{maskin2014arrow} that no such function exists.

\subsection{The Brundtland Commission and sustainability}
The Brundtland Commission is formally known as World Commission on Environment and Development (WCED) and aims to unite countries in the pursuit of sustainable development. The commission's report of 1987 titled ``Our Common Future" \cite{brundtland1987report} had a great impact at the Earth Summit of 1992 in Rio de Janeiro. According to the report sustainability is defined as  
\begin{quotation}
``Sustainable development is development that meets the needs of the present without compromising the ability of future generations to meet their own needs"
\end{quotation}
This definition of sustainability includes two major concepts. First is the concept of ``needs", in particular the essential needs of the world's poor, to which overriding priority should be given. Second is the idea of limitations imposed by the state of technology and social organization on the environment's ability to meet present and future needs.

\chapter{Network topology and spectrum of the Laplacian matrix}

\label{app:laplacian}

In Chapter \ref{chap:n}, we mentioned the influence of the network topology on the eigenvalues of the \emph{Laplacian} matrix. There exist various formal results \cite{chung1997spectral, spielman2007spectral, gross2004handbook} on the interpretation of the spectrum of the Laplacian and its connection to the network structure. Here, without using any formal jargon, we give a brief illustration of how the two smallest eigenvalues of the Laplacian matrix change as the connectedness of a simple network is varied.

Consider a weighted directed graph $G$ on $n$ vertices. Each edge $(i,j)$ has associated with it a weight $w_{ij} > 0$. Define the in-degree of node $i$ as 
\[
	d^-_i = \sum_{j=1}^n w_{ji},
\]
then the Laplacian matrix is defined as
\[
	\mathcal{L} = \left[ \begin{array}{cccc} d^-_1 & -w_{21} & \dots & -w_{n1} \\ -w_{12} & d^-_2 & \dots & -w_{n2} \\ \vdots & \vdots & \ddots & \vdots \\ -w_{1n} & -w_{2n} & \dots & d^-_n \end{array} \right].
\]
Now denote the eigenvalues of $\mathcal{L}$ as 
\[
	\lambda_1 \leq \lambda_2 \dots \leq \lambda_n.
\]
It is known that $\mathcal{L}$ will have only real eigenvalues and is positive semidefinite. The smallest eigenvalue $\lambda_1$ is always 0 and the corresponding eigenvector is $[1, 1, \dots , 1]^T$. The multiplicity of 0 as an eigenvalue is equal to the number of connected components of $G$.\\
Of particular interest is the second smallest eigenvalue $\lambda_2$. As the multiplicity of 0 as an eigenvector is equal to the number of connected components, if $G$ is connected then $\lambda_2 > 0$. Thus $\lambda_2$ is also called the algebraic connectivity of $G$. If $\lambda_2$ is positive, then the magnitude of $\lambda_2$ increases as graph connectedness is increased. This effect can be observed in Figure \ref{fig:lap_eig} in which $\lambda_2$ increases as edges are added to a weakly connected graph and continues to increase until a strongly connected regular graph is achieved at which point $\lambda_2 = \lambda_n$.  

\begin{figure}[h!]
	\captionsetup{font=normal,width=\textwidth}
%	\vspace{-25pt}
	\begin{center}
		\includegraphics[height=0.75\textheight,keepaspectratio]{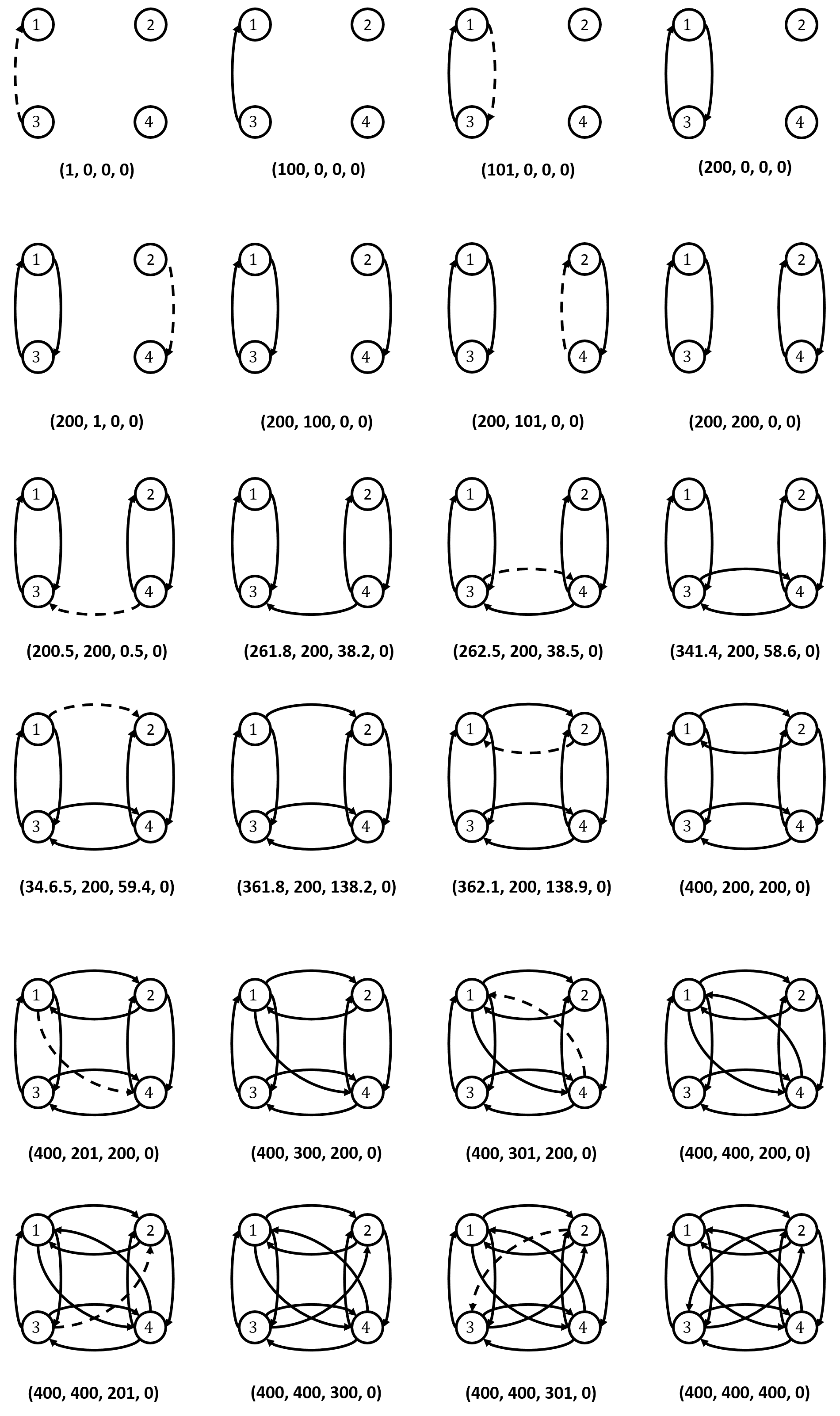}
	\end{center}
%	\vspace{-20pt}
	\caption{Different possible networks of varying connectivity along with the eigenvalues of the respective $\mathcal{L}$ matrix. All dashed links carry weight 1 (weak links), whereas continuously drawn links carry weight 100 (strong links). It can be seen that as network connectivity increases, so does the magnitude of the second smallest eigenvalue.}
%	\vspace{-10pt}
	\label{fig:lap_eig} %laplacian eigenvalues
\end{figure}

\cleardoublepage

\pdfbookmark[-1]{Vita}{ }

\addcontentsline{toc}{chapter}{Vita}
\chapter* {\centering Vita}

\label{chap:vita}

Talha Manzoor was born on May 15, 1988 in Faisalabad, Pakistan. He completed his primary \& middle education from The City School, Capital Campus in Islamabad and matriculation \& intermediate education from Pakistan Air Force (PAF) Public School, Lower Topa, Murree. He graduated in 2010 with a BS in Mechatronics Engineering from National University of Science and Technology (NUST), Islamabad and subsequently joined Lahore University of Management Sciences (LUMS), Lahore for his graduate studies. He completed his MS in Computer Engineering from LUMS in 2013 and has since been pursuing his PhD in the Electrical Engineering Department at LUMS. During his PhD he has spent time as a visiting researcher at the Advanced Systems Analysis program of the International Institute for Applied Systems Analysis (IIASA), Austria and at the Robotics Research Lab of Technical University, Kaiserslautern, Germany. He is generally interested in the application of Systems Analysis techniques to real world problems. His research spans a diverse area, ranging from control and estimation problems in robotics to cybernetic applications in coupled human and natural systems.

%\include{biblio}
%\addcontentsline{toc}{chapter}{Bibliography}
%
%%\renewcommand\cftpartpresnum{~}
%\addcontentsline{toc}{chapter}{Appendices}
%\part*{Appendices}
%\appendix
%\include{appa}
%\include{appb}

%\addtocontents{toc}{\protect\addvspace{2.25em}}
%\bookmarksetup{startatroot}
%\chapter{Chapters/Conclusions}
%
%\cleardoublepage
%
%\pdfbookmark[-1]{Appendices}{appendix}
%\appendix
%\chapter{Appendices/AppendixA}
%\chapter{Appendices/AppendixB}

\end{document}